\documentclass{scrartcl}

\usepackage[ngerman]{babel}
\usepackage[T1]{fontenc}
\usepackage[ansinew]{inputenc}
\usepackage{amsmath, amsthm, amssymb}
\usepackage{dsfont}
\usepackage{graphicx}

\newcommand{\norm}[1]{\left\| #1 \right\|}  
\newcommand{\scprd}[1]{\left\langle #1 \right\rangle}  
\newcommand{\ddt}[1]{\frac{d} {dt} #1}   
\newcommand{\dds}[1]{\frac{d} {ds} #1} 
\newcommand{\ddtau}[1]{\frac{d} {d\tau} #1} 
\newcommand{\ddtk}[1]{\frac{d^k}{dt^k} #1}  
\newcommand{\ddtl}[1]{\frac{d^l}{dt^l} #1}  
\newcommand{\natu}{\mathbb{N}}  
\newcommand{\real}{\mathbb{R}}
\newcommand{\complex}{\mathbb{C}}
\newcommand{\im}{\operatorname{im}}  
\newcommand{\rk}{\operatorname{rk}}
\newcommand{\supp}{\operatorname{supp}}
\newcommand{\spn}{\operatorname{span}}
\newcommand{\dist}{\operatorname{dist}}
\newcommand{\dom}{\operatorname{dom}}
\newcommand{\id}{\operatorname{id}}

\renewcommand{\Re}{\operatorname{Re}}

\newtheorem{thm}{Satz}[section]
\newtheorem{cor}[thm]{Korollar}
\newtheorem{prop}[thm]{Proposition}
\newtheorem{lm}[thm]{Lemma}
\theoremstyle{definition} \newtheorem{ex}[thm]{Beispiel}
\theoremstyle{definition}

\title{
Adiabatensätze mit und ohne Spektrallückenbedingung}

\author{Jochen Schmid\\  
\small Faculty of Mathematics and Physics, University of Stuttgart, Germany\\
\small jochen.schmid@mathematik.uni-stuttgart.de}    

\date{December 29, 2011}

\begin{document}

\maketitle

\begin{abstract}
{\small
In dieser Arbeit 
verallgemeinern wir einige der bisher bekannten Adiabatensätze auf Situationen mit nichtunitären Zeitentwicklungen in Banachräumen. Wir beweisen Adiabatensätze mit gleichmäßiger Spektrallückenbedingung (in Verallgemeinerung eines Satzes von Abou Salem), Adiabatensätze mit nichtgleichmäßiger Spektrallückenbedingung (in Verallgemeinerung eines Satzes von Kato) und qualitative sowie quantitative Adiabatensätze ohne Spektrallückenbedingung 
(in Verallgemeinerung von Sätzen von Avron und Elgart und Teufel). 
Außerdem geben wir eine verallgemeinerte Version eines Adiabatensatzes höherer Ordnung an, der von Nenciu stammt. 
In all diesen Adiabatensätzen müssen die betrachteten Spektralwerte nicht notwendig auf der imaginären Achse liegen und in den Adiabatensätzen mit Spektrallückenbedingung und dem Adiabatensatz höherer Ordnung genügen sogar kompakte Untermengen des Spektrums (insbesondere müssen diese nicht aus Eigenwerten bestehen). 
In zahlreichen Beispielen loten wir die Stärke der hier vorgestellten Adiabatensätze aus. Insbesondere zeigen wir, dass die Sätze der vorliegenden Arbeit allgemeiner sind als die bisher bekannten.

Diese Arbeit wurde im Oktober 2010 fertiggestellt und als Diplomarbeit am Fachbereich Mathematik der Universität Stuttgart eingereicht. Die in der Zwischenzeit (im Juni 2011) erschienenen neueren Resultate von Avron, Fraas, Graf und Grech sind daher in der vorliegenden Arbeit nicht berücksichtigt. Wir weisen aber darauf hin, dass die Sätze der vorliegenden Arbeit allgemeiner sind als die ihnen entsprechenden Sätze von Avron, Fraas, Graf und Grech -- mit Ausnahme des Adiabatensatzes höherer Ordnung, der in keinem logischen Verhältnis zu dem Adiabatensatz höherer Ordnung von Avron, Fraas, Graf und Grech steht. 

Wir sind dabei die wichtigsten Sätze der vorliegenden Arbeit in einem Artikel zusammenzustellen und werden diesen in Kürze auf das arXiv hochladen.
\\

In this work we generalize some of the previously known adiabatic theorems to situations with non-unitary evolutions in Banach spaces. We prove adiabatic theorems with uniform gap condition (generalizing a theorem of Abou Salem), adiabatic theorems with non-uniform gap condition (generalizing a theorem of Kato) and qualitative as well as quantitative adiabatic theorems without gap condition (generalizing theorems of Avron and Elgart, and Teufel).
Additionally, we give a generalized version of an adiabatic theorem of higher order due to Nenciu.
In all these adiabatic theorems the considered spectral values need not lie on the imaginary axis and in the adiabatic theorems with spectral gap condition and the adiabatic theorem of higher order compact subsets of the spectrum are sufficient (in particular, these subsets need not consist of eigenvalues). 
We explore the strength of the presented adiabatic theorems in numerous examples. In particular, we show that the theorems of the present work are  more general than the previously known theorems.

This work was finished in October~2010 and handed in as a diploma thesis at the mathematics department of the University of Stuttgart. The more recent results of Avron, Fraas, Graf und Grech which appeared in the meantime (in June 2011) are therefore not taken into consideration here. We point out, however, that the theorems of the present work are more general than the corresponding theorems of Avron, Fraas, Graf und Grech -- with the exception of the adiabatic theorem of higher order which is in no logical relation to the adiabatic theorem of higher order of Avron, Fraas, Graf and Grech. 

We are about to gather the most important theorems of the present work in an article and will upload it to arXiv soon.}
\end{abstract}

\vspace{1cm}

\noindent {\small \emph{2010 Mathematics Subject Classification:} 34E15 (primary), 34G10,  47D06 (secondary) } \\
{\small \emph{Key words and phrases:} adiabatic theorems in Banach spaces, non-unitary evolutions, spectral gap condition }

\newpage

Ich möchte an dieser Stelle Professor Marcel Griesemer, dem Betreuer dieser Diplomarbeit, sehr herzlich danken: für die sehr interessante und schöne Aufgabenstellung und für viele hilfreiche Diskussionen.

\newpage

\tableofcontents

\newpage

In der ganzen Arbeit bezeichnen wir mit $I$ das Intervall $[0,1]$, mit $X$ einen Banachraum über $\complex$, mit $H$ einen Hilbertraum über $\complex$ und mit $D$ einen dichten Unterraum von $X$ bzw. $H$ (wobei $D$ auch mit $X$ bzw. $H$ übereinstimmen kann).

\section{Worum geht's?} \label{sect: einleitung}
 
\subsection{Adiabatentheorie -- Ausgangslage und -frage}
  
Wir gehen von der folgenden Situation aus, die man als allgemeine \emph{Ausgangssituation} der Adiabatentheorie ansehen kann: 
\begin{itemize}
\item $A(t)$ ist für jedes $t \in I$ eine abgeschlossene lineare Abbildung $D \subset X \to X$, sodass für alle $T \in (0, \infty)$ die Anfangswertprobleme 
\begin{align*}
y' = T A(t) y, \; y(s) = x
\end{align*}
wohlgestellt sind (s. Abschnitt~3.1),
\item $\sigma(t)$ ist für jedes $t \in I$ eine kompakte Untermenge von $\sigma(A(t))$, 
\item $P(t)$ ist für jedes $t \in I$ eine beschränkte Projektion in $X$ mit 
\begin{align*}
A(t) \bigl( P(t)X \cap D \bigr) \subset P(t)X \text{ \; und \; } \sigma \bigl( A(t) \big|_{P(t)X \cap D} \bigr) = \sigma(t).
\end{align*}
\end{itemize}
Zunächst interessieren wir uns eigentlich für die Anfangwertprobleme
\begin{align*}
y' = A\Bigl( \frac{t'}{T} \Bigr) y, \; y(0) = x 
\end{align*}
auf $[0,T]$ (für $T \in (0, \infty)$) mit den für große $T$ langsam zeitveränderlichen linearen Abbildungen $A\bigl( \frac{t'}{T} \bigr)$. Diese Anfangswertprobleme gehen durch Skalierung aber offensichtlich über in die eingangs angeschriebenen Anfangswertprobleme
\begin{align*}
y' = T A(t) y, \; y(0) = x
\end{align*}
auf dem festen Intervall $I$.
\\

Die \emph{Ausgangsfrage} der Adiabatentheorie ist nun: wann -- unter welchen Voraussetzungen an $A(t)$, $\sigma(t)$ und $P(t)$ -- überführt die Zeitentwicklung $U_T$ zu $T A$ (s. Abschnitt~3.1) Vektoren aus dem Unterraum $P(0)X$ in Vektoren, die für große $T$ beinahe in dem Unterraum $P(t)X$ liegen? Anders (und präziser) gefragt: unter welchen Voraussetzungen gilt
\begin{align*}
(1-P(t)) U_T(t) P(0) \longrightarrow 0 \quad (T \to \infty)
\end{align*}
für alle $t \in I$?

Sätze, die solche Voraussetzungen angeben, heißen -- durchaus treffend -- \emph{Adiabatensätze.} Warum treffend? Weil das griechische Wort "`adiabatos"' so viel bedeutet wie "`keine Übergänge zulassend"' und Adiabatensätze besagen eben (grob gesprochen) gerade, dass aus dem Unterraum $P(0)X$ für große $T$ unter der Wirkung von $U_T(t)$ kaum etwas übergeht in den Unterraum $(1-P(t))X$. Wir können das auch so ausdrücken, dass die Zeitentwicklung $U_T$ zu $T A$ für große $T$ in einem gewissen Sinn (s. Abschnitt~3.3) beinahe adiabatisch ist bzgl. $P$. 

Wir unterscheiden zwischen Adiabatensätzen \emph{mit (gleichmäßiger oder nichtgleichmäßiger) Spektrallückenbedingung} und Adiabatensätzen \emph{ohne Spektrallückenbedingung}: bei den einen ist $\sigma(t)$ für jedes $t \in I$ isoliert in $\sigma(A(t))$ (gleichmäßig oder nichtgleichmäßig), bei den andern ist $\sigma(t)$ nicht für jedes $t \in I$ isoliert in $\sigma(A(t))$ (s. Abschnitt~2.1.1).
\\

In den folgenden Zeichnungen sind drei typische Situationen veranschaulicht, und zwar für den Sonderfall, dass $\sigma(A(t)) \subset i \, \real$ und $\sigma(t) = \{ \lambda(t) \}$ für alle $t \in I$. In den ersten beiden liegt eine gleichmäßige bzw. nichtgleichmäßige Spektrallücke vor, in der dritten liegt keine Spektrallücke vor.

\includegraphics{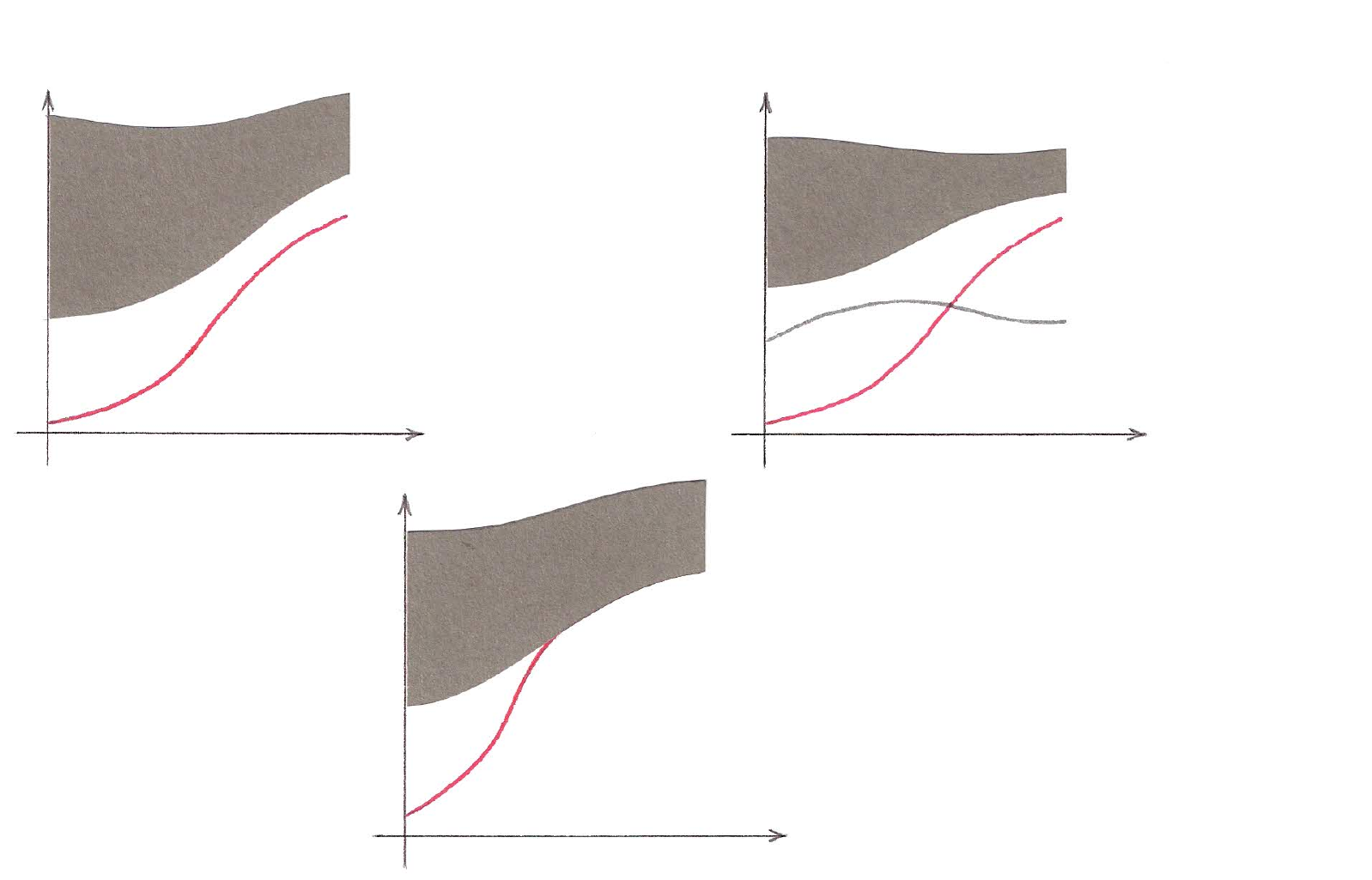}

Die horizontale Achse ist dabei jeweils die reelle Achse ($t$-Achse), die vertikale Achse die imaginäre Achse ($\lambda$-Achse), die Mengen $\sigma(t)$ sind rot und die Mengen $\sigma(A(t)) \setminus \sigma(t)$ grau gefärbt.

\subsection{Was ist bekannt?}

Die Wurzeln der Adiabatentheorie sind natürlich physikalisch, genauer: quantenmechanisch, und daher sind die ersten 
Adiabatensätze auch für schiefselbstadjungierte $A(t)$ formuliert (Schrödingeroperatoren -- multipliziert mit einem Faktor $i$ -- sind ja bekanntlich schiefselbstadjungiert). 
In der Arbeit~\cite{BornFock} von Born, Fock (1928) und in der Arbeit~\cite{Kato 50} von Kato (1950) werden Adiabatensätze mit Spektrallückenbedingung bewiesen (gleichmäßig und nichtgleichmäßig), und zwar eben für schiefselbstadjungierte $A(t)$. Außerdem sind die $\sigma(t)$ in diesen Arbeiten als endlich vorausgesetzt. Diese Voraussetzung schwächt Nenciu in seinem Artikel~\cite{Nenciu 80} aus dem Jahr 1980 ab: für einen Adiabatensatz mit gleichmäßiger Spektrallückenbedingung (und weiterhin für schiefselbstadjungierte $A(t)$) genügt es, wenn $\sigma(t)$ bloß kompakt ist für jedes $t \in I$.

1992 veröffentlichen Nenciu und Rasche (in dem Artikel~\cite{NenciuRasche 92}) einen Adiabatensatz für nichtschiefselbstadjungierte $A(t)$, der jedoch nur in sehr speziellen Situationen (mit gleichmäßiger Spektrallückenbedingung) anwendbar ist. In der Arbeit~\cite{Nenciu 93} von 1993 beweist Nenciu einen Adiabatensatz, der alle bis dahin bekannten Adiabatensätze beliebiger Ordnung enthält (Abschnitt~1 in~\cite{Nenciu 93}). In diesem Satz sind die $A(t)$ als schiefselbstadjungiert vorausgesetzt.

Im Jahr 1998 beweisen Avron und Elgart (in der Arbeit~\cite{AvronElgart 99}) und unabhängig davon Bornemann (in der Arbeit~\cite{Bornemann 98})  Adiabatensätze ohne Spektrallückenbedingung, beidesmal für schiefselbstadjungierte $A(t)$. Die Zugänge der beiden Arbeiten unterscheiden sich deutlich, Bornemanns Zugang (über quadratische Formen) erlaubt es auch Situationen zu behandeln, in denen die Definitionsbereiche $D(A(t))$ von $t$ abhängen. Wir werden das aber in dieser Arbeit nicht weiterverfolgen.

In der Arbeit~\cite{Abou 07} von 2007 schließlich beweist Abou Salem einen ersten allgemeinen Adiabatensatz für nichtschiefselbstadjungierte $A(t)$ (und zwar unabhängig von dem speziellen Satz Nencius und Rasches aus~\cite{NenciuRasche 92}). Die $A(t)$ in diesem Satz sind nur als Erzeuger von Kontraktionshalbgruppen vorausgesetzt.
\\

Wir weisen ausdrücklich darauf hin, dass wir hier nur diejenigen Adiabatensätze herausgegriffen haben, die für diese Arbeit besonders wichtig sind. Einen umfassenderen Abriss der Adiabatentheorie, der auch die Anwendungen stärker berücksichtigt, bietet beispielsweise Abschnitt~3 in~\cite{AvronElgart 99}.

\subsection{Was machen wir in dieser Arbeit? Was ist neu?}

Wir werden in dieser Arbeit einige der eben genannten bisher bekannten Adiabatensätze verallgemeinern und damit die Adiabatentheorie ein klein wenig weiterentwickeln. Alle Adiabatensätze (und auch alle Beispiele) dieser Arbeit ordnen sich in die eingangs beschriebene allgemeine Ausgangssituation der Adiabatentheorie ein.
\\

Im vorbereitenden Abschnitt~\ref{sect: vorbereitungen} stellen wir grundlegende Definitionen und Sätze zusammen, auf die wir in den darauffolgenden Abschnitten immer wieder zurückgreifen werden. In Abschnitt~2.2.2 skizzieren wir einen wenig geläufigen aber sehr direkten Zugang zum Spektralsatz für normale lineare Abbildungen, der auf Bernaus Artikel~\cite{Bernau 65} beruht. 
Der Abschnitt~2 kann zunächst überflogen werden.

Abschnitt~\ref{sect: zeitentwicklungen} handelt von Zeitentwicklungen: zunächst führen wir den Zeitentwicklungsbegriff ein, der uns ja in der oben formulierten Ausgangsfrage der Adiabatentheorie schon begegnet ist. Inbesondere werden wir sehen, dass in der eingangs geschilderten Ausgangssituation der Adiabatentheorie die Zeitentwicklung $U_T$ zu $T A$ wirklich existiert für alle $T \in (0, \infty)$ (Satz~\ref{thm: zeitentwicklung natürlich}). Vor allem aber werden wir in Abschnitt~3 klären, unter welchen Voraussetzungen an die $A(t)$ eine (dann eindeutige) Zeitentwicklung $U_T$ zu $T A$ existiert (Satz~\ref{thm: Dyson} und Satz~\ref{thm: Kato}). Auf Proposition~\ref{prop: zshg. regvor. kato} sei hier besonders hingewiesen: sie erlaubt es die Voraussetzungen des Satzes von Kato in der Version von Yosida (Theorem~XIV.4.1 in~\cite{Yosida: FA}) deutlich einfacher auszusprechen, was bisher nicht bekannt gewesen zu sein scheint. In Abschnitt~3.3 führen wir den Adiabatizitätsbegriff für Zeitentwicklungen ein und eine für die Adiabatensätze in den Abschnitten~\ref{sect: triviale adsätze} bis~\ref{sect: adsätze ohne sl} sehr wichtige Vergleichszeitentwicklung, die adiabatische Zeitentwicklung zu $T A$ und $P$.
\\

In Abschnitt~\ref{sect: triviale adsätze} stellen wir zwei triviale Adiabatensätze vor, bevor wir dann im Abschnitt~\ref{sect: adsätze mit sl} unsere ersten nichttrivialen Adiabatensätze beweisen. In Abschnitt~5.1 zeigen wir zwei Adiabatensätze mit gleichmäßiger Spektrallückenbedingung (Satz~\ref{thm: unhandl adsatz mit sl} und Satz~\ref{thm: handl adsatz mit sl}), die den oben angesprochenen Adiabatensatz von Abou Salem verallgemeinern: die dort getroffene Voraussetzung, dass $\sigma(t) = \{ \lambda(t) \}$ für Eigenwerte $\lambda(t)$ mit algebraischer Vielfachheit $1$, ist überflüssig -- es genügt, wenn $\sigma(t)$ kompakt ist für alle $t \in I$. 
Wir heben hier außerdem Proposition~\ref{prop: zshg isoliert und glm isoliert} hervor, die einen allgemeinen Zusammenhang zwischen Isoliertheit und gleichmäßiger Isoliertheit von kompakten in $\sigma(A(t))$ isolierten Untermengen $\sigma(t)$ von $\sigma(A(t))$ herstellt. 
In Abschnitt~5.2 beweisen wir einen Adiabatensatz mit nichtgleichmäßiger Spektrallückenbedingung (Satz~\ref{thm: unhandl adsatz mit nichtglm sl}), der den oben erwähnten Adiabatensatz Katos verallgemeinert: die Voraussetzung, dass die $A(t)$ schiefselbstadjungiert und die $\sigma(t)$ endlich sind, ist überflüssig -- wir brauchen nur, dass $A(t)$ für jedes $t \in I$ eine stark stetige Halbgruppe erzeugt und $A$ $(M, 0)$-stabil ist (s. Abschnitt~3.2) und $\sigma(t)$ kompakt. 

Der Abschnitt~\ref{sect: adsätze ohne sl} handelt von Adiabatensätzen ohne Spektrallückenbedingung. In Abschnitt~6.1 beweisen wir zwei solche Sätze (Satz~\ref{thm: allg adsatz ohne sl} und Satz~\ref{thm: adsatz ohne sl für normale A(t)}), die den oben genannten Adiabatensatz von Avron und Elgart und einen daran anknüpfenden Adiabatensatz von Teufel (Artikel~\cite{Teufel 01} von 2001) verallgemeinern: wir können auf die Voraussetzung der Schiefselbstadjungiertheit verzichten, stattdessen genügt es, wenn die $A(t)$ wie in den Adiabatensätzen des Abschnitts~\ref{sect: adsätze mit sl} Halbgruppenerzeuger sind mit derselben Stabilitätseigenschaft wie dort. 
Wir haben uns bemüht Satz~\ref{thm: allg adsatz ohne sl} noch weiter zu verallgemeinern. Satz~\ref{thm: erweiterter adsatz ohne sl} ist so eine Verallgemeinerung, die jedoch etwas unbefriedigend ist.
In Abschnitt~6.2 übertragen wir zwei quantitative Adiabatensätze ohne Spektrallückenbedingung von Avron und Elgart und von Teufel auf Situationen mit nichtschiefselbstadjungierten $A(t)$. Diese quantitativen Adiabatensätze sagen im Unterschied zu den Sätzen aus Abschnitt~6.1 nicht nur, \emph{dass} $U_T$ für große $T$ beinahe adiabatisch ist bzgl. $P$, sondern auch \emph{wie} adiabatisch bzw. diabatisch $U_T$ für große $T$ ist, genauer: sie beinhalten eine Aussage über die Konvergenzrate von $(1-P(t))U_T(t)P(0)$.

In Abschnitt~\ref{sect: höhere adsätze} stellen wir eine verallgemeinerte Version (Satz~\ref{thm: höherer adsatz}) des erwähnten Adiabatensatzes aus~\cite{Nenciu 93} von Nenciu vor: auch hier ist die Schiefselbstadjungiertheit nicht vonnöten. Aus Satz~\ref{thm: höherer adsatz} erhalten wir unmittelbar einen Adiabatensatz höherer Ordnung (Korollar~\ref{cor: höherer adsatz für suppP' in (0,1)}) und nebenbei fällt aus diesem Satz -- jedoch (nach Beispiel~\ref{ex: unser höherer adsatz echt allgemeiner als der von nenciu}) nicht schon aus seinem Vorbild aus Nencius Arbeit~\cite{Nenciu 93} -- noch einmal unser Adiabatensatz mit gleichmäßiger Spektrallückenbedingung (Satz~\ref{thm: unhandl adsatz mit sl}) ab.
\\

In all diesen Adiabatensätzen schwächen wir auch die Regularitätsvoraussetzungen an $A(t)$, $\sigma(t)$ und $P(t)$ ab im Vergleich zu den bekannten Adiabatensätzen, aus denen unsere Sätze entstanden sind. So genügt es beispielsweise -- wie wir dies in den Adiabatensätzen der Abschnitte~\ref{sect: adsätze mit sl} und~\ref{sect: adsätze ohne sl} tun -- zu fordern, dass $t \mapsto A(t)x$ stetig differenzierbar ist für alle $x \in D$. Dies ist eine deutliche Vereinfachung der entsprechenden Voraussetzungen der bisher bekannten Adiabatensätze.
\\

Wir ergänzen darüberhinaus unsere Adiabatensätze durch Beispiele, die vor Augen führen sollen, was die Sätze können und was nicht. So zeigen wir durch Beispiele einerseits, dass mit den vorgestellten Adiabatensätzen Situationen behandelt werden können, in denen die bisher bekannten Adiabatensätze nicht anwendbar sind, mit anderen Worten: dass die vorgestellten Adiabatensätze \emph{echt} allgemeiner sind als ihre Vorbilder und nicht nur allgemeiner aussehen. 
Andererseits versuchen wir durch Beispiele auszuloten, was in Situationen geschieht, in denen unsere Adiabatensätze nicht anwendbar sind: ob dann auch die Aussage des Adiabatensatzes verletzt ist oder ob diese trotzdem noch gilt. Wir zeigen insbesondere, dass etwa die Voraussetzungen an $A$, die wir in unseren Adiabatensätzen machen, nicht mehr wesentlich abgeschwächt werden können (Beispiel~\ref{ex: (M,0)-stabilität wesentlich in den adsätzen mit sl} und Beispiel~\ref{ex: (M,0)-stabilität wesentlich in adsatz ohne sl}). 

Viele unserer Beispiele weisen dieselbe einfache Struktur auf, die in dem mit "`Standardbeispiele"' überschriebenen Abschnitt~4.2 verankert ist.
\\

Wir beschließen die Arbeit mit einem Anwendungsbeispiel aus der Neutronentransporttheorie. Die $A(t)$ dort sind nicht schiefselbstadjungiert und auch nicht normal, was belegt, dass Adiabatensätze für nichtschiefselbstadjungierte $A(t)$ nicht nur abstrakte Spielerei sind.


\section{Vorbereitungen: grundlegende Definitionen und Sätze} \label{sect: vorbereitungen}

\subsection{Analysis}

\subsubsection{Stetigkeit mengenwertiger Abbildungen und Isoliertheit}

Sei $X_0$ ein metrischer Raum und $E \subset X_0$. Wir schreiben dann
\begin{align*}
U_r(E) := \{ x \in X_0: \dist(x,E) < r \} \text{ \; und \; } \overline{U}_r(E) := \{ x \in X_0: \dist(x,E) \le r \}
\end{align*}
für alle $r \in (0, \infty)$ und dehnen damit die allgemein übliche entsprechende Schreibweise für einpunktige Mengen $E = \{a\}$ aus: $U_r(\{a\}) = U_r(a)$ und $\overline{U}_r(\{a\}) = \overline{U}_r(a)$.
\\

Wie man sich leicht überlegt, ist $U_r(E)$ offen und $\overline{U}_r(E)$ abgeschlossen und 
\begin{align*}
U_{r_1}\bigl( U_{r_2}(E) \bigr) \subset U_{r_1+r_2}(E) \text{ \; und \; }  \overline{U}_{r_1}\bigl( \overline{U}_{r_2}(E) \bigr) \subset \overline{U}_{r_1+r_2}(E)
\end{align*}
für alle Untermengen $E$ von $X_0$ und alle positiven Zahlen $r$, $r_1$ und $r_2$.
Wenn $X_0$ sogar ein Innenproduktraum ist, dann gelten auch die umgekehrten Inklusionen und darüberhinaus $\overline{U}_r(E) = \overline{ U_r(E) }$.
\\

Sei $X_0$ wieder ein metrischer Raum, $J$ ein Intervall in $\real$ und seien $E \subset F \subset X_0$ und $E(t) \subset F(t) \subset X_0$ für alle $t \in J$. Dann heißt $E$ \emph{isoliert in $F$} genau dann, wenn eine positive Zahl $r$ existiert, sodass
\begin{align*}
U_r(E) \cap F = E.
\end{align*}
Die $E(t)$ heißen \emph{gleichmäßig (bzgl. $t \in J$) isoliert in $F(t)$} genau dann, wenn eine von $t \in J$ unabhängige positive Zahl $r$ existiert, sodass
\begin{align*}
U_r(E(t)) \cap F(t) = E(t)
\end{align*}
für alle $t \in J$.
Sei $E(t)$ für jedes $t \in J$ isoliert in $F(t)$. Wir sagen dann, \emph{$E(t)$ falle an der Stelle $t_0 \in J$ in $F(t) \setminus E(t)$ hinein} genau dann, wenn eine Folge $(t_n)$ in $J$ existiert mit $t_n \longrightarrow t_0 \;\;(n \to \infty)$ und $\dist(E(t_n), \, F(t_n) \setminus E(t_n) ) \longrightarrow 0 \;\;(n \to \infty)$.
\\

Wie man leicht bestätigt, gilt für kompakte Intervalle $J$: $E(t)$ ist gleichmäßig (bzgl. $t \in J$) isoliert in $F(t)$ genau dann, wenn $E(t)$ an keiner Stelle $t_0 \in J$ in $F(t) \setminus E(t)$ hineinfällt.  
\\

Sei $J$ ein Intervall, $X_0$ ein metrischer Raum und $E(t) \subset X_0$ für jedes $t \in J$. Dann heißt die (mengenwertige) Abbildung $t \mapsto E(t)$ \emph{oberhalbstetig} bzw. \emph{unterhalbstetig in $t_0$} genau dann, wenn zu jedem $\varepsilon > 0$ eine in $J$ offene Umgebung $U_{t_0}$ von $t_0$ existiert, sodass
\begin{align*}
E(t) \subset U_{\varepsilon}(E(t_0)) \text{ \; bzw. \; } E(t_0) \subset U_{\varepsilon}(E(t))
\end{align*}
für alle $t \in U_{t_0}$.
Sie heißt \emph{oberhalbstetig} bzw. \emph{unterhalbstetig} genau dann, wenn sie unter- bzw. oberhalbstetig in jedem $t_0 \in J$ ist.

Weiter heißt die Abbildung $t \mapsto E(t)$ \emph{stetig (in $t_0$)} genau dann, wenn sie ober- und unterhalbstetig ist (in $t_0$).
\\


Wie man leicht sieht, ist die mengenwertige Abbildung $t \mapsto E(t) := \{a_1(t), \dots, a_m(t) \}$ stetig im eben definierten Sinn, wenn alle Abbildungen $t \mapsto a_i(t) \in X_0$ stetig sind. Wenn $E(t)$ einpunktig ist gilt auch die Umkehrung: aus der Stetigkeit von $t \mapsto E(t) = \{a(t)\}$ folgt dann schon die Stetigkeit von $t \mapsto a(t)$ im gewöhnlichen Sinn.


Wir merken noch an, dass die Abbildung $t \mapsto E(t)$ für \emph{kompakte} Untermengen $E(t)$ von $X_0$ genau dann stetig ist, wenn sie stetig bzgl. der Hausdorffmetrik ist, wobei diese gegeben ist durch
\begin{align*}
d(E,F) := \max \big\{    \sup_{a \in E} (\dist(a,F)), \, \sup_{b \in F} (\dist(b,E))  \big\} 
\end{align*}
für kompakte Untermengen $E$ und $F$ von $X_0$.

\subsubsection{Einseitig differenzierbare Abbildungen} 

Sei $J$ ein nichttriviales Intervall in $\real$ mit $a:= \inf J$, $b:= \sup J \in \real \cup \{-\infty, \infty \}$ und sei $f$ eine Abbildung $J \to X$. Dann heißt $f$ bekanntlich \emph{rechts- bzw. linksseitig differenzierbar in $t \in J$} genau dann, wenn der limes
\begin{align*}
\partial_+ f(t) := \lim_{h \searrow 0} \frac{f(t+h) - f(t)}{h} \text{ \; bzw. \; } \partial_- f(t) := \lim_{h \searrow 0} \frac{f(t-h) - f(t)}{-h}
\end{align*}
existiert. Weiter nennen wir $f$ \emph{rechts- bzw. linksseitig (stetig) differenzierbar} genau dann, wenn $\partial_+ f(t)$ bzw. $\partial_- f(t)$ für alle $t \in J \setminus \{b\}$ bzw. alle $t \in J \setminus \{a\}$ existiert (und $J \setminus \{b\} \ni t \mapsto \partial_+ f(t)$ bzw. $J \setminus \{a\} \ni t \mapsto \partial_- f(t)$ stetig ist).
\\

Satz~\ref{thm: einseitig db und beidseitig db}, der in den Ausführungen nach Definition~II.3.2 in~\cite{Krein 71} angedeutet ist, stellt einen Zusammenhang her zwischen einseitiger Differenzierbarkeit und Differenzierbarkeit. Wir verallgemeinern zunächst den Mittelwertsatz auf einseitig differenzierbare $X$-wertige Abbildungen. Die Idee ist dieselbe wie die zu Theorem~IV.2.18 in~\cite{AmannEscher}. 

\begin{lm} \label{lm: mws für einseitig db}
Sei $J$ ein offenes Intervall in $\real$ und $f$ eine stetige und rechtsseitig oder linksseitig differenzierbare Abbildung $J \to X$. Dann gilt
\begin{align*}
\norm{ f(b) - f(a) } \le \sup_{t \in [a,b]} \norm{ \partial_{\pm} f(t) } (b-a)
\end{align*}
für alle $a,b \in J$ mit $a < b$.
\end{lm} 

\begin{proof}
Sei $f$ \emph{rechtsseitig} differenzierbar und seien $a,b \in J$ mit $a < b$. Sei $\varepsilon > 0$, 
\begin{align*}
J_{\varepsilon} := \Bigl\{     t \in [a,b]: \norm{ f(t') - f(a) } \le \bigl( \; \sup_{x \in [a,b]}  \norm{ \partial_+ f(x) }  + \varepsilon \; \bigr) (t'-a) \text{ für alle } t' \in [a,t]      \Bigr\}  
\end{align*}
und $t_{\varepsilon} := \sup J_{\varepsilon}$. Dann ist $a \in J_{\varepsilon}$, insbesondere ist $t_{\varepsilon} \in [a,b]$, und 
\begin{align*}
\norm{ f(t) - f(a) } \le \bigl( \; \sup_{x \in [a,b]}  \norm{ \partial_+ f(x) }  + \varepsilon \; \bigr) (t-a)
\end{align*}
für alle $t \in [a, t_{\varepsilon})$ und auch für $t = t_{\varepsilon}$, da $f$ ja (linksseitig) stetig ist in $t_{\varepsilon}$. 

Wir zeigen nun, dass $t_{\varepsilon} = b$. Angenommen, $t_{\varepsilon} < b$. Dann existiert ein $\delta > 0$, sodass $t_{\varepsilon} + \delta \le b$ und 
\begin{align*}
\norm{     \frac{ f(t) - f(t_{\varepsilon}) }{ t - t_{\varepsilon} }    - \partial_+ f(t_{\varepsilon})     }  < \varepsilon
\end{align*}
für alle $t \in (t_{\varepsilon}, t_{\varepsilon} + \delta]$, woraus wir schließen können, dass
\begin{align*}
\norm{ f(t) - f(a) } &\le \norm{ f(t) - f(t_{\varepsilon}) }  +  \norm{ f(t_{\varepsilon}) - f(a) } \\
& \le \bigl( \; \sup_{x \in [a,b]}  \norm{ \partial_+ f(x) } + \varepsilon \; \bigr) \bigl( (t-t_{\varepsilon}) + (t_{\varepsilon}-a) \bigr)
\end{align*}
für alle $t \in (t_{\varepsilon}, t_{\varepsilon} + \delta]$. Also gilt
\begin{align*}
\norm{ f(t) - f(a) } \le   \bigl( \; \sup_{x \in [a,b]}  \norm{ \partial_+ f(x) }  + \varepsilon \; \bigr) (t-a)
\end{align*}
für alle $t \in [a, t_{\varepsilon}]    \cup   (t_{\varepsilon}, t_{\varepsilon} + \delta]$, das heißt, $t_{\varepsilon} + \delta \in J_{\varepsilon}$ im Widerspruch zur Definition von $t_{\varepsilon}$. Also ist tatsächlich $t_{\varepsilon} = b$ und die Behauptung folgt, indem wir $\varepsilon$ gegen $0$ schicken.
\\

Sei nun $f$ \emph{linkseitig} differenzierbar. Dann folgt die Behauptung, indem wir die eben bewiesene Aussage auf die stetige und rechtsseitig differenzierbare Abbildung $\tilde{f}$ andwenden, wobei $\tilde{f}(t) := f(-t)$ für alle $t \in  \tilde{J} := -J$.
\end{proof}

Die Voraussetzung des obigen Lemmas, dass $f$ stetig ist, können wir nicht weglassen. Sei nämlich $f := \chi_{[0, \infty)}$ und $J := \real$, dann ist $f$ zwar rechtsseitig (sogar rechtsseitig stetig) differenzierbar mit $\partial_+ f(t) = 0$ für alle $t \in J$, aber $f$ ist nicht konstant.

\begin{thm} \label{thm: einseitig db und beidseitig db}
Sei $J$ ein nichttriviales (offenes, halboffenes oder abgeschlossenes) Intervall in $\real$, $f$ eine stetige und einseitig (rechts- oder linksseitig) stetig differenzierbare Abbildung $J \to X$ und $\partial_+ f$ bzw. $\partial_- f$ sei stetig fortsetzbar in den rechten bzw. linken Randpunkt von $J$, falls dieser zu $J$ gehört. Dann ist $f$ schon (beidseitig) stetig differenzierbar auf $J$ und $f' = \partial_{\pm} f$.
\end{thm}

\begin{proof}
Sei $f$ \emph{rechtsseitg} stetig differenzierbar. Wir nehmen $J$ zunächst als offen an und zeigen, dass 
\begin{align*}
f(t) = f(t_0) + \int_{t_0}^t \partial_+ f(\tau) \,d\tau
\end{align*}
für alle $t \in J$. Wie man sofort sieht, sind beide Seiten: $f$ und $t \mapsto f(t_0) + \int_{t_0}^t \partial_+ f(\tau) \,d\tau$ rechtsseitig stetig differenzierbar und ihre rechtsseitigen Ableitungen stimmen überein, woraus mit Lemma~\ref{lm: mws für einseitig db} (angewandt auf die Differenz der beiden Seiten) folgt, dass tatsächlich
\begin{align*}
f(t) = f(t_0) + \int_{t_0}^t \partial_+ f(\tau) \,d\tau
\end{align*}
für alle $t \in J$. Da nun $t \mapsto f(t_0) + \int_{t_0}^t \partial_+ f(\tau) \,d\tau$ sogar beidseitig stetig differenzierbar ist, gilt dies auch für $f$ und $f' = \partial_+ f$.

Wir nehmen nun $J$ nicht mehr als offen an. Zunächst ist dann (nach dem, was wir gerade eben gezeigt haben) nur $f \big|_{ J^{\circ} } = f \big|_{(a,b)}$ stetig differenzierbar. Wir müssen also noch zeigen, dass $f$ auch in den Randpunkten (falls sie zu $J$ gehören) differenzierbar ist und die Ableitung $f'$ dort stetig ist. Wenn der linke Randpunkt $a$ zu $J$ gehört, dann ist $f$ (rechtsseitig) differenzierbar in $a$ und 
\begin{align*}
[a,b) \ni t \mapsto f'(t) = \partial_+ f (t)
\end{align*}
ist stetig (in $a$) -- wir \emph{nehmen} ja gerade \emph{an}, dass $f$ rechtsseitig stetig differenzierbar ist.
Wenn der rechte Randpunkt $b$ zu $J$ gehört, dann ist $f$ (linksseitig) differenzierbar in $b$ und $(a,b] \ni t \mapsto f'(t)$ ist stetig (in $b$). Sei nämlich $g$ die stetige Fortsetzung von $\partial_+ f$ in den Randpunkt $b$ hinein. Dann gilt 
\begin{align*}
\frac{f(b-h) - f(b)}{-h} - g(b) &= \lim_{\varepsilon \searrow 0} \frac{f(b-\varepsilon) - f(b-\varepsilon -h)}{h} - g(b) \\
&= \lim_{\varepsilon \searrow 0} \frac{1}{h} \, \int_{b-\varepsilon -h}^{b-\varepsilon} g(\tau) \,d\tau    - g(b) \\
&= \frac{1}{h} \, \int_{b-h}^{b} g(\tau) \,d\tau    - g(b)   \longrightarrow 0 \quad(h \searrow 0),
\end{align*}
weil $f\big|_{(a,b)}$ stetig differenzierbar ist mit $f' \big|_{(a,b)} = \partial_+ f \big|_{(a,b)} = g \big|_{(a,b)}$ und $g$ stetig ist auf ganz $J$.
Also ist $f$ tatsächlich (linksseitig) differenzierbar in $b$ und 
\begin{align*}
(a,b] \ni t \mapsto f'(t) = g(t)
\end{align*} 
ist stetig (in $b$).
\\

Sei $f$ nun \emph{linksseitig} stetig differenzierbar. Wir setzen dann $\tilde{f}(t) := f(-t)$ für alle $t \in \tilde{J} := -J$ und sehen, dass $\tilde{f}$ eine stetige und rechtsseitig stetig differenzierbare Abbildung $\tilde{J} \to X$ ist und $\partial_+ \tilde{f}$ stetig fortsetzbar ist in den rechten Randpunkt von $\tilde{J}$, falls dieser zu $\tilde{J}$ gehört. Also ist $\tilde{f}$ nach dem eben Bewiesenen (beidseitig) stetig differenzierbar und damit auch $f$, und die Ableitung stimmt mit der linksseitigen Ableitung überein.
\end{proof}



\subsubsection{Komplexe Analysis}

Sei $U$ offen in $\complex$ und $f$ eine Abbildung $U \to X$. Dann heißt $f$ \emph{holomorph} genau dann, wenn $f$ in jedem $z \in U$ (komplex) differenzierbar ist. 
\\

Wie man mithilfe des Satzes von Banach, Steinhaus leicht zeigen kann (Theorem~VI.4 in~\cite{RS 1}), 
ist $f$ holomorph genau dann, wenn $\varphi \circ f$ für alle $\varphi \in X'$ eine holomorphe Abbildung $U \to \complex$ ist. Dieser Zusammenhang erlaubt es, Sätze über $X$-wertige holomorphe Abbildungen zurückzuführen auf die entsprechenden Sätze für $X = \complex$. Wir vermerken hier den cauchyschen Satz (in der homologischen Version). Zuvor jedoch noch einige Vereinbarungen zur Sprechweise. 
\\


Wie üblich verstehen wir unter einem \emph{(geschlossenen) Integrationsweg in $U$} eine stückweise stetig differenzierbare Abbildung $\gamma: [a,b] \to U$ (mit $\gamma(b) = \gamma(a)$) und unter einem \emph{Zykel in $U$} eine formale Summe $\gamma_1 + \dotsb + \gamma_m$ (in dem präzisen Sinn von Definition~1.4 in~\cite{FischerLieb}) von geschlossenen Integrationswegen $\gamma_1, \dots, \gamma_m$ in $U$. Die ganze Zahl
\begin{align*}
n(\gamma,z) := \frac{1}{2 \pi i} \, \int_{\gamma} \frac{1}{w-z} \, dw
\end{align*} 
heißt \emph{Umlaufzahl} des geschlossenen Integrationsweges $\gamma$ um $z \notin \im \gamma$ und entsprechend heißt
\begin{align*}
n(\gamma, z) := n(\gamma_1, z)+ \dots + n(\gamma_m, z)
\end{align*}
Umlaufzahl des Zykels $\gamma = \gamma_1 + \dotsb + \gamma_m$ um $z \notin \im \gamma := \cup_{i=1}^m \im \gamma_i$. Wir schreiben für Mengen $K \subset \complex \setminus \im \gamma$ oft $n(\gamma, K) = 0$ oder $1$ und meinen damit natürlich, dass $n(\gamma, z) = 0$ bzw. $1$ für alle $z \in K$. 

Zwei Zykel $\gamma_1$ und $\gamma_2$ in $U$ heißen \emph{homolog in $U$} genau dann, wenn $n(\gamma_1, z) = n(\gamma_2, z)$ für alle $z \in \complex \setminus U$. 

Schließlich verstehen wir unter einem \emph{einfach geschlossenen Integrationsweg} einen geschlossenen Integrationsweg $\gamma: [a,b] \to \complex$, sodass $\gamma|_{[a,b)}$ injektiv ist, und unter einem \emph{positiv einfach geschlossenen Zykel} einen Zykel $\gamma =  \gamma_1 + \dotsb + \gamma_m$ bestehend aus einfach geschlossenen Integrationswegen $\gamma_i$, sodass $n(\gamma, z) \in \{0,1\}$ für alle $z \in \complex \setminus \im \gamma$.

\begin{thm} \label{thm: Cauchy global}
Sei $U$ offen in $\complex$, $f$ eine holomorphe Abbildung $U \to X$ und $\gamma$ ein Zykel in $U$ mit $n(\gamma, \, \complex \setminus U) = 0$. Dann gilt
\begin{align*}
\int_{\gamma} f(z) \,dz = 0
\end{align*}
und insbesondere gilt 
\begin{align*}
n(\gamma, a) f(a) = \frac{1}{2 \pi i} \int_{\gamma} \frac{f(z)}{z-a} \, dz
\end{align*} 
für alle $a \in \complex \setminus \im \gamma$. 
\end{thm}

Aus diesem Satz folgt insbesondere, dass eine Abbildung $f: U \to X$ genau dann holomorph ist, wenn sie analytisch ist (das heißt, lokal um jeden Punkt $a$ von $U$ als eine Potenzreihe mit Entwicklungspunkt $a$ geschrieben werden kann).

Die folgende Proposition werden wir häufig benützen. Sie ist (anschaulich) sehr einleuchtend, allerdings ist sie nicht ganz so einfach zu beweisen.

\begin{prop}  \label{prop: Cauchy für kompakta}
Sei $U$ offen in $\complex$ und $K$ eine kompakte Untermenge von $U$. Dann existiert ein positiv einfach geschlossener Zykel $\gamma$ in $U \setminus K$ mit $n(\gamma, K) = 1$ und $n(\gamma, \complex \setminus U) = 0$.
\end{prop}

\begin{proof}
Das folgt aus Proposition~13.1.8 in~\cite{Conway: complex ana 2}. 
\end{proof}

Eine schwächere Version der obigen Proposition, die nur von Zykeln statt sogar von positiv einfach geschlossenen Zykeln spricht, ist in Satz~IV.3.3 in~\cite{FischerLieb} enthalten. Wir würden auch mit dieser schwächeren Version auskommen, allerdings würde der Beweis von Proposition~\ref{prop: zshg isoliert und glm isoliert}, in dem wir obige Proposition in ihrer vollen Allgemeinheit benutzen, dann um einiges umständlicher.
\\

Wir vermerken noch eine sehr einfache Aussage zur Vertauschbarkeit von Ableitungen und Wegintegralen, die wir später brauchen werden.

\begin{lm} \label{lm: vertauschung von abl und wegintegral}
Sei $J_0$ ein nichttriviales Intervall, $\gamma_0$ ein Zykel, $m \in \natu \cup \{0\}$ und $f$ eine Abbildung $J_0 \times \im \gamma_0 \to X$, sodass $J_0 \ni t \mapsto f(t,z)$  $m$-mal stetig differenzierbar ist für jedes $z \in \im \gamma_0$ und $\im \gamma_0 \ni z \mapsto \partial_1^k f(t,z)$  stetig ist für jedes $t \in J_0$ und 
\begin{align*}
\sup_{(t,z) \in J_0 \times \, \im \gamma_0} \norm{ \partial_1^k f(t,z) } < \infty
\end{align*}  
für alle $k \in \{0,1, \dots, m \}$.
Dann ist $t \mapsto \int_{\gamma_0} f(t,z) \,dz$ $m$-mal stetig differenzierbar und
\begin{align*}
\ddtk{ \int_{\gamma_0} f(t,z) \,dz  } = \int_{\gamma_0} \partial_1^k f(t,z) \, dz
\end{align*}
für alle $k \in \{0,1, \dots, m \}$.
\end{lm}

\begin{proof}
Das folgt mithilfe des lebesgueschen Satzes.
\end{proof}

\subsubsection{Regularitätslemmas für operatorwertige Abbildungen}

Die folgenden Lemmas werden wir im weiteren Verlauf sehr häufig benutzen. Wir beginnen mit einem Lemma, das einen einfachen Zusammenhang zwischen Regularität bzgl. der Normoperatortopologie und Regularität bzgl. der starken Operatortopologie herstellt. Die entsprechende Aussage für holomorphe operatorwertige Abbildungen ist wohlbekannt.

\begin{lm} \label{lm: strong db and db}
Sei $J$ ein nichttriviales Intervall in $\real$ und $m \in \natu \cup \{0, \infty \}$. Sei $A(t)$ für jedes $t \in J$ eine beschränkte lineare Abbildung in $X$, $t \mapsto A(t)$ $(m+1)$-mal stetig differenzierbar bzgl. der starken Operatortopologie von $X$. Dann ist $t \mapsto A(t)$ $m$-mal stetig differenzierbar  bzgl. der Normoperatorotopologie von $X$.
\end{lm}

\begin{proof}
Das folgt sehr leicht mit Induktion aus dem Satz von Banach, Steinhaus (Prinzip der gleichmäßigen Beschränktheit).
\end{proof}

Auch die Aussage des folgenden Lemmas ist sehr leicht einzusehen.

\begin{lm} \label{lm: strong db of products}
Sei $J$ ein nichttriviales Intervall in $\real$ und $t_0 \in J$. Sei $A(t)$ für jedes $t \in J$ eine beschränkte lineare Abbildung in $X$, $t \mapsto A(t)x$ stetig bzw. differenzierbar in $t_0$ für alle $x \in D$, $t \mapsto x(t) \in D$ ebenfalls stetig bzw. differenzierbar in $t_0$ und sei $\sup_{t \in J} \norm{A(t)} < \infty$. Dann ist $t \mapsto A(t)x(t)$ stetig bzw. differenzierbar in $t_0$.
\end{lm}

\begin{proof}
Seien $t \mapsto A(t)x$ und $t \mapsto x(t)$ zunächst stetig in $t_0$ für alle $x \in D$. Dann gilt wegen $\sup_{t \in J} \norm{A(t)} < \infty$ und $x(t_0) \in D$, dass
\begin{align*}
A(t_0&+h)  x(t_0+h) - A(t_0)x(t_0) \\
&= A(t_0+h) \bigl( x(t_0+h)-x(t_0) \bigr) + \bigl( A(t_0+h)x(t_0) - A(t_0)x(t_0) \bigr) 
\longrightarrow 0 \quad (h \to 0),
\end{align*}
die Abbildung $t \mapsto A(t)x(t)$ ist also stetig in $t_0$.

Seien $t \mapsto A(t)x$ und $t \mapsto x(t)$ nun differenzierbar in $t_0$ für alle $x \in D$. Dann ist $t \mapsto A(t)x$ insbesondere stetig in $t_0$ für alle $x \in D$ und wegen $\sup_{t \in J} \norm{A(t)} < \infty$ und der Dichtheit von $D$ in $X$ folgt daraus, dass $t \mapsto A(t)x$ sogar für alle $x \in X$ stetig ist. Aufgrund dieser starken Stetigkeit von $t \mapsto A(t)$ in $t_0$ erhalten wir nun
\begin{align*} 
& \frac{ A(t_0+h)x(t_0+h) - A(t_0)x(t_0) }{h} \\
& \qquad \quad = A(t_0+h) \, \frac{x(t_0+h) - x(t_0)}{h} + \frac{ A(t_0+h)x(t_0) - A(t_0)x(t_0) }{h} \\
& \qquad \quad \longrightarrow A(t_0)x'(t_0) + A'(t_0)x(t_0) \quad (h \to 0),
\end{align*}
wie gewünscht.
\end{proof}

Die Vektoren $x(t) \in D$ werden im folgenden meist gegeben sein als $B(t)x$ mit einem festen Vektor $x \in D$ und beschränkten linearen Abbildungen $B(t)$ in $X$, die $D$ für alle $t \in J$ in sich überführen. 

\begin{lm} \label{lm: continuity of inv}
Die Menge $U$ der bijektiven beschränkten linearen Abbildungen ist offen in der Menge aller beschränkten linearen Abbildungen in $X$ und die Abbildung 
\begin{align*}
U \ni A \mapsto A^{-1}
\end{align*}
ist beliebig oft differenzierbar, insbesondere stetig.
\end{lm}

\begin{proof}
Siehe etwa Satz~VII.7.2 in~\cite{AmannEscher}.
\end{proof}

Das folgende Lemma wird an einigen Stellen entscheidend sein. $A'(t)$ bezeichnet dabei natürlich die lineare Abbildung 
\begin{align*}
D \ni x \mapsto \lim_{h \to 0} \frac{A(t+h)x - A(t)x}{h}.
\end{align*}
Wir weisen darauf hin dass eine Variante dieses Lemmas schon in~\cite{Krein 71} (Lemma~II.1.5) zu finden ist. Wir haben Lemma~\ref{lm: reg of inv} aber unabhängig davon bewiesen.

\begin{lm} \label{lm: reg of inv}
Sei $J$ ein nichttriviales Intervall in $\real$ und $m \in \natu \cup \{0, \infty \}$. Sei $A(t): D \subset X \to X$ für jedes $t \in J$ eine bijektive abgeschlossene lineare Abbildung, sei $t \mapsto A(t)x$ $m$-mal (stetig) differenzierbar für alle $x \in D$ und sei $\sup_{t \in J} \norm{ A(t)^{-1} } < \infty$. Dann ist $t \mapsto A(t)^{-1}x$ $m$-mal (stetig) differenzierbar für alle $x \in X$ und es gilt (wenn $m \ne 0$)
\begin{align*}
\ddt{ A(t)^{-1}x } = -A(t)^{-1} A'(t) A(t)^{-1}x 
\end{align*}
für alle $t \in J$.
\end{lm}

\begin{proof}
Sei zunächst $m = 0$. Dann ist $t \mapsto A(t)y$ also stetig für alle $y \in D$ und daher gilt wegen $\sup_{t \in J} \norm{ A(t)^{-1} } < \infty$, dass
\begin{align*}
A(t+h)^{-1}x - A(t)^{-1}x = A(t+h)^{-1} \, \bigl( A(t+h) - A(t) \bigr) \, A(t)^{-1}x 
\longrightarrow 0 \quad (h \to 0)
\end{align*}
für alle $x \in X$ und alle $t \in J$. Also ist $t \mapsto A(t)^{-1}$ wie behauptet stark stetig. 

Wir zeigen die Aussage des Lemmas nun für $m \in \natu$, genauer zeigen wir zusätzlich, dass die $m$-te starke Ableitung von $t \mapsto A(t)^{-1}$ aus Summanden besteht, die sich zusammensetzen aus den beschränkten linearen Abbildungen $A(t)^{-1}$, $A'(t)A(t)^{-1}$, \dots, $A^{(m)}(t)A(t)^{-1}$ (und skalaren Vorfaktoren). Diese verschärfte Aussage beweisen wir mit Induktion über $m \in \natu$ (für $m = 0$ haben wir sie eben schon eingesehen).

Zuvor jedoch vergewissern wir uns noch kurz, dass die eben angesprochenen linearen Abbildungen $A(t)^{-1}$, $A'(t)A(t)^{-1}$, \dots, $A^{(m)}(t)A(t)^{-1}$ wirklich alle beschränkt sind (wir diese also insbesondere bedenkenlos addieren und multiplizieren können).
Wir haben für alle $k \in \{ 1, \dots, m\}$, dass 
\begin{align*}
A^{(k)}(t)A(t)^{-1} = A^{(k)}(t)A(0)^{-1} \, A(0) A(t)^{-1}
\end{align*} 
für alle $t \in J$. $A(0)A(t)^{-1}$ ist beschränkt (als eine auf ganz $X$ definierte Verkettung einer abgeschlossenen und einer beschränkten linearen Abbildung), $A'(t) A(0)^{-1}$ ist als starker limes der beschränkten linearen Abbildungen 
\begin{align*}
\frac{A(t+h)A(0)^{-1} - A(t) A(0)^{-1}}{h}
\end{align*}
ebenfalls beschränkt (Satz von Banach, Steinhaus) und induktiv (und ganz entsprechend) folgt, dass auch $A^{(k)}(t) A(0)^{-1}$ beschränkt ist für alle $k \in \{ 1, \dots, m\}$.
\\
 
Sei $m = 1$. Dann gilt nach der für den Fall $m=0$ bereits bewiesenen Behauptung, dass
\begin{align*}
\frac{ A(t+h)^{-1}x - A(t)^{-1}x }{h} &= - A(t+h)^{-1} \, \frac{A(t+h) - A(t)}{h} \, A(t)^{-1}x \\
&\longrightarrow - A(t)^{-1} A'(t) A(t)^{-1}x \quad (h \to 0)
\end{align*}
für alle $x \in X$ und alle $t \in J$. Die Abbildung $t \mapsto A(t)^{-1}$ ist also $1$-mal stark differnzierbar und die $1$-te starke Ableitung hat die behauptete Struktur. 

Wenn $t \mapsto A(t)x$ sogar $1$-mal \emph{stetig} differenzierbar ist für alle $x \in D$, dann ist die Ableitung von $t \mapsto A(t)^{-1}$ stark stetig. Zum einen ist nämlich $t \mapsto A(t)^{-1}$ stark stetig (Behauptung im Fall $m = 0$) und zum andern ist auch $t \mapsto A'(t)A(t)^{-1} = A'(t)A(0)^{-1} \, A(0)A(t)^{-1}$ stark stetig. Warum? Zunächst ist natürlich $t \mapsto A'(t)A(0)^{-1}$ stark stetig, aber -- und das ist jetzt entscheidend -- auch
\begin{align*}
t \mapsto A(0)A(t)^{-1} = \bigl( A(t) A(0)^{-1} \bigr)^{-1}
\end{align*} 
ist stark stetig, weil $t \mapsto A(t) A(0)^{-1}$ nach Lemma~\ref{lm: strong db and db} stetig sogar bzgl. der Normoperatortopologie ist und damit nach Lemma~\ref{lm: continuity of inv} auch die Inverse $t \mapsto \bigl( A(t) A(0)^{-1} \bigr)^{-1}$ stetig bzgl. der Normoperatortopologie, mithin insbesondere stark stetig ist. 
\\

Sei $m \in \natu$ und die zu zeigende Aussage stimme für $m-1$. Dann ist $t \mapsto A(t)^{-1}$ $(m-1)$-mal stark (stetig) differenzierbar und die $(m-1)$-te starke Ableitung von $t \mapsto A(t)^{-1}$ besteht aus Summanden, die sich zusammensetzen aus den beschränkten linearen Abbildungen $A(t)^{-1}$, $A'(t)A(t)^{-1}$, \dots, $A^{(m-1)}(t)A(t)^{-1}$ (und skalaren Vorfaktoren). 
Wir sehen mithilfe von Lemma~\ref{lm: strong db of products}, dass $t \mapsto A(t)^{-1}$, $A'(t)A(t)^{-1}, \dots, A^{(m-1)}(t)A(t)^{-1}$ alle noch einmal stark differenzierbar sind und die jeweiligen Ableitungen aus Summanden bestehen, die sich zusammensetzen aus den beschränkten linearen Abbildungen $A(t)^{-1}$, $A'(t)A(t)^{-1}$, \dots, $A^{(m-1)}(t)A(t)^{-1}$, $A^{(m)}(t)A(t)^{-1}$ (und skalaren Vorfaktoren). 
Also ist $t \mapsto A(t)^{-1}$ $m$-mal stark differenzierbar und die $m$-te starke Ableitung von $t \mapsto A(t)^{-1}$ besteht aus Summanden, die sich zusammensetzen aus den beschränkten linearen Abbildungen $A(t)^{-1}$, $A'(t)A(t)^{-1}$, \dots, $A^{(m)}(t)A(t)^{-1}$ (und skalaren Vorfaktoren). 

Wenn $t \mapsto A(t)x$ sogar $m$-mal \emph{stetig} differenzierbar ist, dann sind $t \mapsto A(t)^{-1}$, $A'(t)A(t)^{-1}$, \dots, $A^{(m-1)}(t)A(t)^{-1}$, $A^{(m)}(t)A(t)^{-1}$ stark stetig, denn dann ist $t \mapsto A^{(k)}(t) A(0)^{-1}$ stark stetig für alle $k \in \{1, \dots, m\}$ und auch 
\begin{align*}
t \mapsto A(0) A(t)^{-1} = \bigl( A(t) A(0)^{-1} \bigr)^{-1}
\end{align*}
ist wegen der Stetigkeit von $t \mapsto A(t) A(0)^{-1}$ bzgl. der Normoperatortopologie (Lemma~\ref{lm: strong db and db}) und wegen Lemma~\ref{lm: continuity of inv} stetig, insbesondere stark stetig.
Also ist $t \mapsto A(t)^{-1}$ sogar $m$-mal stark \emph{stetig} differenzierbar.
\end{proof}

Wir merken an, dass die Voraussetzung
\begin{align*}
\sup_{t \in J} \norm{ A(t)^{-1} } < \infty
\end{align*}
des obigen Lemmas ganz von selbst erfüllt ist, wenn $J$ kompakt ist und $t \mapsto A(t)x$ stetig differenzierbar ist für alle $x \in D$. Denn dann ist $t \mapsto A(t)A(0)^{-1}$ stetig differenzierbar bzgl. der starken Operatortopologie. Aufgrund von Lemma~\ref{lm: strong db and db} ist diese Abbildung dann insbesondere stetig bzgl. der Normoperatortopologie und damit ist nach Lemma~\ref{lm: continuity of inv} auch die Inverse $t \mapsto \bigl( A(t)A(0)^{-1} \bigr)^{-1} = A(0)A(t)^{-1}$ und mithin auch $t \mapsto A(t)^{-1}$ stetig bzgl. der Normoperatortopologie.

\subsubsection{Projektionen}

Sei $P$ eine beschränkte lineare Abbildung in $X$. Dann heißt $P$ bekanntlich \emph{beschränkte Projektion} in $X$ genau dann, wenn $P^2 = P$. $P$ heißt \emph{orthogonale Projektion} in $H$ genau dann, wenn $P$ eine beschränkte Projektion in $H$ ist mit $P^* = P$.
\\

Wie man sofort sieht, geben beschränkte Projektionen $P$ eine Zerlegung von $X$ in die abgeschlossenen Unterräume $PX$ und $(1-P)X$,
\begin{align*}
X = PX + (1-P)X \text{ \; und \; } PX \cap (1-P)X = 0,
\end{align*}
und für das Spektrum gilt: $\sigma(P) \subset \{ 0, 1 \}$, denn
\begin{align*}
z-P = (z-1)P + z(1-P)
\end{align*}
für alle $z \in \complex$.
Außerdem ist eine beschränkte Projektion $P$ in $H$ genau dann orthogonal, wenn die Unterräume $PH$ und $(1-P)H$ orthogonal zueinander sind.
\\ 
 
Auf das folgende Lemma, das wir für Banachräume formulieren (statt wie in~\cite{mmqm1} (Lemma~10.2) für Hilberträume), werden wir uns oft berufen.

\begin{lm} \label{lm: rk konst}
Sei $P(t)$ für jedes $t \in I$ eine beschränkte Projektion in $X$ und $t \mapsto P(t)$ stetig im Punkt $t_0 \in I$. Dann ist $t \mapsto \rk P(t) \in \natu \cup \{0, \infty \}$ lokal um $t_0$ konstant. 
\end{lm}

\begin{proof}
Zunächst halten wir fest: für alle $t \in I$ bildet $1 + \bigl( P(t) - P(t_0) \bigr)$ den Unterraum $P(t_0)X$ in $P(t)X$ ab und $1 + \bigl( P(t_0) - P(t) \bigr)$ den Unterraum $P(t)X$ in $P(t_0)X$, weil
\begin{align*}
\bigl( 1 + ( P(t) - P(t_0) ) \bigr) P(t_0) = P(t) P(t_0) \text{ \; und \; } \bigl( 1 + ( P(t_0) - P(t) ) \bigr) P(t) = P(t_0) P(t).
\end{align*}
Da nun $t \mapsto P(t)$ stetig ist in $t_0$, existiert eine in $I$ offene Umgebung $U_{t_0}$, sodass $\norm{ P(t) - P(t_0) } < 1$ für alle $t \in U_{t_0}$. Also ist $1 + \bigl( P(t) - P(t_0) \bigr)$ und $1 + \bigl( P(t_0) - P(t) \bigr)$ invertierbar für alle $t \in U_{t_0}$, das heißt, für $t \in U_{t_0}$ bildet $1 + \bigl( P(t) - P(t_0) \bigr)$ den Unterraum $P(t_0)X$ \emph{injektiv} in $P(t)X$ ab und $1 + \bigl( P(t_0) - P(t) \bigr)$ bildet $P(t)X$ \emph{injektiv} in $P(t_0)X$ ab. Also gilt
\begin{align*}
\dim P(t_0)X \le \dim P(t)X     \text{ \; und \; }   \dim P(t)X \le \dim P(t_0)X 
\end{align*}
für alle $t \in U_{t_0}$, was zu zeigen war.
\end{proof}

Im obigen Lemma genügt es übrigens nicht, nur vorauszusetzen, dass $t \mapsto P(t)$ \emph{stark} stetig ist in $t_0$. Sei nämlich $X := L^p(\real, \complex)$ für ein $p \in [1,\infty)$ und
\begin{align*}
P(t)g := \chi_{[-t,t]} g  
\end{align*}
für alle $g \in X$ und alle $t \in I$.
Dann ist $P(t)$ für jedes $t \in I$ eine beschränkte Projektion in $X$ und $t \mapsto P(t)$ ist stark stetig, aber $\rk P(0) = 0 \ne \infty = \rk P(t)$ für alle $t \in (0,1]$.


\begin{lm} \label{lm: multiplikationsop mit char fkt nur dann stark db nach t wenn schon konst}
Sei $(X_0, \mathcal{A}, \mu)$ ein Maßraum und $X := L^p(X_0, \complex)$ für ein $p \in [1, \infty)$. Sei $E_t \in \mathcal{A}$ für alle $t \in I$ und sei $t \mapsto P(t)g := \chi_{E_t} \, g$ differenzierbar für alle $g \in X$. Dann ist $t \mapsto P(t)$ konstant.
\end{lm}

\begin{proof}
Sei $g \in X$. Wir müssen zeigen, dass $P'(t)g = 0$ für alle $t \in I$.
Sei $t \in I$, dann gilt
\begin{align*}
\frac{P(t+h)g - P(t)g}{h} \longrightarrow P'(t)g \quad (h \to \infty)
\end{align*}
und damit existiert nach einem bekannten maßtheoretischen Satz (s. etwa Satz~VI.4.3 und Korollar~VI.4.13 in~\cite{Elstrodt}) eine Folge $(h_n)$ in $\real \setminus \{0\}$, sodass $h_n \longrightarrow 0 \;\;(n \to \infty)$ und 
\begin{align*}
\Bigl(  \frac{P(t+h_n)g - P(t)g}{h_n}  \Bigr)(x) \longrightarrow \bigl( P'(t)g \bigr)(x) \quad (n \to \infty)
\end{align*}
für fast alle $x \in X_0$. 

Sei $N_t$ die Menge genau der $x \in X_0$, für die das nicht gilt, 
und sei $x \in X_0 \setminus N_t$. 
Wenn $g(x) = 0$, dann gilt
\begin{align*}
\bigl( P'(t)g \bigr)(x) = \lim_{n \to \infty} \frac{ \chi_{E_{t+h_n}}(x) - \chi_{E_t}(x) }{h_n} \, g(x) = 0.
\end{align*}
Und wenn $g(x) \ne 0$, dann gilt ebenfalls
\begin{align*}
\bigl( P'(t)g \bigr)(x) = \lim_{n \to \infty} \frac{ \chi_{E_{t+h_n}}(x) - \chi_{E_t}(x) }{h_n} \, g(x) = 0,
\end{align*}
denn wegen
\begin{align*}
\Big\{ \frac{-1}{h_n}, \, 0, \, \frac{1}{h_n} \Big\} \ni \frac{ \chi_{E_{t+h_n}}(x) - \chi_{E_t}(x) }{h_n} \longrightarrow \frac{1}{g(x)} \, \bigl( P'(t)g \bigr)(x) \quad (n \to \infty)
\end{align*}
muss $\chi_{E_{t+h_n}}(x) - \chi_{E_t}(x)$ ab einem gewissen $n \in \natu$ gleich $0$ sein.

Also gilt insgesamt $(P'(t)g)(x) = 0$ für alle $x \in X_0 \setminus N_t$, das heißt, $P'(t)g = 0$ (in $L^p(X_0)$), wie gewünscht.
\end{proof}

Wenn wir nur wissen, dass $P(t) = \chi_{E_t}$ für \emph{fast} alle $t \in I$, dafür aber sogar wissen, dass $t \mapsto P(t)g$ \emph{stetig} differenzierbar ist für alle $g \in X$, dann bleibt die Aussage des obigen Lemmas bestehen: $t \mapsto P(t)$ ist dann auch hier konstant. 

Sei nämlich $I'$ die Menge genau der $t \in I$, für die wir die Darstellung $P(t) = \chi_{E_t}$ haben, und sei $t \in I'$. Dann existiert eine Folge $(h_n)$ mit denselben Eigenschaften wie im obigen Beweis und der zusätzlichen Eigenschaft, dass $t + h_n \in I'$ für alle $n \in \natu$, schließlich liegen ja fast alle $t \in I$ auch schon in $I'$. 
Und es folgt genau wie oben, dass $P'(t)g = 0$. Weil nun $s \mapsto P'(s)g$ nach Voraussetzung stetig ist und weil $t \in I'$ beliebig war, folgt $P'(s)g = 0$ für alle $s \in I$.

\subsection{Spektraltheorie}

\subsubsection{Spektraltheorie abgeschlossener linearer Abbildungen in Banachräumen}

Sei $A$ eine lineare Abbildung $D \subset X \to X$. Dann heißt
\begin{align*}
\rho(A) :=& \Bigl\{    z \in \complex : z-A \text{ ist eine bijektive lineare Abbildung } D \subset X \to X \\ 
&\qquad \qquad \text{ und } (z-A)^{-1} \text{ ist beschränkt}    \Bigr\} 
\end{align*}
\emph{Resolventenmenge von A} und $\sigma(A) := \complex \setminus \rho(A)$ heißt \emph{Spektrum von $A$}. 
\\

Wir erinnern daran, dass $\rho(A)$ offen und $\sigma(A)$ demnach abgeschlossen ist und die Abbildung 
\begin{align*}
\rho(A) \ni z \mapsto (z-A)^{-1}
\end{align*}
(die \emph{Resolventenabbildung}) holomorph ist. Weiter erinnern wir daran, dass 
\begin{align*}
(z-A)^{-1}(w-A)^{-1} = \frac{1}{w-z} \bigl( (z-A)^{-1} - (w-A)^{-1} \bigr)
\end{align*}
für alle $z,w \in \rho(A)$ mit $z \ne w$.
Jede abgeschlossene Untermenge $\sigma$ von $\complex$ ist Spektrum einer abgeschlossenen linearen Abbildung -- auch in den Sonderfällen $\sigma = \emptyset$ und $\sigma = \complex$ (Beispiel~IV.1.5 in~\cite{EngelNagel}). 

\begin{ex} \label{ex: spektrum von multop}
Sei $(X_0, \mathcal{A}, \mu)$ ein Maßraum, $p \in [1, \infty)$ und $X := L^p(X_0, \complex)$. Sei ferner $f$ eine messbare Abbildung $X_0 \to \complex$. Dann heißt die lineare Abbildung $M_f$, gegeben durch
\begin{align*}
D(M_f) := \{ g \in X: f g \in X \} \text{ \; und \; } M_f \, g := f g,
\end{align*}
\emph{Multiplikationsoperator mit $f$ auf $X$.} $M_f$ ist dicht definiert und abgeschlossen und $\sigma(M_f) \subset \operatorname{ess-im} f$, wobei
\begin{align*}
\operatorname{ess-im} f := \big\{ z \in \complex: \mu\bigl( f^{-1}(U_{\varepsilon}(z)) \bigr) \ne 0 \big \} 
\end{align*}
den wesentlichen Wertebereich von $f$ bezeichnet. Zumindest wenn $(X_0, \mathcal{A}, \mu)$ $\sigma$-endlich ist, gilt sogar $\sigma(M_f) = \operatorname{ess-im} f$.   $\blacktriangleleft$
\end{ex}

Wenn $A: X \to X$ beschränkt ist und $X \ne 0$, dann ist das Spektrum von $A$ nichtleer und beschränkt, genauer gilt
\begin{align*}
r_A = \lim_{n \to \infty} \norm{ A^n }^{\frac{1}{n}}
\end{align*}
für den \emph{Spektralradius} $r_A := \sup \{ |z| : z \in \sigma(A) \}$ von $A$. Insbesondere gilt also $\sigma(A) \subset U_{ \norm{A} }(0)$ für beschränkte $A$. 
\\

Sei $A$ eine lineare Abbildung $D \subset X \to X$ und $P$ eine beschränkte Projektion in $X$. Wie man leicht einsieht, ist dann $P A \subset A P$ gleichbedeutend damit, dass $P D$, $(1-P) D \subset D$ und $A \bigl( PD \bigr) \subset PX$ sowie $A  \bigl( (1-P)D \bigr) \subset (1-P)X$ (damit also, dass die Unterräume $PX$ und $(1-P)X$ in einem gewissen Sinne invariant sind unter $A$). Die folgende Proposition besagt, dass eine Zerlegung von $X$ in 
$A$-invariante Unterräume eine (allerdings nicht notwendig disjunkte) Zerlegung des Spektrums von $A$ nach sich zieht.

\begin{prop} \label{prop: zerl. des spektrums}
Sei $A: D \subset X \to X$ eine abgeschlossene lineare Abbildung und $P$ eine beschränkte Projektion in $X$ mit $P A \subset A P$. Dann sind auch $A \big|_{PD}$ und $A \big|_{(1-P)D}$ abgeschlossen in $PX$ bzw. $(1-P)X$ und 
\begin{align*}
\sigma(A) = \sigma(A \big|_{PD}) \cup \sigma(A \big|_{(1-P)D}).
\end{align*}
\end{prop}

\begin{proof}
Das folgt elementar.
\end{proof}

Der sehr wichtige Satz~\ref{thm: Rieszprojektion} (der sich von Theorem~III.6.17 in~\cite{Kato: Perturbation 80} nur dadurch unterscheidet, dass er auf die dort getroffene nach Proposition~\ref{prop: Cauchy für kompakta} aber überflüssige Voraussetzung, dass gewisse einfach geschlossene rektifizierbare Wege existieren, verzichtet) sagt aus, dass umgekehrt eine Zerlegung des Spektrums von $A$ zu einer Zerlegung von $X$ in $A$-invariante Unterräume führt, genauer: dass diese Zerlegung von $X$ durch eine Rieszprojektion gegeben ist. Was wir darunter genau verstehen wollen, halten wir nun fest.
\\

Sei $A$ eine abgeschlossene lineare Abbildung $D \subset X \to X$ und $\sigma$ eine kompakte in $\sigma(A)$ isolierte Untermenge von $\sigma(A)$, das heißt, das Spektrum von $A$ zerlegt sich in die beiden abgeschlossenen Mengen $\sigma$ und $\sigma(A) \setminus \sigma$ und $\sigma$ ist zudem beschränkt. Dann heißt $P$ \emph{Rieszprojektion von $A$ auf $\sigma$} genau dann, wenn 
\begin{align*}  
P = \frac{1}{2 \pi i} \, \int_{\gamma} (z-A)^{-1} \, dz
\end{align*}
für einen (und damit nach Satz~\ref{thm: Cauchy global} für alle) Zykel $\gamma$ in $\rho(A)$ mit $n(\gamma, \sigma) = 1$ und $n(\gamma, \sigma(A) \setminus \sigma) = 0$.
\\

Zu jeder kompakten und in $\sigma(A)$ isolierten Untermenge $\sigma$ von $\sigma(A)$ \emph{existiert genau eine} Rieszprojektion von $A$ auf $\sigma$, denn wegen der Isoliertheit von $\sigma$ \emph{existiert} (nach Proposition~\ref{prop: Cauchy für kompakta}) ein Zykel $\gamma$ in $\rho(A)$ mit $n(\gamma, \sigma) = 1$ und $n(\gamma, \sigma(A) \setminus \sigma) = 0$, und alle solche Zykel sind homolog in $\rho(A)$, das heißt, die Wegintegrale stimmen (nach Satz~\ref{thm: Cauchy global}) für alle solche Zykel überein. Im Sonderfall $\sigma = \emptyset$ stimmt das noch nicht ganz -- wir müssen da zusätzlich \emph{voraussetzen}, dass $\rho(A) \ne \emptyset$ (denn dies folgt da noch nicht aus der (leeren) Isoliertheitsbedingung an $\sigma = \emptyset$). Unter dieser zusätzlichen Voraussetzung existiert, wie man (mithilfe von Satz~\ref{thm: Cauchy global}) leicht einsieht, auch zu $\sigma = \emptyset$ genau eine Rieszprojektion, und zwar ist diese gleich $0$.

\begin{thm} \label{thm: Rieszprojektion}
Sei $A: D \subset X \to X$ eine abgeschlossene lineare Abbildung mit $\rho(A) \ne \emptyset$, $\sigma$ eine kompakte in $\sigma(A)$ isolierte Untermenge von $\sigma(A)$ und $P$ die Rieszprojektion von $A$ auf $\sigma$. Dann ist $P$ eine beschränkte Projektion in $X$ mit $P A \subset A P$, $A P$ ist beschränkt und 
\begin{align*}
\sigma(A \big|_{PD}) = \sigma \text{ \; und \; } \sigma(A \big|_{(1-P)D}) = \sigma(A) \setminus \sigma.
\end{align*}
\end{thm}

\begin{proof}
Sei zunächst $\sigma = \emptyset$. Dann gilt $P = 0$, woraus sofort $P A \subset A P$ folgt und 
\begin{align*}
\sigma(A \big|_{PD}) = \sigma_{PX}(0) = \emptyset = \sigma \text{ \; sowie \; } \sigma(A \big|_{(1-P)D}) = \sigma(A) = \sigma(A) \setminus \sigma.
\end{align*}

Sei nun $\sigma \ne \emptyset$. Sei $r_0$ eine positive Zahl mit $U_{r_0}(\sigma) \setminus \sigma    \subset \rho(A)$, die wegen der Isoliertheit von $\sigma$ auch wirklich existiert. Sei $\gamma_1$ ein Zykel in $U_{\frac{r_0}{2}}(\sigma) \setminus \sigma$ mit $n(\gamma_1, \sigma) = 1$ und $n(\gamma_1, \complex \setminus U_{\frac{r_0}{2}}(\sigma)) = 0$ und sei $\gamma_2$ ein Zykel in $U_{r_0}(\sigma) \setminus \overline{U}_{\frac{r_0}{2}}(\sigma)$ mit $n(\gamma_2, \overline{U}_{\frac{r_0}{2}}(\sigma)) = 1$ und $n(\gamma_2, \complex \setminus U_{r_0}(\sigma)) = 0$. Solche Zykel existieren wegen $\sigma \ne \emptyset$ (beachte: für $\sigma = \emptyset$ wäre $U_{r_0}(\sigma) = \emptyset$) nach Proposition~\ref{prop: Cauchy für kompakta} 
Dann liegen $\gamma_1$ und $\gamma_2$ in $\rho(A)$ und beide umlaufen $\sigma$ einmal und $\sigma(A) \setminus \sigma$ keinmal, das heißt 
\begin{align*}
P^2 &= \frac{1}{2 \pi i} \, \int_{\gamma_1} (z-A)^{-1} \, dz \; \; \frac{1}{2 \pi i} \, \int_{\gamma_2} (w-A)^{-1} \, dw \\
&= \frac{1}{2 \pi i} \, \int_{\gamma_1} \; \frac{1}{2 \pi i} \, \int_{\gamma_2}  (z-A)^{-1}(w-A)^{-1} \, dw \,dz \\
&= \frac{1}{2 \pi i} \, \int_{\gamma_1} \; \Bigl( \frac{1}{2 \pi i} \, \int_{\gamma_2}  \frac{1}{w-z} \,dw \Bigr) \; (z-A)^{-1} \,dz \\
& \qquad \qquad + \frac{1}{2 \pi i} \, \int_{\gamma_2} \; \Bigl( \frac{1}{2 \pi i} \, \int_{\gamma_1}  \frac{1}{z-w} \,dz \Bigr) \; (w-A)^{-1} \,dw \\
&= \frac{1}{2 \pi i} \, \int_{\gamma_1} \; n(\gamma_2, z) \, (z-A)^{-1} \,dz
+ \frac{1}{2 \pi i} \, \int_{\gamma_2} \;  n(\gamma_1, w) \, (w-A)^{-1} \,dw 
= P,
\end{align*}
$P$ ist also eine beschränkte Projektion in $X$.

Sei $x \in X$. Dann ist $Px \in D$ und
\begin{align*}
A Px = \frac{1}{2 \pi i} \, \int_{\gamma_i} A (z-A)^{-1} \;x \, dz = \frac{1}{2 \pi i} \, \int_{\gamma_i} z (z-A)^{-1} \;x \, dz, 
\end{align*}
denn $\im \gamma_i \ni z \mapsto A (z-A)^{-1} = -1 + z (z-A)^{-1}$ ist stetig, das heißt das zugehörige Wegintegral existiert, und $A$ ist abgeschlossen, weshalb wir $A$ wirklich ins Wegintegral hineinziehen dürfen. Weiter gilt 
\begin{align*}
P Ax = \frac{1}{2 \pi i} \, \int_{\gamma_i} (z-A)^{-1} \; Ax \, dz = \frac{1}{2 \pi i} \, \int_{\gamma_i} A (z-A)^{-1} \;x \, dz = A Px 
\end{align*}
für alle $x \in D$. Also haben wir $P A \subset A P$ und mehr noch: $PX \subset D$ (insbesondere also $PD = PX \cap D = PX$) und $A P$ ist beschränkt in $X$.

Wir müssen nun noch zeigen, dass $\sigma(A \big|_{PD}) = \sigma$ und $\sigma(A \big|_{(1-P)D}) = \sigma(A) \setminus \sigma$.
Sei dazu $\lambda \notin \sigma$ und $\gamma$ ein Zykel in $\rho(A)$ mit $n(\gamma, \sigma) = 1$, $n(\gamma, \sigma(A) \setminus \sigma) = 0$ und $n(\gamma, \lambda) = 0$ (Proposition~\ref{prop: Cauchy für kompakta}!). Dann gilt
\begin{align*}
(\lambda - A) & \Bigl(     \frac{1}{2 \pi i} \, \int_{\gamma} \frac{1}{\lambda - z} \, (z-A)^{-1}  \, dz         \Bigr) \, x \\
&= \frac{1}{2 \pi i} \, \int_{\gamma}  (z-A)^{-1} \, dz \; x +   \frac{1}{2 \pi i} \, \int_{\gamma} \frac{1}{\lambda - z}  \, dz \; x = Px
\end{align*}
für alle $x \in X$ und 
\begin{align*}
\Bigl(     \frac{1}{2 \pi i} \, \int_{\gamma} \frac{1}{\lambda - z} \, (z-A)^{-1} \, dz         \Bigr)  (\lambda - A)x = Px 
\end{align*}
für alle $x \in D$, das heißt $\lambda \notin \sigma(A \big|_{PD})$. Also gilt $\sigma(A \big|_{PD}) \subset \sigma$.

Sei nun $\lambda \notin \sigma(A) \setminus \sigma$ und $\gamma$ ein Zykel in $\rho(A)$ mit $n(\gamma, \sigma) = 1$, $n(\gamma, \sigma(A) \setminus \sigma) = 0$ und $n(\gamma, \lambda) = 1$ (Proposition~\ref{prop: Cauchy für kompakta}!). Dann gilt
\begin{align*}
(\lambda - A) & \Bigl(     \frac{1}{2 \pi i} \, \int_{\gamma} \frac{1}{z - \lambda} \, (z-A)^{-1}  \, dz         \Bigr) \, x \\
& = - \frac{1}{2 \pi i} \, \int_{\gamma}  (z-A)^{-1} \, dz \; x +   \frac{1}{2 \pi i} \, \int_{\gamma} \frac{1}{z - \lambda}  \, dz \; x = (1-P)x
\end{align*}
für alle $x \in X$ und
\begin{align*}
\Bigl(     \frac{1}{2 \pi i} \, \int_{\gamma} \frac{1}{z - \lambda} \, (z-A)^{-1}  \, dz         \Bigr) (\lambda - A) x = (1-P)x
\end{align*}
für alle $x \in D$, das heißt $\lambda \notin \sigma(A \big|_{(1-P)D})$. Also gilt $\sigma(A \big|_{(1-P)D}) \subset \sigma(A) \setminus \sigma$.

Da nach Proposition~\ref{prop: zerl. des spektrums} $\sigma(A) = \sigma(A \big|_{PD}) \cup \sigma(A \big|_{(1-P)D})$ gilt, muss nun sogar
\begin{align*}
\sigma(A \big|_{PD}) = \sigma \text{ \; und \; } \sigma(A \big|_{(1-P)D}) = \sigma(A) \setminus \sigma
\end{align*}
gelten, wie behauptet.
\end{proof}

Aus diesem Satz folgt, dass die Rieszprojektion $P$ von $A$ auf $\sigma$ ($A$ und $\sigma$ wie im Satz) \emph{genau} dann gleich $0$ ist, wenn $\sigma$ leer ist. Insbesondere ist die Rieszprojektion von $A$ auf das ganze Spektrum $\sigma(A)$ (falls dieses beschränkt ist) nicht immer gleich $1$, denn $\sigma(A)$ kann ja leer sein.
\\

Der obige Satz sagt, dass die Rieszprojektion einer linearen Abbildung $A$ eine Zerlegung von $X$ in $A$-invariante Unterräume gibt, die einer (vorgegebenen) Zerlegung von $\sigma(A)$ in $\sigma$ (kompakt) und $\sigma(A) \setminus \sigma$ (abgeschlossen) entspricht.
Die folgende Proposition besagt nun, dass die Rieszprojektion von $A$ auf $\sigma$ die \emph{einzige} solche Projektion ist. 

\begin{prop} \label{prop: rieszproj eind}
Sei $A: D \subset X \to X$ eine abgeschlossene lineare Abbildung mit $\rho(A) \ne \emptyset$ und $P$ eine beschränkte Projektion in $X$, sodass $P A \subset A P$, die Mengen $\sigma(A \big|_{PD})$ und $\sigma(A \big|_{(1-P)D})$ disjunkt sind und $A \big|_{PD}$ beschränkt (auf $PD$) ist. 
Dann ist $P$ die Rieszprojektion von $A$ auf $\sigma(A \big|_{PD})$.
\end{prop}
 
\begin{proof} 
Sei $\sigma := \sigma(A \big|_{PD})$ und $P_0$ die Rieszprojektion von $A$ auf $\sigma$. Diese existiert wirklich, denn wegen der Abgeschlossenheit und der Beschränktheit von $A \big|_{PD}$ haben wir $PD = PX$, das heißt, $\sigma$ ist (als Spektrum der beschränkten linearen Abbildung $A \big|_{PD}$ auf dem Banachraum $PX$) kompakt und daher getrennt von $\sigma(A \big|_{(1-P)D}) = \sigma(A) \setminus \sigma$.

Wir haben dann für jeden Zykel $\gamma$ in $\rho(A)$, der $\sigma$ einmal und $\sigma(A) \setminus \sigma$ keinmal umläuft, dass (wie gewünscht)
\begin{align*}
P_0 &= \frac{1}{2 \pi i} \, \int_{\gamma} (z-A)^{-1} \, P \, dz + \frac{1}{2 \pi i} \, \int_{\gamma} (z-A)^{-1} \, (1-P) \, dz \\
&= \frac{1}{2 \pi i} \, \int_{\gamma} (z-A \big|_{PD})^{-1}  \, dz \, P + \frac{1}{2 \pi i} \, \int_{\gamma} (z-A \big|_{(1-P)D})^{-1}  \, dz \, (1-P) 
= P,
\end{align*} 
denn einerseits ist 
\begin{align*}
\frac{1}{2 \pi i} \, \int_{\gamma} (z-A \big|_{PD})^{-1}  \, dz
\end{align*}
als Rieszprojektion von $A \big|_{PD}$ (beschränkt!) auf $\sigma = \sigma(A \big|_{PD})$ (das ganze Spektrum von $A \big|_{PD}$!) die identische Abbildung auf $PX$ und andererseits ist 
\begin{align*}
\frac{1}{2 \pi i} \, \int_{\gamma} (z-A \big|_{(1-P)D})^{-1}  \, dz
\end{align*}
gleich null, weil $\gamma$ das Spektrum von $A \big|_{(1-P)D}$ nicht umläuft.
\end{proof} 
 
Wie ändert sich die Rieszprojektion, wenn wir die Zerlegung des Spektrums ändern? Die folgende Proposition zeigt, dass dabei nichts Unerwartetes passiert.

\begin{prop}  \label{prop: rechenregeln rieszprojektion}
Sei $A: D \subset X \to X$ abgeschlossen mit $\rho(A) \ne \emptyset$, seien $\sigma_1$, $\sigma_2$ zwei kompakte in $\sigma(A)$ isolierte Untermengen von $\sigma(A)$ und $P_1$, $P_2$ die zugehörigen Rieszprojektionen. Dann ist $P_1 P_2 = P_2 P_1$ die Rieszprojektion von $A$ auf $\sigma_1 \cap \sigma_2$. Wenn $\sigma_1$ und $\sigma_2$ zusätzlich disjunkt sind, dann ist $P_1 + P_2$ die Rieszprojektion von $A$ auf $\sigma_1 \cup \sigma_2$.  
\end{prop}

\begin{proof}
Zunächst bemerken wir, dass $\sigma_{12} := \sigma_1 \cap \sigma_2$ eine kompakte und in $\sigma(A)$ isolierte Untermenge von $\sigma(A)$ ist und dass damit auch $\sigma_i \setminus \sigma_{12}$ kompakt und isoliert ist in $\sigma(A)$ ($i \in \{ 1,2\}$). Also existiert eine positive Zahl $r_0$, sodass $U_{r_0} (\sigma_{12})$,  $U_{r_0} (\sigma_1 \setminus \sigma_{12})$ und $U_{r_0} (\sigma_2 \setminus \sigma_{12})$ paarweise disjunkt sind und darüberhinaus das Spektrum $\sigma(A)$ nur in $\sigma_{12}$, $\sigma_1 \setminus \sigma_{12}$ bzw. $\sigma_2 \setminus \sigma_{12}$ selbst schneiden. 
Sei $\gamma_{12}$ ein Zykel in $U_{r_0} (\sigma_{12}) \setminus \sigma_{12}$ mit
\begin{align*}
n(\gamma_{12},\, \sigma_{12}) = 1 \text{ \; und \; } n(\gamma_{12}, \, \complex \setminus U_{r_0} (\sigma_{12} ) ) = 0 
\end{align*}
und sei $\gamma_i$ ein Zykel in $U_{r_0} (\sigma_i \setminus \sigma_{12}) \setminus \bigl( \sigma_i \setminus \sigma_{12} \bigr)$ mit
\begin{align*}
n(\gamma_i, \, \sigma_i \setminus \sigma_{12} ) = 1 \text{ \; und \; } n(\gamma_i, \, \complex \setminus U_{r_0} (\sigma_i \setminus \sigma_{12}) ) = 0
\end{align*}
(Proposition~\ref{prop: Cauchy für kompakta}!). Dann ist
\begin{align*}
P_{12} :=  \frac{1}{2 \pi i} \, \int_{\gamma_{12}} (z-A)^{-1}  \, dz 
\end{align*}
die Rieszprojektion von $A$ auf $\sigma_{12}$ und für die Rieszprojektionen $P_i$ gilt
\begin{align*}
P_i = \frac{1}{2 \pi i} \, \int_{\gamma_i + \gamma_{12}} (z-A)^{-1}  \, dz = \frac{1}{2 \pi i} \, \int_{\gamma_i} (z-A)^{-1}  \, dz + \frac{1}{2 \pi i} \, \int_{\gamma_{12}} (z-A)^{-1}  \, dz.
\end{align*} 
Da die Zykel $\gamma_{12}$, $\gamma_1$ und $\gamma_2$ sich nicht gegenseitig umlaufen, verschwinden die Produkte aus den Projektionen $P_{12}$, $\frac{1}{2 \pi i} \, \int_{\gamma_1} (z-A)^{-1}  \, dz$ und $\frac{1}{2 \pi i} \, \int_{\gamma_2} (z-A)^{-1}  \, dz$ (sofern sie natürlich aus verschiedenen Faktoren bestehen), und daraus folgt schließlich, dass $P_1 P_2 = P_{12}$ und $P_2 P_1 = P_{12}$, wie gewünscht.
\\

Seien nun $\sigma_1$ und $\sigma_2$ zusätzlich disjunkt. Dann existiert eine positive Zahl $r_0$, sodass $U_{r_0}(\sigma_1)$ und $U_{r_0}(\sigma_2)$ disjunkt sind und darüberhinaus das Spektrum $\sigma(A)$ nur in $\sigma_1$ bzw. $\sigma_2$ selbst schneiden. Sei nun $\gamma_i$ ein Zykel in $U_{r_0}(\sigma_i) \setminus \sigma_i$, der $\sigma_i$ einmal und $\complex \setminus U_{r_0}(\sigma_i)$ keinmal umläuft ($i \in \{ 1,2\}$). Dann ist $\gamma_1 + \gamma_2$ ein Zykel in $\rho(A)$ mit $n (\gamma_1 + \gamma_2, \, \sigma_1 \cup \sigma_2 ) = 1$ und $n (\gamma_1 + \gamma_2, \, \sigma(A) \setminus (\sigma_1 \cup \sigma_2) ) = 0$, das heißt, die Rieszprojektion von $A$ auf $\sigma_1 \cup \sigma_2$ ist gegeben durch 
\begin{align*}
\frac{1}{2 \pi i} \, \int_{\gamma_1 + \gamma_2} (z-A)^{-1}  \, dz = \frac{1}{2 \pi i} \, \int_{\gamma_1} (z-A)^{-1}  \, dz + \frac{1}{2 \pi i} \, \int_{\gamma_2} (z-A)^{-1}  \, dz = P_1 + P_2,
\end{align*}
wie behauptet.
\end{proof}

Abschließend noch ein wichtiger Satz zum Sonderfall isolierter Spektralwerte $\lambda$ abgeschlossener linearer Abbildungen $A$ in $X$ (Theorem~VIII.8.3 und Theorem~VIII.8.4 in~\cite{Yosida: FA}). 
Wir nennen $\lambda$ dabei einen \emph{isolierten Spektralwert der Ordnung $m \in \natu \cup \{ \infty \}$} genau dann, wenn $\lambda$ -- als isolierter Spektralwert offensichtlich eine isolierte Singularität der Resolventenabbildung -- ein $m$-facher Pol von $(\, . \, - A)^{-1}$ ist (wobei wir unter einem $\infty$-fachen Pol natürlich eine wesentliche Singularität verstehen).

Die \emph{algebraische Vielfachheit von $\lambda$} ist als $\dim PX$ definiert, wobei $P$ die Rieszprojektion von $A$ auf $\{\lambda\}$ ist. Der isolierte Spektralwert $\lambda$ heißt \emph{halbeinfach} genau dann, wenn die geometrische Vielfachheit von $\lambda$ mit der algebraischen übereinstimmt, kurz: wenn $\dim \ker (A-\lambda) = \dim PX$. Insbesondere ist $\lambda$ dann ein Eigenwert. Schließlich nennen wir $\lambda$ \emph{einfach}, wenn die algebraische Vielfachheit von $\lambda$ gleich $1$ ist.

\begin{thm} \label{thm: isol spektralwerte}
Sei $A: D \subset X \to X$ eine abgeschlossene lineare Abbildung, $\lambda$ ein isolierter Spektralwert von $A$ der Ordnung $m \in \natu \cup \{ \infty \}$ und $P$ die Rieszprojektion von $A$ auf $\{ \lambda \}$.
\begin{itemize}
\item [(i)] Sei $m \in \natu$. Dann ist $\lambda$ ein Eigenwert von $A$ und  
\begin{align*}
P X = \ker (A-\lambda)^n \text{ \; und \; } (1-P)X = \im (A-\lambda)^n
\end{align*}
für alle natürlichen Zahlen $n \ge m$. 
\item [(ii)] Sei $m = \infty$. Dann ist die algebraische Vielfachheit von $\lambda$ unendlich. 
\end{itemize}
\end{thm}

\begin{proof}
Sei 
\begin{align*}
U_n := \frac{1}{2 \pi i} \, \int_{\gamma} \frac{ (z-A)^{-1} }{(z-\lambda)^{n+1}} \, dz
\end{align*}
für alle $n \in \mathbb{Z}$, wobei $\gamma$ einen (kreisförmigen) Zykel in $\rho(A)$ bezeichnet, der $\{ \lambda \}$ einmal und $\sigma(A) \setminus \{ \lambda \}$ keinmal umläuft. $U_n$ ist der $n$-te Laurentkoeffizient von $(\,. \,- A)^{-1}$ bei $\lambda$.
Wir sehen sofort, dass
\begin{align*}
U_n A \subset A U_n \text{ \; und \; } U_n X \subset D
\end{align*} 
für alle $n \in \mathbb{Z}$. Außerdem sehen wir (durch direktes Nachrechnen), dass
\begin{align*}
(A-\lambda) U_n = \begin{cases} U_{n-1},&  n \in \mathbb{Z} \setminus \{0\} \\
																	  U_{-1} - 1,& n = 0
											\end{cases},
\end{align*}
woraus sich wiederum ergibt, dass
\begin{align} \label{eq: isol spektralwerte 1}
U_{-n-1} = (A-\lambda)^n U_{-1} = (A-\lambda)^n P \text{ \; und \; } P-1 = U_{-1} -1 = (A-\lambda)^n U_{n-1}
\end{align}
für alle $n \in \natu$.
\\

(i) Sei nun $m \in \natu$, dann gilt $U_{-n-1} = 0$ für alle $n \ge m$. Aus~\eqref{eq: isol spektralwerte 1} folgt damit, dass
\begin{align*}
PX \subset \ker (A- \lambda)^n \text{ \; und \; } (1-P)X \subset \im (A-\lambda)^n
\end{align*}
für alle $n \ge m$, woraus sich erneut mithilfe von~\eqref{eq: isol spektralwerte 1} ergibt, dass 
$(1-P) \bigl( \ker(A-\lambda)^n \bigr) = 0$ und $P \bigl( \im (A-\lambda)^n \bigr) = 0$ und daher
\begin{align*}
PX = \ker (A- \lambda)^n \text{ \; und \; } (1-P)X = \im (A-\lambda)^n
\end{align*}
für alle $n \ge m$, wie gewünscht.
\\

(ii) Sei $m = \infty$, dann gilt nach~\eqref{eq: isol spektralwerte 1}
\begin{align*}
(A-\lambda)^n P = U_{-n-1} \ne 0
\end{align*} 
für unendlich viele $n \in \natu$. Die lineare Abbildung $(A-\lambda) \big|_{PX}$ ist also nicht nilpotent. Wegen $\sigma \bigl((A-\lambda) \big|_{PX}\bigr) = \{0\}$ (Satz~\ref{thm: Rieszprojektion}) kann $PX$ daher nicht endlichdimensional sein, denn nach einem elementaren Satz der linearen Algebra (über die Trigonalisierbarkeit linearer Abbildungen in endlichdimensionalen Räumen) ist jede lineare Abbildung in einem endlichdimensionalen Raum nilpotent, wenn ihr Spektrum gleich $\{0\}$ ist.  
\end{proof}

\subsubsection{Spektraltheorie normaler linearer Abbildungen}


Sei $A$ eine lineare Abbildung in $H$. Dann heißt $A$ bekanntlich \emph{normal} genau dann, wenn $A$ dicht definiert und abgeschlossen ist und $A^* A = A A^*$ gilt. 

$A$ heißt \emph{(schief)symmetrisch} genau dann, wenn $A$ dicht definiert ist und $A \subset A^*$ ($A \subset - A^*$) gilt. Und $A$ heißt \emph{(schief)selbstadjungiert} genau dann, wenn $A$ dicht definiert ist und $A = A^*$ ($A = - A^*$) gilt, mit anderen Worten: wenn $A$ (schief)symmetrisch ist und $\sigma(A) \subset \real$ ($\sigma(A) \subset i \, \real$) (Theorem~2.5 in~\cite{mmqm1} oder Theorem~X.1 in~\cite{RS 2}). 

Schließlich heißt eine lineare Abbildung $U: H \to H$ \emph{unitär} genau dann, wenn $U$ isometrisch und surjektiv ist, mit anderen Worten: wenn $U^* U = 1 = U U^*$.
\\

(Schief)selbstadjungierte und unitäre lineare Abbildungen sind also offensichtlich normal. Und auch der folgende Satz ist für (schief)selbstadjungierte und unitäre $A$ offensichtlich.

\begin{thm} \label{thm: eigenschaften normaler A}
Sei $A: D(A) \subset H \to H$ normal. Dann gilt $D(A^*) = D(A)$ und $\norm{A^* x} = \norm{Ax}$ für alle $x \in D(A)$. Insbesondere gilt $\ker A^* = \ker A$.
\end{thm}

\begin{proof}
Theorem~13.32 in~\cite{Rudin: FA}.
\end{proof}

Sei $(X_0, \mathcal{A})$ ein messbarer Raum und sei $P_E$ für jedes $E \in \mathcal{A}$ eine orthogonale Projektion in $H$. Dann heißt $P$ \emph{Spektralmaß auf $(X_0, \mathcal{A}, H)$} (Definition~IX.1.1 in~\cite{Conway: fana} oder Definition~12.17 in~\cite{Rudin: FA}) 
genau dann, wenn gilt:
\begin{itemize}
\item [(i)] $P_{\emptyset} = 0$ und $P_{X_0} = \id_H$
\item [(ii)] $P_{E \cap F} = P_E \, P_F$ für alle $E, F \in \mathcal{A}$
\item [(iii)] $P_{\cup_{n=1}^{\infty} E_n} x = \sum_{n=1}^{\infty} P_{E_n}x$ für alle $x \in H$ und alle paarweise disjunkten $E_n \in \mathcal{A}$. 
\end{itemize}

Jedes Spektralmaß $P$ auf $(X_0, \mathcal{A}, H)$ gibt für alle $x, y \in H$ ein komplexes Maß $P_{x,y}$ auf $(X_0, \mathcal{A})$, und zwar so: $P_{x,y}(E) := \scprd{ x, P_E y}$ für alle $E \in \mathcal{A}$. Insbesondere ist $P_{x,x}$ für jedes $x \in H$ mit $\norm{x} = 1$ ein Wahrscheinlichkeitsmaß auf $(X_0, \mathcal{A})$.
\\

Wir erinnern an Spektralintegrale. Sei $f$ eine einfache messbare Abbildung $X_0 \to \complex$, das heißt, $f$ ist messbar und nimmt nur endlich viele Werte $\alpha_1, \dots, \alpha_m \in \complex$ an:
\begin{align*}
f = \sum_{k=1}^m \alpha_k \, \chi_{E_k}
\end{align*} 
für messbare Mengen $E_1, \dots, E_m$. Dann heißt
\begin{align*}
\int f \, dP := \sum_{k=1}^m \alpha_k \, P_{E_k}
\end{align*}
\emph{Spektralintegral von $f$ bzgl. $P$}, wobei $P$ ein Spektralmaß auf $(X_0, \mathcal{A}, H)$ ist. Der definierende Ausdruck rechts hängt nicht von der (nichteindeutigen) Darstellung von $f$ als endliche Linearkombination charakteristischer Funktionen ab.

Sei $f$ nun eine beliebige messbare Abbildung $X_0 \to \complex$. Dann ist das \emph{Spektralintegral $\int f \, dP$ von $f$ bzgl. $P$} die  Abbildung in $H$, gegeben durch
\begin{align*}
D\Bigl( \int f \, dP \Bigr) := D_f := \Big \{ x \in H: \int |f|^2 \, dP_{x,x} < \infty \Big \}
\end{align*}
und 
\begin{align*}
\int f \, dP x := \lim_{n \to \infty} \int f_n \, dP x
\end{align*}
für alle $x \in D\bigl( \int f \, dP \bigr)$, wobei die $f_n$ einfache messbare Abbildungen $X_0 \to \complex$ sind mit $f_n(x) \longrightarrow f(x) \;\; (n \to \infty)$ und $|f_n(x)| \le |f(x)|$ für alle $x \in X$.
Solche einfache messbare Abbildungen $f_n$ existieren 
und der definierende Ausdruck auf der rechten Seite hängt nicht von der speziellen Wahl solcher $f_n$ ab. 
\\

Spektralintegrale sind sehr handlich, wie der folgende Satz belegt.

\begin{thm} \label{thm: rechenregeln spektralintegral}
Sei $P$ ein Spektralmaß auf $(X_0, \mathcal{A}, H)$ und seien $f, g: X_0 \to \complex$ messbare Abbildungen. Dann gilt:
\begin{itemize}
\item [(i)] $\int f \,dP$ ist linear, dicht definiert und abgeschlossen.
\item [(ii)] $\norm{ \int f\,dP x}^2 = \int |f|^2 \, dP_{x,x}$ für alle $x \in D_f$. 
\item [(iii)] $D\bigl( \int f \,dP + \int g\,dP \bigr) = D_{f+g} \cap D_g = D_{f+g} \cap D_f$ und
\begin{align*}
\int f \,dP + \int g \,dP \subset \int f+g \,dP.
\end{align*}
$D\bigl( \int f \,dP \, \int g\,dP \bigr) = D_{fg} \cap D_g$ und
\begin{align*}
\Bigl( \int f \, dP \Bigr) \, \Bigl( \int g \,dP \Bigr) \subset \int fg \,dP.
\end{align*} 
Insbesondere gilt $P_E \, \bigl( \int f\,dP \bigr) \subset \bigl( \int f \,dP \bigr) \, P_E$ für alle $E \in \mathcal{A}$.
\item [(iv)] $\ker \int f \, dP = P_{ \{f=0\} }H$. Insbesondere ist $\int f \, dP$ genau dann injektiv, wenn $P_{ \{f = 0\} } = 0$. 
\item[(v)] $\bigl( \int f \,dP \bigr)^* = \int \overline{f} \, dP$.
\end{itemize}
\end{thm} 

\begin{proof}
Die Aussagen (i) bis (iii) und (v) folgen aus Theorem~13.24 in~\cite{Rudin: FA}. Zur (sehr einfachen) Aussage (iv)! Die Inklusion $P_{ \{f=0\} }H \subset \ker \int f \, dP$ folgt aus (iii). Sei umgekehrt $x \in \ker \int f \, dP$, dann gilt nach (ii)
\begin{align*}
0 = \norm{ \int f \,dP x}^2 = \int |f|^2 \, dP_{x,x},
\end{align*}
das heißt, $f = 0$ $P_{x,x}$-fast überall. Also gilt
\begin{align*}
\norm{ P_{ \{f \ne 0\} } x }^2 = \scprd{ x, P_{ \{f \ne 0\} } x } = P_{x,x}(\{f \ne 0\}) = 0
\end{align*}
und damit $x = P_{ \{f = 0\} } x \in P_{ \{f=0\} }H$, wie gewünscht.
\end{proof}

Aus diesem Satz folgt, dass jedes Spektralintegral $\int f \, dP$ normal ist:
\begin{align*}
\Bigl( \int f \,dP \Bigr)^* \, \Bigl( \int f \, dP \Bigr) = \int |f|^2 \, dP = \Bigl( \int f \, dP \Bigr) \, \Bigl( \int f \,dP \Bigr)^*.
\end{align*}
Der folgende überaus wichtige Satz, der Spektralsatz (in der Spektralmaßversion), besagt nun, dass umgekehrt jede normale lineare Abbildung als Spektralintegral dargestellt werden kann. 

\begin{thm} [Spektralsatz] \label{thm: spektralsatz}
Sei $A$ eine normale lineare Abbildung in $H$. Dann existiert genau ein Spektralmaß auf $(\complex, \mathcal{B}_{\complex}, H)$ mit
\begin{align*}
A = \int \id \, dP
\end{align*}
und für dieses gilt $\supp P = \sigma(A)$, wobei 
\begin{align*}
\supp P := \big \{ z \in \complex: P_{U_{\varepsilon}(z)} \ne 0 \text{ für alle } \varepsilon > 0 \big \}
\end{align*}
der Träger von $P$ ist.
\end{thm}

\begin{proof}
Theorem~13.33 in~\cite{Rudin: FA} oder Theorem~X.4.11 in~\cite{Conway: fana}.
\end{proof}

Das eindeutige Spektralmaß $P$ aus dem obigen Satz heißt das \emph{Spektralmaß von $A$} und wir werden es künftig mit $P^{A}$ bezeichnen.

Wir weisen darauf hin, dass es neben den geläufigen Zugängen zum Spektralsatz für beschränkte normale lineare Abbildungen $A$ -- dem beispielsweise in Rudins Buch~\cite{Rudin: FA} vorgestellten Zugang über $C^*$-Algebren oder dem Zugang über den Sonderfall selsbstadjungierter $A$ und die Bildung des Produktspektralmaßes (s. etwa Satz~VII.1.25 in Werners Buch~\cite{Werner: FA}) 
 -- auch noch andere Zugänge gibt. So ist es beispielsweise möglich  
elementar zu zeigen, dass 
\begin{align*}
\sigma(p(A)) = p(\sigma(A)) = \{ p(z): z \in \sigma(A) \}
\end{align*}
für alle \emph{fast polynomialen Abbildungen} $p: \complex \to \complex$ (Theorem~2 in~\cite{Bernau 65}), worunter wir Abbildungen $p$ verstehen mit
\begin{align*}
p(z) = p_0(z, \overline{z}) = \sum_{i, j =0}^m a_{i j} z^{i} \, \overline{z}^j
\end{align*}
für alle $z \in \complex$ ($p_0$ ein Polynom in zwei Variablen). $p(A)$ steht natürlich abkürzend für $\sum_{i, j =0}^m a_{i j} \,A^{i} \, (A^*)^j$. Aus diesem Spektralabbildungssatz für fast polynomiale Abbildungen folgt, dass die Abbildung 
\begin{gather*}
U:= \{ \text{fast polynomiale Abbildungen } p: \sigma(A) \to \complex \} \subset C(\sigma(A)) \\
\longrightarrow \{ \text{beschränkte lineare Abbildungen in } H \} \\
p \mapsto p(A) 
\end{gather*}
wohldefiniert ist 
und dass sie isometrisch ist:
\begin{align*}
\norm{ p(A) } = \sup \{ |w|: w \in \sigma(p(A)) \} = \sup \{ |p(z)|: z \in \sigma(A) \} = \norm{ p }_{C(\sigma(A))}
\end{align*} 
für alle $p \in U$, denn $p(A)$ ist normal und daher gilt
\begin{align*}
\norm{ p(A) } = r_{p(A)}
\end{align*}
nach Satz~VI.1.7 in~\cite{Werner: FA}.
Außerdem ist $U$ offensichtlich eine Unteralgebra von $C(\sigma(A), \complex)$, die $1$ enthält, trennend und konjugationsstabil ist, und $U \ni p \mapsto p(A)$ ist ein Algebrenhomomorphismus. Der Satz von Stone, Weierstraß (Theorem~V.4.7 in~\cite{AmannEscher}) ergibt also, dass $U$ dicht liegt in $C(\sigma(A), \complex)$, das heißt, der angesprochene isometrische Algebrenhomomorphismus ist fortsetzbar zu einem ebenfalls isometrischen Algebrenhomomorphismus $C(\sigma(A)) \ni f \mapsto f(A)$, dem sog. stetigen Funktionalkalkül von $A$. Zu beachten ist dabei, dass wir nicht mit der (hinsichtlich des Spektralabbildungssatzes) viel einfacheren Unteralgebra der polynomialen Abbildungen $p_0: \sigma(A) \to \complex$ arbeiten können, weil diese eben für allgemeine normale $A$ nicht konjugationsstabil ist. 

Aus dem stetigen Funktionalkalkül ergibt sich der Spektralsatz für beschränkte normale $A$ dann ohne größere Schwierigkeiten: der rieszsche Satz (über den bijektiven Zusammenhang zwischen $i \in C_0(\sigma(A))' = C(\sigma(A))'$ und den regulären komplexen Maßen $\mu$ auf $(\sigma(A), \mathcal{B}_{\sigma(A)})$ (Theorem~C.18 in~\cite{Conway: fana})) liefert das gesuchte Spektralmaß von $A$.

\begin{prop} \label{prop: rieszproj für normale A}
Sei $A$ normal mit Spektralmaß $P^{A}$. Dann gilt:
\begin{itemize}
\item [(i)] $\lambda$ ist Eigenwert von $A$ genau dann, wenn $P_{ \{\lambda\} }^{A} \ne 0$, mehr noch:
\begin{align*}
\ker (A-\lambda) = P_{ \{\lambda\} }^{A} H
\end{align*}
für alle $\lambda \in \complex$.
\item [(ii)] $\norm{ (z-A)^{-1} } = \bigl( \dist( z, \sigma(A) ) \bigr)^{-1}$ für alle $z \in \rho(A)$.
\item [(iii)] Wenn $\sigma$ eine kompakte in $\sigma(A)$ isolierte Untermenge von $\sigma(A)$ ist, dann ist die Rieszprojektion von $A$ auf $\sigma$ gleich $P_{\sigma}^{A}$ (eine orthogonale Projektion). Insbesondere stimmt die algebraische Vielfachheit eines isolierten Spektralwerts $\lambda$ von $A$ mit der geometrischen Vielfachheit überein und $\lambda$ ist also ein Eigenwert von $A$.
\end{itemize}
\end{prop}

\begin{proof}
Aussage (i) folgt aus Satz~\ref{thm: rechenregeln spektralintegral}~(iv) und Aussage (ii) und (iii) folgen aus 
\begin{align*}
(z-A)^{-1} = \int \frac{1}{z-w} \, dP^{A}(w) = \int_{\sigma(A)} \frac{1}{z-w} \, dP^{A}(w),
\end{align*}
was sich wiederum aus Satz~\ref{thm: rechenregeln spektralintegral}~(iii) und $\supp P^{A} = \sigma(A)$ ergibt.
\end{proof}

\subsection{Stark stetige Halbgruppen}

Wir erinnern zunächst an die Definition stark stetiger Halbgruppen und der Erzeuger stark stetiger Halbgruppen.
\\

Sei $T(t)$ für jedes $t \in [0,\infty)$ eine beschränkte lineare Abbildung in $X$. Dann heißt $T$ \emph{stark stetige Halbgruppe auf $X$} genau dann, wenn
\begin{itemize}
\item [(i)] $T(0) = 1$ und $T(t+s) = T(t) T(s)$ für alle $s,t \in [0, \infty)$
\item [(ii)] die Abbildung $[0,\infty) \ni t \mapsto T(t)x$ (rechtsseitig) stetig ist in $0$ für alle $x \in X$, das heißt
$T(h)x \longrightarrow x \;\; (h \searrow 0)$.
\end{itemize}

Sei $T$ eine stark stetige Halbgruppe auf $X$. Dann heißt die lineare Abbildung $A$, die gegeben ist durch
\begin{gather*}
D(A) := \big \{ x \in X: [0,\infty) \ni t \mapsto T(t)x \text{ ist (rechtsseitig) differenzierbar in $0$} \big \}, \\ 
Ax := \lim_{h\searrow 0} \frac{T(h)x - x}{h}, 
\end{gather*}
der \emph{Erzeuger von $T$}. Wir sagen auch, $A$ \emph{erzeuge $T$}.

\begin{ex} \label{ex: beschränkte A sind halbgrerzeuger}
Sei $A$ eine beschränkte lineare Abbildung in $X$ und
\begin{align*}
T(t) := e^{A t} = \sum_{n=0}^{\infty} \frac{A^n}{n!} \, t^n
\end{align*}
für alle $t \in [0, \infty)$. Wie man leicht nachprüft, ist $T$ dann eine stark stetige Halbgruppe auf $X$ und $A$ ist der Erzeuger von $T$. Jede beschränkte lineare Abbildung in $X$ erzeugt also eine stark stetige Halbgruppe auf $X$.  $\blacktriangleleft$
\end{ex}

\begin{prop} \label{prop: T stark stetig und exp abschätzung für T(t)}
Sei $T$ eine stark stetige Halbgruppe auf $X$. Dann gibt es Zahlen $M \in [1, \infty)$ und $\omega \in \real$, sodass
\begin{align*}
\norm{ T(t) } \le M e^{\omega t}
\end{align*}
für alle $t \in [0, \infty)$, und $[0,\infty) \ni t \mapsto T(t)x$ ist stetig für alle $x \in X$.
\end{prop}

\begin{proof}
Proposition~I.5.3 und Proposition~I.5.5 in~\cite{EngelNagel}.
\end{proof}

Wir nennen ein $\omega$ wie in der obigen Proposition einen \emph{Wachstumsexponenten von $T$} und das Infimum der Wachstumsexponenten
\begin{align*}
\omega_T 
:= \inf \big \{ \omega \in \real: \text{ es gibt ein $M_{\omega} \in [1, \infty)$, sodass } \norm{T(t)} \le M_{\omega} \, e^{\omega t} \text{ für alle } t \in [0, \infty) \big \}
\end{align*}
heißt \emph{Wachstumsschranke von $T$}. Sie ist nach obiger Proposition enthalten in $\real \cup \{ -\infty \}$ und kann jeden der Werte darin annehmen (Beispiel~I.5.7~(i) und~(iii) in~\cite{EngelNagel}). Zu beachten ist, dass die Wachstumsschranke kein Wachstumsexponent sein muss, wie man etwa an dem Beispiel 
\begin{align*}
T(t) := e^{A t} \text{ für alle $t \in [0, \infty)$ mit } A := \begin{pmatrix} 0 & 1 \\ 0 & 0 \end{pmatrix}  
\end{align*}
sieht, in dem die Wachstumsschranke gleich $0$ ist und daher kein Wachstumsexponent ist.
$T$ heißt \emph{Quasikontraktionshalbruppe} bzw. \emph{Kontraktionshalbgruppe} genau dann, wenn ein $\omega \in \real$ bzw. ein $\omega \in (-\infty, 0]$ existiert, sodass $\norm{T(t)} \le e^{\omega t}$ für alle $t \in [0, \infty)$.
\\

Sei $T$ eine stark stetige Halbgruppe auf $X$ und $A$ ihr Erzeuger. Dann ist die Abbildung $[0, \infty) \ni t \mapsto T(t)x$ rechtsseitig differenzierbar für alle $x \in D(A)$ mit rechtsseitiger Ableitung $t \mapsto T(t)Ax$, denn es gilt ja $T(t+h) = T(t)T(h)$ für alle $t, h \in [0, \infty)$. Insbesondere ist $T(t)x \in D(A)$ für alle $x \in D(A)$ und es gilt
\begin{align*}
A T(t)x = 
\lim_{h \searrow 0} \frac{T(t+h)x - T(t)x}{h} = T(t)Ax
\end{align*}
für alle $x \in D(A)$ und alle $t \in [0, \infty)$, kurz: $A T(t) \supset T(t) A$ für alle $t \in [0, \infty)$. 
Außerdem ist die eben erwähnte rechtsseitige Ableitung nach Proposition~\ref{prop: T stark stetig und exp abschätzung für T(t)} stetig, mithilfe von Satz~\ref{thm: einseitig db und beidseitig db} folgt also, dass $[0, \infty) \ni t \mapsto T(t)x$ beidseitig stetig differenzierbar ist für alle $x \in D(A)$, und daraus wiederum ergibt sich mithilfe von Lemma~\ref{lm: strong db of products} und der Dichtheit von $D(A)$ (Satz~\ref{thm: eigenschaften von erzeugern}) leicht, dass $A$ außer $T$ keine weitere stark stetige Halbgruppe auf $X$ erzeugt (s. Theorem~II.1.4 in~\cite{EngelNagel}). Wir können daher von \emph{der} von $A$ erzeugten stark stetigen Halbgruppe sprechen und diese symbolisch auch mit $e^{A \, . \,}$ bezeichnen (was sich mit Beispiel~\ref{ex: beschränkte A sind halbgrerzeuger} verträgt).

Außerdem können wir deshalb die Wachstumsschranke von $T$ auch mit $\omega_A$ bezeichnen: $\omega_A := \omega_{ e^{A \, . \,} } = \omega_T$.

\begin{thm} \label{thm: eigenschaften von erzeugern}
Sei $T$ eine stark stetige Halbgruppe auf $X$ mit Erzeuger $A$ und seien $M \in [1, \infty)$ und $\omega \in \real$, sodass 
$\norm{ T(t) } \le M e^{\omega t}$
für alle $t \in [0, \infty)$. Dann gilt: 
\begin{itemize}
\item[(i)] $A$ ist eine dicht definierte abgeschlossene lineare Abbildung 
\item[(ii)] $\{ z \in \complex:  \Re z > \omega \} \subset \rho(A)$ und es gilt 
\begin{align*}
\norm{ (z-A)^{-n} } \le \frac{M}{ (\Re z - \omega)^n }
\end{align*}
für alle $n \in \natu$ und alle $z \in \complex$ mit $\Re z > \omega$.
\end{itemize}
\end{thm}

\begin{proof}
Theorem~II.1.4, Theorem~II.1.10 und Korollar~II.1.11 in~\cite{EngelNagel}. 
\end{proof}

Aus diesem Satz folgt unmittelbar, dass
\begin{align*}
\sigma(A) \subset \{ z \in \complex: \Re z \le \omega_A \},
\end{align*}
mit anderen Worten: die Spektralschranke $s_A$ von $A$ ist kleiner oder gleich $\omega_A$, wobei
\begin{align*}
s_A := \sup \{ \Re z: z \in \sigma(A) \}.
\end{align*}
Spektralschranke und Wachstumsschranke stimmen (in Anwendungen) häufig überein (Korollar~IV.3.11, Korollar~IV.3.12 und Lemma~V.1.9 in~\cite{EngelNagel}), so etwa auch in unserem Anwendungsbeispiel aus Abschnitt~\ref{sect: anwendbsp} (nach Theorem~VI.1.15 in~\cite{EngelNagel}).

\begin{ex} \label{ex: mult- und translhalbgruppe}
(i) Sei $(X_0, \mathcal{A}, \mu)$ ein Maßraum, $X := L^p(X_0, \complex)$ und $f$ eine messbare Abbildung $X_0 \to \complex$, sodass $\Re f(x) \le \omega$ für fast alle $x \in X_0$ (für eine von $x$ unabhängige reelle Zahl $\omega$). Sei 
\begin{align*}
\bigl( T(t)g \bigr)(x) := e^{f(x) \, t} g(x)
\end{align*}
für alle $g \in X$ und alle $t \in [0, \infty)$. Dann ist $T$ eine stark stetige Halgbgruppe (lebesguescher Satz!), eine sog. Multiplikationshalbgruppe, und es gilt $\norm{T(t)} \le e^{\omega t}$ für alle $t \in [0, \infty)$. Sei $A$ der Erzeuger von $T$, dann gilt $A = M_f$. Sei nämlich $g \in D(A)$. Dann gilt
\begin{align*}
\frac{ T(h)g - g}{h} \longrightarrow A g \quad (h \searrow 0)
\end{align*}
und daher existiert (beispielsweise nach Satz~VI.4.3 und Korollar~VI.4.13 in~\cite{Elstrodt}) eine Folge $(h_n)$ in $(0, \infty)$, sodass $h_n \longrightarrow 0 \;\;(n \to \infty)$ und
\begin{align*}
\frac{e^{f(x) h_n} - 1}{h_n} \, g(x) =  \Bigl( \frac{ T(h_n)g - g}{h_n} \Bigr)(x) \longrightarrow \bigl( A g \bigr)(x) \quad (n \to \infty)
\end{align*}
für fast alle $x \in X_0$.
Also gilt $f g = A g$ und damit $A \subset M_f$. Da erstens $\{ \Re z > \omega\} \subset \rho(A)$ nach Satz~\ref{thm: eigenschaften von erzeugern} und zweitens $\{ \Re z > \omega\} \subset \rho(M_f)$ nach Beispiel~\ref{ex: spektrum von multop}, 
gilt $\rho(A) \cap \rho(M_f) \ne \emptyset$ und daher (nach Aufgabe~IV.1.21 (5)) $A = M_f$, wie behauptet.
\\

(ii) Sei $X := L^p(\real, \complex)$ und sei 
\begin{align*}
\bigl( T(t)f \bigr)(x) := f(x+t)
\end{align*} 
für alle $f \in X$ und alle $t \in [0, \infty)$. Dann ist $T$, wie man leicht sieht (s. Beispiel~I.5.4 in~\cite{EngelNagel}), eine stark stetige Halbgruppe auf $X$, eine sog. Translationshalbgruppe, und es gilt $\norm{T(t)} \le 1$ für alle $t \in [0, \infty)$. Sei $A$ der Erzeuger von $T$, dann ist $A$ gegeben durch die schwache Ableitung auf $W^{1,p}(\real)$, genauer: $A = B$, wobei 
\begin{align*}
D(B) := W^{1,p}(\real, \complex) := \{ f \in L^p(\real, \complex): f \text{ schwach differenzierbar und } f' \in L^p(\real) \} 
\end{align*}
und $B f := f'$ (schwache Ableitung). 
Sei nämlich $f \in D(A)$. Dann gilt einerseits
\begin{align*}
\int \Bigl( \frac{T(h)f - f}{h} \Bigr)(x) \varphi(x) \, dx \longrightarrow \int \bigl( A f \bigr)(x) \varphi(x) \, dx \quad (h \searrow 0)
\end{align*} 
für alle $\varphi \in C_c^{\infty}(\real, \complex)$ und andererseits
\begin{align*}
\int \Bigl( \frac{T(h)f - f}{h} \Bigr)(x) \varphi(x) \, dx &= \int f(x) \, \frac{\varphi(x-h) - \varphi(x)}{h} \, dx \\
&\longrightarrow - \int f(x) \varphi'(x) \, dx \quad (h \to 0)
\end{align*}
für alle $\varphi \in C_c^{\infty}(\real, \complex)$, woraus sich ergibt, dass $f$ schwach differenzierbar ist mit schwacher Ableitung $f' = A f$. Also ist $A \subset B$ und aus Satz~\ref{thm: eigenschaften von erzeugern} und Proposition~II.2.10 in~\cite{EngelNagel} folgt, dass $\rho(A) \cap \rho(B) \ne \emptyset$ und daher wieder mithilfe von Aufgabe~IV.1.21 (5), dass $A = B$, wie gewünscht.  $\blacktriangleleft$
\end{ex}

Der folgende ungemein wichtige Satz von Hille, Yosida kehrt die Aussage von Satz~\ref{thm: eigenschaften von erzeugern} um.

\begin{thm}[Hille, Yosida] \label{thm: Hille, Yosida}
Sei $A$ eine lineare Abbildung in $X$, die (i) und (ii) aus Satz~\ref{thm: eigenschaften von erzeugern} erfüllt. Dann erzeugt $A$ (genau) eine stark stetige Halbgruppe auf $X$.
\end{thm}

\begin{proof}
Theorem~II.3.8 in~\cite{EngelNagel}.
\end{proof}

Dieser Satz vereinfacht sich deutlich, wenn $A$ normal ist. Die Voraussetzung, dass (i) und (ii) aus Satz~\ref{thm: eigenschaften von erzeugern} erfüllt sind, wird dann nämlich zu 
\begin{align*}
\sigma(A) \subset \{ z \in \complex: \Re z \le \omega \}
\end{align*}
(nach Satz~\ref{thm: rechenregeln spektralintegral}~(ii) und~(iii)) und die von $A$ erzeugte stark stetige Halbgruppe $T$ ist gegeben durch ein Spektralintegral:
\begin{align*}
T(t) = e^{A t} = \int e^{z t} \, dP^{A}(z)
\end{align*}
für alle $t \in [0, \infty)$, wovon man sich ebenfalls mithilfe von Satz~\ref{thm: rechenregeln spektralintegral}~(ii) und Aufgabe~IV.1.21 (5) in~\cite{EngelNagel} überzeugt. Wenn $A$ sogar schiefselbstadjungiert ist, dann ist $\sigma(A) \subset i \, \real$, das heißt, 
\begin{align*}
U(t) := \int e^{z t} \, dP^{A}(z)
\end{align*}
ist für alle $t \in \real$ unitär (nach Satz~\ref{thm: rechenregeln spektralintegral}~(iii) und~(v)) und $\real \ni t \mapsto U(t)x$ ist stark stetig, kurz: $U$ ist eine stark stetige unitäre Gruppe.

Der Satz von Stone (Theorem~X.5.6 in~\cite{Conway: fana}) besagt, dass umgekehrt jede stark stetige unitäre Gruppe einen schiefselbstadjungierten Erzeuger hat.
\\

Der nachfolgende Satz von Lumer, Phillips ist eine oft nützliche Charakterisierung von Kontraktionshalbgruppenerzeugern, die ohne Resolventenabschätzungen auskommt. Sie benützt stattdessen den Begriff der Dissipativität, an den wir kurz erinnern wollen. Sei $A$ eine lineare Abbildung in $H$. Dann heißt $A$ \emph{dissipativ} genau dann, wenn
\begin{align*}
\Re \scprd{ x, Ax } \le 0
\end{align*} 
für alle $x \in D(A)$ (s. Proposition~II.3.23 und Aufgabe~II.3.25 in~\cite{EngelNagel}).
 
\begin{thm} [Lumer, Phillips] \label{thm: Lumer, Phillips}
Sei $A$ eine dicht definierte lineare Abbildung in $H$. Dann erzeugt $A$ genau dann eine Kontraktionshalbgruppe auf $H$, wenn $A$ dissipativ ist und $\lambda - A$ surjektiv ist für ein $\lambda \in (0, \infty)$.
\end{thm}

\begin{proof}
Theorem~II.3.15 und Proposition~II.3.14 in~\cite{EngelNagel}.
\end{proof}

Schließlich führen wir noch eine abkürzende Sprechweise ein, die im Zusammenhang mit Satz~\ref{thm: triv adsatz 2}, einem unserer trivialen Adiabatensätze, natürlich erscheint. Sei $T$ eine stark stetige Halbgruppe auf $X$ mit Erzeuger $A$. Wir nennen dann
\begin{align*}
\omega_T', \omega_A' := \inf \big \{ \omega \in \real: \norm{ T(t) } \le e^{\omega t} \text{ für alle } t \in [0, \infty) \big \}
\end{align*} 
die \emph{quasikontraktive Wachstumsschranke von $T$}. Sie ist enthalten in $\real \cup \{ \infty \}$ und ist kleiner als $\infty$ genau dann, wenn $T$ eine Quasikontraktionshalbgruppe erzeugt. In diesem Fall ist sie ein Wachstumsexponent von $T$. Insbesondere ist die Wachstumsschranke von $T$ kleiner oder gleich der quasikontraktiven Wachstumsschranke: $\omega_A \le \omega_A'$. 

Darüberhinaus besteht aber kein allgemeiner Zusammenhang zwischen Wachstumsschranke und quasikontraktiver Wachstumsschranke, denn zu jedem $(\omega_1, \omega_2) \in \real^2$ mit $\omega_1 \le \omega_2$ existiert eine stark stetige Halbgruppe $T$, sodass $\omega_T = \omega_1$ und $\omega_T' = \omega_2$. Warum? Sei 
\begin{align*}
A(\alpha) := \begin{pmatrix} \omega_1 & \alpha \\ 0 & \omega_1 \end{pmatrix}
\end{align*}
für alle $\alpha \in [0, \infty)$. Dann ist $\omega_{A(\alpha)} = s_{A(\alpha)} = \omega_1$ nach Korollar~IV.3.12 in~\cite{EngelNagel} für alle $\alpha \in [0, \infty)$ und die Abbildung $\alpha \mapsto \omega_{A(\alpha)}'$  ist stetig, monoton wachsend und unbeschränkt und $\omega_{A(0)}' = \omega_1$, das heißt, es gibt ein $\alpha_0 \in [0, \infty)$, sodass $\omega_{A(\alpha_0)}' = \omega_2$, womit wir gezeigt haben, was zu zeigen war.

\subsection{Störungstheorie}

Wir vermerken zunächst einen elementaren Satz zur stetigen Abhängigkeit des Spektrums beschränkter linearer Abbildungen $A(t)$ vom Parameter, der besagt, dass das Spektrum von jeder festen Stelle $t_0$ aus gesehen nicht plötzlich (unkontrolliert) wächst.

\begin{prop} \label{prop: sigma(A(t)) oberhstet}
Sei $A(t)$ für jedes $t \in I$ eine beschränkte lineare Abbildung in $X$ und $t \mapsto A(t)$ stetig. Dann ist $t \mapsto \sigma(A(t))$ oberhalbstetig.
\end{prop}

\begin{proof}
Das folgt mithilfe der neumannschen Reihe.
\end{proof}

I.a. ist $t \mapsto \sigma(A(t))$ allerdings nicht auch unterhalbstetig, wie Beispiel~IV.3.8 in~\cite{Kato: Perturbation 80} zeigt: dort existiert eine Stelle $t_0 \in I$, sodass $\sigma(A(t_0)) = \overline{U}_1(0)$, während $\sigma(A(t)) = \partial U_1(0)$ für alle $t$ in einer punktierten Umgebung von $t_0$. Von $t_0$ aus gesehen schrumpft dort das Spektrum also plötzlich. 

Was wir zusätzlich voraussetzen müssen, um auch die Unterhalbstetigkeit (und damit wegen obiger Proposition die Stetigkeit) von $t \mapsto \sigma(A(t))$ zu bekommen, klären wir in Proposition~\ref{prop: char sigma(A(t)) unterhstet}.
\\

Seien $A_n$, $A$ abgeschlossene lineare Abbildungen in $X$ (nicht notwendig dicht definiert). Dann \emph{konvergieren die $A_n$ im verallgemeinerten Sinn gegen $A$} genau dann, wenn $\hat{d}(A_n, A) \longrightarrow 0 \; \;(n \to \infty)$, wobei $\hat{d}$ die in Abschnitt~IV.2.4 von~\cite{Kato: Perturbation 80} definierte Metrik auf der Menge aller abgeschlossenen linearen Abbildungen in $X$ ist. Wir schreiben dafür kurz
\begin{align*}
A_n \longrightarrow A \quad (n \to \infty) \text{ \; im verallgemeinerten Sinn}.
\end{align*}

Sei $J$ ein Intervall in $\real$ und sei $A(t)$ für jedes $t \in J$ eine abgeschlossene lineare Abbildung in $X$. Dann heißt die Abbildung $t \mapsto A(t)$ \emph{stetig (in $t_0 \in J$) im verallgemeinerten Sinn} genau dann, wenn sie stetig (in $t_0$) bzgl. der Metrik $\hat{d}$ ist. Anders gesagt: $t \mapsto A(t)$ ist stetig in $t_0$ im verallgemeinerten Sinn genau dann, wenn $A(t_n) \longrightarrow A(t_0) \;\;(n \to \infty)$ im verallgemeinerten Sinn für jede Folge $(t_n)$ in $J$ mit $t_n \longrightarrow t_0 \;\; (n \to \infty)$.
\\

Wir geben nun eine Charakterisierung der verallgemeinerten Konvergenz (Theorem IV.2.25 in~\cite{Kato: Perturbation 80}) unter der Voraussetzung $\rho(A) \ne \emptyset$. Diese Zusatzvoraussetzung ist in Anwendungen wohl meistens erfüllt -- wenigstens für Halbgruppenerzeuger (insbesondere schiefselbstadjungierte) $A$ ist sie \emph{immer} erfüllt (Satz~\ref{thm: eigenschaften von erzeugern}) -- und Reed, Simon dient folgende Charakterisierung denn auch als Definition der verallgemeinerten Konvergenz (Definition in Abschnitt~VIII.7 von~\cite{RS 1}).

\begin{thm} \label{thm: char verallg konv}
Seien $A$, $A_n$ abgeschlossene lineare Abbildungen in $X$ und sei $\rho(A) \ne \emptyset$. Dann sind folgende Aussagen äquivalent: 
\begin{itemize}
\item [(i)] $A_n \longrightarrow A \;\; (n \to \infty)$ im verallgemeinerten Sinn 
\item [(ii)] für alle $z \in \rho(A)$ gilt: $z \in \rho(A_n)$ für genügend große $n$ und 
$(A_n - z)^{-1} \longrightarrow (A - z)^{-1} \;\; (n \to \infty)$ 
\item [(iii)] für ein $z \in \rho(A)$ gilt: $z \in \rho(A_n)$ für genügend große $n$ und 
$(A_n - z)^{-1} \longrightarrow (A - \nolinebreak z)^{-1} \;\; (n \to \infty)$ 
\end{itemize} 
\end{thm}

Insbesondere folgt aus diesem Satz mithilfe von Lemma~\ref{lm: continuity of inv}, dass $A_n \longrightarrow A \quad (n \to \infty)$ im verallgemeinerten Sinn für \emph{beschränkte} $A_n$, $A$ genau dann gilt, wenn $A_n \longrightarrow A \quad (n \to \infty)$ im gewöhnlichen Sinn (das heißt bzgl. der Normoperatortopologie). 
\\

Der nachfolgende Satz ist, wie man mithilfe des obigen Satzes~\ref{thm: char verallg konv} (und der Tatsache, dass die Spektren $\sigma(A(t))$ unter den Voraussetzungen von Proposition~\ref{prop: sigma(A(t)) oberhstet} gleichmäßig in $t \in I$ beschränkt sind) leicht sieht, eine Verallgemeinerung von Proposition~\ref{prop: sigma(A(t)) oberhstet}.
 
\begin{thm} \label{thm: sigma(A(t)) oberhstet für unbeschr A(t)}
Sei $A(t)$ für jedes $t \in I$ eine abgeschlossene lineare Abbildung in $X$ und $t \mapsto A(t)$ stetig an der Stelle $t_0 \in I$ im verallgemeinerten Sinn. Sei weiter $K$ eine kompakte Untermenge von $\rho(A(t_0))$. Dann existiert eine in $I$ offene Umgebung $U_{t_0}$ von $t_0$, sodass $K \subset \rho(A(t))$ für alle $t \in U_{t_0}$. 
\end{thm}

\begin{proof}
Das folgt aus Theorem~IV.3.1 in Katos Buch~\cite{Kato: Perturbation 80}: dieser Satz besagt, dass ein $\delta > 0$ existiert, sodass $K \subset \rho(B)$ für alle abgeschlossenen linearen Abbildungen $B$ in $X$ mit $\hat{d} (B,A) < \delta$ (beachte den in Abschnitt~IV.2.4 von~\cite{Kato: Perturbation 80} beschriebenen Zusammenhang zwischen $\hat{d}$ und $\hat{\delta}$), und zu diesem $\delta$ existiert wegen der verallgemeinerten Stetigkeit von $t \mapsto A(t)$ in $t_0$ eine in $I$ offene Umgebung $U_{t_0}$, sodass $\hat{d} (A(t),A) < \delta$ für alle $t \in U_{t_0}$. 
\end{proof}

Auf den folgenden Satz -- eine Verallgemeinerung des obigen Satzes -- werden wir im Abschnitt über Adiabatensätze mit Spektrallücke häufig zurückgreifen. Im Fall beschränkter $A(t)$ folgt er (sehr elementar) genau wie Proposition~\ref{prop: sigma(A(t)) oberhstet}. 

\begin{thm} \label{thm: (A(t)-z)^{-1} stetig in (t,z)}
Sei $A(t)$ für jedes $t \in I$ eine abgeschlossene lineare Abbildung in $X$ und $t \mapsto A(t)$ stetig im verallgemeinerten Sinn. Dann ist 
\begin{align*}
U := \{ (t,z) \in I \times \complex : z \in \rho(A(t)) \}
\end{align*}
offen in $I \times \complex$ und $U \ni (t,z) \mapsto (A(t) - z)^{-1}$ ist stetig.
\end{thm} 

\begin{proof}
Sei $(t_0, z_0) \in U$. Dann existiert nach Theorem~IV.3.15 in~\cite{Kato: Perturbation 80} zu jedem $\varepsilon > 0$ ein $\delta > 0$, sodass $z \in \rho(B)$ und $\norm{  (B-z)^{-1} - (A(t_0) - z_0)^{-1}  } < \varepsilon$ für alle abgeschlossenen linearen Abbildungen $B$ in $X$ mit $\hat{d}(B,A) < \delta$ und alle $z \in \complex$ mit $|z-z_0| < \delta$. Wegen der verallgemeinerten Stetigkeit von $t \mapsto A(t)$ in $t_0$ folgt nun die Behauptung.  
\end{proof}

Jetzt können wir, wie nach Proposition~\ref{prop: sigma(A(t)) oberhstet} angekündigt, (topologisch und analytisch) klären, wann genau unter den Voraussetzungen dieser Proposition sogar Unterhalbstetigkeit (und damit Stetigkeit) von $t \mapsto \sigma(A(t))$ vorliegt.

\begin{prop} \label{prop: char sigma(A(t)) unterhstet}
Sei $A(t)$ wie in Proposition~\ref{prop: sigma(A(t)) oberhstet}. Dann sind folgende Aussagen äquivalent: 
\begin{itemize}
\item [(i)] $t \mapsto \sigma(A(t))$ ist unterhalbstetig 
\item [(ii)] $\bigcup_{t \in I} \{t\} \times \complex \setminus U_{\varepsilon} \bigl( \sigma(A(t)) \bigr)$ ist abgeschlossen für jedes $\varepsilon > 0$
\item [(iii)] zu jedem $\varepsilon > 0$ existiert eine Zahl $M_{\varepsilon}$, sodass
$\norm{  (A(t) - z)^{-1}  } \le M_{\varepsilon}$
für alle $z \in \complex \setminus U_{\varepsilon} \bigl( \sigma(A(t)) \bigr)$ und $t \in I$. 
\end{itemize}
Insbesondere ist $t \mapsto \sigma(A(t))$ unterhalbstetig, wenn die $A(t)$ alle normal oder die Spektren $\sigma(A(t))$ alle endlich sind.
\end{prop}

\begin{proof}
Sei (i) erfüllt. Sei $\varepsilon > 0$ und $(t_0,z_0) \in I \times \complex$ sowie $(t_n, z_n)$ eine Folge in $\bigcup_{t \in I} \{t\} \times \complex \setminus U_{\varepsilon} \bigl( \sigma(A(t)) \bigr)$ mit $(t_n, z_n) \longrightarrow (t_0, z_0) \;\; (n \to \infty)$. Angenommen, $z_0 \in U_{\varepsilon}(\sigma(A(t_0))$. Dann existiert ein $\varepsilon' \in (0, \varepsilon)$ mit $z \in U_{\varepsilon'}(\sigma(A(t_0))$ und daher existiert auch ein $\varepsilon_0 > 0$, sodass $U_{\varepsilon_0}(z) \subset U_{\varepsilon'}(\sigma(A(t_0))$. Da einerseits $z_n \longrightarrow z_0 \;\;(n \to \infty)$, gilt $z_n \in U_{\varepsilon_0}(z)$ für genügend große $n$, und da andererseits $t \mapsto \sigma(A(t))$ nach (i) unterhalbstetig in $t_0$ ist, gilt $U_{\varepsilon'}(\sigma(A(t_0)) \subset U_{\varepsilon}(\sigma(A(t_n))$ für genügend große $n$. Also gilt insgesamt
\begin{align*}
z_n \in U_{\varepsilon_0}(z) \subset U_{\varepsilon'}(\sigma(A(t_0)) \subset U_{\varepsilon}(\sigma(A(t_n))
\end{align*}
für $n$ groß genug, was im Widerspruch dazu steht, dass $(t_n,z_n) \in \bigcup_{t \in I} \{t\} \times \complex \setminus U_{\varepsilon} \bigl( \sigma(A(t)) \bigr)$ für alle $n \in \natu$. Unsere Annahme war also falsch und die Abgeschlossenheit von $\bigcup_{t \in I} \{t\} \times \complex \setminus U_{\varepsilon} \bigl( \sigma(A(t)) \bigr)$ folgt.
\\

Sei (ii) erfüllt. Sei $R$ eine positive Zahl mit $\sup_{t \in I} \norm{A(t)} \le R$. Dann gilt $\sigma(A(t)) \subset \overline{U}_R(0)$ für alle $t \in I$ und 
\begin{align*}
\norm{ (z-A(t))^{-1} } = \norm{ \frac{1}{z} \, \Bigl( 1- \frac{A(t)}{z} \Bigr)^{-1}  } \le \frac{1}{R}
\end{align*}
für alle $z \in \complex \setminus \overline{U}_{2R}(0)$ und alle $t \in I$ (neumannsche Reihe!).

Sei nun $\varepsilon > 0$. Da $t \mapsto A(t)$ stetig und damit stetig im verallgemeinerten Sinn ist, ist nach Satz~\ref{thm: (A(t)-z)^{-1} stetig in (t,z)} $U := \{ (t,z) \in I \times \complex : z \in \rho(A(t)) \}$ offen und die Abbildung $U \ni (t,z) \mapsto (A(t)-z)^{-1}$ stetig und daher lokal beschränkt. Insbesondere ist diese Abbildung beschränkt auf der nach (ii) kompakten Untermenge 
\begin{align*}
\bigcup_{t \in I} \{t\} \times \complex \setminus U_{\varepsilon} \bigl( \sigma(A(t)) \bigr)   \cap  I \times \overline{U}_{2R}(0) 
= \bigcup_{t \in I} \{t\} \times \overline{U}_{2R}(0) \setminus U_{\varepsilon} \bigl( \sigma(A(t)) \bigr)
\end{align*}
von $U$. 

Also haben wir insgesamt, dass $(t,z) \mapsto (A(t)-z)^{-1}$ beschränkt ist auf $\bigcup_{t \in I} \{t\} \times \complex \setminus U_{\varepsilon} \bigl( \sigma(A(t)) \bigr)$, wie gewünscht.
\\

Sei (iii) erfüllt. Sei $t_0 \in I$ und sei $\varepsilon > 0$. Dann existiert eine in $I$ offene Umgebung $U_{t_0}$, sodass 
\begin{align*}
\norm{ A(t) - A(t_0) } \le \frac{1}{2 M_{\varepsilon}}
\end{align*}
für alle $t \in U_{t_0}$.

Sei $t \in U_{t_0}$ und sei $z \notin U_{\varepsilon}(\sigma(A(t))$. Dann gilt nach (iii) $\norm{  (A(t) - z)^{-1}  } \le M_{\varepsilon}$, und daraus folgt, dass $1 - (A(t_0)-A(t))(z-A(t))^{-1}$ und damit auch
\begin{align*}
z - A(t_0) = \bigl(  1 - (A(t_0)-A(t))(z-A(t))^{-1}  \bigr) (z-A(t))
\end{align*}
invertierbar ist. 

Also haben wir $\sigma(A(t_0)) \subset U_{\varepsilon}(\sigma(A(t))$ für alle $t \in U_{t_0}$, was die gewünschte Unterhalbstetigkeit von $t \mapsto \sigma(A(t))$ ergibt.
\\

Seien nun die $A(t)$ alle normal. Dann ist Aussage (iii) erfüllt mit 
$M_{\varepsilon} := \frac{1}{\varepsilon}$ 
(Proposition~\ref{prop: rieszproj für normale A}) und die Unterhalbstetigkeit folgt aus dem eben Bewiesenen.
\\

Seien schließlich die Spektren $\sigma(A(t))$ der $A(t)$ alle endlich. Wir zeigen durch Widerspruch, dass $t \mapsto \sigma(A(t))$ dann tatsächlich unterhalbstetig ist. Angenommen also, $t \mapsto \sigma(A(t))$ sei nicht unterhalbstetig in der Stelle $t_0$. Dann existiert, da $\sigma(A(t_0))$ ja endlich ist, ein $\lambda_0 \in \sigma(A(t_0))$, eine positive Zahl $\varepsilon_0 > 0$ und eine Folge $(t_n)$ in $I$, sodass $t_n \longrightarrow t_0 \;\; (n \to \infty)$ und $\lambda_0 \notin U_{\varepsilon_0}(\sigma(A(t_n)))$ für alle $n \in \natu$ und $\sigma(A(t_0)) \cap U_{\varepsilon_0}(\lambda_0) = \{ \lambda_0 \}$. 
Wir sehen, dass $\sigma(A(t_n)) \cap U_{\varepsilon_0}(\lambda_0) = \emptyset$ für alle $n \in \natu$, und darüberhinaus nach Satz~\ref{thm: sigma(A(t)) oberhstet für unbeschr A(t)}, dass eine in $I$ offene Umgebung $U_{t_0}$ von $t_0$ existiert, sodass $\partial U_{ \frac{\varepsilon_0}{2} } (\lambda_0) \subset \rho(A(t))$ für alle $t \in U_{t_0}$.
Sei nun
\begin{align*}
P_0(t) := \frac{1}{2 \pi i} \, \int_{  \partial U_{ \frac{\varepsilon_0}{2} } (\lambda_0)  } (z-A(t))^{-1} \, dz
\end{align*}
für alle $t \in U_{t_0}$. Dann ist $U_{t_0} \ni t \mapsto P_0(t)$ stetig, denn $U_{t_0} \times \partial U_{ \frac{\varepsilon_0}{2} } (\lambda_0) \ni (t,z) \mapsto (z-A(t))^{-1}$ ist ja stetig nach Satz~\ref{thm: (A(t)-z)^{-1} stetig in (t,z)} und wir können den lebesgueschen Satz anwenden. Weiter ist $P_0(t)$ für jedes $t \in U_{t_0}$ die Rieszprojektion von $A(t)$ auf $\sigma(A(t)) \cap U_{ \frac{\varepsilon_0}{2} } (\lambda_0)$, und daher gilt $P_0(t_n) = 0$ für alle $n \in \natu$ und $P_0(t_0) \ne 0$. Das widerspricht nach Lemma~\ref{lm: rk konst} der Stetigkeit von $U_{t_0} \ni t \mapsto P_0(t)$. 
\end{proof}

Schließlich noch ein einfacher Störungssatz für Halbgruppenerzeuger.

\begin{thm} \label{thm: störungssatz halbgruppenerz}
Sei $A$ Erzeuger einer stark stetigen Halbgruppe auf $X$, sodass $\norm{e^{A t} } \le M e^{\omega t}$ für alle $t \in [0, \infty)$, und sei $B$ eine beschränkte lineare Abbildung in $X$. Dann ist auch $A+B$ Erzeuger einer stark stetigen Halbgruppe auf $X$ und es gilt
\begin{align*}
\norm{ e^{(A+B)t} } \le M e^{(\omega + M \norm{B})t}
\end{align*}
für alle $t \in [0, \infty)$.
\end{thm}

\begin{proof}
Theorem~III.1.3 in~\cite{EngelNagel}.
\end{proof}

Wenn $A$ sogar eine Kontraktionshalbgruppe erzeugt, dann haben wir einen weiter gehenden Störungssatz (Theorem~III.2.7 in~\cite{EngelNagel}), der den Satz von Kato, Rellich (Theorem~V.4.3 in~\cite{Kato: Perturbation 80} oder Theorem~X.12 in~\cite{RS 2}) verallgemeinert: es genügt dann, wenn die Störung $B$ (wie $A$) dissipativ ist und $A$-beschränkt mit $A$-Schranke kleiner als $1$ (Definition in Abschnitt~IV.1.1 von~\cite{Kato: Perturbation 80}).


\section{Zeitentwicklungen} \label{sect: zeitentwicklungen}

\subsection{Der Zeitentwicklungsbegriff}

In diesem Abschnitt bezeichnet $J := [a,b]$ stets ein nichttriviales kompaktes Intervall in $\real$ und $\Delta_J := \{ (s,t) \in J^2: s \le t \}$. Außerdem schreiben wir von nun an immer $\Delta$ für $\Delta_I$.
\\

Wir vereinbaren zunächst, was wir unter wohlgestellten Anfangswertproblemen verstehen wollen. Die folgende Definition ist angelehnt an Definition~II.3.2 in~\cite{Krein 71}.
\\

Sei $A(t)$ für jedes $t \in J$ eine lineare Abbildung $D \subset X \to X$. Wir sagen, die Anfangswertprobleme zu $A$ seien \emph{wohlgestellt}, genau dann, wenn gilt: 
\begin{itemize}
\item [(i)] das Anfangswertproblem 
\begin{align*}
y' = A(t)y, \; y(s) = x
\end{align*}
ist eindeutig lösbar auf $[s,b]$ für alle $x \in D$ und alle $s \in [a,b)$. $y(\,. \,,s,x)$ bezeichnet die eindeutige Lösung des obigen Anfangswertproblems für $s \in [a,b)$ und $x \in D$ und $y(b,b,x) := x$ für alle $x \in D$.
\item [(ii)] $\Delta_J \ni (s,t) \mapsto y(t,s,x)$ ist stetig für alle $x \in D$ und $y(t,s,x_n) \longrightarrow 0 \quad (n \to \infty)$ gleichmäßig in $(s,t) \in \Delta_J$ für alle $(x_n)$ in $D$ mit $x_n \longrightarrow 0 \quad (n \to \infty)$.
\end{itemize}

Die Anfangswertprobleme zu $A$ sind also, wie man das erwartet, wohlgestellt genau dann, wenn sie eindeutig lösbar sind und deren Lösung (in einem gewissen Sinne) stetig von den Anfangsdaten abhängt.

Jetzt führen wir den für alles weitere sehr wichtigen Zeitentwicklungsbegriff ein.
\\

Sei $A(t)$ für jedes $t \in J$ eine lineare Abbildung $D \subset X \to X$ und $U(t,s)$ für jedes $(s,t) \in \Delta_J$ eine beschränkte lineare Abbildung in $X$. Dann heißt $U$ \emph{Zeitentwicklung zu $A$} genau dann, wenn gilt: 
\begin{itemize}
\item [(i)] $[s,b] \ni t \mapsto U(t,s)x$ löst das Anfangswertproblem 
\begin{align*}
y' = A(t)y, \; y(s) = x
\end{align*}
für alle $x \in D$ und alle $s \in [a,b)$ und $U(t,s)U(s,r) = U(t,r)$ für alle $(r,s), (s,t) \in \Delta_J$  
\item [(ii)] $\Delta_J \ni (s,t) \mapsto U(t,s)x$ ist stetig für alle $x \in X$.
\end{itemize}
$U$ heißt \emph{Zeitentwicklung} (schlechthin) in $X$ genau dann, wenn ein dichter Unterraum $D$ von $X$ und lineare Abbildungen $A(t): D \subset X \to X$ existieren, sodass $U$ eine Zeitentwicklung zu $A$ ist.

Schließlich schreiben wir auch $U(t)$ für $U(t,0)$, falls $J = [0,b]$ und $U$ eine Zeitentwicklung (auf $J$) ist.
\\

Dieser Zeitentwicklungsbegriff ist sehr natürlich, wie der folgende Satz zeigt (vgl. die Ausführungen in~\cite{Krein 71}, die sich an Definition~II.3.2 anschließen).

\begin{thm} \label{thm: zeitentwicklung natürlich}
Sei $A(t)$ für jedes $t \in J$ eine lineare Abbildung $D \subset X \to X$. Dann sind die Anfangswertprobleme zu $A$ genau dann wohlgestellt, wenn eine Zeitentwicklung zu $A$ existiert. 
\end{thm}

\begin{proof}
Seien die Anfangswertprobleme zu $A$ wohlgestellt. Sei $U_0(t,s)x := y(t,s,x)$ für alle $(s,t) \in \Delta_J$ und alle $x \in D$, wobei die $y(\,. \,,s,x)$ die eindeutigen Lösungen der zu $A$ gehörigen Anfangswertprobleme seien und $y(b,b,x) := x$. Dann ist $U_0(t,s)$ für alle $(s,t) \in \Delta_J$ offensichtlich eine lineare Abbildung $D \to X$, die zudem beschränkt ist. Sei nämlich $(x_n)$ in $D$ mit $x_n \longrightarrow 0 \quad (n \to \infty)$, dann gilt wegen der Wohlgestelltheit $U_0(t,s)x_n  = y(t,s,x_n) \longrightarrow 0 \quad (n \to \infty)$, was die Beschränktheit von $U_0(t,s)$ beweist.
Weiter gilt 
\begin{align*}
\sup_{(s,t) \in \Delta_J} \norm{ U_0(t,s) } < \infty.
\end{align*}
Andernfalls gäbe es nämlich eine Folge $(s_n,t_n)$ in $\Delta_J$ und eine Folge $(x_n)$ in $D$, sodass
\begin{align*}
x_n \longrightarrow 0 \quad (n \to \infty)  \text{ \; und \; } \norm{ U_0(t_n,s_n)x_n } \ge 1 \text{ \; für alle } n \in \natu,
\end{align*}
was der Wohlgestelltheit widerspräche.
Wir können also $U_0(t,s)$ für alle $(s,t) \in \Delta_J$ fortsetzen zu einer beschränkten linearen Abbildung $U(t,s)$ auf ganz $X$ und für diese gilt
\begin{align*}
\sup_{(s,t) \in \Delta_J} \norm{ U(t,s) } = \sup_{(s,t) \in \Delta_J} \norm{ U_0(t,s) } < \infty.
\end{align*}
Daraus folgt, da $(s,t) \mapsto y(t,s,x) = U(t,s)x$ wegen der Wohlgestelltheit stetig ist für alle $x \in D$, dass $(s,t) \mapsto U(t,s)x$ auch für alle $x \in X$ stetig ist. 
Außerdem gilt $U(t,s)U(s,r) = U(t,r)$ für alle $(r,s), (s,t) \in \Delta_J$, denn für beliebiges $(r,s) \in \Delta_J$ und $x \in D$ lösen
die Abbildungen $[s,b] \ni t \mapsto U(t,s)U(s,r)x$ und $[s,b] \ni t \mapsto U(t,r)x$ beide das Anfangswertproblem 
\begin{align*}
y' = A(t)y, \; y(s) = U(s,r)x
\end{align*}
auf $[s,b]$, das aufgrund der Wohlgestelltheit eindeutig lösbar ist auf $[s,b]$.
Also ist $U$ eine Zeitentwicklung zu $A$ ist.
\\

Sei umgekehrt eine Zeitentwicklung $U$ zu $A$ gegeben. Dann ist für jedes $x \in D$ und jedes $s \in [a,b)$ das Anfangswertproblem 
\begin{align*}
y' = A(t)y, \; y(s) = x
\end{align*}
lösbar auf $[s,b]$, und zwar durch $[s,b] \ni t \mapsto U(t,s)x$. 
Wir müssen zeigen, dass dieses Anfangswertproblem auch eindeutig lösbar sind. Sei also $y$ irgendeine Lösung dieses Anfangswertproblems. 
Dann ist für beliebiges $t \in (s,b]$ die Abbildung $[s,t] \ni \tau \mapsto U(t,\tau)y(\tau)$ rechtsseitig differenzierbar und die rechtsseitige Ableitung verschwindet, denn 
\begin{align*}
\frac{U(t, \tau + h)z - U(t, \tau)z}{h} = -U(t, \tau + h) \, \frac{ U(\tau + h, \tau)z - z}{h} \longrightarrow -U(t,\tau) A(\tau)z \quad (h \searrow 0)
\end{align*}
für alle $z \in D$ und alle $\tau \in [s,t)$. Aus Lemma~\ref{lm: mws für einseitig db} folgt damit, dass $[s,t] \ni \tau \mapsto U(t,\tau)y(\tau)$ konstant ist und also 
\begin{align*}
y(t) - U(t,s)x = U(t,\tau)y(\tau) \big|_{\tau = s}^{\tau = t} = 0 
\end{align*}
für alle $t \in [s,b]$. Das Anfangswertproblem ist daher eindeutig lösbar und aufgrund der starken Stetigkeit von $U$ haben wir auch die stetige Abhängigkeit (im Sinne unserer Definition von Wohlgestelltheit). 
Die Anfangsertprobleme zu $A$ sind also tatsächlich wohlgestellt.
\end{proof}

Die folgende Proposition ergibt sich unmittelbar aus der Definition und obigem Satz.

\begin{prop}  \label{prop: ex höchstens eine zeitentw}
Sei $A(t)$ für jedes $t \in J$ eine lineare Abbildung $D \subset X \to X$. Dann existiert höchstens eine Zeitentwicklung zu $A$.
\end{prop}

\begin{proof}
Seien $U$ und $V$ zwei Zeitentwicklungen zu $A$, dann stimmen die Abbildungen $[s,b] \ni t \mapsto U(t,s)x$ und $[s,b] \ni t \mapsto V(t,s)x$
für alle $x \in D$ und alle $s \in [a,b)$ überein, da die Anfangswertprobleme nach Satz~\ref{thm: zeitentwicklung natürlich} eindeutig lösbar sind. Also gilt
\begin{align*}
U(t,s)x = V(t,s)x
\end{align*} 
für alle $x \in D$ und alle $(s,t) \in \Delta_J$ mit $s \in [a,b)$. Wegen der Beschränktheit der $U(t,s)$ und $V(t,s)$ und starken Stetigkeit von $U$ und $V$ gilt dies auch für $x \in X$ und $s = b$.
\end{proof}

Auf die nächste Proposition werden wir immer wieder zurückgreifen: sie erlaubt es zwei Zeitentwicklungen miteinander zu vergleichen.

\begin{prop} \label{thm: char zeitentwicklung}  
Sei $A(t)$ für jedes $t \in J$ eine lineare Abbildung $D \subset X \to X$ und $U(t,s)$ für jedes $(s,t) \in \Delta_J$ eine beschränkte lineare Abbildung in $X$, sodass $[s,b] \ni t \mapsto U(t,s)x$ das Anfangswertproblem 
\begin{align*}
y' = A(t)y, \; y(s) = x
\end{align*}
für alle $x \in D$ und alle $s \in [a,b)$ löst, und $\Delta_J \ni (s,t) \mapsto U(t,s)x$ für alle $x \in X$ stetig ist. Sei darüberhinaus $t \mapsto A(t)x$ stetig für alle $x \in D$.
Dann sind folgende Aussagen äquivalent:
\begin{itemize}
\item [(i)] $U(t,s)U(s,r) = U(t,r)$ für alle $(r,s), (s,t) \in \Delta_J$. 
\item [(ii)] $[a,t] \ni s \mapsto U(t,s)x$ ist differenzierbar und 
\begin{align*}
\dds{ U(t,s)x } = - U(t,s) A(s)x
\end{align*}
für alle $s \in [a,t]$, $x \in D$ und alle $t \in (a,b]$. 
\end{itemize}
\end{prop}

\begin{proof}
Sei (i) erfüllt. Sei $x \in D$, $t \in (a,b]$ und $f(s) := U(t,s)x$ für alle $s \in [a,t]$. Wir verifizieren die Voraussetzungen von Satz~\ref{thm: einseitig db und beidseitig db}. 
Zunächst ist $f$ stetig, denn Zeitentwicklungen sind stark stetig. Weiter ist $f$ rechtsseitig differenzierbar in jedem $s \in [a,t)$ und 
\begin{align*}
\partial_+ f(s) = - U(t,s) A(s)x
\end{align*}
(wobei (i) in der gleichen Weise wie im Beweis von Satz~\ref{thm: zeitentwicklung natürlich} eingeht).
Nun ist $[a,t] \ni s \mapsto - U(t,s) A(s)x$ nach unserer Zusatzvoraussetzung stetig (auch an der Stelle $t$), $f$ ist also rechtsseitig stetig differenzierbar und $\partial_+ f$ ist stetig fortsetzbar in den rechten Randpunkt $t$ von $[a,t]$. Also ist $f$ nach Satz~\ref{thm: einseitig db und beidseitig db} beidseitig stetig differenzierbar und 
\begin{align*}
\dds{ U(t,s)x } = f'(s) = - U(t,s) A(s)x,
\end{align*}
für alle $s \in [a,t]$, was (ii) beweist.
\\

Sei (ii) erfüllt. Sei $x \in D$, $s \in [a,b)$ und sei $y$ irgendeine auf ganz $[s,b]$ definierte Lösung des Anfangswertproblems
\begin{align*}
y' = A(t)y, \; y(s) = x.
\end{align*}
Dann ist $y(t) = U(t,s)x$ für alle $t \in [s, b]$. Sei nämlich $t \in (s,b]$ (für $t = s$ ist nichts zu zeigen). Dann ist $[s,t] \ni \tau \mapsto U(t, \tau) y(\tau)$ nach (ii) und Lemma~\ref{lm: strong db of products} differenzierbar und
\begin{align*}
\ddtau{    U(t, \tau) y(\tau)   } = -U(t,\tau)A(\tau) \; y(\tau) + U(t,\tau) \; A(\tau)y(\tau)  = 0
\end{align*}
für alle $\tau \in [s,t]$. Also gilt tatsächlich
\begin{align*}
y(t) - U(t,s)x = U(t, \tau) y(\tau) \big|_{\tau=s}^{\tau=t} = 0,
\end{align*}
woraus sich aufgrund von Satz~\ref{thm: zeitentwicklung natürlich} die Aussage~(i) ergibt.
\end{proof}

Wir machen darauf aufmerksam, dass wir die Stetigkeit von $t \mapsto A(t)x$ nur gebraucht haben, um zeigen zu können, dass Aussage (i) Aussage (ii) nach sich zieht (genauer: um sogar die in Aussage~(ii) behauptete beidseitige Differenzierbarkeit zu bekommen -- die rechtsseitige Differenzierbarkeit folgt auch ohne die zusätzliche Stetigkeitsvoraussetzung, wie der Beweis von Satz~\ref{thm: zeitentwicklung natürlich} zeigt).
\\

Die folgenden Beispiele befassen sich mit dem Sonderfall, dass $A(t) = A_0$ ist für alle $t \in J$. Sie zeigen, dass eine Zeitentwicklung zu solch einem $A$ genau dann existiert, wenn $A_0$ abschließbar ist und $\overline{A_0}$ eine stark stetige Halbgruppe erzeugt.

\begin{ex}
Sei $A_0$ eine lineare Abbildung $D \subset X \to X$ und $A(t) := A_0$ für alle $t \in J$. Wenn $A_0$ eine stark stetige Halbgruppe $T$ auf $X$ erzeugt, dann existiert eine (und damit genau eine) Zeitentwicklung zu $A$ und diese ist gegeben durch
\begin{align*}
U(t,s) := T(t-s) = e^{A_0 (t-s)}
\end{align*}
für alle $(s,t) \in \Delta_J$. Dass dies wirklich eine Zeitentwicklung zu $A$ ist, folgt z.B. mithilfe von Proposition~\ref{thm: char zeitentwicklung}, denn Aussage~(iii) aus diesem Satz ist hier offensichtlich erfüllt.  $\blacktriangleleft$
\end{ex}

\begin{ex}
Seien $A_0$ und $A(t)$ wie im obigen Beispiel. Wenn eine Zeitentwicklung $U$ zu $A$ existiert, dann ist $A_0$ abschließbar und $\overline{A_0}$ erzeugt eine stark stetige Halbgruppe auf $X$. 
Sei nämlich 
\begin{align*}
T(t) := U(b,a)^m \, U(t - m(b-a) + a, a)
\end{align*}
für alle $t \in \bigl[ m(b-a), (m+1)(b-a) \bigr)$, $m \in \natu \cup \{0\}$. Dann ist $T$ eine stark stetige Halbgruppe auf $X$. (Um zu zeigen, dass tatsächlich $T(t+s) = T(t)T(s)$ für alle $s, t \in [0, \infty)$, zeigt man diese Halbgruppeneigenschaft zunächst für $s, t \in [0, (b-a)]$ und beweist damit den Fall allgemeiner $s,t \in [0, \infty)$.) Sei $A_T$ ihr Erzeuger und $x \in D$. Dann gilt:
\begin{align*}
\frac{T(h)x-x}{h} = \frac{U(h + a, a)x - U(a,a)x}{h} \longrightarrow A(a)U(a,a)x = A_0 x \quad (h \searrow 0),
\end{align*} 
das heißt, $x \in D(A_T)$ und $A_T x = A_0 x$, also $A_0 \subset A_T$. Wegen 
\begin{align*}
T(t) D = U(b,a)^m \, U(t - m(b-a) + a, a) \,D \subset D
\end{align*}
ist $D$ nach Proposition~II.1.7 in~\cite{EngelNagel} sogar ein core für $A_T$, mit anderen Worten: $A_0 = A_T \big|_D$ ist abschließbar und $\overline{A_0} = \overline{A_T \big|_D} = A_T$ erzeugt die stark stetige Halbgruppe $T$.   $\blacktriangleleft$
\end{ex}

Übrigens muss $A_0$ in der Situation des obigen Beispiels nicht sogar abgeschlossen sein. Sei nämlich $D$ ein \emph{echter} (und, wie immer, dichter) Unterraum von $X$ und $A_0 := \operatorname{id}_D$, dann existiert zwar eine Zeitentwicklung $U$ zu $A$, die gegeben ist durch $U(t,s) := e^{t-s}$, aber $A_0$ ist eben nicht abgeschlossen. 
\\

Im folgenden werden wir uns fast ausschließlich auf Aussagen über das Intervall $J := I$ beschränken, insbesondere werden wir unsere Adiabatensätze nur für das Einheitsintervall $I$ formulieren. Adiabatensätze für allgemeine (kompakte nichttriviale) Intervalle $J$ können wir aber mithilfe des nächsten Lemmas aus den Adiabatensätzen für $I$ folgern. In Lemma~\ref{lm: unhandl adsatz für trf intervall} werden wir das exemplarisch für den ersten Adiabatensatz mit Spektrallückenbedingung, Satz~\ref{thm: unhandl adsatz mit sl}, (sozusagen einmal für allemal) ausführen.

\begin{lm} \label{lm: trf. der zeitentw.}
Sei $A(t)$ für jedes $t \in J = [a,b]$ eine lineare Abbildung $D \subset X \to X$ und $U$ eine Zeitentwicklung zu $A$. Sei $B(t) := (b-a) A(a+t(b-a))$ für alle $t \in I$ und $V(t,s) := U(a+t(b-a), a+s(b-a))$ für alle $(s,t) \in \Delta$. Dann ist $V$ eine Zeitentwicklung zu $B$.
\end{lm}

\begin{proof}
Das ist (fast) offensichtlich.
\end{proof}

\subsection{Wann existieren Zeitentwicklungen?}

Wir haben in Satz~\ref{thm: zeitentwicklung natürlich} gesehen, dass die Anfangswertprobleme zu vorgegebenem $A$ genau dann wohlgestellt sind, wenn eine (und damit genau eine) Zeitentwicklung zu $A$ existiert. Aber wann, unter welchen (leicht verifizierbaren) Voraussetzungen an $A$ existiert denn eine Zeitentwicklung zu $A$? Satz~\ref{thm: Dyson} und der sehr viel tiefer liegende Satz~\ref{thm: Kato} geben solche Voraussetzungen an -- der eine für den Fall beschränkter $A(t)$, der andere für den Fall nicht notwendig beschränkter $A(t)$. 
\\

Zunächst führen wir einen Stabilitätsbegriff ein, der auf Kato~\cite{Kato 70} zurückgeht. 
\\

Sei $J$ ein Intervall und $A(t)$ für jedes $t \in J$ Erzeuger einer stark stetigen Halbgruppe auf $X$ (dessen Definitionsbereich durchaus von $t$ abhängen darf), $M \in [1, \infty)$ und $\omega \in \real$. Dann heißt $A$ $(M,\omega)$\emph{-stabil} genau dann, wenn 
\begin{align*}
\norm{   e^{A(t_n) s_n}  \, \dotsm \, e^{A(t_1) s_1}   } \le M e^{\omega (s_n + \, \dotsb \, + s_1) } 
\end{align*}
für alle $s_1, \dots, s_n \in [0,\infty)$, alle $t_1, \dots, t_n \in J$ mit $t_1 \le \dotsb \le t_n$ und alle $n \in \natu$.
\\

Wir sehen sofort: wenn die Familie $A = (A(t))_{t \in I}$ von Halbgruppenerzeugern $(M, \omega)$-stabil ist, dann erzeugt jedes $A(t)$ eine stark stetige Halbgruppe 
mit $\norm{  e^{A(t) s} } \le M e^{\omega s}$ für alle $s \in [0,\infty)$. Die Umkehrung dieser Aussage gilt beispielsweise im Fall $M = 1$, das heißt im Fall, dass die $A(t)$ allesamt Quasikontraktionshalbgruppen mit einer quasikontraktiven Wachstumsschranke kleiner oder gleich $\omega$ erzeugen. 

Wenn die $A(t)$ allesamt normale lineare Abbildungen sind, dann fallen $(M,\omega)$- und $(1, \omega)$-Stabilität zusammen, wie folgendes Beispiel zeigt.

\begin{ex} \label{ex: (M,w)-stabilität = (1,w)-stabilität für normale A}
Sei $A(t)$ für jedes $t \in I$ eine \emph{normale} lineare Abbildung $D \subset H \to H$. Dann ist $A$ $(M, \omega)$-stabil für irgendein $M \in [1, \infty)$ genau dann, wenn $A$ $(1, \omega)$-stabil ist. Sei nämlich $A$  $(M, \omega)$-stabil. Dann gilt für alle $t \in I$:
\begin{align*}
\norm{  e^{A(t) s} } \le M e^{\omega s}
\end{align*}
für alle $s \in [0, \infty)$, das heißt nach von Satz~\ref{thm: eigenschaften von erzeugern}, dass $\sigma(A(t)) \subset \{z \in \complex: \Re \, z \le \omega \}$ und damit 
\begin{align*}
\norm{ e^{A(t) s} x}^2 = \norm{   \int_{\sigma(A(t))} e^{z \, s} \,dP^{A(t)}(z) \,x   }^2 = \int_{\sigma(A(t))}   \big| e^{z \, s}  \big|^2   \,dP_{x,x}^{A(t)}(z) \le \Bigl( e^{\omega s} \norm{x} \Bigr)^2
\end{align*}
für alle $x \in H$ und alle $s \in [0, \infty)$. Also haben wir sogar
\begin{align*}
\norm{  e^{A(t) s} } \le e^{\omega s}
\end{align*}
für alle $s \in [0, \infty)$, woraus, wie eben bemerkt, die $(1, \omega)$-Stabilität von $A$ folgt.  $\blacktriangleleft$
\end{ex} 

Im allgemeinen fallen $(M, \omega)$- und $(1,\omega)$-Stabilität aber (natürlich) nicht zusammen. Sei beisbielsweise $A_0$ der Erzeuger der Translationshalbgruppe mit Sprung aus Beispiel~I.5.7~(iii) in~\cite{EngelNagel} und $A(t) = A_0$ für alle $t \in I$. Dann erzeugt $A_0$ keine Quasikontraktionshalbgruppe, weshalb $A$ auch nicht $(1,\omega)$-stabil ist, aber $A$ ist wegen 
\begin{align*}
e^{A(t_n) s_n} \, \dotsm \, e^{A(t_1) s_1} = e^{A_0 (s_n + \, \dotsb \, + s_1)}
\end{align*}
nach Beispiel~I.5.7~(iii) $(M,0)$-stabil für alle $M \in [2, \infty)$.    
\\

Die folgende Proposition (Proposition~3.5 aus~\cite{Kato 70}) ist eine Verallgemeinerung des im vorbereitenden Abschnitt erwähnten Störungssatzes, Satz~\ref{thm: störungssatz halbgruppenerz}, für Halbgruppenerzeuger.

\begin{prop} \label{prop: störung (M,w)-stabilität}
Sei $A(t)$ für jedes $t \in I$ Erzeuger einer stark stetigen Halbgruppe auf $X$, $A$ $(M,\omega)$-stabil und $B(t)$ für jedes $t \in I$ eine beschränkte lineare Abbildung in $X$, für die $b := \sup_{t \in I} \norm{B(t)} < \infty$. Dann ist $A+B$ $(M, \omega + M b)$-stabil. 
\end{prop}

Ein Beweis dieser Aussage (und auch eine Charakterisierung von $(M, \omega)$-Stabilität) findet sich z.B. in~\cite{Nickel 00}.
\\

Der nächste sehr einfache Satz (vgl. Theorem~X.69 in~\cite{RS 2}) besagt, dass im Fall beschränkter $A(t)$ eine Zeitentwicklung zu $A$ existiert, wenn $t \mapsto A(t)x$ stetig ist für alle $x \in X$ -- dies folgt größtenteils auch schon mithilfe des Satzes von Picard, Lindelöf. Außerdem liefert er eine Reihendarstellung für die Zeitentwicklung und eine Abschätzung der Zeitentwicklung (in deren Beweis die grundlegende Idee von Lemma~\ref{lm: kato, evolution 2} schon anklingt).  

\begin{thm} \label{thm: Dyson}
Sei $A(t)$ für jedes $t \in I$ eine beschränkte lineare Abbildung in $X$ und $t \mapsto A(t)x$ stetig für alle $x \in X$. Dann existiert genau eine Zeitentwicklung $U$ zu $A$ und diese ist gegeben durch
\begin{align*}
U(t,s)x = x + \int_s^t A(&t_1)x \,dt_1 +  \int_s^t \int_s^{t_1} A(t_1)A(t_2)x \, dt_2 \, dt_1 \\ 
&+ \int_s^t \, \int_s^{t_1} \! \int_s^{t_2} A(t_1)A(t_2)A(t_3)x \, dt_3 \, dt_2 \, dt_1 + \dotsb,
\end{align*}
für alle $(s,t) \in \Delta$ und alle $x \in X$, die Dysonreihe zu $A$. \\
Sei zusätzlich $A$ $(M, \omega)$-stabil, $B(t)$ für jedes $t \in I$ eine beschränkte lineare Abbildung in $X$ und $t \mapsto B(t)x$ stetig für alle $x \in X$. Dann gilt für die (durch Dysonreihen gegebenen) Zeitentwicklungen $U$ und $V$ zu $A$ bzw. $A+B$:
\begin{align*}
\norm{U(t,s)} \le M \, e^{ \omega (t-s) } \text{\; und \;} \norm{V(t,s)} \le M \, e^{ (\omega + M b) (t-s) } \text{ \; für alle } (s,t) \in \Delta.
\end{align*}
\end{thm}

\begin{proof}
Sei
\begin{align*}
U(t,s)x = x + \int_s^t A(&t_1)x \,dt_1 +  \int_s^t \int_s^{t_1} A(t_1)A(t_2)x \, dt_2 \, dt_1 \\ 
&+ \int_s^t \, \int_s^{t_1} \! \int_s^{t_2} A(t_1)A(t_2)A(t_3)x \, dt_3 \, dt_2 \, dt_1 + \dotsb,
\end{align*}
für alle $(s,t) \in \Delta$ und alle $x \in X$. Der limes rechts existiert gleichmäßig in $(s,t) \in \Delta$, weil
\begin{align*}
&\norm{   \int_s^t \, \int_s^{t_1} \dotsi \int_s^{t_{n-1}} A(t_1)A(t_2)\, \dotsm \, A(t_n)x \, dt_n \, \dots \, dt_2 \, dt_1    } \\
& \qquad \qquad \le c^n \; \int_s^t \, \int_s^{t_1} \dotsi \int_s^{t_{n-1}} 1 \; \; dt_n \, \dots \, dt_2 \, dt_1 \; \norm{x} \\
& \qquad \qquad = c^n \, \frac{(t-s)^n}{n!} \, \norm{x}  \le \frac{c^n}{n!} \, \norm{x}
\end{align*}
für alle $n \in \natu$ und alle $(s,t) \in \Delta$, wobei $c := \sup_{t \in I} \norm{A(t)}$. Wegen der gleichmäßigen Konvergenz ist $\Delta \ni (s,t) \mapsto U(t,s)x$ stetig für alle $x \in X$. Weiter ist 
\begin{align*}
I \ni t \mapsto \int_s^t \, \int_s^{t_1} \dotsi \int_s^{t_{n-1}} A(t_1)A(t_2)\, \dotsm \, A(t_n)x \, dt_n \, \dots \, dt_2 \, dt_1
\end{align*}
differenzierbar mit
\begin{align*}
& \ddt{  \biggl( \int_s^t \, \int_s^{t_1} \dotsi \int_s^{t_{n-1}} A(t_1)A(t_2)\, \dotsm \, A(t_n)x \, dt_n \, \dots \, dt_2 \, dt_1   \biggr) } \\
& \qquad \qquad = A(t) \, \int_s^{t} \dotsi \int_s^{t_{n-1}} A(t_2)\, \dotsm \, A(t_n)x \, dt_n \, \dots \, dt_2 \\
& \qquad \qquad = A(t) \, \int_s^{t} \dotsi \int_s^{t_{n-2}} A(t_1)\, \dotsm \, A(t_{n-1})x \, dt_{n-1} \, \dots \, dt_1 
\end{align*}
für alle $s,t \in I$, und auch 
\begin{align*}
I \ni s \mapsto & \int_s^t \, \int_s^{t_1} \dotsi \int_s^{t_{n-1}} A(t_1)A(t_2)\, \dotsm \, A(t_n)x \, dt_n \, \dots \, dt_2 \, dt_1  \\
& = \int_s^t \, \int_{t_n}^t \dotsi \int_{t_2}^t A(t_1) \, \dotsm \, A(t_{n-1}) A(t_n)x \, dt_1 \, \dots \, dt_{n-1} \, dt_n
\end{align*}
ist differenzierbar mit
\begin{align*}
& \dds{  \biggl( \int_s^t \, \int_s^{t_1} \dotsi \int_s^{t_{n-1}} A(t_1)A(t_2)\, \dotsm \, A(t_n)x \, dt_n \, \dots \, dt_2 \, dt_1  \biggr)  } \\
& \qquad \qquad = - \int_{s}^t \dotsi \int_{t_2}^t A(t_1) \, \dotsm \, A(t_{n-1}) A(s)x \, dt_1 \, \dots \, dt_{n-1} \\
& \qquad \qquad = - \int_s^{t} \dotsi \int_s^{t_{n-2}} A(t_1)\, \dotsm \, A(t_{n-1}) A(s) x \, dt_{n-1} \, \dots  \, dt_1
\end{align*}
für alle $s,t \in I$ und alle $x \in X$. Wir sehen nun (da auch die Reihen der Ableitung nach $t$ bzw. $s$ gleichmäßig konvergieren) mit Proposition~\ref{thm: char zeitentwicklung}, dass $U$ eine Zeitentwicklung zu $A$ ist.  
\\

Wir müssen nun noch die Abschätzung der (nach dem bis jetzt Bewiesenen auch wirklich existierenden) Zeitentwicklungen $U$ und $V$ herleiten. Dabei genügt es, dies für $U$, die ungestörte Zeitentwicklung, zu tun: die allgemeinere Abschätzung für die gestörte Zeitentwicklung $V$ folgt dann mit Proposition~\ref{prop: störung (M,w)-stabilität}.

Sei für jedes $k \in \natu$
\begin{align*}
U_k(t,s) := e^{A \left( \frac{\lfloor ks \rfloor} {k} \right) (t-s) }
\end{align*}
für alle $(s,t) \in \Delta$, die in ein und demselben (abgeschlossenen) Intervall der $\frac{1}{k}$-Zerlegung von $I$ liegen, und
\begin{align*}
U_k(t,s) := e^{A \left( \frac{\lfloor kt \rfloor} {k} \right) (t-\frac{\lfloor kt \rfloor} {k}) }\; 
e^{A \left( \frac{\lfloor kt \rfloor -1} {k} \right) \frac{1}{k} }\; \dotsm \; e^{A \left( \frac{\lfloor ks \rfloor} {k} \right) (\frac{\lfloor ks \rfloor + 1} {k} - s) }
\end{align*}  
für alle $(s,t) \in \Delta$, die nicht in ein und demselben (abgeschlossenen) Intervall der $\frac{1}{k}$-Zerlegung von $I$ liegen. 
Dann gilt
\begin{align*}
\norm{ U_k(t,s) } \le M \, e^{ \omega (t-s) }
\end{align*}
für alle $(s,t) \in \Delta$, da $A$ ja $(M, \omega)$-stabil ist. 
Wir zeigen nun, dass 
\begin{align*}
U_k(t,s)x \longrightarrow U(t,s)x \quad (k \to \infty)
\end{align*}
für alle $(s,t) \in \Delta$ und alle $x \in X$. Dann folgt die behauptete Abschätzung für $U$.
Sei also $x \in X$ und $(s,t) \in \Delta$ mit $s \ne t$ (für $s = t$ ist nichts zu zeigen). Sei weiter $k \in \natu$ und $\bigl(t_n \bigr)_{ n \in \{ 0, \dots , m \} }$ eine Zerlegung des Intervalls $[s,t]$ (das heißt, $s = t_0 < t_1 < \dots < t_m = t$), sodass $t_1, \dots, t_{m-1}$ genau die in $(s,t)$ enthaltenen Zerlegungsstellen der $\frac{1}{k}$-Zerlegung von $I$ sind. Dann ist
\begin{align*}
(t_{n-1},t_n) \ni \tau \mapsto U_k(t, \tau) U(\tau,s)x 
\end{align*}
nach Lemma~\ref{lm: strong db of products} differenzierbar und 
\begin{align*}
\ddtau{ \bigl(  U_k(t, \tau) U(\tau,s)x  \bigr)   } = U_k(t,\tau) \biggl(   A(\tau) - A\Bigl( \frac{\lfloor k \tau \rfloor} {k} \Bigr)   \biggr) U(\tau,s)x
\end{align*}
für alle $\tau \in (t_{n-1},t_n)$ und alle $n \in \{1, \dots, m \}$, das heißt
\begin{align*}
(t_{n-1},t_n) \ni \tau \mapsto \ddtau{ \bigl(  U_k(t, \tau) U(\tau,s)x  \bigr)   }
\end{align*}
ist stetig und beschränkt durch $M e^{\omega} \; 2c \; \sup_{\tau \in [s,1]} \norm{U(\tau,s)} \, \norm{x}$. 
Also gilt
\begin{align*}
U(t,s)x - U_k(t,s)x &= U_k(t, \tau) U(\tau,s)x \big|_{\tau=s}^{\tau=t} \\
&= \int_s^t U_k(t,\tau) \biggl(   A(\tau) - A\Bigl( \frac{\lfloor k \tau \rfloor} {k} \Bigr)   \biggr) U(\tau,s)x \, d\tau
\end{align*}
für alle $k \in \natu$, woraus wir mit dem lebesgueschen Satz schließlich
\begin{align*}
U(t,s)x &- U_k(t,s)x \\
&\longrightarrow \int_s^t \lim_{k \to \infty} U_k(t,\tau) \biggl(   A(\tau) - A\Bigl( \frac{\lfloor k \tau \rfloor} {k} \Bigr)   \biggr) U(\tau,s)x \, d\tau  = 0 \quad(k \to \infty)
\end{align*}
erhalten.
\end{proof}

Im Sonderfall vertauschender $A(t)$ vereinfacht sich der Ausdruck für die Zeitentwicklung aus obigem Satz stark.

\begin{cor} \label{cor: Dyson}
Sei $A(t)$ für jedes $t \in I$ eine beschränkte lineare Abbildung, $t \mapsto A(t)x$ stetig für alle $x \in X$ und die $A(t)$ mögen paarweise vertauschen, das heißt, $A(t')A(t) = A(t)A(t')$ für alle $t, t' \in I$. Dann gilt für die Zeitentwicklung $U$ zu $A$:
\begin{align*}
U(t,s) = e^{\int_s^t A(\tau) \,d\tau}
\end{align*}
für alle $(s,t) \in \Delta$, wobei das Integral natürlich im starken Sinne zu verstehen ist.
\end{cor}

\begin{proof}
Sei $x \in X$, $n \in \natu$ und 
\begin{align*}
f(t_1, t_2, \dots, t_n) := A(t_1) A(t_2) \dotsm A(t_n)x
\end{align*}
für alle $(t_1, t_2, \dots, t_n) \in I^n$.  
Sei weiter $(s,t) \in \Delta$ und 
\begin{align*}
E_{\sigma} := \{ (t_1, t_2, \dots, t_n) \in I^n: t \ge t_{\sigma(1)} \ge t_{\sigma(2)} \ge \dots \ge t_{\sigma(n)} \ge s \}
\end{align*}
für alle $\sigma \in S_n$ (Menge der Permutationen von $\{1, \dots, n\}$) sowie $E := E_{\operatorname{id}}$. 
Dann gilt wegen der Vertauschbarkeit der $A(t)$, dass
\begin{align*}
f(t_1, \dots, t_n) = f(t_{\sigma(1)}, \dots, t_{\sigma(n)} )
\end{align*}
für alle $(t_1, \dots, t_n) \in I^n$ und alle $\sigma \in S_n$, und darüberhinaus gilt 
\begin{align*}
E_{\sigma} = \tilde{\sigma}^{-1}(E)
\end{align*}
für alle $\sigma \in S_n$, wobei $\tilde{\sigma}$ die lineare Abbildung $\real^n \to \real^n$ ist, die $(t_1, \dots, t_n)$ auf $(t_{\sigma(1)}, \dots, t_{\sigma(n)})$ schickt (Permutation der Koordinaten). Also gilt
\begin{align*}
&\int_{E_{\sigma}} f(t_1, \dots, t_n) \, d(t_1, \dots, t_n) =  \int_{E_{\sigma}} f(t_{\sigma(1)}, \dots, t_{\sigma(n)}) \, d(t_1, \dots, t_n) \\
& \qquad \quad = \int_{\tilde{\sigma}^{-1}(E)} f\bigl( \tilde{\sigma}(t_1, \dots, t_n) \bigr) \, d(t_1, \dots, t_n)
= \int_{E} f(t_1, \dots, t_n) \, d(t_1, \dots, t_n),  
\end{align*}
denn $|\det \tilde{\sigma}|$ ist ja gleich $1$. 
Wir erhalten damit
\begin{align*}
&\int_s^t \, \int_s^{t_1} \dotsi \int_s^{t_{n-1}} A(t_1)A(t_2)\, \dotsm \, A(t_n)x \, dt_n \, \dots \, dt_2 \, dt_1 \\
& \qquad \quad = \int_{E} f(t_1, \dots, t_n) \, d(t_1, \dots, t_n) = \frac{1}{n!} \, \sum_{\sigma \in S_n}  \int_{E_{\sigma}} f(t_1, \dots, t_n) \, d(t_1, \dots, t_n) \\
& \qquad \quad = \frac{1}{n!} \, \int_{[s,t]^n} f(t_1, \dots, t_n) \, d(t_1, \dots, t_n) = \frac{1}{n!} \, \Bigl( \int_{[s,t]} A(\tau) \, d\tau \Bigr)^n x,
\end{align*}
weil die Vereinigung der $E_{\sigma}$ gleich $[s,t]^n$ ist und sich die $E_{\sigma}$ nur in Nullmengen schneiden. 

Aus der Dysonreihenreihendarstellung der Zeitentwicklung $U$ zu $A$ aus dem obigen Satz wird nun   
\begin{align*}
U(t,s)x = \sum_{n=0}^{\infty} \frac{1}{n!} \, \Bigl( \int_{[s,t]} A(\tau) \, d\tau \Bigr)^n x = e^{ \int_s^t A(\tau) \, d\tau} \, x,
\end{align*}
wie gewünscht.
\end{proof}

Jetzt kommen wir zu Zeitentwicklungen zu Familien unbeschränkter $A(t)$.

\begin{lm} \label{lm: kato, evolution 1}
Sei $A(t): D \subset X \to X$ für jedes $t \in I$ eine bijektive abgeschlossene lineare Abbildung und 
\begin{align*}
\{ (s',t') \in I^2: s' \ne t' \} \ni (s,t) \mapsto \frac{1}{t-s} \, C(t,s)x
\end{align*}
sei beschränkt für alle $x \in X$, wobei 
\begin{align*}
C(t,s) := A(t) A(s)^{-1} - 1
\end{align*}
für alle $(s,t) \in I^2$. 
Dann ist $t \mapsto A(t)x$ stetig für alle $x \in D$ und $t \mapsto C(t,t_0)x$ und $t \mapsto C(t_0,t)x$ sind stetig für alle $x \in X$ und $t_0 \in I$.
\end{lm}

\begin{proof}
Sei $x \in D$ und $t \in I$. Dann ist 
\begin{align*}
A(t+h)x - A(t)x = C(t+h,t)A(t)x
\end{align*}
und daher
\begin{align*}
\norm{A(t+h)x - A(t)x} \le c \norm{A(t)x} |h|
\end{align*}
für alle $h \in \real$ mit $t + h \in I$, wobei 
\begin{align*}
c := \sup_{(s,t) \in \{ s' \ne t' \} } \norm{ \frac{1}{t-s} \,C(t,s) },
\end{align*}
was (nach unserer Voraussetzung und dem Satz von Banach, Steinhaus) eine reelle Zahl ist. Also ist $t' \mapsto A(t')x$ stetig in $t$.
\\

Sei $x \in X$, $t_0 \in I$ und $t \in I$. Dann folgt mit der eben bewiesenen Stetigkeitsaussage
\begin{align*}
C(t+h,t_0)x - C(t,t_0)x = \bigl( A(t+h) - A(t) \bigr) \; A(t_0)^{-1} x \longrightarrow 0 \quad (h \to 0)
\end{align*}
und wegen $\sup_{(s,t) \in I^2 } \norm{ C(t,s) } \le c$
\begin{align*}
C(t_0, t+h)x - C(t_0,t)x &= A(t_0)A(t+h)^{-1} \bigl( A(t) - A(t+h) \bigr) A(t)^{-1}x \\
&= \bigl(C(t_0, t+h) + 1 \bigr) \bigl( A(t) - A(t+h) \bigr) A(t)^{-1}x \\
&\longrightarrow 0 \quad (h \to 0),
\end{align*} 
wie gewünscht.
\end{proof}

Das nächste Lemma ist eine geringfügige Abwandlung von Theorem~XIV.4.1 aus Yosidas Buch~\cite{Yosida: FA}, das im wesentlichen auf Kato zurückgeht (Theorem~4 in~\cite{Kato 53}). Worin unterscheidet sich unser Lemma von Theorem~XIV.4.1? Wir setzen nur $(M,0)$-Stabiltät voraus (anstelle von $(1,0)$-Stabilität)
und wir ersetzen die Voraussetzung, dass $(s,t) \mapsto \frac{1}{t-s} \, C(t,s)x$ gleichmäßig stetig ist, durch die (jedenfalls in der Situation von Satz~\ref{thm: Kato}) etwas leichter nachprüfbare Voraussetzung, dass $\lim_{t \searrow 0} C(t)x$ existiert. Am Beweis ändert sich dadurch allerdings kaum etwas, weshalb wir auch nur \emph{die} Argumente genauer ausführen, die abgewandelt werden müssen. Wir halten fest, dass Lemma~\ref{lm: kato, evolution 2} im Fall $M =1$ äquivalent ist zu Yosidas Theorem~XIV.4.1. Dies folgt aus Proposition~\ref{prop: zshg. regvor. kato} und den sich daran anschließenden Ausführungen.

Allgemeinere Aussagen als die von Lemma~\ref{lm: kato, evolution 2} finden sich in Katos Arbeiten~\cite{Kato 70} (Theorem~6.1) und~\cite{Kato 73} (Theorem~1) und Kobayasis Arbeit~\cite{Kobayasi79} (in der das eben genannte Theorem~6.1 Katos noch einmal geringfügig verallgemeinert wird). 
Die größere Allgemeinheit des Satzes in der Arbeit Kobayasis sieht man beispielsweise daran, dass dieser unseren Satz~\ref{thm: Dyson} enthält, 
 was für Theorem~XIV.4.1 nicht gilt: nach den Ausführungen im Anschluss an Proposition~\ref{prop: zshg. regvor. kato} ist dieses Theorem nämlich auf bloß  stark stetige (aber nicht auch stark stetig differenzierbare) $A$ nicht anwendbar.

\begin{lm} \label{lm: kato, evolution 2}
Sei $A(t): D \subset X \to X$ für jedes $t \in I$ eine bijektive lineare Abbildung, die eine stark stetige Halbgruppe auf $X$ erzeugt, sei $A$ $(M,0)$-stabil und für alle $x \in X$ gelte:\\
(i) \begin{align*}
\{ (s',t') \in I^2: s' \ne t' \} \ni (s,t) \mapsto \frac{1}{t-s} \, C(t,s)x
\end{align*}
ist beschränkt ($C(t,s)$ wie in Lemma~\ref{lm: kato, evolution 1}) \\
(ii) \begin{align*}
C(t)x := \lim_{h \searrow 0} \frac{1}{h} C(t, t-h)x
\end{align*}
existiert gleichmäßig in $t \in (0,1]$, das heißt, $C(t)x := \lim_{h \searrow 0} \frac{1}{h} C(t, t-h)x$ existiert für alle $t \in (0,1]$ und 
\begin{align*}
\sup_{t \in [h,1]} \norm{ \frac{1}{h} C(t, t-h)x - C(t)x } \longrightarrow 0 \quad (h \searrow 0)
\end{align*}
(iii) $(0,1] \ni t \mapsto C(t)x$ ist stetig und $C(0)x := \lim_{t \searrow 0} C(t)x$ existiert. \\
Dann existiert genau eine Zeitentwicklung $U$ zu $A$ und für diese gilt:
\begin{align*}
\norm{U(t,s)} \le M \text{ \; für alle } (s,t) \in \Delta.
\end{align*}
\end{lm}

\begin{proof}
Sei für jedes $k \in \natu$
\begin{align*}
U_k(t,s) := e^{A \left( \frac{\lfloor ks \rfloor} {k} \right) (t-s) }
\end{align*}
für alle $(s,t) \in \Delta$, die in ein und demselben (abgeschlossenen) Intervall der $\frac{1}{k}$-Zerlegung von $I$ liegen, und
\begin{align*}
U_k(t,s) := e^{A \left( \frac{\lfloor kt \rfloor} {k} \right) (t-\frac{\lfloor kt \rfloor} {k}) }\; 
e^{A \left( \frac{\lfloor kt \rfloor -1} {k} \right) \frac{1}{k} }\; \dotsm \; e^{A \left( \frac{\lfloor ks \rfloor} {k} \right) (\frac{\lfloor ks \rfloor + 1} {k} - s) }
\end{align*}  
für alle $(s,t) \in \Delta$, die nicht in ein und demselben (abgeschlossenen) Intervall der $\frac{1}{k}$-Zerlegung von $I$ liegen.
Sei weiter
\begin{align*}
W_k(t,s) := A(t) U_k(t,s) A(s)^{-1}
\end{align*} 
für alle $(s,t) \in \Delta$ und alle $k \in \natu$.

Wir sehen dann sofort, dass
\begin{align*}
U_k(t,s) U_k(s,r) = U_k(t,r)
\end{align*}
und
\begin{align*}
\norm{U_k(t,s)} \le M
\end{align*}
für alle $(r,s),(s,t) \in \Delta$, dass $\Delta \ni (s,t) \mapsto U_k(t,s)x$ stetig ist für alle $x \in X$, und dass die $W_k(t,s)$ beschränkte lineare Abbildungen in $X$ sind.
Weiter sehen wir (vgl. Yosidas Beweis), dass 
\begin{align} \label{eq: Kato 1}
W_k(t,s) = \biggl( 1 + C\Bigl(t, \frac{\lfloor kt \rfloor} {k} \Bigr) \biggr)   \Bigl( W_k^{(0)} (t,s) + \dotsb + & W_k^{(m_k(t,s))} (t,s) \Bigr) \notag \\
& \qquad  \biggl( 1 + C\Bigl( \frac{\lfloor ks \rfloor} {k}, s \Bigr) \biggr) 
\end{align}
für alle $(s,t) \in \Delta$ und alle $k \in \natu$, wobei $W_k^{(0)} (t,s) := U_k(t,s)$ und 
\begin{align*}
W_k^{(m)} (t,s) := \sum_{l_1 = \lfloor ks \rfloor + m}^{\lfloor kt \rfloor} \; & \sum_{l_2 = \lfloor ks \rfloor + m-1}^{l_1 -1} \dotsi \sum_{l_m = \lfloor ks \rfloor + 1}^{l_{m-1} -1} U_k \Bigl( t, \frac{l_1}{k} \Bigr) C_{ \frac{l_1}{k}, \frac{l_1 -1}{k} } \; \cdot \\
& \cdot \, U_k \Bigl( \frac{l_1}{k}, \frac{l_2}{k} \Bigr) C_{ \frac{l_2}{k}, \frac{l_2 -1}{k} } \; \dotsb \;
U_k \Bigl( \frac{l_{m-1}}{k}, \frac{l_m}{k} \Bigr) C_{ \frac{l_m}{k}, \frac{l_m -1}{k} }  U_k \Bigl( \frac{l_{m}}{k}, s \Bigr),
\end{align*}
\begin{align*}
C_{ \frac{l_i}{k}, \frac{l_i -1}{k} } := C\Bigl( \frac{l_i}{k}, \frac{l_i -1}{k} \Bigr) 
\end{align*}
für alle $m \in \natu$ und
\begin{align*}
m_k(t,s) := \max \{ k \in \natu: \lfloor ks \rfloor + m \le \lfloor kt \rfloor \} \in \natu \cup \{ 0 \},
\end{align*}
die Anzahl der Zerlegungspunkte der $\frac{1}{k}$-Zerlegung von $I$, die im Intervall $(s,t]$ liegen -- das heißt, $m_k(t,s) +1$ ist die Anzahl der in $U_k(t,s)$ vorkommenden Faktoren. Die $W_k^{(m)} (t,s)$ definierende Summe ist leer (und damit gleich $0$) genau dann, wenn in $(s,t]$ weniger als $m$ Zerlegungspunkte der $\frac{1}{k}$-Zerlegung liegen -- anders gesagt, wenn in $U_k(t,s)$ weniger als $m + 1$ Faktoren auftreten.

Sei
\begin{align*}
c := \sup_{(s,t) \in \{ s' \ne t' \} } \norm{ \frac{1}{t-s} \,C(t,s) },
\end{align*}
was nach Voraussetzung eine reelle Zahl ist. Dann gilt
\begin{align*} 
\norm{W_k^{(m)}(t,s)} \le &\: M^{m+1} \, \Bigl( \frac{c}{k} \Bigr)^m   \sum_{l_1 = \lfloor ks \rfloor + m}^{\lfloor kt \rfloor} \;  \sum_{l_2 = \lfloor ks \rfloor + m-1}^{l_1 -1} \dotsi \sum_{l_m = \lfloor ks \rfloor + 1}^{l_{m-1} -1}   1  \\
\le &\: M \frac{(Mc)^m}{k^m} \; \frac{ \bigl( \lfloor kt \rfloor - (\lfloor ks \rfloor +1) \bigr)^m }{m!}  \\ 
\le &\: M \frac{(Mc \, (t-s))^m}{m!}
\end{align*}
für alle $m \in \natu$ und damit
\begin{align}  \label{eq: Kato 2}
\norm{W_k(t,s)} &\le \Bigl( 1 + \frac{c}{k} \Bigr)  \biggl(  \sum_{m=0}^{m_k(t,s)}  M \frac{(Mc \, (t-s))^m}{m!}  \biggr)  \Bigl( 1 + \frac{c}{k} \Bigr) \notag \\
&\le M \Bigl( 1 + \frac{c}{k} \Bigr)^2 \, e^{Mc \, (t-s)}
\end{align}
für alle $(s,t) \in \Delta $, $k \in \natu $.

Wir können jetzt zeigen, dass $\bigl( U_k(t,s)x \bigr)$ für alle $x \in X$ eine Cauchyfolge ist, und zwar gleichmäßig in $(s,t) \in  \Delta$.

Sei zunächst $x \in D$ und $(s,t) \in \Delta$ mit $s \ne t$. Seien $k$, $l$ natürliche Zahlen und $\bigl(t_n \bigr)_{ n \in \{ 0, \dots , m \} }$ eine Zerlegung des Intervalls $[s,t]$ (das heißt, $s = t_0 < t_1 < \dots < t_m = t$), sodass $t_1, \dots, t_{m-1}$ genau die in $(s,t)$ enthaltenen Zerlegungsstellen der $\frac{1}{k}$-Zerlegung und der $\frac{1}{l}$-Zerlegung von $I$ sind. Dann ist 
\begin{align*}
(t_{n-1},t_n) \ni \tau \mapsto U_l(t, \tau)U_k(\tau, s)x
\end{align*}
(nach Lemma~\ref{lm: kato, evolution 1} sowie Proposition~\ref{thm: char zeitentwicklung} und nach Lemma~\ref{lm: strong db of products}) differenzierbar und
\begin{align*}
\ddtau{ \bigl( U_l(t, \tau)U_k(\tau, s)x \bigr) } = U_l(t,\tau) \biggl(    C\Bigl( \frac{\lfloor k \tau \rfloor}{k}, \tau \Bigr) - C\Bigl( \frac{\lfloor l \tau \rfloor}{l}, \tau \Bigr)   \biggr)      W_k(\tau,s)      \bigl( 1+C(s,0) \bigr)      A(0)x
\end{align*}
für alle $\tau \in (t_{n-1}, t_n)$ und alle $n \in \{1, \dots, m \}$. 
Weil die Abbildungen $(t_{n-1},t_n) \ni \tau \mapsto \lfloor k \tau \rfloor$, $\lfloor l \tau \rfloor$, $m_k(\tau,s)$ konstant sind, ist 
\begin{align*}
(t_{n-1},t_n) \ni \tau \mapsto W_k(\tau,s)y
\end{align*}
wegen~\eqref{eq: Kato 1} und Lemma~\ref{lm: kato, evolution 1} stetig für alle $y \in X$ und damit auch
\begin{align*}
(t_{n-1},t_n) \ni \tau \mapsto \ddtau{ \bigl( U_l(t, \tau)U_k(\tau, s)x \bigr) }
\end{align*}
stetig (erneut wegen Lemma~\ref{lm: kato, evolution 1}). Außerdem haben wir aufgrund von~\eqref{eq: Kato 2}
\begin{align*}
\norm{  \ddtau{ \bigl( U_l(t, \tau)U_k(\tau, s)x \bigr) }   }   \le   M \, \Bigl( \frac{c}{k} + \frac{c}{l} \Bigr) \, M (1 + c)^2 \, e^{Mc} \: (1 + c) \norm{A(0)x}
\end{align*}
für alle $\tau \in (t_{n-1}, t_n)$ und alle $n \in \{1, \dots, m \}$. Also gilt
\begin{align*}
U_l(t, t_n)U_k(t_n,s)x &- U_l(t,t_{n-1})U_k(t_{n-1},s)x \\
&= \lim_{\varepsilon \searrow 0} U_l(t, t_n - \varepsilon)U_k(t_n - \varepsilon,s)x - U_l(t,t_{n-1} + \varepsilon)U_k(t_{n-1} + \varepsilon,s)x \\
&= \lim_{\varepsilon \searrow 0} \int_{t_{n-1} + \varepsilon}^{t_n - \varepsilon}  \ddtau{ \bigl( U_l(t, \tau)U_k(\tau, s)x \bigr) } \, d\tau \\
&= \int_{t_{n-1}}^{t_n}  \ddtau{ \bigl( U_l(t, \tau)U_k(\tau, s)x \bigr) } \, d\tau
\end{align*}
für alle $n \in \{1, \dots, m \}$ und damit
\begin{align*}
\norm{ U_k(t,s)x - U_l(t,s)x } \le M^2 \, (1 + c)^3 \, e^{Mc} \: \Bigl( \frac{c}{k} + \frac{c}{l} \Bigr) \,\norm{A(0)x} \, (t-s), 
\end{align*}
da ja
\begin{align*}
U_k(t,s)x - U_l(t,s)x &= U_l(t,t)U_k(t,s)x - U_l(t,s)U_k(s,s)x \\
&= \sum_{n=1}^{m} U_l(t, t_n)U_k(t_n,s)x - U_l(t,t_{n-1})U_k(t_{n-1},s)x.
\end{align*}
Also ist $\bigl( U_k(t,s)x \bigr)$ tatsächlich für alle $x \in D$ eine Cauchyfolge gleichmäßig in $(s,t) \in \Delta$ und wegen $\norm{U_k(t,s)} \le M$  gilt dies auch für alle $x \in X$.

Sei nun 
\begin{align*}
U(t,s)x := \lim_{k \to \infty} U_k(t,s)x
\end{align*}
für alle $(s,t) \in \Delta$ und alle $x \in X$. Dann gilt
\begin{align*}
\norm{U(t,s)} \le M,
\end{align*}
\begin{align*}
U(t,s) U(s,r) = U(t,r)
\end{align*}
für alle $(r,s), (s,t) \in \Delta$ und $(s,t) \mapsto U(t,s)x$ ist stetig für alle $x \in X$.

Wir haben gerade gesehen, dass $\bigl( W_k^{(0)}(t,s)x \bigr) = \bigl( U_k(t,s)x \bigr)$ für alle $(s,t) \in \Delta$ und alle $x \in X$ gegen $W^{(0)}(t,s)x := U(t,s)x$ konvergiert, jetzt wollen wir zeigen, dass $\bigl( W_k^{(m)}(t,s)x \bigr)$ auch für $m \in \natu$ für alle $(s,t) \in \Delta$ und alle $x \in X$ konvergiert. 

Sei also $m \in \natu$, $x \in X$ und $(s,t) \in \Delta$. Sei weiter
\begin{align*}
f_k^{(m)}(t_1, t_2, \dots, t_m) \\
:= \sum_{l_1 = \lfloor ks \rfloor + m}^{\lfloor kt \rfloor} \; & \sum_{l_2 = \lfloor ks \rfloor + m-1}^{l_1 -1} \dotsi \sum_{l_m = \lfloor ks \rfloor + 1}^{l_{m-1} -1} U_k \Bigl( t, \frac{l_1}{k} \Bigr)\, k\,C_{ \frac{l_1}{k}, \frac{l_1 -1}{k} } \; \cdot \\
& \cdot \, U_k \Bigl( \frac{l_1}{k}, \frac{l_2}{k} \Bigr) \, k\, C_{ \frac{l_2}{k}, \frac{l_2 -1}{k} } \; \dotsb \;
U_k \Bigl( \frac{l_{m-1}}{k}, \frac{l_m}{k} \Bigr) \,k\, C_{ \frac{l_m}{k}, \frac{l_m -1}{k} }  U_k \Bigl( \frac{l_{m}}{k}, s \Bigr) x \; \cdot \\
& \cdot \, \chi_{ \bigl[ \frac{l_1-1}{k}, \frac{l_1}{k} \bigr)  } (t_1) \; \chi_{ \bigl[ \frac{l_2-1}{k}, \frac{l_2}{k} \bigr)  } (t_2) \; \dotsm \; \chi_{ \bigl[ \frac{l_m-1}{k}, \frac{l_m}{k} \bigr)  } (t_m),
\end{align*}
und
\begin{align*}
f^{(m)}(t_1, t_2, \dots, t_m) := U(t,t_1) C(t_1) U(t_1,t_2) C(t_2) \; \dotsm \; U(t_{m-1},t_m) C(t_m) U(t_m,s) x
\end{align*}
für alle $(t_1, t_2, \dots, t_m) \in \Delta_{[s,t]}^{(m)} := \{ (t_1', t_2', \dots, t_m') \in [s,t]^m: t_1' \ge t_2' \ge \dotsb \ge t_m' \}$ und alle $k \in \natu$ und sei
\begin{align*}
W^{(m)}(t,s) x := \int_s^t  \int_s^{t_1} \; \dotsi \; \int_s^{t_{m-1}} f^{(m)}(t_1, t_2, \dots, t_m) \, dt_m \dots dt_2 \, dt_1.
\end{align*}
Dieses Integral existiert, weil $t' \mapsto C(t')$ nach Voraussetzung stark stetig ist, auch in $0$.
Wir können nun $W_k^{(m)}(t,s)x$ als Integral ausdrücken:
\begin{align*}
W_k^{(m)}(t,s)x = \int_{ \frac{ \lfloor ks \rfloor + m-1 }{k} }^{ \frac{\lfloor kt \rfloor}{k}  } \; \int_{ \frac{ \lfloor ks \rfloor + m-2 }{k} }^{ \frac{\lfloor k t_1 \rfloor}{k}  }  \dotsi \int_{ \frac{ \lfloor ks \rfloor }{k} }^{ \frac{\lfloor k t_{m-1} \rfloor}{k}  }     f_k^{(m)}(t_1, t_2, \dots, t_m)   \, dt_m \dots dt_2 \, dt_1.
\end{align*}
Weiter gilt
\begin{align*}
\norm{ f_k^{(m)}(t_1, t_2, \dots, t_m) } \le M \, (Mc)^m \, \norm{x}
\end{align*}
für alle $(t_1, t_2, \dots, t_m) \in \Delta_{[s,t]}^{(m)}$ und alle $k \in \natu$ und 
\begin{align*}
f_k^{(m)}(t_1, t_2, \dots, t_m) \longrightarrow f^{(m)}(t_1, t_2, \dots, t_m) \quad (k \to \infty)
\end{align*}
für alle $(t_1, t_2, \dots, t_m) \in \Delta_{[s,t]}^{(m)}$ mit $s < t_m  < \dotsb < t_2 < t_1 < t$, wobei wir benutzt haben, dass
\begin{align*}
\sup_{(s',t') \in \Delta} \norm{ U_k(t',s')y - U(t',s')y } \longrightarrow 0 \quad (k \to \infty),
\end{align*}
$\Delta \ni (s',t') \mapsto U(t',s')y$ stetig ist,
\begin{align*}
\sup_{ t' \in \bigl[ \frac{1}{k}, 1 \bigr] }   \norm{ k \, C\Bigl(t', t'-\frac{1}{k} \Bigr)y - C(t')y }  \longrightarrow  0 \quad (k \to \infty)
\end{align*}
und $(0,1] \ni t' \mapsto C(t')y$ stetig ist für alle $y \in X$.
Also erhalten wir mit dem lebesgueschen Satz, dass
\begin{align*}
\int_{ \Delta_{[s,t]}^{(m)} }   f_k^{(m)}(t_1, t_2, \dots, t_m) \,d(t_1, t_2, \dots, t_m)  & \longrightarrow   \int_{ \Delta_{[s,t]}^{(m)} }   f^{(m)}(t_1, t_2, \dots, t_m) \,d(t_1, t_2, \dots, t_m) \\ 
& =  W^{(m)}(t,s)x  \quad (k \to \infty)
\end{align*}
und damit auch 
\begin{align*}
W_k^{(m)}(t,s)x  \longrightarrow  W^{(m)}(t,s)x \quad (k \to \infty).
\end{align*}

Sei 
\begin{align*}
W(t,s)x := \sum_{m=0}^{\infty} W^{(m)}(t,s)x
\end{align*}
für alle $(s,t) \in \Delta$ und alle $x \in X$. Dann ist $(s,t) \mapsto W(t,s)x$ stetig für alle $x \in X$, weil $(s,t) \mapsto W^{(m)}(t,s)x$ stetig ist und die Reihe gleichmäßig konvergiert, und es gilt
\begin{align*}
W_k(t,s)x & = \biggl( 1 + C\Bigl(t, \frac{\lfloor kt \rfloor} {k} \Bigr) \biggr)   \biggl(  \sum_{m=0}^{\infty} W_k^{(m)} (t,s)  \biggr)    \biggl( 1 + C\Bigl( \frac{\lfloor ks \rfloor} {k}, s \Bigr) \biggr) x \\
& \longrightarrow W(t,s)x \quad (k \to \infty)
\end{align*} 
für alle $(s,t) \in \Delta$ und alle $x \in X$, weil $W_k^{(m)}(t,s)y  \longrightarrow  W^{(m)}(t,s)y \;\; (k \to \infty)$ und $\norm{W_k^{(m)}(t,s)} \le M \frac{(Mc \, (t-s))^m}{m!}$ für alle $y \in X$ und alle $(s,t) \in \Delta$.

Jetzt können wir (endlich) zeigen, dass $U$ eine Zeitentwicklung zu $A$ ist. Sei dazu $x \in D$. Dann gilt
\begin{align*}
A(t) U_k(t,s)x = W_k(t,s) A(s)x \longrightarrow W(t,s) A(s)x \quad (k \to \infty),
\end{align*} 
das heißt, 
\begin{align*}
U(t,s)x \in D \text{ \, und \, } A(t) U(t,s)x = W(t,s) A(s)x
\end{align*}
($A(t)$ ist ja abgeschlossen) für alle $(s,t) \in \Delta$ und $(s,t) \mapsto A(t)U(t,s)x$ ist stetig.
Weiter gilt
\begin{align*}
U_k(t,s)x - x & = \int_s^t \ddtau{ U_k (\tau,s)x } \, d\tau = \int_s^t \biggl( 1 + C\Bigl( \frac{\lfloor k \tau \rfloor} {k}, \tau \Bigr) \biggr) \, W_k(\tau,s) A(s)x \, d\tau \\
& \longrightarrow \int_s^t W(\tau, s) A(s)x \, d\tau = \int_s^t A(\tau) U(\tau,s)x \, d\tau \quad(k \to \infty),
\end{align*}
woraus sich ergibt, dass $[s,1] \ni t \mapsto U(t,s)x$ differenzierbar ist und
\begin{align*}
\ddt{ U(t,s)x } = A(t) U(t,s)x
\end{align*}
für alle $t \in [s,1]$ und alle $s \in [0,1)$. Die übrigen Zeitentwicklungseigenschaften folgen mithilfe von Propostion~\ref{thm: char zeitentwicklung}, da $t \mapsto A(t)x$ nach Lemma~\ref{lm: kato, evolution 1} stetig ist für alle $x \in D$.
\end{proof}

Die Voraussetzung (ii) im obigen Lemma hätten wir übrigens auch durch die schwächere Bedingung
\begin{align*}
C(t)x := \lim_{k \to \infty} k \: C\Bigl( t, t-\frac{1}{k} \Bigr)x \text{ \; existiert gleichmäßig in } t \in (0,1]
\end{align*}
ersetzen können, denn nur diese haben wir im Beweis des Lemmas gebraucht.
\\

Wir können jetzt den folgenden sehr wichtigen Satz (s. Satz~VI.9.5 in~\cite{EngelNagel}) beweisen -- der Schlüssel dazu ist Lemma~\ref{lm: kato, evolution 2}.

\begin{thm} \label{thm: Kato}
Sei $A(t)$ für jedes $t \in I$ eine lineare Abbildung $D \subset X \to X$, die eine stark stetige Halbgruppe auf $X$ erzeugt, sei $A$ $(M,\omega)$-stabil und $t \mapsto A(t)x$ stetig differenzierbar für alle $x \in D$. Sei $B(t)$ für jedes $t \in I$ eine beschränkte lineare Abbildung in $X$, $t \mapsto B(t)x$ stetig differenzierbar für alle $x \in X$ und $b := \sup_{t \in I} \norm{B(t)}$. Dann existiert genau eine Zeitentwicklung $U$ zu $A$ und genau eine Zeitentwicklung $V$ zu $A + B$ und für diese gilt:
\begin{align*}
\norm{U(t,s)} \le M \, e^{ \omega (t-s) }, \; \norm{V(t,s)} \le M \, e^{ (\omega + M b) (t-s) } \text{ \; für alle } (s,t) \in \Delta.
\end{align*}
\end{thm}


\begin{proof}
Sei $\tilde{A}(t) := A(t) - (\omega + 1)$ für alle $t \in I$. Wir zeigen, dass die Voraussetzungen von Lemma~\ref{lm: kato, evolution 2} für $\tilde{A}$ (statt $A$) erfüllt sind. 
Sei $C(t,s) := \tilde{A}(t) \tilde{A}(s)^{-1} - 1$ für alle $(s,t) \in I^2$ (beachte, dass $\omega + 1 \in (\omega, \infty) \subset \rho(A(t))$ nach Satz~\ref{thm: eigenschaften von erzeugern}) und sei $x \in X$. Dann ist $\tilde{A}$ offensichtlich $(M, 0)$-stabil und
\begin{align*}
\frac{1}{t-s} \, C(t,s)x = \frac{1}{t-s} \, (\tilde{A}(t) - \tilde{A}(s)) \, \tilde{A}(s)^{-1}x = \frac{1}{t-s} \, \int_s^t \tilde{A}'(\tau) \tilde{A}(s)^{-1}x \,d\tau
\end{align*} 
für alle $s,t \in I$ mit $s \ne t$. $\tilde{A}'(t) \tilde{A}(0)^{-1}$ ist als starker limes der beschränkten linearen Abbildungen $\frac{1}{h} (\tilde{A}(t+h) - \tilde{A}(t)) \, \tilde{A}(0)^{-1}$ selbst beschränkt in $X$ für alle $t \in I$ und $I \ni t \mapsto \tilde{A}'(t) \tilde{A}(0)^{-1}$ ist stark stetig. $\tilde{A}(0) \tilde{A}(s)^{-1}$ ist als auf ganz $X$ definierte Verkettung einer abgeschlossenen mit einer beschränkten linearen Abbildung beschränkt für alle $s \in I$, weiter ist $I \ni s \mapsto \bigl( \tilde{A}(0) \tilde{A}(s)^{-1} \bigr)^{-1} = \tilde{A}(s) \tilde{A}(0)^{-1}$ stark stetig differenzierbar, das heißt, $s \mapsto \bigl( \tilde{A}(0) \tilde{A}(s)^{-1} \bigr)^{-1}$ ist stetig (nach Lemma~\ref{lm: strong db and db}), woraus mithilfe von Lemma~\ref{lm: continuity of inv} folgt, dass auch  $s \mapsto  \tilde{A}(0) \tilde{A}(s)^{-1}$ stetig und insbesondere stark stetig ist. Also ist
\begin{align*}
I^2 \ni (s,t) \mapsto  \tilde{A}'(t) \tilde{A}(s)^{-1}x = \tilde{A}'(t) \tilde{A}(0)^{-1}  \tilde{A}(0) \tilde{A}(s)^{-1}x
\end{align*}
stetig und damit gleichmäßig stetig. Wir erhalten daraus (mithilfe der obigen Integraldarstellung), dass $\{ (s',t') \in I^2: s' \ne t' \} \ni (s,t) \mapsto \frac{1}{t-s} \, C(t,s)x$ beschränkt ist und 
\begin{align*}
\frac{1}{h} \, C(t,t-h)x = \frac{1}{h} \, \int_{t-h}^t \tilde{A}'(\tau) \tilde{A}(t-h)^{-1}x \,d\tau \longrightarrow C(t)x := \tilde{A}'(t) \tilde{A}(t)^{-1}x \quad(h \searrow 0)
\end{align*}
für alle $t \in (0,1]$. Wir erhalten sogar 
\begin{align*}
\sup_{t \in [h,1]} \norm{    \frac{1}{h} \, C(t,t-h)x - C(t)x    } &= \sup_{t \in [h,1]} \norm{    \frac{1}{h} \, \int_{t-h}^t \tilde{A}'(\tau) \tilde{A}(t-h)^{-1}x - \tilde{A}'(t) \tilde{A}(t)^{-1}x  \,d\tau   }  \\
&\longrightarrow 0 \quad (h \searrow 0).
\end{align*}
Schließlich existiert $\lim_{t\searrow 0} C(t)x = \lim_{t \searrow 0} \tilde{A}'(t) \tilde{A}(t)^{-1}x = \tilde{A}'(0) \tilde{A}(0)^{-1}x$, das heißt die Voraussetzungen von Lemma~\ref{lm: kato, evolution 2} sind tatsächlich alle erfüllt. Also existiert (genau) eine Zeitentwicklung $\tilde{U}$ zu $\tilde{A}$. Sei 
\begin{align*}
U(t,s) := \tilde{U}(t,s)\, e^{(\omega + 1)(t-s)}
\end{align*}
für alle $(s,t) \in \Delta$. Dann ist $U$, wie man leicht sieht, eine Zeitentwicklung zu $A$ und für diese gilt
\begin{align*}
\norm{ U(t,s) } = \lim_{k \to \infty} \norm{ \tilde{U_k}(t,s)  \, e^{(\omega + 1)(t-s)}   } \le M \,e^{ \omega (t-s) },
\end{align*}
wobei $\tilde{U_k}(t,s)$ wie im Beweis von Lemma~\ref{lm: kato, evolution 2} (mit $A$ ersetzt durch $\tilde{A}$) definiert ist.
\\

Wir müssen nur noch die entsprechenden Aussagen für $A + B$ statt $A$ zeigen. Aber diese folgen sofort aus dem eben Bewiesenen -- wir müssen nur beachten, dass $A + B$ nach Proposition~\ref{prop: störung (M,w)-stabilität} $(M, \omega + M b)$-stabil ist und $t \mapsto A(t)x + B(t)x$ stetig differenzierbar ist für alle $x \in D$.
\end{proof}

Wie hängen die Voraussetzungen (i) bis (iii) aus Lemma~\ref{lm: kato, evolution 2} mit der in Satz~\ref{thm: Kato} getroffenen (im Vergleich) bemerkenswert einfachen Voraussetzung
\begin{align*}
t \mapsto A(t)x \text{\; ist stetig differenzierbar für alle } x \in D
\end{align*}
zusammen? Zunächst mag diese einfache Voraussetzung einschränkender erscheinen als jene Voraussetzungen (i) bis (iii) -- aber das ist tatsächlich \emph{nicht} so, wie die folgende Proposition zeigt. 

\begin{prop} \label{prop: zshg. regvor. kato}
Sei $A(t): D \subset X \to X$ für jedes $t \in I$ eine bijektive abgeschlossene lineare Abbildung und für alle $x \in X$ seien (i) bis (iii) aus Lemma~\ref{lm: kato, evolution 2} erfüllt. Dann ist $t \mapsto A(t)x$ stetig differenzierbar für alle $x \in D$.
\end{prop}

\begin{proof}
Sei  $x \in D$ und $f(t) := A(t)x$ für alle $t \in I$. Wir zeigen, dass $f$ die Voraussetzungen von Satz~\ref{thm: einseitig db und beidseitig db} erfüllt.
Zunächst ist $f$ nach Lemma~\ref{lm: kato, evolution 1} stetig. Weiter ist $f$ linksseitig differenzierbar in jedem $t \in (0,1]$. Sei nämlich $t \in (0,1]$. Dann gilt
\begin{align*}
\frac{f(t-h) - f(t)}{-h} = \frac{1}{h} \, C(t, t-h) A(t-h)x \longrightarrow C(t)A(t)x \quad (h \searrow 0)
\end{align*}
wegen Lemma~\ref{lm: strong db of products}, das heißt, $f$ ist tatsächlich linksseitig differenzierbar in $t$ und 
\begin{align*}
\partial_- f(t) = C(t)A(t)x.
\end{align*}
Da nun $[0,1] \ni t \mapsto C(t)A(t)x$ stetig ist, ist $f$ sogar linksseitig stetig differenzierbar und $\partial_- f$ ist stetig fortsetzbar in den Randpunkt $0$ hinein. Also ist $f$ nach Satz~\ref{thm: einseitig db und beidseitig db} (beidseitig) stetig differenzierbar.
\end{proof}

Aus dieser Proposition folgt, dass Theorem~XIV.4.1 aus~\cite{Yosida: FA} (ebenso wie Theorem~X.70 aus~\cite{RS 2}) nicht stärker ist als Satz~\ref{thm: Kato}, denn die in Theorem~XIV.4.1 getroffene Voraussetzung, $(s,t) \mapsto \frac{1}{t-s} \, C(t,s)x$ sei gleichmäßig stetig, führt dazu, dass sogar $\lim_{h \searrow 0} \frac{1}{h} C(t, t-h)x$ gleichmäßig in $t \in (0,1]$ existiert (und nicht nur $C(t)x := \lim_{k \to \infty} k \; C\Bigl( t, t-\frac{1}{k} \Bigr)x$) und dass $C(0)x := \lim_{t \searrow 0} C(t)x$ existiert.
Wir könnten also in Theorem~XIV.4.1 genauso gut 
\begin{align*}
t \mapsto A(t)x \text{\; ist stetig differenzierbar für alle } x \in D
\end{align*}
voraussetzen, ohne Allgemeinheit einzubüßen. Dies scheint -- so entnehmen wir~\cite{Schnaubelt email} -- bisher nicht bekannt zu sein.
\\

Satz~\ref{thm: Kato} legt die Voraussetzungen an $A$ in den nichttrivialen Adiabatensätzen unten (Abschnitt~\ref{sect: adsätze mit sl} bis~\ref{sect: höhere adsätze}) schon weitgehend fest: dort werden wir immer voraussetzen, dass $A(t): D \subset X \to X$ für jedes $t \in I$ eine stark stetige Halbgruppe auf $X$ erzeugt, dass $A$ $(M,0)$-stabil ist und schließlich dass $t \mapsto A(t)x$ stetig differenzierbar ist für alle $x \in D$. Die einzige Verschärfung gegenüber den Voraussetzungen von Satz~\ref{thm: Kato} besteht also darin, dass $\omega$ in den Adiabatensätzen kleiner oder gleich $0$ ist -- und auch sein muss (nach Beispiel~\ref{ex: (M,0)-stabilität wesentlich in den adsätzen mit sl} und Beispiel~\ref{ex: (M,0)-stabilität wesentlich in adsatz ohne sl}).

Diese Voraussetzungen an $A$ verbessern nun diejenigen etwa aus~\cite{AvronSeilerYaffe 87} und~\cite{Griesemer 07}. Vor allem natürlich, weil dort sogar verlangt wird, dass die $A(t)$ schiefselbstadjungiert sind. Aber auch weil die Voraussetzung der stetigen Differenzierbarkeit von $t \mapsto A(t)x$ ein klein wenig besser ist als die dortigen Regularitätsvoraussetzungen: sie ist nicht nur (vom theoretischen und erst recht vom praktischen Standpunkt her) einfacher, sondern auch allgemeiner als die in~\cite{AvronSeilerYaffe 87} und~\cite{Griesemer 07} getroffenen Regularitätsannahmen. Dies kann man durch ähnliche Argumente wie die für Lemma~\ref{lm: strong db of products} und Lemma~\ref{lm: reg of inv} einsehen.

Und auch die Voraussetzungen aus~\cite{AvronElgart 99}, \cite{Teufel 01} und~\cite{Abou 07} implizieren (gemäß den Ausführungen nach Proposition~\ref{prop: zshg. regvor. kato}) die stetige Differenzierbarkeit von $t \mapsto A(t)x$ für alle $x \in D$, sofern es diese Voraussetzungen wirklich -- wie dort behauptet -- erlauben Theorem~XIV.4.1 aus~\cite{Yosida: FA} anzuwenden. 
\\

Wir merken noch an, dass wir aus einer exponentiellen Abschätzung einer Zeitentwicklung $U$ (zu $A$) eine entsprechende exponentielle Abschätzung einer gestörten Zeitentwicklung $V$ (zu $A + B$) herleiten können, und zwar mithilfe einer wenig erstaunlichen Störungsreihendarstellung von $V$ (Dysonreihe). Wir müssen dabei nur voraussetzen, dass die Zeitentwicklung zu $A + B$ überhaupt existiert und die Störungen $B(t)$ beschränkt sind sowie stark stetig von $t$ abhängen. (Vgl. dies mit den schärferen Voraussetzungen von Satz~\ref{thm: Kato}, in dem es anders als hier um die \emph{Existenz} der gestörten Zeitentwicklung geht.) 

\begin{prop} \label{prop: abschätzung für gestörte zeitentw}
Sei $A(t)$ für jedes $t \in I$ eine lineare Abbildung $D \subset X \to X$ und $t \mapsto A(t)x$ stetig für alle $x \in D$. Sei $B(t)$ eine beschränkte lineare Abbildung in $X$ und sei $t \mapsto B(t)$ stark stetig. Weiter gebe es eine Zeitentwicklung $U$ zu $A$ und eine Zeitentwicklung $V$ zu $A+B$. Dann gilt
\begin{align*}
V(t,s)x = \sum_{n=0}^{\infty} V_n(t,s)x
\end{align*} 
für alle $(s,t) \in \Delta$ und alle $x \in X$, wobei
\begin{align*}
V_0(t,s)x := U(t,s)x \text{ \; und \; } V_{n+1}(t,s)x := \int_s^t U(t,t_1)B(t_1)V_n(t_1,s)x \, dt_1.
\end{align*}
Insbesondere gilt 
\begin{align*}
\norm{V(t,s)} \le M e^{(\omega + Mb)(t-s)}
\end{align*}
für alle $(s,t) \in \Delta$ (wobei $b := \sup_{t\in I}\norm{B(t)}$), wenn die entsprechende Abschätzung $\norm{U(t,s)} \le M e^{\omega(t-s)}$ für $U$ gilt.
\end{prop}

\begin{proof}
Sei 
\begin{align*}
M_0 := \sup_{(s,t) \in \Delta} \norm{U(t,s)},
\end{align*}
was eine reelle Zahl ist ($U$ ist ja als Zeitentwicklung stark stetig).
Dann erhalten wir induktiv, dass
\begin{align*}
\norm{V_n(t,s)} \le M_0^{n+1} \, b^n \, \frac{(t-s)^n}{n!}
\end{align*}
für alle $(s,t) \in \Delta$ und alle $n \in \natu \cup \{0\}$. Also konvergiert die Reihe
\begin{align*}
\tilde{V}(t,s)x := \sum_{n=0}^{\infty} V_n(t,s)x
\end{align*}
gleichmäßig in $(s,t) \in \Delta$ für alle $x \in X$ und die Abbildung $(s,t) \mapsto \tilde{V}(t,s)x$ ist stetig für alle $x \in X$ (die $V_n$ sind ja alle stark stetig).

Wir zeigen nun, dass $V(.,s)x$ und $\tilde{V}(.,s)x$ für alle $s \in I$ und $x \in X$ dieselbe (volterrasche) Integralgleichung erfüllen.
Sei $s \in I$ und $x \in X$. Dann gilt einerseits
\begin{align}     \label{eq: abschätzung für gestörte zeitentw 1}
V(t,s)x = U(t,s)x + U(t,\tau)V(\tau,s)x \big|_{\tau=s}^{\tau=t} = U(t,s)x + \int_s^t U(t,\tau)B(\tau)V(\tau,s)x \,d\tau
\end{align}
für alle $t \in [s,1]$ (zunächst gilt dies für $x \in D$, denn $[s,t] \ni \tau \mapsto U(t,\tau)V(\tau,s)x$ ist für diese $x$ nach Proposition~\ref{thm: char zeitentwicklung} und Lemma~\ref{lm: strong db of products} stetig differenzierbar mit Ableitung $\tau \mapsto U(t,\tau)B(\tau)V(\tau,s)x$,  wegen der Dichtheit von $D$ gilt es dann aber auch für beliebige $x \in X$) und andererseits
\begin{align}         \label{eq: abschätzung für gestörte zeitentw 2}
\tilde{V}(t,s)x = V_0(t,s)x + \sum_{n=0}^{\infty} V_{n+1}(t,s)x = U(t,s)x + \int_s^t U(t,\tau) B(\tau) \tilde{V}(\tau,s)x \, d\tau 
\end{align}
für alle $t \in [s,1]$.

Aus~\eqref{eq: abschätzung für gestörte zeitentw 1} und~\eqref{eq: abschätzung für gestörte zeitentw 2} folgt nun leicht, dass $V(t,s)x = \tilde{V}(t,s)x$ für alle $t \in \cup_{i=1}^m [t_{i-1},t_i] = [s,1]$ und alle $x \in X$, wobei 
\begin{align*}
t_0 := s \text{ \; und \; } t_i := \min \big\{ t_{i-1}  + \, \frac{1}{2 M_0 \,b + 1}, \, 1 \big\}
\end{align*}
und $m$ eine (die kleinste) natürliche Zahl ist mit $t_m = 1$.

Also haben wir
\begin{align*}
V(t,s)x = \sum_{n=0}^{\infty} V_n(t,s)x
\end{align*}
für alle $x \in X$, wie behauptet.
\\

Schließlich nehmen wir zusätzlich an, dass wir für $U$ sogar eine exponentielle Abschätzung haben,
\begin{align*}
\norm{U(t,s)} \le M e^{\omega(t-s)}
\end{align*}
für alle $(s,t) \in \Delta$. Wir bekommen dann induktiv 
\begin{align*}
\norm{V_n(t,s)} \le M^{n+1} \, b^n \, \frac{(t-s)^n}{n!} \, e^{\omega(t-s)}
\end{align*}
für alle $n \in \natu \cup \{0\}$, und daraus, wie gewünscht,
\begin{align*}
\norm{V(t,s)} \le \sum_{n=0}^{\infty} M^{n+1} \, b^n \, \frac{(t-s)^n}{n!} \, e^{\omega(t-s)} = M e^{(\omega + Mb)(t-s)}
\end{align*}
für alle $(s,t) \in \Delta$.
\end{proof}

Ähnlich wie im zeitunabhängigen Fall ($A(t) = A_0$ für alle $t \in I$) (stonescher Satz!) besteht auch im zeitabhängigen Fall ein enger Zusammenhang zwischen Schiefsymmetrie und Isometrie.  

\begin{prop} \label{prop: zeitentw unitär}
(i) Sei $A(t)$ für jedes $t \in I$ eine lineare Abbildung $D \subset H \to H$, $t \mapsto A(t)x$ stetig für alle $x \in D$ und $U$ eine Zeitentwicklung zu $A$. Dann ist $U(t,s)$ isometrisch für alle $(s,t) \in \Delta$ genau dann, wenn $A(t)$ schiefsymmetrisch ist für alle $t \in I$. 

(ii) Sei $A(t): D \subset H \to H$ für jedes $t \in I$ schiefselbstadjungiert, $t \mapsto A(t)x$ stetig (im Fall $D = H$) bzw. stetig differenzierbar (im allgemeinen Fall) für alle $x \in D$ und $U$ die Zeitentwicklung zu $A$. Dann ist $U(t,s)$ unitär für alle $(s,t) \in \Delta$.
\end{prop}

\begin{proof}
(i) Seien $x, y \in D$, dann ist $[s,1] \ni t \mapsto \scprd{ U(t,s)x, U(t,s)y }$ differenzierbar und
\begin{align*}
\ddt{   \scprd{ U(t,s)x, U(t,s)y }   } &= \scprd{ A(t)U(t,s)x, U(t,s)y } + \scprd{ U(t,s)x, A(t)U(t,s)y } \\
&= \scprd{ A(t)U(t,s)x, U(t,s)y } + \scprd{ A(t)^* \, U(t,s)x, U(t,s)y }
\end{align*}
für alle $t \in [s,1]$ und alle $s \in [0,1)$.

Sei $U(t,s)$ isometrisch für alle $(s,t) \in \Delta$. Dann gilt
\begin{align*}
0 = \ddt{   \scprd{ U(t,s)x, U(t,s)y }   } \big|_{t=s} &= \scprd{ A(t)U(t,s)x, U(t,s)y } + \scprd{ A(t)^* \, U(t,s)x, U(t,s)y }  \Big|_{t=s} \\
&= \scprd{ A(s)x, y } + \scprd{ A(s)^* x, y },
\end{align*}
für alle $x, y \in D$ und alle $s \in [0,1)$, das heißt, $A(s)$ ist schiefsymmetrisch für alle $s \in [0,1)$. Da $t \mapsto A(t)x$ stetig ist für alle $x \in D$, erhalten wir mithilfe von Proposition~\ref{thm: char zeitentwicklung}, dass
\begin{align*}
0 = \dds{  \scprd{ U(1,s)x, U(1,s)y }   } = \scprd{ -U(1,s)A(1)x, U(1,s)y } + \scprd{ U(1,s)x, -U(1,s)A(1)y }
\end{align*}
für alle $s \in [0,1]$ und damit insbesondere
\begin{align*}
0 = \scprd{ -A(1)x, y } + \scprd{ x, -A(1)y }
\end{align*}
für alle $x, y \in D$, woraus sich die Schiefsymmetrie von $A(1)$ ergibt.

Sei umgekehrt $A(t)$ schiefsymmetrisch für alle $t \in I$. Dann gilt
\begin{align*}
\ddt{   \scprd{ U(t,s)x, U(t,s)y }   } = \scprd{ A(t)U(t,s)x, U(t,s)y } + \scprd{ A(t)^* \, U(t,s)x, U(t,s)y } = 0,
\end{align*}
folglich 
\begin{align*}
\scprd{ U(t,s)x, U(t,s)y } = \scprd{ U(t,s)x, U(t,s)y } \big|_{t=s} = \scprd{ x, y}
\end{align*}
für alle $x, y \in D$, alle $t \in [s,1]$ und alle $s \in [0,1)$, das heißt, $U(t,s)$ ist isometrisch für alle $t \in [s,1]$ und alle $s \in [0,1)$ und wegen $U(1,1) = 1$ auch für $(s,t) = (1,1)$.
\\

(ii) Wir wissen schon aus (i), dass die Zeitentwicklung $U$ zu $A$ isometrisch ist, es bleibt also zu zeigen, dass 
\begin{align*}
U(t,s) U(t,s)^* = 1
\end{align*}
für alle $(s,t) \in \Delta$. Wir tun dies nur im Fall von beliebigem $D$, in dem wir sogar die stetige Differenzierbarkeit der Abbildungen $t \mapsto A(t)x$ vorausgesetzt haben. Für den (einfacheren) Fall $D = H$ verweisen wir auf Theorem~X.69 in~\cite{RS 2}. 

Sei $\tilde{A}(t) := A(t) - 1$ für alle $t \in I$, $U_k(t,s)$ wie im Beweis von Lemma~\ref{lm: kato, evolution 2}, 
\begin{align*}
\tilde{U}_k(t,s) := U_k(t,s) e^{-(t-s)} \text{ \; und \; } \tilde{U}(t,s) := U(t,s) e^{-(t-s)}
\end{align*}
sowie 
\begin{align*}
\tilde{V}_k(t,s) := \tilde{A}(s) \, U_k(t,s)^* \,  \tilde{A}(t)^{-1}
\end{align*}
für alle $(s,t) \in \Delta$.
Dann ist $U_k(t,s)$ unitär (die $A(t')$ sind ja schiefselbstadjungiert, erzeugen also unitäre Gruppen auf $H$) für alle $k \in \natu$ und (nach dem Beweis von Lemma~\ref{lm: kato, evolution 2}, dessen Voraussetzungen nach dem Beweis von Satz~\ref{thm: Kato} ja wirklich erfüllt sind)
\begin{align*}
U_k(t,s)x = \tilde{U}_k(t,s) e^{(t-s)} x \longrightarrow  \tilde{U}(t,s) e^{(t-s)} x = U(t,s)x \quad(k \to \infty)
\end{align*}
für alle $x \in H$, und damit insbesondere
\begin{align*}
\scprd{ U(t,s) U_k(t,s)^* x , y } & = \scprd{ x , U_k(t,s) U(t,s)^* y } \\
& \longrightarrow \scprd{ x , U(t,s) U(t,s)^* y } = \scprd{ U(t,s) U(t,s)^* x , y } \quad(k \to \infty)
\end{align*}
für alle $x, y \in H$ und alle $(s,t) \in \Delta$.
Weiter gilt, erneut wegen der Schiefselbstadjungiertheit der $A(t')$, dass
\begin{align*}
\tilde{V}_k(t,s)  =& \, \tilde{A}(s) \, U_k(t,s)^* \, \tilde{A}(t)^{-1}  \\
=& \, \tilde{A}(s) \; e^{ A \left( \frac{\lfloor ks \rfloor} {k} \right) (s - \frac{\lfloor ks \rfloor + 1} {k} ) } \; e^{ A \left( \frac{\lfloor ks \rfloor + 1} {k} \right) (-\frac{1} {k} ) } \; \dotsm \; e^{ A \left( \frac{\lfloor kt \rfloor} {k} \right) ( \frac{\lfloor kt \rfloor} {k} - t) }\; \tilde{A}(t)^{-1}    \\
=& \, \biggl( 1 + C\Bigl(s, \frac{\lfloor ks \rfloor} {k} \Bigr) \biggr) \; 
e^{ A \left( \frac{\lfloor ks \rfloor} {k} \right) (s - \frac{\lfloor ks \rfloor + 1} {k} ) } \; \cdot \\   
& \cdot \; \biggl( 1 + C\Bigl( \frac{\lfloor ks \rfloor} {k} , \frac{\lfloor ks \rfloor + 1} {k} \Bigr)  \biggr)   \;  
e^{ A \left( \frac{\lfloor ks \rfloor + 1} {k} \right) (-\frac{1} {k} ) } \; \dotsm \\
& \dotsm \; \biggl( 1 + C\Bigl( \frac{\lfloor kt \rfloor - 1} {k}, \frac{\lfloor kt \rfloor} {k} \Bigr) \biggr) \;
e^{ A \left( \frac{\lfloor kt \rfloor} {k} \right) ( \frac{\lfloor kt \rfloor} {k} - t) }\; 
\biggl( 1 + C\Bigl( \frac{\lfloor kt \rfloor} {k}, t \Bigr) \biggr) 
\end{align*}
und daher
\begin{align*}
\norm{ \tilde{V}_k(t,s) } & \le \Bigl( 1 + \frac{c}{k} \Bigr)^{m_k(t,s) + 2}  \le  \Bigl( 1 + \frac{c}{k} \Bigr)^{k(t-s)} \, \Bigl( 1 + \frac{c}{k} \Bigr)^3  \\
& \le e^{c(t-s)} \, (1 + c)^3 
\end{align*}
für alle $k \in \natu$ und alle $(s,t) \in \Delta$, wobei $C(t,s) := \tilde{A}(t) \tilde{A}(s)^{-1} - 1$,
\begin{align*}
c := \sup_{(s,t) \in \{ s' \ne t' \} } \norm{ \frac{1}{t-s} \,C(t,s) }.
\end{align*}
und $m_k(t,s)$ wie im Beweis von Lemma~\ref{lm: kato, evolution 2} erklärt sei.

Wir zeigen nun, dass 
\begin{align*}
U(t,s) U_k(t,s)^*x \longrightarrow x \quad (k \to \infty)
\end{align*}
für alle $x \in D$ und alle $(s,t) \in \Delta$.

Sei also $(s,t) \in \Delta$ mit $s \ne t$ (für $s = t$ ist die Behauptung klar) und $x \in D$. Sei weiter $k \in \natu$ und $\bigl(t_n \bigr)_{ n \in \{ 0, \dots , m \} }$ eine Zerlegung des Intervalls $[s,t]$, sodass $t_1, \dots, t_{m-1}$ genau die in $(s,t)$ enthaltenen Zerlegungsstellen der $\frac{1}{k}$-Zerlegung von $I$ sind. Dann ist
\begin{align*}
(t_{n-1}, t_n) \ni \tau \mapsto U_k(t,\tau) U(\tau,s) U_k(t,s)^*x
\end{align*}
wegen $U_k(t,s)^* D \subset D$ differenzierbar (nach Lemma~\ref{lm: strong db of products}) und 
\begin{align*}
\ddtau{ \bigl( U_k(t,\tau) U(\tau,s)  U_k(t,s)^*x \bigr) } &= U_k(t,\tau) \biggl( -A\Bigl( \frac{\lfloor k \tau \rfloor} {k} \Bigr) + A(\tau) \biggr) U(\tau,s) U_k(t,s)^*x \\
&= -U_k(t,\tau) \, C\Bigl( \frac{\lfloor k \tau \rfloor} {k}, \tau \Bigr)\, \tilde{W}(\tau,s) e^{\tau -s} \; \tilde{V}_k(t,s) \, \tilde{A}(t)x
\end{align*}
für alle $\tau \in (t_{n-1},t_n)$ und alle $n \in \{ 1, \dots , m \}$, wobei $\tilde{W}(t',s') := \tilde{A}(t')\tilde{U}(t',s')\tilde{A}(s')^{-1}$. 
Nach dem Beweis von Lemma~\ref{lm: kato, evolution 2} ist $\Delta \ni (s',t') \mapsto \tilde{W}(t',s')$ stark stetig und es gilt
\begin{align*}
\norm{ \tilde{W}(t',s') } \le (1 + c)^2 \; e^{c(t'-s')}
\end{align*}
für alle $(s',t') \in \Delta$.
Wir sehen damit und mithilfe von Lemma~\ref{lm: kato, evolution 1}, dass 
\begin{align*}
(t_{n-1}, t_n) \ni \tau \mapsto U_k(t,\tau) U(\tau,s) U_k(t,s)^*x
\end{align*}
sogar \emph{stetig} differenzierbar ist und
\begin{align*}
\norm{    \ddtau{ \bigl( U_k(t,\tau) U(\tau,s)  U_k(t,s)^*x  \bigr) }    }  
&\le \frac{c}{k} \, \, (1 + c)^2 e^{(c+1)(\tau - s)} \; \; e^{c(t-s)} \, (1 + c)^3 \, \, \norm{ \tilde{A}(t)x }  \\
&\le \frac{c}{k} \, (1 + c)^5 \, e^{2c + 1} \, \norm{ \tilde{A}(t)x }
\end{align*}
für alle $\tau \in (t_{n-1},t_n)$ und alle $n \in \{ 1, \dots , m \}$, das heißt
\begin{align*}
\norm{ U(t,s)U_k(t,s)^*x - x } &= \norm{ U_k(t,\tau) U(\tau,s)  U_k(t,s)^*x \big|_{\tau=s}^{\tau=t}  } \\
&\le \frac{c}{k} \, (1 + c)^5 \, e^{2c + 1} \, \, (t-s) \, \norm{ \tilde{A}(t)x }.
\end{align*}
Also gilt tatsächlich 
\begin{align*}
U(t,s) U_k(t,s)^*x \longrightarrow x \quad (k \to \infty)
\end{align*}
und damit 
\begin{align*}
\scprd{ U(t,s) U(t,s)^*x, y } = \lim_{k \to \infty} \scprd{ U(t,s) U_k(t,s)^*x , y} = \scprd{x,y}
\end{align*}
für alle $x \in D$, $y \in H$ und alle $(s,t) \in \Delta$, woraus sich schließlich $U(t,s) U(t,s)^* = 1$ für alle $(s,t) \in \Delta$ ergibt.
\end{proof}

Wir weisen darauf hin, dass aus
\begin{gather*}
U_k(t,s) \text{ unitär für alle } k \in \natu \text{ \; und \; } \\
U_k(t,s)x \longrightarrow U(t,s)x \quad (k \to \infty) \text{ für alle } x \in X
\end{gather*}
allein noch \emph{nicht} folgt, dass auch $U(t,s)$ unitär ist (Bemerkung~II.4.10 in~\cite{Takesaki}). Wir müssen also wirklich etwas mehr arbeiten, was den langen Beweis von Aussage~(ii) des obigen Satzes erklärt.

\subsection{Adiabatische Zeitentwicklungen}

Wir erinnern uns: ein erstes Ziel der Adiabatentheorie ist es zu klären, wann $(1-P(t)) U_T P(0)$ in der in Abschnitt~\ref{sect: einleitung} beschriebenen Ausganssituation gegen $0$ konvergiert für $T$ gegen $\infty$, kurz gefasst: wann die Zeitentwicklung $U_T$ 
 beinahe adiabatisch ist bzgl. $P$.
Was genau wir unter bzgl. $P$ adiabatischen Zeitentwicklungen verstehen wollen, legen wir nun fest. (Vgl. dies auch mit Abschnitt~IV.3.1 und~IV.3.2 in~\cite{Krein 71}.)
\\

Sei $U$ eine Zeitentwicklung in $X$ und $P(t)$ für jedes $t \in I$ eine beschränkte Projektion in $X$. Dann heißt $U$ \emph{adiabatisch bzgl. $P$ im engeren bzw. weiteren Sinne} genau dann, wenn
\begin{align*}
(1-P(t)) U(t,s) P(s) = 0 \text{ \; bzw. \; } (1-P(t)) U(t,0) P(0) = 0
\end{align*} 
für alle $(s,t) \in \Delta$ bzw. für alle $t \in I$.
\\

Zeitentwicklungen, die adiabatisch sind bzgl. einer Familie $P$ von beschränkten Projektionen, lassen also \emph{keine Übergänge} von den Unterräumen $P(s)X$ bzw. dem Unterraum $P(0)X$ in die Unterräume $P(t)X$ zu. In diesem Sinne sind sie \emph{adiabatisch}.

Wie man leicht einsieht, ist eine Zeitentwicklung $U$ in $X$ i. e. S. adiabatisch bzgl. $P$ \emph{und} $1-P$ genau dann, wenn 
\begin{align*}
P(t) U(t,s) = U(t,s) P(s)
\end{align*}
für alle $(s,t) \in \Delta$.
\\

Zunächst ein offensichtliches Lemma.

\begin{lm} \label{lm: skalierung und (M,w)-stabilität}
Sei $A(t)$ für jedes $t \in I$ Erzeuger einer stark stetigen Halbgruppe auf $X$ und sei $A$ $(M, \omega)$-stabil. Dann ist $T A$ $(M, T \omega)$-stabil für alle $T \in (0, \infty)$.
\end{lm}

\begin{proof}
Wir müssen nur beachten, dass $T A(t)$ für alle $T \in (0, \infty)$ tatsächlich eine stark stetige Halbgruppe auf $X$ erzeugt, die gegeben ist durch
\begin{align*}
e^{(T A(t)) \,s} = e^{A(t) \, (T \,s)}
\end{align*}
für alle $s \in [0, \infty)$.
\end{proof}

Wie die folgende Proposition nahelegt, sollten wir (in der in Abschnitt~\ref{sect: einleitung} beschriebenen Ausgangssituation) von der Zeitentwicklung $U_T$ zu $T A$ höchstens erwarten, dass sie \emph{beinahe} adiabatisch ist bzgl. $P$. Wir können nicht erwarten, dass sie (echt) adiabatisch ist (i. e. S. bzgl. $P$ und $1-P$). 

\begin{prop} \label{prop: U_T fast nie adiabatisch}
Sei $A(t)$ für jedes $t \in I$ eine lineare Abbildung $D \subset X \to X$, die eine stark stetige Halbgruppe auf $X$ erzeugt, sei $A$ $(M,\omega)$-stabil und $t \mapsto A(t)x$ stetig (falls $D = X$) bzw. stetig differenzierbar (falls $D$ ein beliebiger dichter Unterraum ist) für alle $x \in D$. Sei $P(t)$ für jedes $t \in I$ eine beschränkte Projektion in $X$ mit $P(t)A(t) \subset A(t)P(t)$ und sei $t \mapsto P(t)x$ differenzierbar für alle $x \in X$. Wenn die Zeitentwicklung $U_T$ zu $T A$ (die nach Lemma~\ref{lm: skalierung und (M,w)-stabilität} und Satz~\ref{thm: Dyson} bzw. Satz~\ref{thm: Kato} existiert) i. e. S. adiabatisch ist bzgl. $P$ und $1-P$, dann gilt $P' = 0$, die durch $P$ gegebene Zerlegung von $X$ hängt dann also nicht von $t \in I$ ab.
\end{prop}

\begin{proof}
Sei die Zeitentwicklung $U_T$ zu $T A$ i. e. S. adiabatisch bzgl. $P$ und $1-P$. Dann gilt 
\begin{align*}
P(t) U_T(t,s) = U_T(t,s) P(s)
\end{align*}
für alle $(s,t) \in \Delta$. Sei nun $s \in [0,1)$, dann erhalten wir durch Ableiten nach $t$, dass
\begin{align*}
P'(t) \, U_T(t,s)x + P(t) \, T A(t) U_T(t,s)x = T A(t) U_T(t,s)\, P(s)x
\end{align*}
für alle $t \in [s,1]$ und damit insbesondere
\begin{align*}
P'(s) x = T \bigl( A(s) P(s)x - P(s) A(s)x \bigr) = 0 
\end{align*}
für alle $x \in D$. 
Wir müssen nun nur noch zeigen, dass auch $P'(1) = 0$. Wir wissen, dass
\begin{align*}
P(1) U_T(1,s) = U_T(1,s) P(s)
\end{align*}
für alle $s \in [0,1]$. Durch Ableiten nach $s$ (beachte Proposition~\ref{thm: char zeitentwicklung}) erhalten wir
\begin{align*}
P(1) \bigl( -U_T(1,s) \, T A(s)x \bigr) = -U_T(1,s) \, T A(s) \, P(s)x + U_T(1,s)\, P'(s)x
\end{align*} 
für alle $s \in [0,1]$ und damit insbesondere
\begin{align*}
P'(1)x = T \bigl( A(1) P(1)x - P(1) A(1)x \bigr) = 0,
\end{align*}
für alle $x \in D$, woraus die Behauptung folgt.
\end{proof}

Wie können wir nun zeigen, dass die Zeitentwicklung zu $T A$ (wenigstens) \emph{beinahe} adiabatisch ist? 
Der nächste Satz, den wir Kato~\cite{Kato 50} sowie Dalecki und Krein (s. die Anmerkungen zu Kapitel~IV in~\cite{Krein 71}) verdanken, ist ein erster wichtiger Schritt dorthin: dieser Satz gibt uns eine bzgl. $P$ und $1-P$ adiabatische Zeitentwicklung $V_T$, die -- wie wir in den Adiabatensätzen in den Abschnitten~\ref{sect: triviale adsätze} bis~\ref{sect: adsätze ohne sl} sehen werden -- die uns eigentlich interessierende Zeitentwicklung $U_T$ gut approximiert.

\begin{thm} \label{thm: intertwining relation}
Sei $A(t)$ für jedes $t \in I$ eine lineare Abbildung $D \subset X \to X$ und sei $t \mapsto A(t)x$ stetig für alle $x \in D$. Sei $P(t)$ für jedes $t \in I$ eine beschränkte Projektion in $X$ mit $P(t)A(t) \subset A(t)P(t)$ und sei $t \mapsto P(t)x$ stetig differenzierbar für alle $x \in X$. Schließlich existiere für jedes $T \in (0, \infty)$ die Zeitentwicklung zu $T A + [P',P]$ und diese sei mit $U_{a,T}$ bezeichnet. 
Dann gilt
\begin{align*}
P(t) U_{a,T}(t,s) = U_{a,T}(t,s) P(s)
\end{align*}
(die intertwining relation) für alle $(s,t) \in \Delta$ und alle $ T \in (0, \infty)$. 
\end{thm}

\begin{proof}
Zunächst gilt
\begin{align*}
P' = \left(P^2\right)' = P'P + PP',
\end{align*}
woraus durch Anwendung von $P$ von rechts und von links folgt, dass 
\begin{align*}
P P' P = 0.
\end{align*}
Sei $(s,t) \in \Delta$ mit $s \ne t$ (für $s = t$ ist die intertwining relation offensichtlich) und sei $x \in D$. Wir erhalten dann wegen 
\begin{align*}
P(t')D \subset D \text{ \; für alle } t' \in I
\end{align*}
mithilfe von Lemma~\ref{lm: strong db of products} und Proposition~\ref{thm: char zeitentwicklung}, dass $[s,t] \ni \tau \mapsto U_{a,T}(t, \tau) P(\tau) U_{a,T}(\tau,s)x$ differenzierbar ist und 
\begin{align*}
&\ddtau{ \Bigl( U_{a,T}(t, \tau) P(\tau) U_{a,T}(\tau,s)x \Bigr)  } \\
& \qquad \qquad = U_{a,T}(t, \tau) \Bigl(      -\bigl(T A(\tau) + [P'(\tau),P(\tau)] \bigr) P(\tau) + P'(\tau) \\
& \qquad \qquad \qquad \qquad \qquad \qquad \qquad \qquad + P(\tau) \bigl(T A(\tau) + [P'(\tau),P(\tau)] \bigr)    \Bigr) U_{a,T}(\tau,s)x \\
&\qquad \qquad = U_{a,T}(t, \tau) \Bigl(   -T A(\tau) P(\tau) + P(\tau) T A(\tau)     \Bigr) = 0
\end{align*}
für alle $\tau \in [s,t]$ und alle $T \in (0, \infty)$.
Also gilt 
\begin{align*}
P(t) U_{a,T}(t,s)x - U_{a,T}(t,s) P(s)x = U_{a,T}(t, \tau) P(\tau) U_{a,T}(\tau,s)x \Big|_{\tau=s}^{\tau=t} = 0,
\end{align*}
wie gewünscht.
\end{proof}

Wenn sie denn existieren, werden wir die Zeitentwicklungen zu $T A$ und $T A + [P',P]$ künftig immer mit $U_T$ bzw. $U_{a,T}$ bezeichnen, wobei wir natürlich die Abhängigkeit von $A$ und $P$ unterdrücken. 

Der obige Satz zeigt, dass die Zeitentwicklung $U_{a,T}$ zu $T A + [P',P]$ unter den -- abgesehen von der Invarianzbedingung $P(t)A(t) \subset A(t)P(t)$ (s. Proposition~\ref{prop: zerl. des spektrums}) -- nicht sehr einschränkenden Voraussetzungen dieses Satzes adiabatisch ist bzgl. $P$ und $1-P$. Wir werden sie daher auch als \emph{adiabatische Zeitentwicklung zu $TA$ und $P$} 
bezeichnen. 
Wegen $[(1-P)', (1-P)] = [P',P]$ stimmen die adiabatische Zeitentwicklung zu $TA$ und $P$ und diejenige zu $TA$ und $1-P$ überein.
\\

Abschließend sei noch angemerkt, dass die adiabatische Zeitentwicklung durch die intertwining relation aus obigem Satz bei weitem nicht eindeutig bestimmt ist (Theorem~IV.3.1 in~\cite{Krein 71}): die Zeitentwicklung $V_T$ zu $T A + B$ ist (unter den Voraussetzungen des obigen Satzes und der Voraussetzung, dass die $B(t)$ beschränkte lineare Abbildungen sind, die stark stetig von $t \in I$ abhängen) nämlich genau dann adiabatisch bzgl. $P$ und $1-P$, wenn für jedes $t \in I$ eine beschränkte mit $P(t)$ vertauschende lineare Abbildung $C(t)$ existiert, sodass $B(t) = [P'(t),P(t)] + C(t)$. Dies sieht man wie im Beweis von Satz~\ref{thm: intertwining relation}, indem man für $(s,t) \in \Delta$ mit $s \ne t$ und $x \in D$ die Ableitung von $[s,t] \ni \tau \mapsto V_T(t, \tau) P(\tau) V_T(\tau,s)x$ bestimmt.


\section{Zwei triviale Adiabatensätze und Standardbeispiele} \label{sect: triviale adsätze}

\subsection{Die beiden Sätze}

Der folgende Satz ist angesichts von Satz~\ref{thm: intertwining relation} ziemlich trivial.

\begin{thm} \label{thm: triv adsatz 1}
Sei $A(t)$ für jedes $t \in I$ eine lineare Abbildung $D \subset X \to X$, die eine stark stetige Halbgruppe auf $X$ erzeugt, sei $A$ $(M,\omega)$-stabil und $t \mapsto A(t)x$ stetig (falls $D = X$) bzw. stetig differenzierbar (falls $D$ ein beliebiger dichter Unterraum ist) für alle $x \in D$. Sei $P(t)$ für jedes $t \in I$ eine beschränkte Projektion in $X$ mit $P(t)A(t) \subset A(t)P(t)$ und sei $t \mapsto P(t)$ konstant. 
Dann gilt
\begin{align*}
(1-P(t)) U_T(t,s) P(s) = 0 \text{ \; und \; } P(t) U_T(t,s) (1-P(s)) = 0
\end{align*}
für alle $(s,t) \in \Delta$.
\end{thm}

\begin{proof}
Wegen $P' = 0$ stimmt die Zeitentwicklung $U_T$ zu $T A$ mit der Zeitentwicklung $U_{a,T}$ zu $T A + [P',P]$ überein (Proposition~\ref{prop: ex höchstens eine zeitentw}) und für diese gilt nach Satz~\ref{thm: intertwining relation} die intertwining relation, die äquivalent ist zur behaupteten Adiabatizität i.e.S. bzgl. $P$ und $1-P$.
\end{proof}

Zusammen mit Proposition~\ref{prop: U_T fast nie adiabatisch} ergibt dieser Satz, dass die eigentlich interessierende Zeitentwicklung $U_T$ zu $T A$ genau dann i.e.S. adiabatisch ist bzgl. $P$ und $1-P$, wenn $P$ konstant ist.
\\

Der nächste Satz ist etwas weniger offensichtlich. Dieser Satz steht im wesentlichen auch in~\cite{Krein 71} (Theorem~IV.1.7), jedoch haben wir ihn unabhängig davon bewiesen.

\begin{thm} \label{thm: triv adsatz 2}
Sei $A(t)$ für jedes $t \in I$ eine lineare Abbildung $D \subset X \to X$, die eine stark stetige Halbgruppe auf $X$ erzeugt, sei $A$ $(M,\omega)$-stabil für eine (negative) Zahl $\omega \in (-\infty, 0)$ und sei $t \mapsto A(t)x$ stetig (falls $D =X$) bzw. stetig differenzierbar (falls $D$ ein beliebiger dichter Unterraum ist) für alle $x \in D$. Sei $P(t)$ für jedes $t \in I$ eine beschränkte Projektion in $X$ mit $P(t)A(t) \subset A(t)P(t)$ und sei $t \mapsto P(t)x$ stetig differenzierbar für alle $x \in X$. Schließlich existiere die Zeitentwicklung $U_{a,T}$ zu $TA + [P',P]$ für jedes $T \in (0, \infty)$.
Dann gilt
\begin{align*}
\sup_{t \in I} \norm{ U_{a,T}(t)-U_T(t)} = O\Bigl( \frac{1}{| \omega |\, T} \Bigr) \quad (T \to \infty).
\end{align*}
Wenn zusätzlich $\supp P' \subset (r_0, 1]$ für eine positive Zahl $r_0$, dann gilt sogar
\begin{align*}
\sup_{t \in I} \norm{ U_{a,T}(t)-U_T(t)} = O\Bigl( e^{-r_0 \, |\omega| \, T} \Bigr) \quad (T \to \infty).
\end{align*}
\end{thm}

\begin{proof}
Sei $x \in D$. Dann gilt 
\begin{align} \label{eq: triv adsatz 2 1}
U_{a,T}(t)x - U_T(t)x &= U_T(t,s)U_{a,T}(s)x \big|_{s=0}^{s=t} \notag \\
&= \int_s^t U(t,s)\, [P'(s),P(s)] \, U_{a,T}(s)x \, ds
\end{align}
für alle $t \in I$ und alle $T \in (0, \infty)$, denn die Abbildung $[0,t] \ni s \mapsto U_T(t,s)U_{a,T}(s)x$ ist nach Proposition~\ref{thm: char zeitentwicklung} und Lemma~\ref{lm: strong db of products} stetig differenzierbar mit 
\begin{align*}
\dds{ U_T(t,s)U_{a,T}(s)x } = U_T(t,s)\, [P'(s),P(s)] \, U_{a,T}(s)x
\end{align*}
für alle $s \in [0,t]$.

Weiter gilt nach Satz~\ref{thm: Dyson} bzw. Satz~\ref{thm: Kato} und Lemma~\ref{lm: skalierung und (M,w)-stabilität}
\begin{align*}
\norm{U_T(t,s)} \le M e^{T \omega(t-s)}
\end{align*}
für alle $(s,t) \in \Delta$ und daher nach Proposition~\ref{prop: abschätzung für gestörte zeitentw} auch
\begin{align*}
\norm{U_{a,T}(t,s)} \le M e^{(T \omega + Mc)(t-s)}
\end{align*}
für alle $(s,t) \in \Delta$, wobei wir $c := \sup_{s \in I} \norm{ [P'(s),P(s)] }$ setzen.

Also haben wir
\begin{align*}
\norm{ U_{a,T}(t)x - U_T(t)x } &\le \int_0^t M e^{T \omega (t-s)} \, c \, M e^{ (T \omega + Mc)s } \, ds \, \norm{x} \\
&\le M^2 \,c \, e^{Mc} \, e^{- | \omega | \, T t} \, t \, \norm{x}
\end{align*}
für alle $t \in I$, $T \in (0, \infty)$ und alle $x \in D$, und damit 
\begin{align} \label{eq: triv adsatz 2 2}
\norm{ U_{a,T}(t) - U_T(t) } \le M^2 \,c \, e^{Mc} \, e^{- | \omega | \, T t} \, t
\end{align}
für alle $t \in I$ und alle $T \in (0, \infty)$. 

Weil nun $[0, \infty) \ni \xi \mapsto e^{-\xi} \, \xi$ beschränkt ist, ergibt sich daraus
\begin{align*}
\sup_{t \in I} \norm{ U_{a,T}(t)-U_T(t)} = O\Bigl( \frac{1}{| \omega |\, T} \Bigr) \quad (T \to \infty),
\end{align*}
wie gewünscht.
\\

Sei schließlich $\supp P' \subset (r_0, 1]$ für eine positive Zahl $r_0$. Dann ist $U_{a,T}(t) - U_T(t) = 0$ für alle $t \in [0, r_0]$ nach \eqref{eq: triv adsatz 2 1} und nach \eqref{eq: triv adsatz 2 2} ist
\begin{align*}
\norm{ U_{a,T}(t) - U_T(t) } \le M^2 \,c \, e^{Mc} \, e^{- | \omega | \, T t} \, t \le M^2 \,c \, e^{Mc} \, e^{- | \omega | \, T \, r_0} 
\end{align*}
für alle $t \in (r_0, 1]$. Also folgt auch
\begin{align*}
\sup_{t \in I} \norm{ U_{a,T}(t)-U_T(t)} = O\Bigl( e^{-r_0 \, |\omega| \, T} \Bigr) \quad (T \to \infty)
\end{align*}
und wir sind fertig.
\end{proof}

\subsection{Standardbeispiele}

Wir werden den Adiabatensätzen in den Abschnitten~\ref{sect: adsätze mit sl} bis~\ref{sect: höhere adsätze} stets Beispiele zur Seite stellen, um aufzuzeigen, was die Sätze können und was nicht, genauer: um erstens vor Augen zu führen, in welchen Situationen die Voraussetzugen (und damit auch die Aussagen) dieser Sätze erfüllt sind, und um zweitens vor Augen zu führen, was (mit der Aussage dieser Sätze) in Situationen geschieht, in denen die Voraussetzungen nicht erfüllt sind. Insbesondere sollen die Beispiele zeigen, dass die Adiabatensätze der folgenden Abschnitte nicht trivial sind.
\\

$A$ und $P$ werden in unseren Beispielen meist von der folgenden einfachen Struktur sein (weshalb wir auch von Standardbeispielen sprechen): $A(t)$ geht für jedes $t \in I$ durch eine isometrische Ähnlichkeitstransformation $R(t)$ aus der linearen Abbildung $A_0(t)$ hervor und $P(t)$ geht durch dieselbe Ähnlichkeitstransformation aus der beschränkten Projektion $P_0(t)$ hervor. 
Wir haben dann die folgenden einfachen Zusammenhänge.

\begin{prop} \label{prop: gemeinsamkeiten der bsp}
Sei $A_0(t)$ für jedes $t \in I$ eine beschränkte lineare Abbildung in $X$, 
$P_0(t)$ eine beschränkte Projektion in $X$, $R(t)$ eine surjektive isometrische lineare Abbildung in $X$ und $t \mapsto R(t)x$ zweimal stetig differenzierbar für alle $x \in X$. Sei weiter 
\begin{align*}
A(t):= R(t)^{-1} A_0(t) R(t) \text{ \; und \; } P(t):= R(t)^{-1} P_0(t) R(t)
\end{align*}
für alle $t \in I$. 
Dann stimmt die quasikontraktive Wachstumsschranke von $A(t)$ für jedes $t \in I$ mit der von $A_0(t)$ überein und mit $t \mapsto A_0(t)x$ ist auch $t \mapsto A(t)x$ stetig differenzierbar für alle $x \in X$. Weiter ist $P(t)$ für jedes $t \in I$ eine beschränkte Projektion in $X$, die mit $A(t)$ genau dann vertauscht, wenn $P_0(t)$ mit $A_0(t)$ vertauscht, und mit $t \mapsto P_0(t)x$ ist auch $t \mapsto P(t)x$ zweimal stetig differenzierbar für alle $x \in X$.  
\end{prop}

\begin{proof}
Die quasikontraktiven Wachstumsschranken von $A(t)$ und $A_0(t)$ stimmen für alle $t \in I$ überein, weil $e^{A(t) \, . \,} = R(t)^{-1} \, e^{A_0(t) \, . \,} \, R(t)$ und weil $R(t)$ und damit auch $R(t)^{-1}$ isometrisch ist für alle $t \in I$. 

Die übrigen Aussagen folgen wegen $\sup_{t \in I} \norm{ R(t)^{-1} } < \infty$ aus Lemma~\ref{lm: strong db of products} und Lemma~\ref{lm: reg of inv}.
\end{proof}

Weiter wird in unseren Beispielen meist $X = \ell^p(I_d)$ sein für ein $d \in \natu \cup \{ \infty \}$, wobei
\begin{align*}
I_d := \begin{cases} \{1, 2, \dots, d\}, & d \in \natu \\
										 \natu, & d = \infty.
			 \end{cases}						 
\end{align*}
Wir haben dann also die kanonische Schauderbasis $\{e_n: n \in I_d \}$ (s. beispielsweise Aufgabe~IV.7.18 in~\cite{Werner: FA}) (im Fall $p = 2$ eine Orthonormalbasis) und können jede beschränkte lineare Abbildung $A$ in $X$ darstellen durch das (endliche oder unendliche) Zahlenschema
\begin{align*}
\begin{pmatrix}
{e_1}^*(A e_1)   &   {e_1}^* (A e_2)  & \cdots   & \cdots    \\
{e_2}^*(A e_1)   &  {e_2}^* (A e_2)  & \cdots   & \cdots    \\       
{e_3}^* (A e_1)  &  {e_3}^* (A e_2)  &  \cdots   & \cdots    \\      
\vdots        &  \vdots         &       &                  \\
\vdots        &  \vdots         &       &                          
\end{pmatrix},
\end{align*} 
wobei ${e_n}^*$ die beschränkte lineare Abbildung $X \to \complex$ mit $x \mapsto x_n$ (der Koeffizient in der (eindeutigen!) Reihenentwicklung von $x$ nach der Schauderbasis $\{e_n: n \in I_d \}$) bezeichnet. Wir werden im folgenden beschränkte lineare Abbildungen in $X$ mit dem sie darstellenden Zahlenschema identifizieren. So schreiben wir etwa
\begin{align*}
\begin{pmatrix} 
0             & 0             &               &        &  \\
1             & 0             & \ddots        &        & \\
0             & 1             & 0             & \ddots             &  \\
             & \ddots             & 1             & 0             & \ddots \\
              &               & \ddots        & \ddots        & \ddots
\end{pmatrix}
\text{ \; bzw. \; }
\begin{pmatrix}
0             & 1             & 0             &               &  \\
0             & 0             & 1             & \ddots        & \\
              & \ddots        & 0             & 1             & \ddots  \\
              &               & \ddots        & 0             & \ddots \\
              &               &               & \ddots        & \ddots
\end{pmatrix}
\end{align*} 
für den Shift nach rechts bzw. links, die beschränkte lineare Abbildung in $\ell^p(I_{\infty})$ also mit 
\begin{align*}
(x_1, x_2, x_3, \dots) \mapsto (0, x_1, x_2, \dots ) \text{ bzw. } (x_2, x_3, x_4, \dots )
\end{align*}
für alle $x = (x_1, x_2, x_3, \dots ) \in \ell^p(I_{\infty})$. 
Zu beachten ist dabei, dass nicht jedes beliebige unendliche Zahlenschema eine beschränkte lineare Abbildung auf $X = \ell^p(I_{\infty})$ (bzgl. der Schauderbasis $\{e_n: n \in I_{\infty} \}$) darstellt. Die Zahlenschemas in unseren Beispielen werden aber immer (wie bei den Shiftoperatormatrizen eben) ziemlich offensichtlich eine beschränkte lineare Abbildung in $X$ wiedergeben.

\begin{ex}  \label{ex: spektrum der shifts}
Sei $A_1$ bzw. $A_2$ der Shift nach rechts bzw. links auf $X := \ell^2(I_{\infty})$. Dann gilt 
\begin{align*}
\sigma(A_1) = \overline{U}_1(0) = \sigma(A_2).
\end{align*}
Warum? Zunächst ist $\norm{A_1} = 1 = \norm{A_2}$, das heißt, $\sigma(A_i) \subset \overline{U}_1(0)$. Weiter gilt für $z \in U_1(0)$, dass
\begin{align*}
x:=(1,z,z^2,z^3, \dots) \in X \text{ \; und \; } (A_1 - z) x = 0,
\end{align*}
das heißt, $U_1(0) \subset \sigma_p(A_1) \subset \sigma(A_1)$. Daraus folgt wegen $A_2^* = A_1$ und der Invarianz von $U_1(0)$ unter komplexer Konjugation, dass $U_1(0) \subset \sigma_r(A_2) \subset \sigma(A_2)$. Insgesamt sehen wir nun, dass $\sigma(A_1)$ und $\sigma(A_2)$ gleich $\overline{U}_1(0)$ sind.  $\blacktriangleleft$
\end{ex}

Schließlich führen wir noch Zahlen $\lambda_d$ ein, die in unseren Standardbeispielen oft vorkommen werden:
\begin{align*}
\lambda_d := \sup \big\{ \lambda \in \real: \Re \scprd{ x , A_d(\lambda) x } \le 0 \text{ für alle } x \in \ell^2(I_d) \big\}, 
\end{align*}
wobei 
\begin{align*}
A_d(\lambda) := 
\begin{pmatrix}
\lambda       & 1             & 0             &               &  \\
0             & \lambda       & 1             & \ddots        & \\
              & \ddots        & \lambda       & 1             & \ddots  \\
              &               & \ddots        & \lambda       & \ddots \\
              &               &               & \ddots        & \ddots
\end{pmatrix}
\end{align*}
für alle $\lambda \in \real$ und alle $d \in \natu \cup \{ \infty \}$.

\begin{prop} \label{prop: eigenschaften der lambda(d)}
Die Zahlen $\lambda_d$ liegen in dem Intervall $[-1,0]$, $\lambda_1 = 0$, $\lambda_{\infty} = -1$, es gilt $\lambda_d \longrightarrow \lambda_{\infty} \;\; (d \to \infty)$ und $\lambda_{d+1} \le \lambda_d$ für alle $d \in \natu$.
\end{prop}

\begin{proof}
Zunächst gilt
\begin{align*}
\Re \scprd{ x , A_d(-1) x } &= \sum_{k=1}^d \Re ( \overline{x}_k \, (-x_k)) + \sum_{k=1}^{d-1} \Re (\overline{x}_k \, x_{k+1}) \\
& \le - \sum_{k=1}^d |x_k|^2 + \Bigl( \sum_{k=1}^{d-1} |x_k|^2 \Bigr)^{\frac{1}{2}} \, \Bigl( \sum_{k=1}^{d-1} |x_{k+1}|^2 \Bigr)^{\frac{1}{2}} 
\le 0
\end{align*}
für alle $x \in \ell^2(I_d)$ und alle $d \in \natu \cup \{ \infty \}$, woraus wir ersehen, dass $\lambda_d \ge -1$ für alle $d \in \natu \cup \{ \infty \}$. 
Weiter gilt
\begin{align*}
\Re \scprd{ e_1 , A_d(\lambda) e_1 } = \lambda,
\end{align*}
für alle $\lambda \in \real$ und alle $d \in \natu \cup \{ \infty \}$, woraus folgt, dass $\lambda_d \le 0$ für alle $d \in \natu \cup \{ \infty \}$.

Wir sehen darüberhinaus sofort, dass $\lambda_1 = 0$, und auch $\lambda_{\infty} = -1$ ist nicht schwer einzusehen. 
Sei nämlich $\varepsilon > 0$. Dann haben wir für alle $d \in \natu$
\begin{align*}
&\Re \scprd{ e_1 + \dotsb + e_d, A_{\infty}(-1 + \varepsilon) (e_1 + \dotsb + e_d) } \\
& \qquad \qquad = \Re \scprd{ e_1 + \dotsb + e_d, A_d(-1 + \varepsilon) (e_1 + \dotsb + e_d) } = d \, \varepsilon - 1, 
\end{align*} 
was für genügend große $d$ positiv ist. Also gilt $\lambda_{\infty} \le -1 + \varepsilon$ und $\lambda_d \le -1 + \varepsilon$ für genügend große $d$, das heißt, $\lambda_{\infty} = -1$ und $\lambda_d \longrightarrow -1 \;\; (d \to \infty)$.

Schließlich bleibt zu zeigen, dass die $\lambda_d$ monoton fallend sind in $d$. Sei $\varepsilon >0$ und $d \in \natu$.  Dann existiert nach Definition von $\lambda_d$ ein Vektor $x \in \ell^2(I_d)$ mit $\Re \scprd{ x , A_d(\lambda_d + \varepsilon) x } > 0$. Also gilt für $y := (x, 0) \in \ell^2(I_{d+1})$, dass
\begin{align*}
\Re \scprd{ y , A_{d+1}(\lambda_d + \varepsilon) y } = \Re \scprd{ x , A_d(\lambda_d + \varepsilon) x } > 0,
\end{align*}
woraus hervorgeht, dass $\lambda_{d+1} \le \lambda_d + \varepsilon$ und damit, dass $\lambda_{d+1} \le \lambda_d$.
\end{proof}

Wir merken an, dass beispielsweise $\lambda_2 = - \frac{1}{2}$ und $\lambda_3 = - \frac{1}{\sqrt{2}}$ (was man mithilfe von Lagrangemultiplikatoren einsehen kann). Insbesondere gilt nach der obigen Proposition $\lambda_d \le - \frac{1}{2} < 0$ für alle $d \in \{2,3, \dots \}$.
\\

Der Witz an den Zahlen $\lambda_d$ ist, dass $A_d(\lambda_d)$ gerade noch dissipativ ist und damit (nach dem Satz von Lumer, Phillips (Satz~\ref{thm: Lumer, Phillips})) gerade noch eine Kontraktionshalbgruppe auf $\ell^2(I_d)$ erzeugt, womit wir natürlich meinen, dass die quasikontraktive Wachstumsschranke von $A_d(\lambda_d)$ gleich $0$ ist. Aus den $A_d(\lambda_d)$ können wir daher mithilfe von Proposition~\ref{prop: gemeinsamkeiten der bsp} (nichtnormale) beschränkte lineare Abbildungen $A(t)$ auf $X := \ell^2(I_{d'})$ konstruieren, sodass $A$ zwar $(1,0)$-stabil aber nicht sogar $(1, \omega)$-stabil ist für negative Zahlen $\omega$.
Solche $A(t)$ fallen also wenigstens nicht unter die etwas schwächere Version von Satz~\ref{thm: triv adsatz 2} mit $M = 1$: die Aussage des Adiabatensates ist für solche $A(t)$ also wenigstens nicht ganz trivial erfüllt.

Aber natürlich ist es denkbar, dass solche $A$ trotzdem $(M, \omega)$-stabil sind für eine negative Zahl $\omega$ und ein $M \in (1, \infty)$. Schließlich stimmt ja die Wachstumsschranke von $A_d(\lambda_d)$ mit der Spektralschranke überein (Korollar~IV.3.12 in~\cite{EngelNagel}), die für $d \in \natu$ wiederum gleich $\lambda_d$ 
und damit für $d \ne 1$ echt kleiner als $0$ ist, woraus wir ersehen, dass für jede negative Zahl $\omega \in (\lambda_d, 0)$ eine Zahl $M_{\omega}$ (notwendig echt größer als $1$) existiert, sodass 
\begin{align*}
\norm{ e^{A_d(\lambda_d) s} } \le M_{\omega} \, e^{\omega \, s}
\end{align*}
für alle $s \in [0, \infty)$.

Wir weisen indes darauf hin, dass wir in allen unseren Beispielen (zu den Adiabatensätzen) der folgenden Abschnitte -- mit Ausnahme von Beispiel~\ref{ex: adsatz mit nichtglm sl, endl viele überschneidungen} und Beispiel~\ref{ex: motivierendes bsp für erweiterten adsatz ohne sl} -- sicher sagen können, dass nicht nur keine $(1,\omega)$- sondern auch keine $(M, \omega)$-Stabilität vorliegt für negative Zahlen $\omega$. Und durch geringfügige Abwandlung der genannten Ausnahmebeispiele können wir erreichen, dass auch dort sicher keine $(M,\omega)$-Stabilität vorliegt für negative $\omega$.


\section{Adiabatensätze mit Spektrallückenbedingung}  \label{sect: adsätze mit sl}

\subsection{Adiabatensätze mit gleichmäßiger Spektrallückenbedingung}

Wir beginnen mit der folgenden wichtigen Aussage.

\begin{lm}  \label{lm: A stetig im verallg sinn}
Sei  $J$ ein Intervall und $A(t): D \subset X \to X$ für jedes $t \in J$ Erzeuger einer stark stetigen Halbgruppe auf $X$, sei $A$ $(M, \omega)$-stabil und $t \mapsto A(t)x$ stetig differenzierbar für alle $x \in D$. Dann ist $t \mapsto A(t)$ stetig im verallgemeinerten Sinn.
\end{lm}

\begin{proof}
$A(t) - (\omega+1)$ ist für jedes $t \in J$ eine bijektive abgeschlossene lineare Abbildung $D \subset X \to X$ und 
\begin{align*}
\sup_{t \in J} \norm{  \bigl( A(t) - (\omega+1) \bigr)^{-1}  } \le \frac{M}{(\omega+1) - \omega} < \infty,
\end{align*}
weil $A$ ja $(M, \omega)$-stabil ist (Satz~\ref{thm: eigenschaften von erzeugern}), und $t \mapsto \bigl( A(t) - (\omega+1) \bigr)x$ ist stetig differenzierbar für alle $x \in D$. Also ist $t \mapsto  \bigl( A(t) - (\omega+1) \bigr)^{-1} \, x$ nach Lemma~\ref{lm: reg of inv} stetig differenzierbar für alle $x \in X$ und damit $t \mapsto \bigl( A(t) - (\omega+1) \bigr)^{-1}$ nach  Lemma~\ref{lm: strong db and db} stetig, das heißt, $t \mapsto A(t)$ ist folgenstetig im verallgemeinerten Sinn nach Satz~\ref{thm: char verallg konv} und damit auch stetig im verallgemeinerten Sinn.
\end{proof}

Wir können nun einen ersten nichttrivialen Adiabatensatz beweisen. Diesen Satz (wie auch Satz~\ref{thm: handl adsatz mit sl}) haben wir aus Theorem~2.2 in Abou Salems Arbeit~\cite{Abou 07} gewonnen. Später haben wir festgestellt, dass die Sätze~\ref{thm: unhandl adsatz mit sl} und~\ref{thm: handl adsatz mit sl} und insbesondere Abou Salems Theorem~2.2 auch aus einer verallgemeinerten Version (Satz~\ref{thm: höherer adsatz}) eines Adiabatensatzes von Nenciu folgen. Wir unterstreichen aber, dass Abou Salems Satz (und damit erst recht Satz~\ref{thm: unhandl adsatz mit sl} und Satz~\ref{thm: handl adsatz mit sl}) aus Nencius ursprünglichem Satz noch \emph{nicht} folgt (Beispiel~\ref{ex: unser höherer adsatz echt allgemeiner als der von nenciu}).

Abou Salems Satz ist nur für einpunktige Untermengen $\sigma(t) = \{ \lambda(t) \}$ des Spektrums $\sigma(A(t))$ mit einem einfachen Eigenwert $\lambda(t)$ formuliert, und $t \mapsto \lambda(t)$ ist dort als einmal stetig differenzierbar vorausgesetzt, wohingegen die Sätze~\ref{thm: unhandl adsatz mit sl} und~\ref{thm: handl adsatz mit sl} für allgemeine kompakte Untermengen $\sigma(t)$ von $\sigma(A(t))$ formuliert sind und auch mit der Stetigkeit von $t \mapsto \sigma(t)$ auskommen.

\begin{thm} \label{thm: unhandl adsatz mit sl}
Sei $A(t)$ für jedes $t \in I$ eine lineare Abbildung $D \subset X \to X$, die eine stark stetige Halbgruppe auf $X$ erzeugt, sei $A$ $(M,0)$-stabil und $t \mapsto A(t)x$ stetig differenzierbar für alle $x \in D$. Sei $\sigma(t)$ für jedes $t \in I$ eine kompakte in $\sigma(A(t))$ isolierte Untermenge von $\sigma(A(t))$ und zu jedem $t_0 \in I$ gebe es einen Zykel $\gamma_{t_0}$ und eine in $I$ offene Umgebung $U_{t_0}$ von $t_0$, sodass $\im \gamma_{t_0} \subset \rho(A(t))$ und $n(\gamma_{t_0}, \sigma(t)) = 1$ und $n(\gamma_{t_0}, \sigma(A(t)) \setminus \sigma(t)) = 0$ für alle $t \in U_{t_0}$. Sei $P(t)$ für jedes $t \in I$ die Rieszprojektion von $A(t)$ auf $\sigma(t)$ und sei $t \mapsto P(t)x$ zweimal stetig differenzierbar für alle $x \in X$. Dann gilt
\begin{align*}
\sup_{t \in I} \norm{ U_{a,T}(t) - U_T(t) } = O\Bigl( \frac{1}{T} \Bigr) \quad (T \to \infty)
\end{align*}
und insbesondere
\begin{align*}
\sup_{t \in I} \norm{ (1-P(t)) U_T(t) P(0) } = O\Bigl( \frac{1}{T} \Bigr) \quad (T \to \infty).
\end{align*}
\end{thm}

\begin{proof}
Sei für jedes $t \in I$ 
\begin{align*}
B(t) := \frac{1}{2 \pi i} \, \int_{\gamma_t} (A(t)-z)^{-1} P'(t) (A(t)-z)^{-1} \, dz,
\end{align*}
wobei $\gamma_t$ ein Zykel in $\rho(A(t))$ sei mit $n(\gamma_t, \sigma(t)) = 1$ und $n(\gamma_t, \sigma(A(t)) \setminus \sigma(t)) = 0$. So ein Zykel existiert nach Voraussetzung und der Wert des obigen Wegintegrals ist (nach Satz~\ref{thm: Cauchy global}) für alle solchen Zykel gleich, weil 
\begin{align*}
\rho(A(t)) \ni z \mapsto (A(t)-z)^{-1}\, P'(t)\, (A(t)-z)^{-1}
\end{align*}
holomorph ist und je zwei solche Zykel homolog in $\rho(A(t))$ sind.

Sei nun $t_0 \in I$ und $U_{t_0}$ eine in $I$ offene Umgebung von $t_0$, sodass $\im \gamma_{t_0} \subset \rho(A(t))$ und $n(\gamma_{t_0}, \sigma(t)) = 1$ und $n(\gamma_{t_0}, \sigma(A(t)) \setminus \sigma(t)) = 0$ für alle $t \in U_{t_0}$. Wir zeigen, dass die Abbildung $U_{t_0} \ni t \mapsto B(t)x$ stetig differenzierbar ist für alle $x \in X$. 

Sei $x \in X$ und sei $V_{t_0}$ eine in $I$ offene Umgebung von $t_0$ 
mit $\overline{V}_{t_0} \subset U_{t_0}$. 
Zunächst beobachten wir, dass
\begin{align*}
B(t) &= \frac{1}{2 \pi i} \, \int_{\gamma_t} (A(t)-z)^{-1}\, P'(t)\, (A(t)-z)^{-1} \, dz \\
&= \frac{1}{2 \pi i} \, \int_{\gamma_{t_0}} (A(t)-z)^{-1}\, P'(t)\, (A(t)-z)^{-1} \, dz
\end{align*}
für alle $t \in U_{t_0}$, denn $\gamma_t$ \emph{und} $\gamma_{t_0}$ sind für jedes $t \in U_{t_0}$ Zykel in $\rho(A(t))$, die wegen
\begin{align*}
n(\gamma_{t}, \sigma(t)) = 1 = n(\gamma_{t_0}, \sigma(t)) 
\end{align*}
und
\begin{align*}
n(\gamma_t, \sigma(A(t)) \setminus \sigma(t)) = 0 = n(\gamma_{t_0}, \sigma(A(t)) \setminus \sigma(t))
\end{align*}
zudem homolog in $\rho(A(t))$ sind.
Sei $U := \{ (t,z) \in I \times \complex : z \in \rho(A(t)) \}$. Dann ist $U$ nach Satz~\ref{thm: (A(t)-z)^{-1} stetig in (t,z)} und Lemma~\ref{lm: A stetig im verallg sinn} offen in $I \times \complex$ und $U \ni (t,z) \mapsto (A(t)-z)^{-1}$ stetig, insbesondere lokal beschränkt, und daher gilt
\begin{align*}
\sup_{(t,z) \in V_{t_0} \times \im \gamma_{t_0}} \norm{ (A(t)-z)^{-1} } < \infty.
\end{align*}

Wir sehen daher mit Lemma~\ref{lm: reg of inv}, dass $V_{t_0} \ni t \mapsto (A(t)-z)^{-1}$ stark stetig differenzierbar ist für alle $z \in \im \gamma_{t_0}$, das heißt, 
\begin{align*}
V_{t_0} \ni t \mapsto f(t,z) := (A(t)-z)^{-1} \, P'(t) \, (A(t)-z)^{-1} \,x
\end{align*}
ist stetig differenzierbar für alle $z \in \im \gamma_{t_0}$ (nach Lemma~\ref{lm: strong db of products}). Weiter ist $\im \gamma_{t_0} \ni z \mapsto \partial_1 f(t,z)$ stetig 
für alle $t \in V_{t_0}$ und 
\begin{align*}
\sup_{(t,z) \in V_{t_0} \times \im \gamma_{t_0}} \norm{ \partial_1 f(t,z) } < \infty.
\end{align*}
Die Abbildung $U \ni (t,z) \mapsto \partial_1 f(t,z)$ ist nämlich stetig. Warum? Weil 
\begin{align*}
\ddt{  (A(t)-z)^{-1} \,y  } &= - (A(t)-z)^{-1} \, A'(t) \, (A(t)-z)^{-1} \,y \\
&= - (A(t)-z)^{-1} \, A'(t)A(0)^{-1} \, A(0)(A(t)-z)^{-1} \,y
\end{align*}
für alle $(t,z) \in U$ und alle $y \in X$, und weil $U \ni (t,z) \mapsto (A(t)-z)^{-1}$ stetig ist, $I \ni t \mapsto A'(t)A(0)^{-1}$ stark stetig ist und auch $U \ni (t,z) \mapsto A(0)(A(t)-z)^{-1} = \bigl( (A(t)-z) A(0)^{-1} \bigr)^{-1}$ wegen der Stetigkeit von $U \ni (t,z) \mapsto (A(t)-z) A(0)^{-1}$ (Lemma~\ref{lm: strong db and db}!) nach Lemma~\ref{lm: continuity of inv} stetig ist.

Aus Lemma~\ref{lm: vertauschung von abl und wegintegral} folgt nun, dass $V_{t_0} \ni t \mapsto \frac{1}{2 \pi i} \, \int_{\gamma_{t_0}} f(t,z) \, dz = B(t)x$ stetig differenzierbar ist mit
\begin{align*}
B'(t)x = \frac{1}{2 \pi i} \, \int_{\gamma_{t_0}} \partial_1 f(t,z) \, dz
\end{align*}
für alle $t \in V_{t_0}$. Also ist $I \ni t \mapsto B(t)x$ tatsächlich stetig differenzierbar für alle $x \in X$.  
\\

Jetzt kommen wir zum wesentlichen Schritt. Sei $x \in D$, dann haben wir für alle $t \in I$ die folgende (sogenannte) Kommutatorgleichung:
\begin{align*} \label{eq: commutator eq}
B(t)A(&t)x - A(t)B(t)x \\
=& \: \frac{1}{2\pi i} \int_{\gamma_{t}} (A(t) - z)^{-1}\, P'(t)\, (A(t) - z)^{-1} A(t) x \,dz  \\
&- \frac{1}{2\pi i} \int_{\gamma_{t}} A(t) (A(t) - z)^{-1}\, P'(t)\, (A(t) - z)^{-1} x \,dz  \\
=& \: \frac{1}{2\pi i} \int_{\gamma_{t}} (A(t)-z)^{-1} P'(t) x \,dz + \frac{1}{2\pi i} \int_{\gamma_{t}} (A(t)-z)^{-1} P'(t)\; z(A(t)-z)^{-1}x \,dz  \\
&- \frac{1}{2\pi i} \int_{\gamma_{t}} P'(t) (A(t) - z)^{-1}x \,dz - \frac{1}{2\pi i} \int_{\gamma_{t}} z(A(t) - z)^{-1} P'(t) (A(t)-z)^{-1}x \,dz  \\
= &\:  - \biggl( \frac{1}{2\pi i} \int_{\gamma_{t}} (z - A(t))^{-1} \,dz \biggr) P'(t) x + P'(t)\left( \frac{1}{2\pi i} \int_{\gamma_{t}}  (z - A(t))^{-1} \,dz \right) x  \\
= & \: P'(t)P(t)x - P(t)P'(t)x, 
\end{align*} 
wobei wir benutzt haben, dass $A(t)$ abgeschlossen ist. Insbesondere zeigt diese Gleichung, dass $B(t) D \subset D$.
Wir erhalten damit, da $t \mapsto B(t)$ -- wie eben gezeigt -- stark stetig differenzierbar ist, dass
\begin{align*}
\bigl( U_{a,T}&(t) - U_T(t) \bigr)x =  U_T (t,s) U_{a,T} (s)x \big|_{s = 0}^{s = t} = \int_0^t U_T (t,s) [P'(s), P(s)] U_{a,T} (s)x \,ds \\
= \: & \frac{1}{T} \int_0^t U_T (t,s) \Bigl( B(s)TA(s) - TA(s)B(s) \Bigr) U_{a,T} (s)x \,ds  \\
= \: & \frac{1}{T} \int_0^t U_T (t,s) \Bigl( -TA(s)B(s) + B'(s) + B(s)\bigl( TA(s) + [P'(s), P(s)] \bigr) \Bigr) U_{a,T} (s)x \,ds \\
& - \frac{1}{T} \int_0^t U_T (t,s) \Bigl( B'(s) + B(s)[P'(s), P(s)] \Bigr) U_{a,T} (s)x \,ds \\
= \: & \frac{1}{T} \int_0^t \dds{\Bigl( U_T (t,s) B(s) U_{a,T} (s)x \Bigr)} \,ds \\
&- \frac{1}{T} \int_0^t U_T (t,s) \Bigl( B'(s) + B(s)[P'(s), P(s)] \Bigr) U_{a,T} (s)x \,ds 
\end{align*}
für alle $t \in I$ und alle $T \in (0, \infty)$.
Weiter sehen wir, dass
\begin{align*}
\norm{ U_T (t,s) } \le M \text{ \; und \; } \norm{ U_{a,T} (t,s) } \le M e^{Mc (t-s) }
\end{align*}
für alle $(s,t) \in \Delta$  \emph{und alle} $T \in (0,\infty)$, weil $A$ ja nach Voraussetzung $(M,0)$-stabil ist und damit $T A$ und $T A + [P',P]$ nach Lemma~\ref{lm: skalierung und (M,w)-stabilität} und Proposition~\ref{prop: störung (M,w)-stabilität} $(M, 0)$-stabil bzw. $(M, 0 + Mc)$-stabil ist -- dabei bezeichnet $c$ eine Zahl mit
\begin{align*}
\norm{ [P'(s), P(s)] }, \norm{B(s)} \text{\,und\,} \norm{B'(s)} \le c
\end{align*}
für alle $s \in I$. Also gilt, da $x$ beliebig war in $D$, 
\begin{align*}
\sup_{t \in I}\norm{U_{a,T}(t) - U_T(t)} \le \frac{1}{T} \bigl( 2 \, Mc \,M e^{Mc} + Mc\, M e^{Mc} + Mc^2 \,Me^{Mc} \bigr) 
\end{align*}
für alle $T \in (0, \infty)$, woraus mit Satz~\ref{thm: intertwining relation} insbesondere auch
\begin{align*}
\sup_{t \in I} \norm{(1-P(t))U_T(t) P(0)} = O\Bigl( \frac{1}{T} \Bigr) \quad (T \to \infty)
\end{align*}
folgt, die $P(t)$ vertauschen ja als Rieszprojektionen von $A(t)$ tatsächlich mit $A(t)$ (nach Satz~\ref{thm: Rieszprojektion}).
\end{proof}

Im obigen Satz haben wir keine gleichmäßige Spektrallücke vorausgesetzt, sondern nur, dass $\sigma(t)$ für jedes einzelne $t \in I$ isoliert ist in $\sigma(A(t))$ -- das hat ausgereicht. Wie wir aber in Korollar~\ref{cor: sl glm unter den unhandl vor} sehen werden, ist die Spektrallücke unter den Voraussetzungen des obigen Satzes \emph{ganz von selbst} gleichmäßig.
\\

Die folgende Proposition stellt einen Zusammenhang her zwischen Isoliertheit und gleichmäßiger Isoliertheit. Sie besagt, dass eine kompakte in $\sigma(A(t))$ isolierte Untermenge $\sigma(t)$, die oberhalbstetig von $t$ abhängt, sogar gleichmäßig isoliert ist in $\sigma(A(t))$, wenn die Rieszprojektion von $A(t)$ auf $\sigma(t)$ stetig von $t$ abhängt (oder auch nur der Rang der Rieszprojektionen konstant ist in $t$) und $t \mapsto A(t)$ stetig im verallgemeinerten Sinn ist.


\begin{prop} \label{prop: zshg isoliert und glm isoliert}
Sei $J$ ein kompaktes Intervall. Sei $A(t)$ für jedes $t \in J$ eine abgeschlossene lineare Abbildung $D \subset X \to X$ und $t \mapsto A(t)$ stetig im verallgemeinerten Sinn. Sei $\sigma(t)$ für jedes $t \in J$ eine kompakte in $\sigma(A(t))$ isolierte Untermenge von $\sigma(A(t))$, $\sigma(t)$ falle an der Stelle $t_0$ in $\sigma(A(t)) \setminus \sigma(t)$ hinein und $t \mapsto \sigma(t)$ sei oberhalbstetig in $t_0$. Sei schließlich $P(t)$ für jedes $t \in J$ die Rieszprojektion von $A(t)$ auf $\sigma(t)$. Dann ist $t \mapsto P(t)$ nicht stetig in $t_0$ und 
\begin{align*}
\limsup_{n \to \infty} \bigl( \rk P(t_n) \bigr) \le \rk P(t_0) - 1
\end{align*}
für alle $(t_n)$ in $J$ mit $t_n \longrightarrow t_0$ und $\dist ( \sigma(t_n), \sigma(A(t_n)) \setminus \sigma(t_n) ) \longrightarrow 0 \quad (n \to \infty)$.
\end{prop}

\begin{proof}
Sei $\gamma_{t_0}$ ein positiv einfach geschlossener Zykel in $\rho(A(t_0))$ mit $n(\gamma_{t_0}, \sigma(t_0)) = 1$ und $n(\gamma_{t_0}, \sigma(A(t_0)) \setminus \sigma(t_0)) = 0$. So ein Zykel existiert nach Proposition~\ref{prop: Cauchy für kompakta}, weil $\sigma(t_0)$ isoliert ist in $\sigma(A(t_0))$. Sei $V := \{  z \in \complex \setminus \im \gamma_{t_0} : n(\gamma_{t_0}, z) = 1 \}$ (das Innere von $\gamma_{t_0}$) und sei 
\begin{align*}
\sigma_0(t) := \sigma(A(t)) \cap V \text{ \; sowie \; } \tau (t) := \sigma_0(t) \setminus \sigma(t)
\end{align*}
für alle $t \in J$. Dann ist $V$ eine offene 
Umgebung von $\sigma(t_0)$, das heißt, es gibt eine positive Zahl $r_0$, sodass $U_{2 r_0}(\sigma(t_0)) \subset V$. Weiter existiert (beachte Satz~\ref{thm: sigma(A(t)) oberhstet für unbeschr A(t)}) eine in $J$ offene Umgebung $U_{t_0}$ von $t_0$, sodass $\im \gamma_{t_0} \subset \rho(A(t))$ und $\sigma(t) \subset U_{r_0}(\sigma(t_0)) \subset V$ für alle $t \in U_{t_0}$, denn $t \mapsto A(t)$ ist ja stetig im verallgemeinerten Sinn (in $t_0$) und $t \mapsto \sigma(t)$ ist oberhalbstetig in $t_0$. 

Wir erhalten damit, dass $\sigma_0(t) = \sigma(t) \cup \tau(t)$ für alle $t \in U_{t_0}$, und außerdem können wir uns leicht davon überzeugen, dass $\sigma_0(t)$ und $\tau(t)$ für jedes $t \in U_{t_0}$ kompakte in $\sigma(A(t))$ isolierte Untermengen von $\sigma(A(t))$ sind. 

Sei für jedes $t \in U_{t_0}$ $P_0(t)$ die Rieszprojektion von $A(t)$ auf $\sigma_0(t)$ und $Q(t)$ die Rieszprojektion von $A(t)$ auf $\tau(t)$. Dann gilt $P_0(t) = P(t) + Q(t)$ für alle $t \in U_{t_0}$ nach Proposition~\ref{prop: rechenregeln rieszprojektion} und 
\begin{align*}
P_0(t) = \frac{1}{2 \pi i} \, \int_{\gamma_{t_0}} (z-A(t))^{-1} \, dz,
\end{align*} 
weil $\gamma_{t_0}$ für jedes $t \in U_{t_0}$ ein Zykel in $\rho(A(t))$ ist mit
\begin{align*}
n(\gamma_{t_0}, \sigma_0(t)) = n(\gamma_{t_0}, V) = 1 
\end{align*}
und 
\begin{align*}
n(\gamma_{t_0}, \sigma(A(t)) \setminus \sigma_0(t)) = n\bigl( \gamma_{t_0}, \bigl( \complex \setminus \im \gamma_{t_0} \bigr) \setminus V \bigr) = 0. 
\end{align*}
Im letzten Gleichheitszeichen haben wir benutzt, dass $\gamma_{t_0}$ ein \emph{positiv einfach geschlossener} Zykel ist. Die Abbildung $U_{t_0} \ni t \mapsto P_0(t)$ ist stetig, denn wegen der Stetigkeit von $t \mapsto A(t)$ im verallgemeinerten Sinn gilt
\begin{align*}
\sup_{(t,z) \in V_{t_0} \times \im \gamma_{t_0}} \norm{ (A(t)-z)^{-1} } < \infty
\end{align*}
für jede in $J$ offene Umgebung $V_{t_0}$ von $t_0$ mit $\overline{V}_{t_0} \subset U_{t_0}$ (beachte Satz~\ref{thm: (A(t)-z)^{-1} stetig in (t,z)}), und wir können den lebesgueschen Satz anwenden.

Sei nun $(t_n)$ eine (nach Voraussetzung wirklich existierende) Folge in $J$ mit $t_n \longrightarrow t_0$ und $\dist ( \sigma(t_n), \sigma(A(t_n)) \setminus \sigma(t_n) ) \longrightarrow 0 \quad (n \to \infty)$. Dann gilt
\begin{align*}
\sigma(A(t_n))\setminus\sigma(t_n) \cap U_{r_0}(\sigma(t_n)) \ne \emptyset
\end{align*}
für genügend große $n$ und, da $U_{r_0}(\sigma(t_n)) \subset U_{2 r_0}(\sigma(t_0)) \subset V$ für genügend große $n$ ($t \mapsto \sigma(t)$ ist oberhalbstetig in $t_0$!), gilt auch
\begin{align*}
\tau(t_n) = \sigma(A(t_n)) \setminus \sigma(t_n) \cap V \supset \sigma(A(t_n)) \setminus \sigma(t_n) \cap U_{r_0}(\sigma(t_n)) \ne \emptyset
\end{align*}
für $n$ groß genug. Also folgt $Q(t_n) \ne 0$ für $n$ groß genug, und wir sehen mithilfe von Lemma~\ref{lm: rk konst}, dass $t \mapsto Q(t)$ nicht stetig ist in $t_0$, denn $\tau(t_0) = \sigma_0(t_0)\setminus\sigma(t_0) = \emptyset$ und damit $Q(t_0) = 0$. Da $U_{t_0} \ni t \mapsto P_0(t)$ stetig ist, ist  $U_{t_0} \ni t \mapsto P(t) =  P_0(t) - Q(t)$ also nicht stetig in $t_0$. Weiter ist $\rk P_0(t_n) = \rk P(t_0)$ für $n$ groß genug (Lemma~\ref{lm: rk konst}) und $\rk P(t_n) + \rk Q(t_n) = \rk (P(t_n) + Q(t_n))$ für $n$ groß genug, weil $P(t_n)Q(t_n) = 0$ nach Proposition~\ref{prop: rechenregeln rieszprojektion}. Also haben wir
\begin{align*}
\rk P(t_n) + 1 \le \rk P(t_n) + \rk Q(t_n)  
= \rk P_0(t_0) = \rk P(t_0)
\end{align*}
für $n$ genügend groß, was die behauptete Aussage ergibt.
\end{proof}

In der Situation der obigen Proposition kann die Abbildung 
\begin{align*}
J \setminus \{t_0\} \ni t \mapsto P(t)
\end{align*}
durchaus stetig fortsetzbar sein in $t_0$ (Beispiel~\ref{ex: adsatz mit nichtglm sl, endl viele überschneidungen}). 
Sie muss es aber nicht (Beispiel~\ref{ex: reg von P wesentlich im adsatz mit nichtglm sl} und~\ref{ex: rieszproj muss nicht stetig fb sein (Rellich)}).
\\

Wie angekündigt zeigen wir nun, dass die Spektrallücke unter den Voraussetzungen des Satzes~\ref{thm: unhandl adsatz mit sl} automatisch gleichmäßig ist.

\begin{cor} \label{cor: sl glm unter den unhandl vor}
Seien $A(t)$, $\sigma(t)$ und $P(t)$ wie in Satz~\ref{thm: unhandl adsatz mit sl}. Dann ist $t \mapsto \sigma(t)$ oberhalbstetig, insbesondere ist $\sigma(t)$ gleichmäßig isoliert in $\sigma(A(t))$. 
\end{cor}

\begin{proof}
Wir zeigen, dass $t \mapsto \sigma(t)$ oberhalbstetig ist. Dann folgt mithilfe von Proposition~\ref{prop: zshg isoliert und glm isoliert}, dass $\sigma(t)$ an keiner Stelle $t_0 \in I$ in $\sigma(A(t)) \setminus \sigma(t)$ hineinfällt und daher gleichmäßig isoliert ist in  $\sigma(A(t))$. 

Wir können annehmen, dass $P(t) \ne 1$ für alle $t \in I$. Andernfalls ist nämlich $1-P(t) = 0$ für \emph{alle} $t \in I$ nach Lemma~\ref{lm: rk konst} und daher $\sigma(t) = \sigma( A(t) \big|_{P(t)D} ) \cup \sigma( A(t) \big|_{(1-P(t))D} ) = \sigma(A(t))$ trivialerweise gleichmäßig isoliert in $\sigma(A(t))$.

Sei $t_0 \in I$ und seien $\gamma_{t_0}$, $U_{t_0}$ wie in der Voraussetzung von Satz~\ref{thm: unhandl adsatz mit sl}. Dann ist 
\begin{align*}
U_{t_0} \ni t \mapsto A(t)P(t) = \frac{1}{2 \pi i} \, \int_{\gamma_t} z \, (z-A(t))^{-1} \, dz = \frac{1}{2 \pi i} \, \int_{\gamma_{t_0}} z \, (z-A(t))^{-1} \, dz
\end{align*}
stetig nach Satz~\ref{thm: (A(t)-z)^{-1} stetig in (t,z)}, $U_{t_0} \ni t \mapsto \sigma(A(t)P(t))$ also oberhalbstetig nach Proposition~\ref{prop: sigma(A(t)) oberhstet}. Weiter bemerken wir, dass
\begin{align*}
\sigma(A(t)P(t)) = \sigma \bigl( A(t)P(t) \big|_{P(t)D} \bigr) \cup \sigma \bigl( A(t)P(t) \big|_{(1-P(t))D} \bigr) = \sigma(t) \cup \{ 0 \}
\end{align*}
für alle $t \in I$ (nach Proposition~\ref{prop: zerl. des spektrums}). Also existiert zu jedem $\varepsilon > 0$ ein $\delta > 0$, sodass 
\begin{align*}
\sigma(t) \subset \sigma(t) \cup \{ 0 \} \subset U_{\varepsilon} \bigl(  \sigma(t_0) \cup \{ 0 \}  \bigr) 
\end{align*}
für alle $t \in U_{\delta} (t_0) \cap I$.
 
Wenn nun $0 \in \sigma(t_0)$, dann folgt die Oberhalbstetigkeit von $t \mapsto \sigma(t)$ in $t_0$ sofort aus der obigen Beziehung.

Wenn $0 \notin \sigma(t_0)$, erhalten wir die Oberhalbstetigkeit mit einem Widerspruchsargument. Angenommen, $t \mapsto \sigma(t)$ ist nicht oberhalbstetig in $t_0$. Dann existiert ein $\varepsilon_0 > 0$, eine Folge $(t_n)$ in $I$ mit $t_n \longrightarrow t_0 \; (n \to \infty)$ und eine Folge $(\lambda_n)$ mit $\lambda_n \in \sigma(t_n)$ und $\lambda_n \notin U_{\varepsilon_0} (  \sigma(t_0)  )$ für alle $n \in \natu$. Sei $\varepsilon \in (0 , \varepsilon_0]$, dann gilt wegen der Oberhalbstetigkeit von $t \mapsto \sigma(t) \cup \{ 0 \}$
\begin{align*}
\lambda_n \in \sigma(t_n) \subset U_{\varepsilon} \bigl(  \sigma(t_0) \cup \{ 0 \}  \bigr)   \subset   U_{\varepsilon} ( \sigma(t_0) ) \cup U_{\varepsilon}(0)
\end{align*}
und daher $\lambda_n \in U_{\varepsilon}(0)$ für genügend große $n$. Wir sehen also, dass $\lambda_n \longrightarrow 0 \quad (n \to \infty)$, woraus  wegen $\lambda_n \in \sigma(t_n) \subset \sigma(A(t_n))$ und der Stetigkeit von $t \mapsto A(t)$ im verallgemeinerten Sinn folgt, dass $0 \in \sigma(A(t_0))$ (benutze Satz~\ref{thm: (A(t)-z)^{-1} stetig in (t,z)}). Also gilt 
\begin{align*}
n(\gamma_{t_0}, 0 ) = n( \gamma_{t_0}, \sigma(A(t_0)) \setminus \sigma(t_0) ) = 0,
\end{align*}
aber andererseits gilt auch 
\begin{align*}
n(\gamma_{t_0}, 0) = \lim_{n \to \infty} n( \gamma_{t_0}, \lambda_n ) = 1,
\end{align*}
da $n(\gamma_{t_0}, \sigma(t_n)) = 1$ für alle $n \in \natu$ mit $t_n \in U_{t_0}$. Das ist ein Widerspruch und wir sind fertig. 
\end{proof}

Die Voraussetzung in Satz~\ref{thm: unhandl adsatz mit sl}, dass zu jedem $t_0 \in I$ eine Umgebung $U_{t_0}$ und ein Zykel $\gamma_{t_0}$ existiert, der für jedes $t \in U_{t_0}$ in $\rho(A(t))$ verläuft und $\sigma(t)$ einmal und $\sigma(A(t)) \setminus \sigma(t)$ keinmal umläuft, sieht nicht besonders handlich aus. Wie der nächste Satz zeigt, können wir sie durch die handlichere Voraussetzung, dass $t \mapsto \sigma(t)$ stetig ist, ersetzen -- und zwar vermutlich ohne großen Verlust an Allgemeinheit (s. Proposition~\ref{prop: zshg handl und unhandl adsatz}).

\begin{thm} \label{thm: handl adsatz mit sl}
Sei $A(t)$ für jedes $t \in I$ eine lineare Abbildung $D \subset X \to X$, die eine stark stetige Halbgruppe auf $X$ erzeugt, sei $A$ $(M,0)$-stabil und $t \mapsto A(t)x$ stetig differenzierbar für alle $x \in D$. Sei $\sigma(t)$ für jedes $t \in I$ eine kompakte in $\sigma(A(t))$ isolierte Untermenge von $\sigma(A(t))$ und $t \mapsto \sigma(t)$ sei stetig. Sei $P(t)$ für jedes $t \in I$ die Rieszprojektion von $A(t)$ auf $\sigma(t)$ und sei $t \mapsto P(t)x$ zweimal stetig differenzierbar für alle $x \in X$. Dann gilt
\begin{align*}
\sup_{t \in I} \norm{ U_{a,T}(t) - U_T(t) } = O\Bigl( \frac{1}{T} \Bigr) \quad (T \to \infty)
\end{align*}
und insbesondere
\begin{align*}
\sup_{t \in I} \norm{ (1-P(t)) U_T(t) P(0) } = O\Bigl( \frac{1}{T} \Bigr) \quad (T \to \infty).
\end{align*}
\end{thm}

\begin{proof}
Wir bemerken zunächst, dass $\sigma(t)$ nach Proposition~\ref{prop: zshg isoliert und glm isoliert} sogar gleichmäßig isoliert ist, denn $t \mapsto \sigma(t)$ ist ja (insbesondere) oberhalbstetig, $t \mapsto A(t)$ ist stetig im verallgemeinerten Sinne (nach Lemma~\ref{lm: A stetig im verallg sinn}) und $t \mapsto P(t)$ ist stetig nach Lemma~\ref{lm: strong db and db}. Also existiert eine positive Zahl $r_0$, sodass 
\begin{align*}
U_{r_0} (\sigma(t)) \setminus \sigma(t) \subset \rho(A(t))
\end{align*} 
für alle $t \in I$.

Wir zeigen nun, dass die oben angesprochene etwas unhandliche Voraussetzung aus Satz~\ref{thm: unhandl adsatz mit sl} erfüllt ist. Sei also $t_0 \in I$. Da $t \mapsto \sigma(t)$ stetig ist, existiert eine in $I$ offene Umgebung $U_{t_0}$, sodass
\begin{align*}
\sigma(t) \subset U_{\frac{r_0}{3}}(\sigma(t_0)) \text{ \; und \; } \sigma(t_0) \subset U_{\frac{r_0}{3}}(\sigma(t))
\end{align*}
für alle $t \in U_{t_0}$. 
Sei $\gamma_{t_0}$ ein Zykel in $U_{\frac{2}{3}r_0}(\sigma(t_0))   \setminus    \overline{U}_{\frac{r_0}{3}}(\sigma(t_0))$ mit $n\bigl(  \gamma_{t_0}, \overline{U}_{\frac{r_0}{3}}(\sigma(t_0))  \bigr) = 1$ und $n \bigl(  \gamma_{t_0}, \complex \setminus U_{\frac{2}{3}r_0}(\sigma(t_0))  \bigr) = 0$ (Proposition~\ref{prop: Cauchy für kompakta}!). Dann gilt
\begin{align*}
\sigma(t) \subset \overline{U}_{\frac{r_0}{3}}(\sigma(t_0))
\end{align*}
und
\begin{align*}
\sigma(A(t)) \setminus \sigma(t)  \subset  \complex \setminus U_{r_0}(\sigma(t))  \subset \complex \setminus U_{\frac{2}{3}r_0}(\sigma(t_0))
\end{align*}
für alle $t \in U_{t_0}$, 
weshalb $\gamma_{t_0}$ für $t \in U_{t_0}$ in $\rho(A(t))$ verläuft und
\begin{align*}
n(\gamma_{t_0}, \sigma(t)) = n\bigl(   \gamma_{t_0}, \overline{U}_{\frac{r_0}{3}}(\sigma(t_0))   \bigr) = 1 
\end{align*}
und
\begin{align*}
n(\gamma_{t_0}, \sigma(A(t)) \setminus \sigma(t)) = n\bigl(   \gamma_{t_0}, \complex \setminus U_{\frac{2}{3}r_0}(\sigma(t_0))  \bigr) = 0.
\end{align*}
Das zeigt, dass die Voraussetzungen von Satz~\ref{thm: unhandl adsatz mit sl} tatsächlich erfüllt sind, und dieser Satz liefert die Behauptung.
\end{proof}

Wir haben gerade im Beweis gesehen, dass die Voraussetzungen des obigen Satzes diejenigen von Satz~\ref{thm: unhandl adsatz mit sl} nach sich ziehen. Dieser Satz ist also (eher) allgemeiner als der obige etwas handlichere Adiabatensatz. Allerdings ist er wohl nicht \emph{viel} allgemeiner, wie folgende Proposition vermuten lässt.

\begin{prop} \label{prop: zshg handl und unhandl adsatz}
Seien $A(t)$, $\sigma(t)$ und $P(t)$ wie in Satz~\ref{thm: unhandl adsatz mit sl}. Sei ferner $0 \in \sigma(t)$ für alle $t \in I$ oder $0 \notin \sigma(t)$ für alle $t \in I$ und sei $t \mapsto \sigma(A(t)P(t))$ unterhalbstetig, was beispielsweise dann erfüllt ist, wenn alle $A(t)$ normal oder die $\sigma(t)$ alle endlich sind (s. auch die Charakterisierung in Proposition~\ref{prop: char sigma(A(t)) unterhstet}). Dann ist $t \mapsto \sigma(t)$ stetig, das heißt, die Voraussetzungen von Satz~\ref{thm: handl adsatz mit sl} sind erfüllt.
\end{prop}

\begin{proof}
Wir können wie im Beweis von Korollar~\ref{cor: sl glm unter den unhandl vor} zur Vereinfachung annehmen, dass $P(t) \ne 1$ für alle $t \in I$, denn andernfalls ist $P(t) = 1$ für \emph{alle} $t \in I$ und die behauptete Stetigkeit von $t \mapsto \sigma(t) = \sigma(A(t)P(t))$ folgt schon aus der vorausgestzten Unterhalbstetigkeit und Korollar~\ref{cor: sl glm unter den unhandl vor}. Aus unserer vereinfachenden Annahme folgt, dass
\begin{align*}
\sigma(t) \cup \{0\} &= \sigma(A(t)\big|_{P(t)D}) \cup \{ 0 \} = \sigma(A(t)P(t)\big|_{P(t)D}) \cup \sigma(A(t)P(t)\big|_{(1-P(t))D}) \\ 
&= \sigma(A(t)P(t)) 
\end{align*}
für alle $t \in I$.

Sei zunächst $0 \in \sigma(t)$ für alle $t \in I$.
Dann ist $t \mapsto \sigma(t) = \sigma(A(t)P(t))$ unterhalbstetig, also nach Korollar~\ref{cor: sl glm unter den unhandl vor} sogar stetig.

Sei nun $0 \notin \sigma(t)$ für alle $t \in I$. Dann ist $t \mapsto \sigma(t)$ ebenfalls unterhalbstetig und damit nach Korollar~\ref{cor: sl glm unter den unhandl vor} auch stetig, wie ein Widerspruchsargument zeigt. Angenommen nämlich, $t \mapsto \sigma(t)$ ist nicht unterhalbstetig in der Stelle $t_0 \in I$. Dann existiert ein $\varepsilon_0 > 0$ und eine Folge $(t_n)$ in $I$, sodass $t_n \longrightarrow t_0 \;\; (n \to \infty)$ und $\sigma(t_0) \not\subset U_{\varepsilon_0}(\sigma(t_n))$ für alle $n \in \natu$, und damit existiert eine Folge $(\lambda_n)$ in $\sigma(t_0)$, sodass $\lambda_n \notin U_{\varepsilon_0}(\sigma(t_n))$ für alle $n \in \natu$. Aus der Unterhalbstetigkeit von $t \mapsto \sigma(t) \cup \{0\}$ (in $t_0$) folgt nun, dass $\lambda_n \longrightarrow 0 \;\; (n \to \infty)$, und daher (wegen der Abgeschlossenheit von $\sigma(t_0)$), dass $0 \in \sigma(t_0)$. Das ist ein Widerspruch zu unserer Voraussetzung.
\\

Zuletzt: dass $t \mapsto \sigma(A(t)P(t))$ wirklich unterhalbstetig ist, wenn die $A(t)$ alle normal oder die $\sigma(t)$ alle endlich sind, folgt aus Proposition~\ref{prop: char sigma(A(t)) unterhstet}, weil $t \mapsto A(t)P(t)$ (nach dem Beweis von Korollar~\ref{cor: sl glm unter den unhandl vor}) stetig ist und mit $A(t)$ auch $A(t)P(t)$ normal ist ($P(t)$ ist dann ja nach Proposition~\ref{prop: rieszproj für normale A} gleich $P_{\sigma(t)}^{A(t)}$).
\end{proof}

\subsection{Adiabatensätze mit nichtgleichmäßiger Spektrallückenbedingung}

Wir erweitern nun die Adiabatensätze des vorangehenden Abschnitts, in denen ja $\sigma(t)$ (wegen der gleichmäßigen Spektrallücke) 
an keiner Stelle in $\sigma(A(t)) \setminus \sigma(t)$ hineinfallen konnte, auf bestimmte Situationen mit nichtgleichmäßiger Spektrallücke: nämlich auf genau solche, in denen $\sigma(t)$ an endlich vielen Stellen in $\sigma(A(t)) \setminus \sigma(t)$ hineinfällt und die Rieszprojektionen in diesen endlich vielen Stellen zweimal stark stetig differenzierbar fortgesetzt werden können. 

Wir können dabei aber keine Aussage mehr über die Konvergenzrate machen, sondern erhalten nur noch, dass überhaupt Konvergenz vorliegt. 
\\

Zunächst ein  einfaches Lemma, das besagt, dass der erste Adiabatensatz aus dem vorangehenden Abschnitt auch für beliebige Intervalle $[a,b]$ gilt (vgl. die Bemerkung vor Lemma~\ref{lm: trf. der zeitentw.}), was wenig erstaunlich ist.

\begin{lm} \label{lm: unhandl adsatz für trf intervall}
Sei $[a,b]$ ein nichttriviales kompaktes Intervall in $\real$ und seien $A(t)$, $\sigma(t)$ und $P(t)$ wie in Satz~\ref{thm: unhandl adsatz mit sl} mit $I$ überall ersetzt durch $[a,b]$. Dann gilt
\begin{align*}
\sup_{t \in [a,b]} \norm{ U_{a,T}(t,a) - U_T(t,a) } = O\Bigl( \frac{1}{T} \Bigr) \quad (T \to \infty),
\end{align*}
wobei $U_T$, $U_{a,T}$ die (nach Satz~\ref{thm: Kato} und Lemma~\ref{lm: trf. der zeitentw.} wirklich existierende) Zeitentwicklung zu $T A$ bzw. $T A + [P', P]$ bezeichnet.
\end{lm}

\begin{proof}
Sei $\tilde{A}(t) := (b-a) \, A(a+t(b-a))$, $\tilde{\sigma}(t) := (b-a) \, \sigma(a+t(b-a))$ und $\tilde{P}(t) := P(a+t(b-a))$ für alle $t \in I$. Sei weiter 
\begin{align*}
\tilde{U}_T(t,s) := U_T(a+t(b-a), a+s(b-a))
\end{align*} 
und 
\begin{align*}
\tilde{U}_{a,T}(t,s) := U_{a,T}(a+t(b-a), a+s(b-a))
\end{align*}
für alle $(s,t) \in \Delta$. Dann ist nach Lemma~\ref{lm: trf. der zeitentw.} $\tilde{U}_T$ die Zeitentwicklung zu $\bigl( (b-a) \cdot T A(a+t(b-a))  \bigr)_{t \in I} = T \tilde{A}$ und $\tilde{U}_{a,T}$ die Zeitentwicklung zu $\Bigl( (b-a) \cdot \Bigl(  T A(a+t(b-a)) + \bigl[ P'(a+ t(b-a)), P(a+t(b-a)) \bigr]  \Bigr)  \Bigr)_{t \in I} = T \tilde{A} + [\tilde{P}', \tilde{P} ]$, also die adiabatische Zeitentwicklung zu $T \tilde{A}$ und $\tilde{P}$. 

Aus den getroffenen Voraussetzungen an $A(t)$, $\sigma(t)$ und $P(t)$ folgen sehr leicht die entsprechenden Voraussetzungen für $\tilde{A}(t)$, $\tilde{\sigma}(t)$ und $\tilde{P}(t)$. Wir führen das kurz aus.  
$\tilde{A}(t)$ erzeugt für jedes $t \in I$ eine stark stetige Halbgruppe auf $X$, $\tilde{A}$ ist $(M,0)$-stabil (Lemma~\ref{lm: skalierung und (M,w)-stabilität}) und $t \mapsto \tilde{A}(t)x$ ist stetig differenzierbar für alle $x \in D$. Weiter ist $\tilde{\sigma}(t)$ 
für jedes $t \in I$ eine kompakte in $(b-a) \, \sigma(A(a+t(b-a)) = \sigma(\tilde{A}(t))$ isolierte Untermenge von $\sigma(\tilde{A}(t))$ und zu jedem $t_0 \in I$ existiert ein Zykel $\tilde{\gamma}_{t_0}$, nämlich $\tilde{\gamma}_{t_0} := (b-a) \, \gamma_{a+t_0(b-a)}$ (wobei $\gamma_{a+t_0(b-a)}$ wie in der Voraussetzung von Satz~\ref{thm: unhandl adsatz mit sl} sei), und eine in $I$ offene Umgebung $\tilde{U}_{t_0}$ von $t_0$, nämlich $\tilde{U}_{t_0} := \frac{1}{b-a} \, \bigl( U_{a+t_0(b-a)} - a \bigr)$ (wobei auch $U_{a+t_0(b-a)}$ wie in der Voraussetzung von Satz~\ref{thm: unhandl adsatz mit sl} sei), sodass
\begin{align*}
\im \tilde{\gamma}_{t_0} = (b-a) \, \im \gamma_{a+t_0(b-a)} \subset (b-a)\, \rho \bigl( A(a+t(b-a)) \bigr) = \rho( \tilde{A}(t) ),
\end{align*}
\begin{align*}
n( \tilde{\gamma}_{t_0}, \tilde{\sigma}(t) ) &=  n\bigl(  (b-a) \, \gamma_{a+t_0(b-a)}, (b-a) \, \sigma(a+t(b-a))  \bigr) \\ 
&=  n\bigl( \gamma_{a+t_0(b-a)}, \sigma(a+t(b-a)) \bigr) = 1
\end{align*}
und 
\begin{align*}
n\bigl( \tilde{\gamma}_{t_0}, \sigma(\tilde{A}(t)) \setminus \tilde{\sigma}(t) \bigr) &=  n\Bigl(  (b-a)\, \gamma_{a+t_0(b-a)}, \,(b-a)\, \bigl( \sigma(\tilde{A}(t)) \setminus \sigma(a+t(b-a)) \bigr) \Bigr) \\
&=  n\Bigl( \gamma_{a+t_0(b-a)}, \, \sigma\bigl(A(a+t(b-a))\bigr) \setminus \sigma(a+t(b-a)) \Bigr) = 0
\end{align*}
für alle $t \in \tilde{U}_{t_0}$.
Schließlich ist deswegen 
\begin{align*}
\tilde{P}(t) &= P(a+t(b-a)) = \frac{1}{2 \pi i} \, \int_{ \gamma_{a+t(b-a)} } (z-A(a+t(b-a))^{-1} \, dz \\ 
&= \frac{1}{2 \pi i} \, \int_{ \tilde{\gamma}_{t} } (w-\tilde{A}(t))^{-1} \, dw
\end{align*}
für jedes $t \in I$ die Rieszprojektion von $\tilde{A}(t)$ auf $\tilde{\sigma}(t)$ und $t \mapsto \tilde{P}(t)x$ ist zweimal stetig differenzierbar für alle $x \in X$.

Also erfüllen $\tilde{A}(t)$, $\tilde{\sigma}(t)$ und $\tilde{P}(t)$ tatsächlich die Voraussetzungen von Satz~\ref{thm: unhandl adsatz mit sl}, und dieser liefert
\begin{align*}
&\sup_{t \in I} \norm{ U_{a,T}(a+t(b-a),a) - U_T(a+t(b-a),a) } \\
&\qquad \qquad \qquad \qquad  = \sup_{t \in I} \norm{ \tilde{U}_{a,T}(t) - \tilde{U}_T(t) } = O\Bigl( \frac{1}{T} \Bigr) \quad (T \to \infty),
\end{align*}
wie behauptet.
\end{proof}

Die dem folgenden Satz zugrundeliegende Idee ist sehr einfach und entwickelt eine Idee aus Katos Artikel~\cite{Kato 50} weiter.
Wir wollen sie kurz erklären.  

Sei zur Vereinfachung $t_1$ die \emph{einzige} Stelle in $I$, in der $\sigma(t)$ in $\sigma(A(t)) \setminus \sigma(t)$ hineinfällt, und diese liege im \emph{Innern} von $I$. Sei $\varepsilon > 0$ vorgegeben. Wir zerlegen dann $I$ in drei kleinere Intervalle: ein sehr kleines, mittleres Intervall $(t_{1 \delta}^-, t_{1 \delta}^+)$, das die kritische Stelle $t_1$ enthält, und die beiden Intervalle $[0, t_{1 \delta}^-]$ und $[t_{1 \delta}^+, 1]$ links und rechts davon. 

Die Idee ist nun: wir können die Differenz $U_{a,T}(t,t_{1 \delta}^-)x - U_T(t,t_{1 \delta}^-)x$ (für $x \in D$) wie üblich durch ein Integral ausdrücken mit den Integrationsgrenzen $t_{1 \delta}^-$ und $t$ und einem Integranden, der gleichmäßig in $T \in (0, \infty)$ abgeschätzt werden kann, und diese Differenz daher beliebig klein machen, indem wir das Intervall $[t_{1 \delta}^-, t_{1 \delta}^+]$ nur genügend klein machen und $t$ nur dieses Intervall durchlaufen lassen. Dann können wir den Adiabatensatz mit \emph{gleichmäßiger} Spektrallücke nach obigem Lemma auf die beiden äußeren Intervalle anwenden (schließlich enthalten diese die einzige kritische Stelle $t_1$ nicht), und erhalten, dass die Differenzen $U_{a,T}(t,0) - U_T(t,0)$ und $U_{a,T}(t,t_{1 \delta}^+) - U_T(t,t_{1 \delta}^+)$ für $t$ im linken bzw. rechten Teilintervall beliebig klein werden, wenn wir nur $T$ genügend groß machen. Schließlich müssen wir nur noch erkennen, dass wir aus den auf diese Weise abschätzbaren Ausdrücken die uns interessierende Differenz $U_{a,T}(t) - U_T(t)$ aufbauen können.

\begin{thm} \label{thm: unhandl adsatz mit nichtglm sl}
Sei $A(t)$ für jedes $t \in I$ eine lineare Abbildung $D \subset X \to X$, die eine stark stetige Halbgruppe auf $X$ erzeugt, sei $A$ $(M,0)$-stabil und $t \mapsto A(t)x$ stetig differenzierbar für alle $x \in D$. Sei $\sigma(t)$ für jedes $t \in I$ eine kompakte in $\sigma(A(t))$ isolierte Untermenge von $\sigma(A(t))$, es gebe nur endlich viele Stellen $t_1, \dots, t_m$, in denen $\sigma(t)$ in $\sigma(A(t)) \setminus \sigma(t)$ hineinfällt, und zu jedem $t_0 \in I \setminus \{t_1, \dots, t_m \}$ gebe es einen Zykel $\gamma_{t_0}$ und eine in $I$ offene Umgebung $U_{t_0}$, sodass $\im \gamma_{t_0} \subset \rho(A(t))$ und $n( \gamma_{t_0}, \sigma(t)) = 1$ und $n(\gamma_{t_0}, \sigma(A(t)) \setminus \sigma(t)) = 0$ für alle $t \in U_{t_0}$. Sei $P(t)$ für jedes $t \in I$ eine beschränkte Projektion in $X$, für jedes $t \in I \setminus \{t_1, \dots, t_m \}$ die Rieszprojektion von $A(t)$ auf $\sigma(t)$, und sei $t \mapsto P(t)x$ zweimal stetig differenzierbar für alle $x \in X$.
Dann gilt
\begin{align*}
\sup_{t \in I} \norm{ U_{a,T}(t) - U_T(t) } \longrightarrow 0 \quad (T \to \infty)
\end{align*}
und insbesondere
\begin{align*}
(1-P) U_T P(0) \longrightarrow 0 \quad (T \to \infty) \text{ \; gleichmäßig auf } I.
\end{align*}
\end{thm}

\begin{proof}
Wir zeigen die gleichmäßige Konvergenz $U_{a,T} - U_T \longrightarrow 0 \;\; (T \to \infty)$ mit Induktion über die Anzahl $m$ der Stellen, in denen $\sigma(t)$ in $\sigma(A(t)) \setminus \sigma(t)$ hineinfällt. 

Sei $m = 1$. Sei $\varepsilon > 0$ und 
\begin{align*}
\delta := \frac{\varepsilon}{2 M^2 c \, (M e^{Mc})^2},
\end{align*}
wobei $c$ eine Zahl sei mit $\norm{ [P'(t), P(t)] } \le c$ für alle $t \in I$. Sei weiter $t_{1 \delta}^- := \max \{ 0, t_1 - \delta \}$ und $t_{1 \delta}^+ := \min \{ t_1 + \delta, 1 \}$. 

Wir können 
Lemma~\ref{lm: unhandl adsatz für trf intervall} auf das Intervall $[0,t_{1 \delta}^-]$ anwenden (vorausgesetzt es ist nicht trivial), und dieses liefert ein  $T_{\delta}^- \in (0, \infty)$, sodass
\begin{align*}
\sup_{t \in [0,t_{1 \delta}^-]} \norm{U_{a,T}(t) - U_T(t) } \le \frac{\varepsilon}{(M e^{Mc})^2} \le \varepsilon
\end{align*}  
für alle $T \in (T_{\delta}^-, \infty)$. (Wenn das Intervall $[0,t_{1 \delta}^-]$ trivial ist, können wir das Lemma zwar nicht anwenden, aber die Abschätzung gilt wegen $U_T(0), U_{a,T}(0) = 1$ natürlich trotzdem.)

Sei nun $t \in [t_{1 \delta}^-, t_{1 \delta}^+]$. Dann gilt
\begin{align*}
U_{a,T}(t) - U_T(t) = U_{a,T}(t,t_{1 \delta}^-) \bigl( U_{a,T}(t_{1 \delta}^-) - U_T(t_{1 \delta}^-) \bigr) + \bigl( U_{a,T}(t,t_{1 \delta}^-) - U_T(t, t_{1 \delta}^-) \bigr) U_T(t_{1 \delta}^-)
\end{align*}
(Proposition~\ref{thm: char zeitentwicklung}) und für die beiden Ausdrücke in Klammern haben wir 
\begin{align*}
\norm{ U_{a,T}(t_{1 \delta}^-) - U_T(t_{1 \delta}^-) } \le \frac{\varepsilon}{(M e^{Mc})^2}
\end{align*}
für alle $T \in (T_{\delta}^-, \infty)$ sowie
\begin{align*}
\norm{ U_{a,T}(t,t_{1 \delta}^-) - U_T(t, t_{1 \delta}^-) } \le M \, c \, M e^{Mc} \, \cdot 2 \delta,
\end{align*}
da für $x \in D$
\begin{align*}
U_{a,T}(t,t_{1 \delta}^-)x - U_T(t, t_{1 \delta}^-)x &= U_T(t, \tau) U_{a,T}(\tau, t_{1 \delta}^-) x \big|_{\tau=t_{1 \delta}^-}^{\tau=t} \\
&= \int_{t_{1 \delta}^-}^t U_T(t, \tau) [P'(\tau), P(\tau)] U_{a,T}(\tau, t_{1 \delta}^-)x \, d\tau
\end{align*}  
gilt und $A$ $(M,0)$-stabil ist. Also haben wir
\begin{align*}
\norm{ U_{a,T}(t) - U_T(t) } \le M e^{Mc}\, \frac{\varepsilon}{(M e^{Mc})^2} + 2 M^2 c \, M e^{Mc} \, \delta   \le 2 \varepsilon
\end{align*}
für alle $T \in (T_{\delta}^-, \infty)$.

Schließlich können wir Lemma~\ref{lm: unhandl adsatz für trf intervall} auf das Intervall $[t_{1 \delta}^+, 1]$ anwenden (wieder sofern es nicht trivial ist), und erhalten ein $T_{\delta}^+ \in (0, \infty)$, sodass
\begin{align*}
\sup_{t \in [t_{1 \delta}^+, 1]} \norm{ U_{a,T}(t,t_{1 \delta}^+) - U_T(t,t_{1 \delta}^+) } \le \frac{\varepsilon}{M}
\end{align*} 
für alle $T \in (T_{\delta}^+, \infty)$.
Da
\begin{align*}
U_{a,T}(t) - U_T(t) = U_{a,T}(t,t_{1 \delta}^+) \bigl( U_{a,T}(t_{1 \delta}^+) - U_T(t_{1 \delta}^+) \bigr) + \bigl( U_{a,T}(t,t_{1 \delta}^+) - U_T(t, t_{1 \delta}^+) \bigr) U_T(t_{1 \delta}^+)
\end{align*}
für alle $t \in [t_{1 \delta}^+, 1]$, ergibt sich damit und mit der Abschätzung für das mittlere Intervall $[t_{1 \delta}^-, t_{1 \delta}^+]$, dass
\begin{align*}
\sup_{t \in [t_{1 \delta}^+, 1]} \norm{ U_{a,T}(t) - U_T(t) } \le M e^{Mc} \Bigl( \frac{\varepsilon}{M e^{Mc}} + 2 M^2 c \, M e^{Mc} \, \delta \Bigr) + \frac{\varepsilon}{M} \, M     \le 3 \varepsilon
\end{align*}
für alle $T \in (T_{\delta}, \infty)$, wobei $T_{\delta} := \max \{ T_{\delta}^-, T_{\delta}^+ \}$.
Also gilt insgesamt
\begin{align*} 
\sup_{t \in I} \norm{U_{a,T}(t) - U_T(t) } \le 3 \varepsilon
\end{align*}
für alle $T \in (T_{\delta}, \infty)$. Das beweist die Behauptung im Fall $m =1$.
\\

Sei nun $m \in \natu$ und die behauptete Aussage sei wahr für dieses $m$. Wir müssen zeigen, dass sie dann auch für $m+1$ zutrifft.
Sei $\varepsilon > 0$ und wiederum
\begin{align*}
\delta := \frac{\varepsilon}{2 M^2 c \, (M e^{Mc})^2}
\end{align*}
($c$ wie oben). Sei weiter $t_{m+1 \delta}^- := \max \{ t_m, t_{m+1} - \delta \}$ und $t_{m+1 \delta}^+ := \min \{ t_{m+1} + \delta, 1 \}$. Da im Intervall $[0, t_{m+1 \delta}^-]$ nur noch $m$ Stellen liegen, in denen $\sigma(t)$ in $\sigma(A(t))  \setminus \sigma(t)$ hineinfällt, können wir die Induktionsvoraussetzung anwenden, und zwar auf $\tilde{A}(t) := t_{m+1 \delta}^- \, A(t \cdot t_{m+1 \delta}^-)$, $\tilde{\sigma}(t) := t_{m+1 \delta}^- \, \sigma( t \cdot t_{m+1 \delta}^-)$ und $\tilde{P}(t) := P(t \cdot t_{m+1 \delta}^-)$, $t \in I$ (vgl. den Beweis von von Lemma~\ref{lm: unhandl adsatz für trf intervall}). Wir erhalten auf diese Weise ein $T_{\delta}^- \in (0, \infty)$, sodass
\begin{align*}
\sup_{t \in [0,t_{m+1 \delta}^-]} \norm{ U_{a,T}(t) - U_T(t) } = \sup_{t \in I} \norm{  \tilde{U}_{a,T}(t) - \tilde{U}_T(t) } \le \frac{\varepsilon}{(M e^{Mc})^2} \le \varepsilon
\end{align*} 
für alle $T \in (T_{\delta}^-, \infty)$, wobei $\tilde{U}_T$, $\tilde{U}_{a,T}$ natürlich die Zeitentwicklung zu $T \tilde{A}$ bzw. die adiabatische Zeitentwicklung zu $T \tilde{A}$ und $\tilde{P}$ bezeichnet.

Weiter ergibt sich (und zwar genau wie im Induktionsanfang), dass
\begin{align*}
\sup_{t \in [t_{m+1 \delta}^-, t_{m+1 \delta}^+]} \norm{ U_{a,T}(t) - U_T(t) } \le M e^{Mc}\, \frac{\varepsilon}{(M e^{Mc})^2} + 2 M^2 c \, M e^{Mc} \, \delta   \le 2 \varepsilon
\end{align*}
für alle $T \in (T_{\delta}^-, \infty)$.

Wenden wir schließlich noch einmal Lemma~\ref{lm: unhandl adsatz für trf intervall} auf das Intervall $[t_{m+1 \delta}^+,1]$ an, so bekommen wir ein $T_{\delta}^+ \in (0, \infty)$, sodass
\begin{align*}
\sup_{t \in [t_{m+1 \delta}^-, t_{m+1 \delta}^+]} \norm{ U_{a,T}(t) - U_T(t) } \le  M e^{Mc} \Bigl( \frac{\varepsilon}{M e^{Mc}} + 2 M^2 c \, M e^{Mc} \, \delta \Bigr) + \frac{\varepsilon}{M} \, M     \le 3 \varepsilon
\end{align*}
für alle $T \in (T_{\delta}, \infty)$, wobei $T_{\delta} := \max \{ T_{\delta}^-, T_{\delta}^+ \}$.

Also gilt insgesamt wieder
\begin{align*}
\sup_{t \in I} \norm{U_{a,T}(t) - U_T(t) } \le 3 \varepsilon,
\end{align*}
für alle $T \in (T_{\delta}, \infty)$, was belegt, dass die behauptete Aussage auch für $m+1$ stimmt.
\\

Zum Schluss müssen wir uns noch davon überzeugen, dass $P(t) A(t) \subset A(t) P(t)$ für \emph{alle} $t \in I$, 
um Satz~\ref{thm: intertwining relation} anwenden zu können und 
\begin{align*}
(1-P) U_T P(0) \longrightarrow 0 \quad (T \to \infty) \text{ \; gleichmäßig auf } I
\end{align*}
aus dem eben Bewiesenen folgern zu können.

Für $t \in I \setminus \{t_1, \dots, t_m \}$ ist die Vertauschbarkeit von $A(t)$ und $P(t)$ klar, weil $P(t)$ für diese $t$ die Rieszprojektion von $A(t)$ auf $\sigma(t)$ ist.
Sei also $t = t_i$ für ein $i \in \{1, \dots, m \}$ und sei $x \in D$. Dann existiert eine Folge $(t_{i n})$ in $I \setminus \{t_1, \dots, t_m \}$ mit $t_{i n} \longrightarrow t_i \; \; (n \to \infty)$. Wegen der (starken) Stetigkeit von $P$ erhalten wir $P(t_{i n})x \longrightarrow P(t_i)x \; \; (n \to \infty)$. Außerdem gilt $P(t_{i n})x \in D$ für alle $n \in \natu$, da ja $P(t_{i n}) A(t_{i n}) \subset A(t_{i n}) P(t_{i n})$, und es gilt
\begin{align*}
A(t_i)P(t_{i n})x &= A(t_{i n}) P(t_{i n})x + \bigl( A(t_i) - A(t_{i n}) \bigr) P(t_{i n})x \\
&= P(t_{i n}) A(t_{i n})x + \Bigl( \bigl(A(t_i)-1\bigr)\bigl(A(t_{i n})-1\bigr)^{-1} - 1 \Bigr) P(t_{i n}) \bigl(A(t_{i n})-1\bigr)x \\
&\longrightarrow P(t_i) A(t_i)x \quad (n \to \infty),
\end{align*}
denn 
\begin{align*}
\norm{ \bigl(A(t_i)-1\bigr)\bigl(A(t_{i n})-1\bigr)^{-1} - 1  } \le c \, |t_i - t_{i n}|
\end{align*}
für alle $n \in \natu$ (nach dem Beweis von Satz~\ref{thm: Kato}) und $P(t_{i n}) A(t_{i n})x \longrightarrow P(t_i) A(t_i)x \; \; (n \to \infty)$ nach Lemma~\ref{lm: strong db of products}. Also folgt wegen der Abgeschlossenheit von $A(t_i)$, dass $P(t_i)x \in D$ und $A(t_i)P(t_i)x = P(t_i) A(t_i)x$, wie gewünscht.
\end{proof}

Wie im vorigen Abschnitt (Beweis von Satz~\ref{thm: handl adsatz mit sl}) erhalten wir eine schwächere, dafür aber auch etwas handlichere Version des obigen Satzes, indem wir die etwas sperrige Voraussetzung über die Zykel $\gamma_{t_0}$ und die Umgebungen $U_{t_0}$ ersetzen durch die kurze Bedingung
\begin{align*}
t \mapsto \sigma(t) \text{\, ist stetig in allen Stellen } t_0 \in I \setminus \{t_1, \dots, t_m \}.
\end{align*}

\subsection{Wie weit reichen die Sätze? Beispiele}

Wir beginnen mit einem sehr einfachen Beispiel, in dem $X$ endlichdimensional ist.

\begin{ex}  \label{ex: adsatz mit sl, X endlichdim}
Sei $X := \ell^2(I_3)$. Sei 
\begin{align*}
A_0 := \begin{pmatrix}
0         & 0     & 0         \\
0    & \lambda_2  & 1     \\
0     & 0					& \lambda_2      
\end{pmatrix}, \;
R(t):= \begin{pmatrix}
\cos t  & \sin t & 0    \\
-\sin t & \cos t & 0     \\
0       & 0      & 1   
\end{pmatrix},
\end{align*}
$\sigma(t) := \{ 0 \}$ für alle $t \in I$ und $P_0$ die orthogonale Projektion auf $\spn \{e_1\}$. Sei 
\begin{align*}
A(t) := R(t)^* A_0 R(t) \text{ \; und \; } P(t):= R(t)^* P_0 R(t)
\end{align*}
für alle $t \in I$. Dann erfüllen $A$, $\sigma$, $P$ nach Proposition~\ref{prop: gemeinsamkeiten der bsp} die Voraussetzungen von Satz~\ref{thm: handl adsatz mit sl}, denn $A_0$ erzeugt als direkte Summe der Kontraktionshalbgruppenerzeuger 
\begin{align*}
0 \text{ (auf } \spn\{e_1\}) \text{ \; und \; } 
\begin{pmatrix} 
\lambda_2  & 1     \\
0					& \lambda_2
\end{pmatrix}
\text{ (auf } \spn\{e_2,e_3\})
\end{align*}
selbst eine Kontraktionshalbgruppe auf $X$ und $P_0$ ist (wie man sofort nachrechnet oder mithilfe von Proposition~\ref{prop: rieszproj eind} erschließt) die Rieszprojektion von $A_0$ auf $\{0\}$ und daher ist
\begin{align*}
P(t) &= R(t)^* \Bigl( \frac{1}{2 \pi i} \, \int_{\partial U_{\frac{1}{2}}(0)} (z-A_0)^{-1} \, dz \Bigr) R(t) \\
&= \frac{1}{2 \pi i} \, \int_{\partial U_{\frac{1}{2}}(0)} \bigl( z- R(t)^* A_0 R(t) \bigr)^{-1} \, dz
\end{align*}   
die Rieszprojektion von $A(t)$ auf $\sigma(t)$ für alle $t \in I$. 

Also ist die Aussage von Satz~\ref{thm: handl adsatz mit sl} erfüllt, aber nicht schon trivialerweise: erstens ist $t \mapsto P(t)$ nicht konstant, weil
\begin{align*}
t \mapsto P(t)X = \spn\{ R(t)^* e_1 \}
\end{align*}
nicht konstant ist, und zweitens ist $A$ für keine negative Zahl $\omega$ $(M, \omega)$-stabil, weil dazu $\sigma(A(t)) = \{0, \lambda_2\}$ in der offenen linken Halbebene $\{ z \in \complex: \Re z < 0 \}$ enthalten sein müsste. Wir erhalten die Aussage des Adiabatensatzes also nicht schon aus den trivialen Adiabatensätzen (Satz~\ref{thm: triv adsatz 1} und Satz~\ref{thm: triv adsatz 2}) des vorigen Abschnitts. $\blacktriangleleft$
\end{ex}

Jetzt ein Beispiel, in dem $\sigma(t) = \{\lambda(t)\}$ mit Spektralwerten $\lambda(t)$, die keine Eigenwerte sind. Insbesondere zeigt dieses Beispiel, dass Satz~\ref{thm: handl adsatz mit sl} und erst recht Satz~\ref{thm: unhandl adsatz mit sl} echt allgemeiner ist als der Adiabatensatz (Theorem~2.2) aus Abou Salems Artikel~\cite{Abou 07}.

\begin{ex}  \label{ex: adsatz mit sl, lambda spektralwert aber kein eigenwert}
Sei $X := \ell^2(I_{\infty})$. Sei 
\begin{align*}
A_0 := \begin{pmatrix} 
0       & 0             &         &         &         &  \\
0       &-1             & 0             &         &         &  \\
0       & \frac{1}{2!}  & -1            & 0             &         &   \\
        & 0             & \frac{1}{3!}  & -1            & 0             & \\
        &               & 0             & \frac{1}{4!}  & -1            & \ddots \\
        &               &               & \ddots        & \ddots        & \ddots
\end{pmatrix}, \; 
R(t) := \begin{pmatrix}
\cos t  & \sin t & 0        & \cdots\\
-\sin t & \cos t & 0        &      \\
0       & 0      & 1        & \ddots\\
\vdots  &         & \ddots   & \ddots
\end{pmatrix},
\end{align*}
$\sigma(t) := \{-1\}$ für alle $t \in I$ und $P_0$ die orthogonale Projektion auf $\overline{\spn}\{e_2, e_3, \dots \}$. Sei 
\begin{align*}
A(t) := R(t)^* A_0 R(t) \text{ \; und \; } P(t):= R(t)^* P_0 R(t)
\end{align*}
für alle $t \in I$. 

Wir sehen sofort, dass $-1$ kein Eigenwert von $A_0$ ist, wohl aber ein Spektralwert, da $A_0 - (-1)$ nicht surjekiv ist. Wir zeigen nun, dass $\sigma(A_0) = \{0,-1\}$, woraus hervorgeht, dass $-1$ ein isolierter Spektralwert von $A_0$ und damit auch von jedem $A(t)$ ist.

Zunächst gilt $A_0 = A_1 \oplus A_2$, wobei $A_1 := 0$ auf $\spn\{e_1\}$ und $A_2 := -1 + B$ auf $\overline{\spn}\{e_2, e_3, \dots \}$ und
\begin{align*}
B := \begin{pmatrix}
0               & 0          &                  &      & \\
\frac{1}{2!}  & 0            & 0             &         &   \\
0               & \frac{1}{3!}  & 0           & 0             & \\
                & 0             & \frac{1}{4!}  & 0            & \ddots \\
                &               & \ddots        & \ddots        & \ddots
\end{pmatrix}.
\end{align*}
Wegen
\begin{align*}
B^n \, x &= \Bigl( 0, \dots, 0, \, \frac{x_1}{(n+1)! \dotsm 2!} \, , \frac{x_2}{(n+2)! \dotsm 3!} \, , \frac{x_3}{(n+3)! \dotsm 4!} \,, \dots \Bigr) 
\end{align*}
haben wir
\begin{align*}
\norm{B^n \, x}^2 = \sum_{k=1}^{\infty} \biggl( \frac{1}{(n+k)! \dotsm (1+k)!} \biggr)^2 \, |x_k|^2 \le \Bigl( \frac{1}{n!} \Bigr)^2 \, \norm{x}^2
\end{align*}
für alle $x \in X$ und alle $n \in \natu$. Also gilt
\begin{align*}
\norm{B^n} \le \frac{1}{n!}
\end{align*}
für alle $n \in \natu$ und damit gilt 
\begin{align*}
r_B = \lim_{n \to \infty} \norm{B^n}^{\frac{1}{n}} = 0
\end{align*}
für den Spektralradius $r_B$ von $B$. Daraus folgt, dass $\sigma(B) = \{0\}$ und daher 
\begin{align*}
\sigma(A_0) = \sigma(A_1) \cup \sigma(A_2) = \{0\} \cup (-1 + \sigma(B)) = \{0,-1\}
\end{align*}
nach Proposition~\ref{prop: zerl. des spektrums}, 
wie gewünscht.

Schließlich erzeugt $A_0$ eine Kontraktionshalbgruppe (was man wie im Beweis von Proposition~\ref{prop: eigenschaften der lambda(d)} mithilfe des Satzes von Lumer, Phillips einsehen kann) und $P_0$ ist nach Proposition~\ref{prop: rieszproj eind} die Rieszprojektion von $A_0$ auf $\{-1\}$. 

Anhand von Proposition~\ref{prop: gemeinsamkeiten der bsp} sieht man nun, dass alle Voraussetzungen von Satz~\ref{thm: handl adsatz mit sl} erfüllt sind, und zwar wieder nicht schon trivialerweise. $\blacktriangleleft$
\end{ex}

Im nächsten Beispiel ist $\sigma(t)$ nicht mehr einpunktig.

\begin{ex}  \label{ex: adsatz mit sl, sigma nicht einpunktig}
Sei $X := \ell^2(I_{\infty})$. Sei
\begin{align*}
A_0 := \begin{pmatrix} 
i       & 0         & 0              &                &                  &             &       \\
0       & -i        & 0             & 0               &                 &               & \\
        & 0         &-1             & 1             & 0               &               &     \\
        &           & 0             & -1            & 1             & 0               &  \\
			  & 			    &               & 0             & -1            & 1             & \ddots \\
			  & 			    & 			        &               & 0             & -1            & \ddots \\
			  & 			    & 			        & 				       &              & \ddots        & \ddots
\end{pmatrix}, \;
R(t):= \begin{pmatrix}
1          & 0       & 0         & 0        & \cdots \\
0          &\cos t   & \sin t    & 0        & \cdots\\
0          & -\sin t & \cos t    & 0        &        \\
\vdots     & 0       & 0         & 1        & \ddots\\
\vdots     & \vdots &            & \ddots   & \ddots
\end{pmatrix},
\end{align*}
$\sigma(t):= \overline{U}_1(-1)$ für alle $t \in I$ und $P_0$ die orthogonale Projektion auf $\overline{\spn}\{e_3, e_4, \dots \}$. Sei 
\begin{align*}
A(t) := R(t)^* A_0 R(t) \text{ \; und \; } P(t):= R(t)^* P_0 R(t)
\end{align*}
für alle $t \in I$. 

Zunächst haben wir die Zerlegung $A_0 = A_1 \oplus A_2$ in 
\begin{align*}
A_1 := \begin{pmatrix}
i        & 0          \\
0        & i
\end{pmatrix}
\text{ \; und \; }
A_2 := -1 +  B := \begin{pmatrix}
-1             & 1             & 0               &               &     \\
0             & -1            & 1             & 0               &  \\
      & 0             & -1            & 1             & \ddots \\
			&               & 0             & -1            & \ddots \\
      & 				       &              & \ddots        & \ddots
\end{pmatrix}
\end{align*}
auf $\spn\{e_1, e_2\}$ bzw. $\overline{\spn} \{ e_3, e_4, \dots \}$, woraus ersichtlich wird, dass $A_0$ eine Kontraktionshalbgruppe auf $X$ erzeugt (beachte, dass $A_2 = A_{\infty}(-1)$, und s. den Beweis von Proposition~\ref{prop: eigenschaften der lambda(d)}).

Außerdem sehen wir, dass 
\begin{align*}
\sigma(A_0) &= \sigma(A_1) \cup \sigma(A_2) = \{-i, i \} \cup (-1 + \sigma(B)) = \{-i,i\} \cup \bigl( -1 + \overline{U}_1(0) \bigr) \\
&= \{-i,i\} \cup \overline{U}_1(-1)
\end{align*}
nach Proposition~\ref{prop: zerl. des spektrums} und Beispiel~\ref{ex: spektrum der shifts}, $\sigma(t)$ ist also gleichmäßig isoliert in $\sigma(A(t))$.

Schließlich folgt mithilfe von Proposition~\ref{prop: rieszproj eind}, dass $P_0$ die Rieszprojektion von $A_0$ auf $\overline{U}_1(-1)$ ist, und damit
$P(t)$ die Rieszprojektion von $A(t)$ auf $\sigma(t)$. 

Aus Proposition~\ref{prop: gemeinsamkeiten der bsp} entnehmen wir nun, dass alle Voraussetzungen von Satz~\ref{thm: handl adsatz mit sl} erfüllt sind, dessen Aussage also auch -- aber wiederum nicht schon trivialerweise -- erfüllt ist. $\blacktriangleleft$
\end{ex}

Schließlich ein Beispiel, in dem die Spektrallücke nicht gleichmäßig ist: dort fällt ein (nichthalbeinfacher) Eigenwert an endlich vielen Stellen in das übrige Spektrum hinein. Dieses Beispiel zeigt zugleich, dass Satz~\ref{thm: unhandl adsatz mit nichtglm sl} echt allgemeiner ist als der Satz in Katos Arbeit~\cite{Kato 50}, wo die $A(t)$ als schiefselbstadjungiert vorausgesetzt sind und damit jeder Eigenwert halbeinfach ist.

\begin{ex}  \label{ex: adsatz mit nichtglm sl, endl viele überschneidungen}
Sei $X := \ell^2(I_3)$. Sei $\lambda$ eine stetig differenzierbare Abbildung $I \to (-\infty, \lambda_2]$, sodass $\lambda(t) = \lambda_2$ für endlich viele $t \in I$, sei
\begin{align*}
A_0(t) := \begin{pmatrix}
\lambda_2         & 0     & 0         \\
0    & \lambda(t)  & 1     \\
0     & 0					& \lambda(t)      
\end{pmatrix}, \;
R(t):= \begin{pmatrix}
\cos t  & \sin t & 0    \\
-\sin t & \cos t & 0     \\
0       & 0      & 1   
\end{pmatrix},
\end{align*}
$\sigma(t) := \{ \lambda(t) \}$ für alle $t \in I$ und $P_0$ die orthogonale Projektion auf $\spn \{ e_2, e_3 \}$. Sei 
\begin{align*}
A(t) := R(t)^* A_0(t) R(t) \text{ \; und \; } P(t):= R(t)^* P_0 R(t)
\end{align*}
für alle $t \in I$.

Dann ist $\sigma(t)$ für jedes $t \in I$ isoliert in $\{ \lambda_2, \lambda(t) \} = \sigma(A(t))$, die Abbildung $t \mapsto \sigma(t)$ ist stetig und $\sigma(t)$ fällt nur an endlich vielen Stellen, nämlich genau an den Schnittstellen von $\lambda$ und $\lambda_2$, in $\sigma(A(t)) \setminus \sigma(t)$ hinein. Weiter ist $P(t)$ für alle anderen Stellen die Rieszprojektion von $A(t)$ auf $\sigma(t)$ (nach Proposition~\ref{prop: rieszproj eind}). 

Wir sehen damit (und mit Proposition~\ref{prop: gemeinsamkeiten der bsp}), dass alle Voraussetzungen von Satz~\ref{thm: unhandl adsatz mit nichtglm sl} erfüllt sind. 
Die Aussage dieses Satzes gilt allerdings nicht schon ganz trivialerweise, denn $t \mapsto P(t)$ ist nicht konstant und $A(t)$ erzeugt für jedes $t \in I$ nur gerade noch eine Kontraktionshalbgruppe. $\blacktriangleleft$
\end{ex}

Wir zeigen mit dem folgenden Beispiel, dass die Voraussetzung der $(M,0)$-Stabilität für Satz~\ref{thm: unhandl adsatz mit sl}, Satz~\ref{thm: handl adsatz mit sl} und Satz~\ref{thm: unhandl adsatz mit nichtglm sl} entscheidend ist: wenn $A$ nicht $(M,0)$-stabil (sogar wenn $A$ nicht $(1,0)$-stabil) ist, braucht die Aussage dieser Sätze nicht zu gelten. Zu vergleichen ist dieser Sachverhalt mit Satz~\ref{thm: Kato}: für die Existenz der Zeitentwicklung zu $T A$ genügt nach diesem Satz auch $(M, \omega)$-Stabilität für irgendein $\omega \in \real$.

\begin{ex}  \label{ex: (M,0)-stabilität wesentlich in den adsätzen mit sl}
Sei $X := \ell^2(I_2)$. Sei $\lambda$ eine stetig differenzierbare Abbildung $I \to [0, \infty)$, die an keiner Stelle bzw. nur an endlich vielen Stellen $0$ wird, sei  
\begin{align*}
A_0(t) := \begin{pmatrix}
\lambda(t)   & 0      \\
0         & 0 
\end{pmatrix}, \;
R(t) := \begin{pmatrix}
\cos (2\pi t) & \sin (2\pi t) \\
-\sin (2\pi t) & \cos (2\pi t)
\end{pmatrix},
\end{align*}
$\sigma(t) := \{\lambda(t)\}$ für alle $t \in I$ und $P_0$ die orthogonale Projektion auf $\spn\{e_1\}$. Sei weiter wie üblich
\begin{align*}
A(t) := R(t)^* A_0(t) R(t) \text{ \; und \; } P(t):= R(t)^* P_0 R(t)
\end{align*}
für alle $t \in I$.

Dann sind alle Voraussetzungen von Satz~\ref{thm: handl adsatz mit sl} (erst recht die von Satz~\ref{thm: unhandl adsatz mit sl}) bzw. von Satz~\ref{thm: unhandl adsatz mit nichtglm sl} erfüllt mit der einzigen Ausnahme, dass $A$ hier nicht $(M,0)$-stabil ist, weil sonst $\sigma(A(t))$ in $\{ z \in \complex: \Re z \le 0 \}$ enthalten wäre für alle $t \in I$. $A$ ist hier nur $(1, \omega)$-stabil für $\omega := \sup_{t \in I} \lambda(t)$.

Und tatsächlich geht hier die Aussage des Adiabatensatzes schief, wie wir nun zeigen werden. 
Zunächst gilt
\begin{align*}
A(t) =  \lambda(t) \begin{pmatrix} 
\cos^2 (2\pi t) & \cos (2\pi t) \sin (2\pi t) \\
\cos (2\pi t) \sin (2\pi t) & \sin^2 (2\pi t)
\end{pmatrix}
\end{align*}
für alle $t \in I$ und weiter gilt für die Zeitentwicklung $U_T$ zu $T A$ nach Satz~\ref{thm: Dyson} 
\begin{align*}
U_T(t)x = x + T \, \int_0^t A(&t_1)x \,dt_1 +  T^2 \, \int_0^t \int_s^{t_1} A(t_1)A(t_2)x \, dt_2 \, dt_1 + \dotsb
\end{align*}
für alle $x \in X$.

Wir sehen daraus, dass die linearen Abbildungen $A(t)$ und $U_T(t)$ positiv sind für alle $t \in [0,\frac{1}{4}]$ und alle $T \in (0, \infty)$, das heißt, sie überführen $\{x \in \ell^2(I_2, \real): x \ge 0 \}$ (den positiven Kegel des Banachverbands $\ell^2(I_2,\real)$ (s. etwa Bemerkung~X.4.9 in~\cite{AmannEscher})) in sich. 

Außerdem ist $P(t)$ für jedes $t \in I$ die orthogonale Projektion auf $\spn\{R(t)^* e_1 \}$ und damit $1-P(t)$ die orthogonale Projektion auf 
\begin{align*}
\spn\{R(t)^* e_2 \} = \left\{   \begin{pmatrix} -\sin(2\pi t) \\ \cos(2\pi t) \end{pmatrix}   \right\}.
\end{align*}

Sei nun $t_0 := \frac{1}{4}$. Dann gilt $1-P(t_0) = P_0$ und daher
\begin{align*}
(1-P(t_0))U_T (t_0)P(0)e_1 = P_0 \, U_T(t_0) e_1 = \scprd{e_1, U_T(t_0) e_1} e_1
\end{align*} 
für alle $T \in (0, \infty)$. Wegen der Positivität von $U_T(t_0)$ und $A(t)$ für $t \in [0, \frac{1}{4}]$ folgt, dass
\begin{align*}
&\norm{ (1-P(t_0))U_T (t_0)P(0)e_1 } \\
& \qquad \quad = |\scprd{e_1, U_T(t_0) e_1}| = \scprd{e_1, U_T(t_0) e_1} \\
& \qquad \quad = 1 + T \int_0^{t_0} \scprd{e_1, A(t_1)e_1} dt_1 + T^2 \int_0^{t_0} \int_0^{t_1} \scprd{e_1, A(t_1)A(t_2)e_1} dt_2 dt_1 + \dotsb \\
& \qquad \quad \ge 1 + T \int_0^{t_0} \scprd{e_1, A(t_1)e_1} dt_1 = 1 + T \int_0^{t_0} \lambda(t) \cos^2(2\pi t_1) dt_1,
\end{align*} 
was nicht gegen $0$ konvergiert für $T$ gegen $\infty$. Also ist die Aussage des Adiabatensatzes hier nicht erfüllt. $\blacktriangleleft$
\end{ex}

Im folgenden Beispiel sind die $A(t)$ Multiplikationsoperatoren $M_{f_t}$ auf $L^2(\real, \complex)$. Die beiden Voraussetzungen von Satz~\ref{thm: handl adsatz mit sl}, dass $t \mapsto A(t)g$ stetig differenzierbar ist und dass $t \mapsto P(t)g$ zweimal stetig differenzierbar ist für alle $g \in L^2(\real, \complex)$, sind hier nicht erfüllt und auch die Aussage dieses Satzes fällt durch.

\begin{ex}  \label{ex: A(t)= multop, mit sl}
Sei $X := L^2(\real, \complex)$. Sei 
\begin{align*}
f_t := f_0(\, . \, + t) \text{ \; und \; } f_0 := i \, \chi_{[-1,1]},
\end{align*}
sei $A(t) := M_{f_t}$ auf $X$, $\sigma(t) := \{0\}$ und $P(t)$ die Rieszprojektion von $A(t)$ auf $\sigma(t)$ für alle $t \in I$ (beachte, dass $\sigma(t)$ wirklich isoliert ist in $\{ 0, i \} = \sigma(A(t))$. 

Dann ist $A(t)$ für jedes $t \in I$ beschränkt und schiefselbstadjungiert und auch die übrigen Voraussetzungen von Satz~\ref{thm: handl adsatz mit sl} (insbesondere die von Satz~\ref{thm: unhandl adsatz mit sl}) sind erfüllt bis auf zwei: zum einen ist $t \mapsto A(t)g$ nur stetig für alle $g \in X$, aber nicht sogar stetig differenzierbar, und zum andern $t \mapsto P(t)g$ ist nur stetig für alle $g \in X$, aber nicht sogar zweimal stetig differenzierbar (beispielsweise nach Lemma~\ref{lm: multiplikationsop mit char fkt nur dann stark db nach t wenn schon konst}), denn
\begin{align*}
\bigl( P(t)g \bigr)(x) &= \Bigl(   \frac{1}{2 \pi i} \, \int_{ \partial U_{\frac{1}{2}}(0) } (z-A(t))^{-1} \, dz   \Bigr)(x)
= \frac{1}{2 \pi i} \, \int_{ \partial U_{\frac{1}{2}}(0) }  \frac{1}{z-f_t(x)} \, dz \, g(x) \\
&= \chi_{ \{f_t = 0 \} }(x) \, g(x) 
\end{align*}
für alle $g \in X$ und alle $x \in \real$, und $t \mapsto \{f_t = 0 \}$ ist nicht konstant.

Wir zeigen nun, dass dies die Aussage von Satz~\ref{thm: handl adsatz mit sl} (genauer die speziellere Aussage, dass $(1-P) U_T P(0)$ und $P U_T (1-P(0))$ beide für $T$ gegen $\infty$ gleichmäßig auf $I$ gegen $0$ konvergieren) schon zerstört.
Die $A(t)$ vertauschen offensichtlich paarweise, weshalb für die Zeitentwicklung $U_T$ zu $T A$ nach Korollar~\ref{cor: Dyson} gilt:
\begin{align*}
U_T(t) = e^{T \int_0^t A(\tau) \, d\tau}
\end{align*}
für alle $t \in I$ und alle $T \in (0, \infty)$.
Wir haben daher wegen $f_t(\real) \subset i \, \real$, dass
\begin{align*}
&\norm{  (1-P(t))U_T(t)P(0)g - P(t)U_T(t)(1-P(0))g }^2 \\
& \qquad = \int \Big| \bigl( (1-\chi_{E_t}) \, e^{ T \int_0^t f_{\tau}\,d\tau } \, \chi_{E_0} \, g \bigr)(x) - \bigl( \chi_{E_t} \, e^{T \int_0^t f_{\tau}\,d\tau } \, (1-\chi_{E_0}) \, g \bigr)(x) \Big|^2 \, dx \\
& \qquad = \int \Big| \bigl( (1-\chi_{E_t}) \,  \chi_{E_0} \, g \bigr)(x) - \bigl( \chi_{E_t} \,  (1-\chi_{E_0}) \, g \bigr)(x) \Big|^2 \, dx 
= \norm{ P(t)g - P(0)g }^2
\end{align*}
für alle $T \in (0, \infty)$, alle $t \in I$ und alle $g \in X$. Weil nun die rechte Seite nicht von $T$ abhängt, kann die linke Seite nur dann für alle $g \in X$ und alle $t \in I$ gegen $0$ konvergieren für $T$ gegen $\infty$, wenn 
\begin{align*}
P(t)g = P(0)g 
\end{align*} 
für alle $t \in I$ und alle $g \in X$, mit anderen Worten: wenn $t \mapsto P(t)$ konstant ist. Da nun $P$ nach Wahl aber eben nicht konstant ist, können $(1-P) U_T P(0)$ und $P U_T (1-P(0))$ nicht beide gleichmäßig auf $I$ gegen $0$ konvergieren für $T$ gegen $\infty$, wie gewünscht. $\blacktriangleleft$
\end{ex}

Wir weisen darauf hin, dass allgemein die Voraussetzungen von Satz~\ref{thm: handl adsatz mit sl} und Satz~\ref{thm: unhandl adsatz mit sl} für $A(t) = M_{f_t}$ nur erfüllt sein können, wenn $t \mapsto P(t)$ konstant ist (wenn wir also auch schon den trivialen Adiabatensatz, Satz~\ref{thm: triv adsatz 1}, anwenden können). Sei nämlich $X := L^2(X_0,\complex)$ für einen Maßraum $(X_0, \mathcal{A}, \mu)$, sei $A(t) = M_{f_t}$ auf $X$ für messbare Abbildungen $f_t$, für die $\Re f_t(x) \le 0$ für fast alle $x \in X$, sei $\sigma(t)$ eine kompakte in $\sigma(A(t))$ isolierte Untermenge von $\sigma(A(t))$, sodass die Zykelbedingung aus Satz~\ref{thm: unhandl adsatz mit sl} erfüllt ist, und sei $P(t)$ die Rieszprojektion von $A(t)$ auf $\sigma(t)$. Dann gilt
\begin{align*}
\bigl( P(t) g \bigr)(x) &= \Bigl(   \frac{1}{2 \pi i} \, \int_{\gamma_{t_0}} (z-A(t))^{-1}\, g \, dz  \Bigr)(x) = \frac{1}{2 \pi i} \, \int_{\gamma_{t_0}} \frac{1}{z-f_t(x)} \, dz \, g(x) \\
&= \chi_{E_t}(x) \, g(x)
\end{align*}
für alle $t \in U_{t_0}$, alle $g \in X$ und alle $x \in X_0$, wobei $E_t := \{ f_t \in \sigma(t) \}$. Wenn also die Rieszprojektionen $P(t)$ nicht konstant sind, können sie nach Lemma~\ref{lm: multiplikationsop mit char fkt nur dann stark db nach t wenn schon konst} nicht stark differenzierbar, erst recht nicht zweimal stark stetig differenzierbar von $t$ abhängen, wie das in Satz~\ref{thm: unhandl adsatz mit sl} verlangt wird.

Aber 
auch die Aussage dieses Satzes scheint -- wenn überhaupt -- nur sehr selten nichttrivialerweise 
erfüllt zu sein. Wenigstens wenn $A(t) = M_{f_t}$ beschränkt ist, $f_t(X_0) \subset i \, \real$ für alle $t \in I$ und $t \mapsto A(t)g$ stetig ist für alle $g \in X$, können $(1-P) U_T P(0)$ und $P U_T (1-P(0))$ nur dann beide gleichmäßig auf $I$ gegen $0$ konvergieren, wenn $t \mapsto P(t)$ schon konstant ist (wie man auf dieselbe Weise wie im obigen Beispiel einsieht).
\\

Das nächste Beispiel zeigt, dass die Aussage von Satz~\ref{thm: unhandl adsatz mit nichtglm sl} nicht zu gelten braucht, wenn die in diesem Satz getroffene Voraussetzung, dass $t \mapsto P(t)$ zweimal stark stetig differenzierbar ist, als einzige verletzt ist. Auch dann nicht, wenn die $A(t)$ schiefselbstadjungiert sind und $t \mapsto A(t)$ sogar analytisch ist.

\begin{ex}  \label{ex: reg von P wesentlich im adsatz mit nichtglm sl}
Sei $X := \ell^2(I_2)$. Sei $\lambda_1(t), \lambda_2(t) := \mp \bigl(t- \frac{1}{2} \bigr)^2$,
\begin{align*}
A(t):= \begin{pmatrix}
i \lambda_1(t) & 0 \\
0 & i \lambda_2(t)
\end{pmatrix}
\end{align*} 
und sei $\sigma(t) := \{ \lambda(t) \}$ für alle $t \in I$, wobei
\begin{align*}
\lambda(t) := \begin{cases} i \lambda_1(t), & t \in [0, \frac{1}{2}) \\
													  i \lambda_2(t), & t \in [\frac{1}{2}, 1]
							\end{cases}.						  
\end{align*}
Dann ist $A(t)$ für jedes $t \in I$ schiefselbstadjungiert und $t \mapsto A(t)$ ist analytisch (insbesondere einmal stark stetig differenzierbar). Ferner ist $\sigma(t)$ für jedes $t \in I$ isoliert in $\sigma(A(t))$, $\sigma(t)$ fällt an \emph{einer} Stelle, nämlich bei $t = \frac{1}{2}$, in $\sigma(A(t)) \setminus \sigma(t)$ hinein und $t \mapsto \sigma(t)$ ist stetig, weil $t \mapsto \lambda(t)$ stetig ist.
Aber für die Rieszprojektionen $P_0(t)$ von $A(t)$ auf $\sigma(t)$ gilt
\begin{align*}
P_0(t) = \begin{cases} P_1, & t \in [0, \frac{1}{2}) \\
										 1, & t = \frac{1}{2} \\
										 P_2, & t \in (\frac{1}{2}, 1]
			 \end{cases},
\end{align*}
wobei $P_1$, $P_2$ die orthogonale Projektion auf $\spn\{e_1\}$ bzw. $\spn\{e_2\}$ bezeichnet. Und daraus folgt, dass keine beschränkten Projektionen $P(t)$ existieren, die stetig von $t$ abhängen und außer in $t = \frac{1}{2}$ mit $P_0(t)$ übereinstimmen, kurz: $P_0 \big|_{I \setminus \{ \frac{1}{2} \} }$ ist nicht stetig (insbesondere nicht zweimal stark stetig differenzierbar) fortsetzbar.

Also sind zwar die Voraussetzungen an $A$ und $\sigma$ von Satz~\ref{thm: unhandl adsatz mit nichtglm sl} erfüllt aber nich die an $P$. 
Wie wir nun zeigen werden, ist hier auch die Aussage von Satz~\ref{thm: unhandl adsatz mit nichtglm sl} nicht erfüllt.
Sei nämlich $P(t)$ für jedes $t \in I$ eine beschränkte Projektion in $X$ und für jedes $t \in I \setminus \{ \frac{1}{2} \}$ gleich der Rieszprojektion $P_0(t)$ von $A(t)$ auf $\sigma(t)$. 
Dann gilt 
\begin{align*}
U_T(t) = e^{T  \int_0^t A(\tau) \, d\tau } = \begin{pmatrix} e^{i T \int_0^t \lambda_1(\tau)\, d\tau}  & 0 \\ 0  &  e^{i T \int_0^t \lambda_2(\tau)\, d\tau} \end{pmatrix}
\end{align*}
für alle $T \in (0, \infty)$ und alle $t \in I$ (Korollar~\ref{cor: Dyson}), und $1-P(t) = P_1 = P(0)$ für alle $t \in (\frac{1}{2}, 1]$. Wir haben also, dass
\begin{align*}
\norm{(1-P(t)) U_T(t) P(0)} = \norm{ e^{i T \int_0^t \lambda_1(\tau)\, d\tau} \, P_1} = 1 
\end{align*}
für alle $t \in (\frac{1}{2}, 1]$, das heißt, die Aussage des Adiabatensatzes ist hier tatsächlich nicht erfüllt. $\blacktriangleleft$
\end{ex}

Auch im folgenden klassischen Beispiel Rellichs (Beispiel~II.3.5 in~\cite{Kato: Perturbation 80}) sind die Voraussetzungen von Satz~\ref{thm: unhandl adsatz mit nichtglm sl} nicht erfüllt, weil die Rieszprojektionen nicht stetig in eine Überschneidungsstelle hinein fortgesetzt werden können. 
Allerdings wissen wir hier (anders als im obigen Beispiel) nicht,  
ob auch die Aussage dieses Satzes nicht erfüllt ist oder ob sie trotzdem erfüllt ist. Die Schwierigkeit besteht darin, dass es wohl keinen einfachen Ausdruck für die (als Dysonreihe gegebene) Zeitentwicklung gibt -- jedenfalls ist Korollar~\ref{cor: Dyson} nicht anwendbar.

\begin{ex} \label{ex: rieszproj muss nicht stetig fb sein (Rellich)}
Sei $X := \ell^2(I_2)$. Sei
\begin{align*}
A(0):= 0 \text{ \; und \; } 
A(t) := i e^{-\frac{1}{t^2}} \, \begin{pmatrix} \cos \frac{2}{t}  &  \sin \frac{2}{t} \\
																								\sin \frac{2}{t}  & -\cos \frac{2}{t}
																\end{pmatrix}
\end{align*}
für alle $t \in (0,1]$, 
\begin{align*}
\lambda(t) := \begin{cases} 0, & t = 0 \\
														i e^{-\frac{1}{t^2}},& t \in (0,1]
							\end{cases}
\end{align*}
$\sigma(t) := \{ \lambda(t) \}$ und $P_0(t)$ die Rieszprojektion von $A(t)$ auf $\sigma(t)$. 

Dann ist $A(t)$ für jedes $t \in I$ schiefselbstadjungiert, $t \mapsto A(t)$ ist stetig differenzierbar (sogar beliebig oft differenzierbar), $\sigma(A(t)) = \{ \pm \lambda(t) \}$, $\lambda(t)$ ist also ein Eigenwert von $A(t)$ und $t \mapsto \lambda(t)$ ist stetig differenzierbar (sogar beliebig oft differenzierbar). Darüberhinaus ist 
\begin{align*}
P_0(0) =  1 \text{ \; und \; } 
P_0(t) = 	\begin{pmatrix} \cos^2 \frac{1}{t}    &  \cos \frac{1}{t} \sin \frac{1}{t} \\
													\cos \frac{1}{t} \sin \frac{1}{t}   &  \sin^2 \frac{1}{t}
					\end{pmatrix}
\end{align*}
für alle $t \in (0,1]$, woraus hervorgeht, dass keine stetige Abbildung $t \mapsto P(t)$ existiert, die mit Ausnahme von endlich vielen Stellen mit $P_0$ übereinstimmt. Die Voraussetzungen von Satz~\ref{thm: unhandl adsatz mit nichtglm sl} sind hier also tatsächlich nicht erfüllbar. $\blacktriangleleft$
\end{ex}

Wir schließen mit einem Beispiel, das zeigt, dass die Aussage des Adiabatensatzes auch dann gelten kann, wenn die Projektion $P(t)$ für kein $t \in I$ mit $A(t)$ vertauscht. Insbesondere müssen die Projektionen $P(t)$ nicht die Rieszprojektionen von $A(t)$ auf $\sigma(t)$ sein, damit die Aussage von Satz~\ref{thm: unhandl adsatz mit sl} gelten kann. 

\begin{ex}  \label{ex: adsatz auch für P(t)A(t) ne A(t)P(t)}
Sei $X := \ell^2(I_3)$. Sei
\begin{align*}
A_0 := \begin{pmatrix}
\lambda_2  & 1           & 0 \\
0           & \lambda_2  & 0  \\
0           & 0           & \lambda_2
\end{pmatrix}, \;
R(t):= \begin{pmatrix}
1  & 0		 	& 0 \\
0  & \cos t & \sin t \\
0  & -\sin t  & \cos t
\end{pmatrix},
\end{align*}
$\sigma(t) := \{ \lambda_2 \}$ und $P_0$ die orhtogonale Projektion auf $\spn\{e_1, e_3\}$. Sei 
\begin{align*}
A(t) := R(t)^* A_0(t) R(t) \text{ \; und \; } P(t):= R(t)^* P_0 R(t)
\end{align*}
für alle $t \in I$.

Dann gilt
\begin{align*}
A(t) = \begin{pmatrix}
\lambda_2  & \cos t           & \sin t \\
0           & \lambda_2  & 0  \\
0           & 0           & \lambda_2
\end{pmatrix}
\end{align*}
für alle $t \in I$, woraus hervorgeht, dass die $A(t)$ paarweise vertauschen. Wir bekommen also mithilfe von Korollar~\ref{cor: Dyson}
\begin{align*}
U_T(t) &= e^{T \int_0^t A(\tau) \, d\tau } 
= e^{\lambda_2 T t} \, \begin{pmatrix} 1 & T \sin t & T (1-\cos t) \\ & 1 & 0 \\ &  & 1 \end{pmatrix}
\end{align*}
für alle $T \in (0, \infty)$ und alle $t \in I$.
Weiter ist $P(0)$ die orthogonale Projektion auf $\spn\{e_1, e_3 \}$ und $1-P(t)$ ist die orthogonale Projektion auf 
\begin{align*}
\spn \{ R(t)^* e_2 \} = \spn \left\{ \begin{pmatrix} 0 \\ \cos t \\ \sin t \end{pmatrix}   \right\}
\end{align*}
für alle $t \in I$.
Damit erhalten wir erstens
\begin{align*}
(1-P(t)) U_T(t) P(0) e_1 = \scprd{  \begin{pmatrix} 0 \\ \cos t \\ \sin t \end{pmatrix}, U_T(t) e_1  } \, \begin{pmatrix} 0 \\ \cos t \\ \sin t \end{pmatrix} = 0
\end{align*}
für alle $t \in I$ und alle $T \in (0, \infty)$, und zweitens
\begin{align*}
(1-P(t)) U_T(t) P(0) e_3 = \scprd{  \begin{pmatrix} 0 \\ \cos t \\ \sin t \end{pmatrix}, U_T(t) e_3  } \, \begin{pmatrix} 0 \\ \cos t \\ \sin t \end{pmatrix}
= \sin t \, \, e^{\lambda_2 T t} \, \begin{pmatrix} 0 \\ \cos t \\ \sin t \end{pmatrix}
\end{align*}
für alle $t \in I$ und alle $T \in (0, \infty)$, woraus sich wegen der Beschränktheit von $t \mapsto \frac{\sin t}{t}$ und der Negativität von $\lambda_2$ ergibt, dass
\begin{align*}
\sup_{t \in I} \norm{  (1-P(t)) U_T(t) P(0)  } = O\Bigl( \frac{1}{T} \Bigr) \quad (T \to \infty),
\end{align*}
also die Aussage des Adiabatensatzes. $\blacktriangleleft$
\end{ex}


\section{Adiabatensätze ohne Spektrallückenbedingung} \label{sect: adsätze ohne sl}

In diesem Abschnitt beweisen wir Adiabatensätze, die ohne Spektrallückenbedingung auskommen, und zwar zunächst solche, in denen nur ausgesagt wird, dass $U_T$ überhaupt beinahe adiabatisch ist bzgl. $P$ (qualitative Adiabatensätze ohne Spektrallückenbedingung), und anschließend solche, in denen auch gesagt wird, wie adiabatisch bzw. diabatisch $U_T$ bzgl. $P$ ist (quantitative Adiabatensätze ohne Spektrallückenbedingung). 

Anders als in den Adiabatensätzen mit Spektrallückenbedingung des vorigen Abschnitts müssen die $\sigma(t)$ hier einpunktig sein: $\sigma(t) = \{ \lambda(t) \}$ für alle $t \in I$, und die Spektralwerte $\lambda(t)$ müssen Eigenwerte sein, die immerhin noch nach einer Seite hin isoliert sind in $\sigma(A(t))$: $\lambda(t) + \varepsilon e^{i \vartheta_0} \in \rho(A(t))$ für alle $\varepsilon \in (0, \varepsilon_0]$ und alle $t \in I$. Schließlich brauchen wir noch Abschätzungen an $\norm{ \bigl( \lambda(t) + \varepsilon e^{i \vartheta_0} - A(t) \bigr)^{-1}  }$.

\subsection{Qualitative Adiabatensätze ohne Spektrallückenbedingung}

Wir beginnen mit einigen vorbereitenden Aussagen. Das nachfolgende Lemma ist der entscheidende Schlüssel für Satz~\ref{thm: allg adsatz ohne sl} und Satz~\ref{thm: adsatz ohne sl für normale A(t)}. 

\begin{lm} \label{lm: zentrales lm für adsatz ohne sl}
Sei $A$ eine abgeschlossene lineare Abbildung $D \subset X \to X$ und sei $\lambda \in \complex$ mit 
\begin{equation*}
\lambda + \varepsilon e^{i \vartheta_0} \in \rho(A) \quad \text{und} \quad \norm{(\lambda + \varepsilon e^{i \vartheta_0} - A)^{-1}} \le \frac{M_0}{\varepsilon} \quad \text{für alle} \; \varepsilon \in (0, \varepsilon_0],
\end{equation*}
wobei $\varepsilon_0 \in (0, \infty)$ und $\vartheta_0 \in \real$.
Dann gilt
\begin{equation*}
\varepsilon (\lambda + \varepsilon e^{i \vartheta_0} - A)^{-1}y \longrightarrow 0 \quad (\varepsilon \searrow 0)
\end{equation*}
für alle $y \in \overline{\im(A-\lambda)}$.
\end{lm}

\begin{proof}
Sei $y \in \im(A-\lambda)$, dann $y = (A-\lambda)x$ für ein $x \in D$ und 
\begin{equation*} \begin{split}
\varepsilon (\lambda + \varepsilon e^{i \vartheta_0} - A)^{-1}y &= \varepsilon (\lambda + \varepsilon e^{i \vartheta_0} - A)^{-1}(A-\lambda)x \\
&= -\varepsilon x + \varepsilon \left( \varepsilon e^{i \vartheta_0}  (\lambda + \varepsilon e^{i \vartheta_0} - A)^{-1} \right) x \longrightarrow 0 \quad (\varepsilon \searrow 0).
\end{split} \end{equation*}
Weil nun $\norm{\varepsilon (\lambda + \varepsilon e^{i \vartheta_0} - A)^{-1}} \le M_0$ für alle $\varepsilon \in (0, \varepsilon_0]$, gilt 
\begin{equation*}
\varepsilon (\lambda + \varepsilon e^{i \vartheta_0} - A)^{-1}y \longrightarrow 0 \quad (\varepsilon \searrow 0)
\end{equation*}
auch für alle $y \in \overline{\im(A-\lambda)}$.
\end{proof}

Das nächste ziemlich offensichtliche Lemma erlaubt es unter gewissen Voraussetzungen von der Konvergenz in der starken Operatortopologie auf Konvergenz bzgl. der Normoperatortopologie zu schließen. 

\begin{lm} \label{lm: aus starker konv mach normkonv}
Seien $A_n$, $A$ und $B$ beschränkte lineare Abbildungen in $X$ für alle $n \in \natu$, sodass $A_n \longrightarrow A \;\;(n \to \infty)$ bzgl. der starken Operatortopologie von $X$, und sei $\rk B < \infty$. Dann gilt $A_n B \longrightarrow A B \;\;(n \to \infty)$ bzgl. der Normoperatortopologie. 
\end{lm}

\begin{proof}
Sei $\{ y_1, \dots, y_m \}$ eine Basis des (endlichdimensionalen!) Unterraums $BX$ von $X$ und sei $\kappa$ die lineare Abbildung $BX \to \complex^m$, die jedem Vektor $y \in BX$ seinen Koordinatenvektor bzgl. der Basis $\{ y_1, \dots, y_m \}$ zuordnet. Dann ist $\kappa$ beschränkt (als eine auf einem endlichdimensionalen normierten Raum definierte lineare Abbildung) und wir erhalten, da
\begin{align*}
A_n Bx - A Bx = \sum_{i=1}^m \kappa(Bx)^{i} \bigl(A_n - A \bigr)y_i 
\end{align*}
für alle $x \in X$, dass  
\begin{align*}
\norm{A_n B - A B}^2 &= \sup_{x  \in S_X} \norm{ A_n Bx - A Bx }^2 \\
&\le \norm{\kappa}^2 \norm{B}^2 \, \sum_{i=1}^m \norm{ \bigl(A_n - A \bigr)y_i }^2 \longrightarrow 0 \quad (n \to \infty),
\end{align*} 
wie gewünscht.
\end{proof}

Wir weisen darauf hin, dass wir unter denselben Voraussetzungen \emph{nicht} auf 
\begin{align*}
B A_n  \longrightarrow B A \quad (n \to \infty) \text{ \; bzgl. der Normoperatortopologie }
\end{align*} 
schließen können. 
Sei nämlich $X := \ell^2(\natu)$, $A_n := {A_0}^n$ für alle $n \in \natu$, $A := 0$ und $B$ die orthogonale Projektion auf $\spn \{e_1 \}$. $A_0$ bezeichnet den Shift nach links, das heißt,
\begin{align*}
A_0 (x_1, x_2, x_3, \dots ) = (x_2, x_3, \dots )
\end{align*}
für alle $x \in X$. Dann gilt zwar $A_n x \longrightarrow 0 = Ax \;\;(n \to \infty)$ für alle $x \in X$ und $\rk B < \infty$, aber $({A_n}^*B^*e_1) = (e_{n+1})$ konvergiert nicht, erst recht konvergiert $(B A_n)$ nicht bzgl. der Normoperatortopologie.

\begin{lm} \label{lm: abschätzung für abl der resolvente}
Sei $A(t)$ für jedes $t \in I$ eine abgeschlossene lineare Abbildung $D \subset X \to X$ und sei $t \mapsto A(t)x$ stetig differenzierbar für alle $x \in D$. Sei $\lambda(t) \in \complex$ mit $\lambda(t) + \varepsilon e^{i \vartheta_0} \in \rho(A(t))$ und 
\begin{align*}
\norm{(\lambda(t) + \varepsilon e^{i \vartheta_0} - A(t))^{-1}} \le \frac{M_0}{\varepsilon} 
\end{align*}
für alle $\varepsilon \in (0, \varepsilon_0]$ und alle $t \in I$ (wobei $\varepsilon_0$, $\vartheta_0$ sowie $M_0$ positive von $t$ unabhängige Zahlen bezeichnen), und sei $t \mapsto \lambda(t)$ stetig differenzierbar.
Dann ist $t \mapsto (\lambda(t) + \varepsilon e^{i \vartheta_0} - A(t))^{-1}x$ stetig differenzierbar und es gibt eine Zahl $M_0'$, sodass
\begin{align*}
\norm{ \ddt{\left( (\lambda(t) + \varepsilon e^{i \vartheta_0} - A(t))^{-1}x \right) } } \le \frac{M_0'}{\varepsilon^2} \norm{x} 
\end{align*}
für alle $t \in I$, $\varepsilon \in (0, \varepsilon_0]$ und alle $x \in X$.
\end{lm}

\begin{proof}
Aus Lemma~\ref{lm: reg of inv} folgt sofort, dass $t \mapsto \bigl( \lambda(t) + \varepsilon e^{i \vartheta_0} - A(t) \bigr)^{-1}x$ stetig differenzierbar ist und 
\begin{align*}
&\ddt{\left( (\lambda(t) + \varepsilon e^{i \vartheta_0} - A(t))^{-1}x \right) } \\
& \qquad = \bigl( \lambda(t) + \varepsilon e^{i \vartheta_0} - A(t) \bigr)^{-1}   \bigl( A'(t)(A(t)-1)^{-1} \bigr)  \, (A(t)-1) \bigl( \lambda(t) + \varepsilon e^{i \vartheta_0} - A(t) \bigr)^{-1} x \\
& \qquad \quad \; - \bigl( \lambda(t) + \varepsilon e^{i \vartheta_0} - A(t) \bigr)^{-1}   \, \lambda'(t) \,   \bigl( \lambda(t) + \varepsilon e^{i \vartheta_0} - A(t) \bigr)^{-1} x 
\end{align*}
für alle $x \in X$ und alle $t \in I$. Wegen $\sup_{t \in I} \norm{ A'(t)(A(t)-1)^{-1} } < \infty$ (Beweis von Lemma~\ref{lm: reg of inv}) folgt nun die gewünschte Abschätzung.
\end{proof}
 
Wir können nun unseren ersten Adiabatensatz ohne Spektrallückenbedingung beweisen, der einen entsprechenden Satz von Avron und Elgart (Theorem~1 bzw. Theorem~3 in~\cite{AvronElgart 99}) und einen daran anknüpfenden Satz von Teufel aus~\cite{Teufel 01} verallgemeinert: dort sind die $A(t)$ als schiefselbstadjungiert vorausgesetzt, unser Satz 
geht auch für lineare Abbildungen $A(t)$ (wie immer mit $t$-unabhängigem Definitionsbereich), die jeweils nur Erzeuger einer stark stetiger Halbgruppe auf einem Banachraum $X$ sind und für die $A$ $(M,0)$-stabil ist. 
Wir orientieren uns am Beweis in~\cite{Griesemer 07}.

\begin{thm} \label{thm: allg adsatz ohne sl}
Sei $A(t)$ für jedes $t \in I$ eine lineare Abbildung $D \subset X \to X$, die eine stark stetige Halbgruppe auf $X$ erzeugt, sei $A$ $(M,0)$-stabil und sei $t \mapsto A(t)x$ stetig differenzierbar für alle $x \in D$. 
Sei $\lambda(t)$ für jedes $t \in I$ ein Eigenwert von $A(t)$, seien $\varepsilon_0$, $\vartheta_0$ sowie $M_0$ positive Zahlen, sodass $\lambda(t) + \varepsilon e^{i \vartheta_0} \in \rho(A(t))$ und 
\begin{align*}
\norm{\bigl( \lambda(t) + \varepsilon e^{i \vartheta_0} - A(t) \bigr)^{-1} } \le \frac{M_0}{\varepsilon}
\end{align*}
für alle $\varepsilon \in (0, \varepsilon_0]$ und alle $t \in I$, und sei $t \mapsto \lambda(t)$ stetig differenzierbar.
Sei $P(t)$ für jedes $t \in I$ eine beschränkte Projektion in $X$ mit $P(t)A(t) \subset A(t)P(t)$, sodass 
\begin{align*}
P(t)X \subset \ker (A(t)-\lambda(t))
\end{align*} für alle $t \in I$ und \begin{align*}(1-P(t))X \subset \overline{\im(A(t)-\lambda(t))}
\end{align*}
für fast alle $t \in I$, sei $\rk P(0) < \infty$ oder $\rk (1-P(0)) < \infty$ und sei $t \mapsto P(t)$ stetig differenzierbar.
Dann gilt
\begin{align*}
(1-P) U_T P(0) \longrightarrow 0 \quad (T \to \infty) \text{ \; gleichmäßig auf } I.
\end{align*}
\end{thm}

\begin{proof}
Sei $V_T$ die nach Satz~\ref{thm: Dyson} existierende Zeitentwicklung zu $T \lambda + [P',P]$. (Wir können, wann immer diese existiert, auch mit der Zeitentwicklung zu $T A + [P',P]$ arbeiten, denn diese stimmt \emph{auf $P(0)X$} mit $V_T$ überein). Dann gilt nach demselben Satz und Lemma~\ref{lm: skalierung und (M,w)-stabilität} 
\begin{align*}
\norm{V_T(t,s)} \le e^{c (t-s)}
\end{align*}
für alle $(s,t) \in \Delta$ und alle $T \in (0, \infty)$ (wobei $c$ eine Zahl ist, für die $\norm{P(s)}$, $\norm{P'(s)}$ und $\norm{[P'(s),P(s)]} \le c$ für alle $s \in I$) und nach Satz~\ref{thm: intertwining relation} gilt
\begin{align*}
P(t)V_T(t) = V_T(t)P(0)
\end{align*}
für alle $t \in I$ und alle $T \in (0,\infty)$. 
Also haben wir
\begin{align*}
(1-P(t))U_T(t) P(0) = (1-P(t)) \bigl( U_T(t)-V_T(t) \bigr) P(0).
\end{align*}

Sei nun $x \in X$. Dann ist $[0,t] \ni s \mapsto U_T(t,s)V_T(s)P(0)x = U_T(t,s) \, P(s) \, V_T(s)x$ differenzierbar nach Proposition~\ref{thm: char zeitentwicklung} und Lemma~\ref{lm: strong db of products} (beachte, dass $P(s)X \subset D$ für alle $s \in I$ nach Voraussetzung) mit
\begin{align*}
&\dds{ U_T(t,s)V_T(s)P(0)x } \\
& \qquad \qquad = U_T(t,s) \bigl( T (\lambda(s)-A(s)) \bigr) V_T(s)P(0)x + U_T(t,s) \, [P'(s),P(s)] \, V_T(s) P(0)x \\
&\qquad \qquad = U_T(t,s) \, P'(s) \, V_T(s) P(0)x
\end{align*}
für alle $s \in [0,t]$ (benutze im zweiten Gleichheitszeichen die obige intertwining relation für $V_T$), woraus hervorgeht, dass $[0,t] \ni s \mapsto U_T(t,s)V_T(s)P(0)x$ sogar \emph{stetig} differenzierbar ist. Also gilt
\begin{align} \label{eq: allg adsatz ohne sl 1}
V_T(t)&P(0)x - U_T(t)P(0)x = U_T(t,s)V_T(s)P(0)x \big|_{s=0}^{s=t} = \int_0^t U_T(t,s) \, P'(s) \, V_T(s) P(0)x \,ds \notag \\
&= \int_0^t U_T(t,s) \, (\lambda(s)-A(s)) \, \bigl( \lambda(s) + \varepsilon e^{i \vartheta_0} - A(s) \bigr)^{-1} \, P'(s) \, V_T(s) P(0)x \, ds \notag \\
&\quad + \int_0^t U_T(t,s) \, \varepsilon e^{i \vartheta_0} \, \bigl( \lambda(s) + \varepsilon e^{i \vartheta_0} - A(s) \bigr)^{-1} \, P'(s) \, V_T(s) P(0)x \, ds
\end{align} 
für alle $\varepsilon \in (0, \varepsilon_0]$, alle $t \in I$ und alle $T \in (0, \infty)$.

Sei nun 
\begin{align*}
\tilde{U}_T(t,s) := U_T(t,s) \, e^{-T \int_s^t \lambda(\tau)\,d\tau} \text{ \; und \; } \tilde{V}_T(t,s) := V_T(t,s) \, e^{-T \int_s^t \lambda(\tau)\,d\tau}
\end{align*}
für alle $(s,t) \in \Delta$ und alle $T \in (0, \infty)$. 
Dann ist, wie man sofort sieht, $\tilde{U}_T$ die Zeitentwicklung zu $T (A-\lambda)$ und $\tilde{V}_T$ die Zeitentwicklung zu $[P',P]$.

Sei weiter 
\begin{align*}
Q_n(t) := \int_0^1 \gamma_{\frac{1}{n}} (t-r) \, P'(r) \, dr = \bigl( \gamma_{\frac{1}{n}} \ast (P' \chi_{[0,1]}) \bigr) (t) 
\end{align*} 
für alle $t \in I$ und alle $n \in \natu$, wobei $\gamma_{\frac{1}{n}} := n \, \gamma(n \, . \,)$ für eine Abbildung $\gamma \in C_c^1(\real, [0, \infty))$ mit $\supp \gamma \subset [-1,1]$ und $\int \gamma(r) \,dr = 1$.
Dann gilt $Q_n(t) \longrightarrow P'(t) \;\;(n \to \infty)$ für alle $t \in (0,1)$, 
\begin{align*}
\norm{Q_n(t)} \le \sup_{r \in I} \norm{P'(r)} \le c
\end{align*}
für alle $t \in I$ und alle $n \in \natu$ und die Abbildung $t \mapsto Q_n(t)$ ist (stark) stetig differenzierbar für alle $n \in \natu$.

Wir sehen nun, dass $[0,t] \ni s \mapsto \tilde{U}_T(t,s) \, \bigl( \lambda(s) + \varepsilon e^{i \vartheta_0} - A(s) \bigr)^{-1} \, Q_n(s) \, \tilde{V}_T(s) P(0)x$ stetig differenzierbar ist (Lemma~\ref{lm: reg of inv}) mit der Ableitung
\begin{align*}
s \mapsto & \; \tilde{U}_T(t,s) \, T \bigl(\lambda(s)-A(s)\bigr) \, \bigl( \lambda(s) + \varepsilon e^{i \vartheta_0} - A(s) \bigr)^{-1} \,Q_n(s) \, \tilde{V}_T(s) P(0)x \\
& + \tilde{U}_T(t,s) \Bigl( \dds{ \bigl( \lambda(s) + \varepsilon e^{i \vartheta_0} - A(s) \bigr)^{-1} } \Bigr) Q_n(s) \, \tilde{V}_T(s) P(0)x \\
& + \tilde{U}_T(t,s) \bigl( \lambda(s) + \varepsilon e^{i \vartheta_0} - A(s) \bigr)^{-1} \, \Bigl( Q_n'(s) + Q_n(s) \, [P'(s),P(s)] \Bigr) \tilde{V}_T(s) P(0)x
\end{align*}
und daher gilt (partielle Integration)
\begin{align*}
&\int_0^t U_T(t,s) \, (\lambda(s)-A(s)) \, \bigl( \lambda(s) + \varepsilon e^{i \vartheta_0} - A(s) \bigr)^{-1} \, Q_n(s) \, V_T(s) P(0)x \, ds \\
& \quad  = \frac{1}{T} \, e^{T \int_0^t \lambda(\tau)\,d\tau} \; \int_0^t  \tilde{U}_T(t,s) \, T \bigl(\lambda(s)-A(s)\bigr) \\
& \qquad \qquad \qquad \qquad \qquad \qquad \qquad \qquad \bigl( \lambda(s) + \varepsilon e^{i \vartheta_0} - A(s) \bigr)^{-1} \, Q_n(s) \, \tilde{V}_T(s) P(0)x \, ds \\
& \quad = \frac{1}{T} \, U_T(t,s) \bigl( \lambda(s) + \varepsilon e^{i \vartheta_0} - A(s) \bigr)^{-1} \, Q_n(s) \, V_T(s) P(0)x \big|_{s=0}^{s=t} \\
& \quad \quad - \frac{1}{T} \, \int_0^t  U_T(t,s) \Bigl( \dds{ \bigl( \lambda(s) + \varepsilon e^{i \vartheta_0} - A(s) \bigr)^{-1} } \Bigr) Q_n(s) \, V_T(s) P(0)x \,ds \\
& \quad \quad - \frac{1}{T} \, \int_0^t U_T(t,s) \bigl( \lambda(s) + \varepsilon e^{i \vartheta_0} - A(s) \bigr)^{-1} \\
& \qquad \qquad \qquad \qquad \qquad \qquad \qquad \Bigl( Q_n'(s) + Q_n(s) \, [P'(s),P(s)] \Bigr) V_T(s) P(0)x \,ds
\end{align*}
für alle $\varepsilon \in (0,\varepsilon_0]$, $n \in \natu$ und alle $t \in I$ sowie $T \in (0,\infty)$.

Schätzen wir nun die rechte Seite von \eqref{eq: allg adsatz ohne sl 1} ab!
Zunächst gilt
\begin{align} \label{eq: allg adsatz ohne sl 2}
&\norm{ \int_0^t U_T(t,s) \, (\lambda(s)-A(s)) \, \bigl( \lambda(s) + \varepsilon e^{i \vartheta_0} - A(s) \bigr)^{-1} \, \bigl(P'(s)-Q_n(s)\bigr) \, V_T(s) P(0)x \, ds } \notag \\
& \qquad \le \int_0^t M \, (1+M_0) \, \norm{ \bigl( P'(s)-Q_n(s) \bigr) P(s) } \, e^c  \, ds \, \norm{x}.
\end{align}
Weiter gilt nach der eben ausgeführten partiellen Integration
\begin{align} \label{eq: allg adsatz ohne sl 3}
&\norm{    \int_0^t U_T(t,s) \, (\lambda(s)-A(s)) \, \bigl( \lambda(s) + \varepsilon e^{i \vartheta_0} - A(s) \bigr)^{-1} \, Q_n(s) \, V_T(s) P(0)x \, ds   } \notag \\
& \qquad \qquad  \le \frac{2}{T} \, \Bigl( M \, \frac{M_0}{\varepsilon} \, c \, e^c \, c \Bigr) \norm{x} \notag \\
&\qquad \qquad \quad + \frac{1}{T} \, \int_0^t M \, \frac{M_0'}{\varepsilon^2} \, c \, e^c\, c \, ds \, \norm{x}  \notag \\
&\qquad \qquad \quad + \frac{1}{T} \, \int_0^t M \, \frac{M_0}{\varepsilon} \, \bigl(c_n + c^2 \bigr) \, e^c\, c \, ds \, \norm{x},
\end{align}
wobei wir Lemma~\ref{lm: abschätzung für abl der resolvente} benutzt haben und $c_n := \sup_{s \in I} \norm{Q_n'(s)}$ gesetzt haben.
Schließlich gilt wegen $P' P = (1-P) P' P$
\begin{align}  \label{eq: allg adsatz ohne sl 4}
& \norm{   \int_0^t U_T(t,s) \, \varepsilon e^{i \vartheta_0} \, \bigl( \lambda(s) + \varepsilon e^{i \vartheta_0} - A(s) \bigr)^{-1} \, P'(s) \, V_T(s) P(0)x \, ds   } \notag \\
& \qquad \qquad  \le \int_0^t M \, \norm{  \varepsilon  \, \bigl( \lambda(s) + \varepsilon e^{i \vartheta_0} - A(s) \bigr)^{-1} \, (1-P(s)) \, P'(s) P(s)   } \, e^c  \, ds \, \norm{x}.
\end{align}

Sei $\varepsilon_T := T^{-\frac{1}{3}}$ für alle $T \in [\varepsilon_0^3,\infty)$. Dann folgt mit den Abschätzungen \eqref{eq: allg adsatz ohne sl 2}, \eqref{eq: allg adsatz ohne sl 3}, \eqref{eq: allg adsatz ohne sl 4} aus \eqref{eq: allg adsatz ohne sl 1}, dass 
\begin{align}  \label{eq: allg adsatz ohne sl 5}
&\norm{V_T(t)P(0) - U_T(t)P(0)} \notag \\
&\qquad \qquad \le \int_0^1 M \, (1+M_0) \, \norm{ \bigl( P'(s)-Q_n(s) \bigr) P(s) } \, e^c  \, ds \notag \\
& \qquad \qquad \quad +  2 M  c^2\,e^c \, \frac{M_0}{ \varepsilon_T \, T}    +    M  c^2 \, e^c \, \frac{M_0'}{\varepsilon_T^2 \, T}    +    M  c \, e^c \, \bigl( c_n + c^2 \bigr) \, \frac{M_0}{\varepsilon_T \, T}  \notag \\      
& \qquad \qquad \quad + \int_0^1 M \, \norm{  \varepsilon_T  \, \bigl( \lambda(s) + \varepsilon_T e^{i \vartheta_0} - A(s) \bigr)^{-1} \, (1-P(s)) \, P'(s) P(s)   } \, e^c  \, ds
\end{align}
für alle $n \in \natu$, alle $t \in I$ und alle $T \in [\varepsilon_0^3,\infty)$.

Wir bemerken nun erstens, dass
\begin{align*}
\int_0^1 M \, (1+M_0) \, \norm{ \bigl( P'(s)-Q_n(s) \bigr) P(s) } \, e^c  \, ds \longrightarrow 0 \quad (n \to \infty)
\end{align*}
nach dem lebesgueschen Satz. Zweitens gilt ebenfalls nach dem lebesgueschen Satz
\begin{align*}
\int_0^1 M \, \norm{  \varepsilon  \, \bigl( \lambda(s) + \varepsilon e^{i \vartheta_0} - A(s) \bigr)^{-1} \, (1-P(s)) \, P'(s) P(s)   } \, e^c  \, ds \longrightarrow 0 \quad (\varepsilon \searrow 0),
\end{align*}
denn nach Lemma~\ref{lm: zentrales lm für adsatz ohne sl} gilt
\begin{align*}
\varepsilon  \, \bigl( \lambda(s) + \varepsilon e^{i \vartheta_0} - A(s) \bigr)^{-1} \, (1-P(s)) \, P'(s) P(s)   \longrightarrow 0 \quad (\varepsilon \searrow 0)
\end{align*}
bzgl. der \emph{starken} Operatortopologie für fast alle $s \in I$ und nach Lemma~\ref{lm: rk konst} gilt $\rk P(s) < \infty$ für alle $s \in I$ oder $\rk P(s) < \infty$ für alle $s \in I$, woraus dann wegen Lemma~\ref{lm: aus starker konv mach normkonv} folgt, dass sogar
\begin{align*}
\norm{  \varepsilon  \, \bigl( \lambda(s) + \varepsilon e^{i \vartheta_0} - A(s) \bigr)^{-1} \, (1-P(s)) \, P'(s) P(s)   }   \longrightarrow 0 \quad (\varepsilon \searrow 0)
\end{align*} 
für fast alle $s \in I$.

Aus \eqref{eq: allg adsatz ohne sl 5} ersehen wir nun, dass wir 
\begin{align*}
\sup_{t \in I} \norm{V_T(t)P(0) - U_T(t)P(0) } 
\end{align*}
beliebig klein machen können, indem wir zunächst eine so große natürliche Zahl $n_0$ wählen, dass
\begin{align*}
\int_0^1 M \, (1+M_0) \, \norm{ \bigl( P'(s)-Q_{n_0}(s) \bigr) P(s) } \, e^c  \, ds 
\end{align*}
klein genug ist, und dann $T_0 \in [\varepsilon_0^3,\infty)$ so groß wählen, dass 
\begin{align*}
2 M  c^2\,e^c \, \frac{M_0}{ \varepsilon_T \, T}    +    M  c^2 \, e^c \, \frac{M_0'}{{\varepsilon_T}^2 \, T}    +    M  c \, e^c \, \bigl( c_{n_0} + c^2 \bigr) \, \frac{M_0}{\varepsilon_T \, T} 
\end{align*}
sowie
\begin{align*}
\int_0^1 M \, \norm{  \varepsilon_T  \, \bigl( \lambda(s) + \varepsilon_T e^{i \vartheta_0} - A(s) \bigr)^{-1} \, (1-P(s)) \, P'(s) P(s)   } \, e^c  \, ds
\end{align*}
klein genug ist für alle $T \in [T_0, \infty)$. 
Wir haben also
\begin{align*}
\sup_{t \in I} \norm{V_T(t)P(0) - U_T(t)P(0) } \longrightarrow 0 \quad (T \to \infty)
\end{align*}
und damit auch $\sup_{t\in I} \norm{\bigl(1-P(t)\bigr)  U_T(t) P(0) }  \longrightarrow 0 \; \; (T \to \infty)$, wie behauptet.
\end{proof}

Wenn wir die Voraussetzung 
\begin{align*}
\rk P(0) < \infty \text{ \; oder \; } \rk (1-P(0)) < \infty
\end{align*}
des obigen Adiabatensatzes durch die Bedingung $\rk P(0) < \infty$ ersetzen, dann gilt die Aussage des Satzes auch dann noch, wenn wir $P$ nur als einmal \emph{stark} stetig differenzierbar voraussetzen. 

Dies können wir durch eine leichte Abwandlung der oben gegebenen Argumentation einsehen. Die Zeitentwicklung $V_T$ zu $T \lambda + [P',P]$ existiert nach wie vor und die oben benützten Eigenschaften bleiben erhalten, und auch die Definition der beschränkten linearen Abbildungen $Q_n(t)$ können wir übernehmen, wenn wir das definierende Integral als starkes Integral auffassen. 
Dann ist $t \mapsto Q_n(t)$ wieder stark stetig differenzierbar für alle $n \in \natu$. 

Allerdings haben wir jetzt nur noch, dass $Q_n(s) \longrightarrow P'(s) \;\;(n \to \infty)$ bzgl. der \emph{starken} Operatortopologie für alle $s \in (0,1)$, und auch die Existenz der Integrale auf den rechten Seiten von \eqref{eq: allg adsatz ohne sl 2} und \eqref{eq: allg adsatz ohne sl 4} ist zunächst fraglich -- wir wissen jetzt ja nur noch, dass $s \mapsto P'(s)$ stark stetig ist. Aber tatsächlich existieren die Integrale \emph{doch}: weil $\rk P(s) < \infty$ für alle $s \in I$ (nach unserer neuen Voraussetzung), folgt mit Lemma~\ref{lm: aus starker konv mach normkonv} und Lemma~\ref{lm: strong db and db}, dass die Integranden $s \mapsto \bigl( P'(s) - Q_n(s) \bigr) P(s)$ und $s \mapsto  \varepsilon  \, \bigl( \lambda(s) + \varepsilon e^{i \vartheta_0} - A(s) \bigr)^{-1} \, (1-P(s)) \, P'(s) P(s)$ sogar stetig bzgl. der Normoperatortopologie sind.

Außerdem folgt aus Lemma~\ref{lm: aus starker konv mach normkonv} wegen unserer neuen Voraussetzung auch $\bigl( Q_n(s) - P'(s) \bigr) P(s) \longrightarrow 0 \;\;(n \to \infty)$ für alle $s \in (0,1)$, weswegen wir immer noch
\begin{align*}
\int_0^1 M \, (1+M_0) \, \norm{ \bigl( P'(s)-Q_n(s) \bigr) P(s) } \, e^c  \, ds \longrightarrow 0 \quad (n \to \infty)
\end{align*}
haben.
 
Alle übrigen Argumente können wir wörtlich aus dem obigen Beweis übernehmen.
\\

Wir wollen jetzt untersuchen, wann gewisse Voraussetzungen des Satzes~\ref{thm: allg adsatz ohne sl} (automatisch) erfüllt sind. Zunächst zur stetigen Differenzierbarkeit der Eigenwertkurve $t \mapsto \lambda(t)$.

\begin{prop} \label{prop: lambda automatisch stetig db}
Sei $X$ ein Hilbertraum. Dann ergibt sich die Voraussetzung von Satz~\ref{thm: allg adsatz ohne sl}, dass $t \mapsto \lambda(t)$ stetig differenzierbar ist, aus den übrigen Voraussetzungen dieses Satzes, sofern $P(0) \ne 0$. 
\end{prop} 

\begin{proof}
Seien alle von der infrage stehenden Voraussetzung verschiedenen Voraussetzungen von Satz~\ref{thm: allg adsatz ohne sl} erfüllt und sei $P(0) \ne 0$ (wenn $P(0)=0$ ist, dann ist dieser Satz natürlich nicht mehr interessant). 
Sei $t_0 \in I$. Dann ist $P(t_0) \ne 0$ nach Lemma~\ref{lm: rk konst} und Lemma~\ref{lm: strong db and db}, das heißt, es existiert ein $x_0 \in X$ mit $P(t_0) x_0 \ne 0$. Wegen der Stetigkeit von $t \mapsto P(t)x_0$ gilt also $P(t)x_0 \ne 0$ auch für alle $t \in U_{t_0}$ für eine in $I$ offene Umgebung $U_{t_0}$ von $t_0$. 
Weiter gilt
\begin{align*}
\frac{1}{\lambda(t)-1} = \frac{       \scprd{ P(t)x_0 \,, \, (A(t)-1)^{-1} \, P(t)x_0   }     }{      \scprd{  P(t)x_0, P(t)x_0  }    }
\end{align*}
für alle $t \in U_{t_0}$ und daraus ergibt sich, da die rechte Seite stetig differenzierbar von $t$ abhängt (Lemma~\ref{lm: reg of inv}), dass $U_{t_0} \ni t \mapsto \lambda(t)$ stetig differenzierbar ist, wie gewünscht.
\end{proof}

Jetzt zu der wesentlich interessanteren Voraussetzung
\begin{align*}
\lambda(t) + \varepsilon e^{i \vartheta_0} \in \rho(A(t))    \text{\; und \;}    \norm{(\lambda(t) + \varepsilon e^{i \vartheta_0} - A(t))^{-1}} \le \frac{M_0}{\varepsilon} \quad \text{für alle} \; \varepsilon \in (0, \varepsilon_0].
\end{align*}
Diese folgt (mit $\vartheta_0 := 0$ und $\varepsilon_0 \in (0, \infty)$ beliebig) beispielsweise dann, wenn $\lambda(t)$ für alle $t \in I$ auf der imaginären Achse liegt (insbesondere also dann, wenn die $A(t)$ alle schiefselbstadjungiert sind), schon aus den übrigen Voraussetzungen von Satz~\ref{thm: allg adsatz ohne sl}, genauer: aus der $(M,0)$-Stabilität von $A$ (nach Satz~\ref{thm: eigenschaften von erzeugern}).  
\\

Im Sonderfall normaler $A(t)$ können wir mehr sagen: dort ergibt sich die obige Resolventenabschätzung schon dann aus den übrigen Voraussetzungen des  Satzes~\ref{thm: allg adsatz ohne sl}, wenn nur ein (kleiner) in $\lambda(t)$ angehefteter Sektor
\begin{align*}
\lambda(t) + \varepsilon_0 \, S_{(\vartheta_0^-, \vartheta_0^+)} := \{      \lambda(t) + \varepsilon e^{i \vartheta}: \varepsilon \in (0, \varepsilon_0), \, \vartheta \in (\vartheta_0^-, \vartheta_0^+)         \}
\end{align*}
in $\rho(A(t))$ liegt für alle $t \in I$. Diese Sektorbedingung ist nicht sehr einschränkend, auch wenn sie natürlich -- selbst bei von $t$ unabhängigem Spektrum $\sigma(A(t))$ -- nicht immer erfüllt ist (etwa dann nicht, wenn ein cusp (s. etwa Abschnitt~6.1 in~\cite{Dobrowolski}) bei $\lambda(t)$ vorliegt).

\begin{prop}  \label{prop: resolvabsch automatisch erf}
Sei $A(t)$ für jedes $t \in I$ eine normale lineare Abbildung $D \subset H \to H$, $\lambda(t) \in \sigma(A(t))$ und seien $\varepsilon_0 \in (0,\infty)$ und $\vartheta_0^-$, $\vartheta_0^+ \in [0, 2 \pi)$ mit $\vartheta_0^- < \vartheta_0^+$, sodass $\lambda(t) + \varepsilon_0 \, S_{(\vartheta_0^-, \vartheta_0^+)} \subset \rho(A(t))$ für alle $t \in I$. 
Dann existieren positive Zahlen $\varepsilon_0'$ und $M_0$, sodass $\lambda(t) + \varepsilon e^{i \vartheta_0} \in \rho(A(t))$ und
\begin{align*}
\norm{(\lambda(t) + \varepsilon e^{i \vartheta_0} - A(t))^{-1}} \le \frac{M_0}{\varepsilon}
\end{align*}
für alle $\varepsilon \in (0, \varepsilon_0']$ und alle $t \in I$, wobei $\vartheta_0 := \frac{1}{2} \,\bigl( \vartheta_0^+ + \vartheta_0^- \bigr)$.
\end{prop}

\begin{proof}
Sei $\varepsilon \in (0, \varepsilon_0]$ und $\beta_0 := \frac{1}{2} \,\bigl( \vartheta_0^+ - \vartheta_0^- \bigr) \in (0,\pi)$ (halber Öffnungswinkel des Sektors). Dann gilt
\begin{align*}
\dist \bigl( \lambda(t) + \varepsilon e^{i \vartheta_0} , \, \sigma(A(t)) \bigr) 
&\ge \dist \Bigl( \lambda(t) + \varepsilon e^{i \vartheta_0}  , \,  \complex \setminus \bigl( \lambda(t) + \varepsilon_0 \, S_{(\vartheta_0^-, \vartheta_0^+)} \bigr)  \Bigr) \\
&= \min \{ \varepsilon \sin \beta_0, \varepsilon_0 - \varepsilon \},
\end{align*}
wie man anhand einer Zeichnung leicht (elementargeometrisch) einsieht.

Seien nun
\begin{align*}
\varepsilon_0' := \frac{\varepsilon_0}{1 + \sin \beta_0} \text{ \; und \; } M_0 := \frac{1}{\sin \beta_0}
\end{align*}
und sei $\varepsilon \in (0, \varepsilon_0']$. Dann gilt
\begin{align*}
\lambda(t) + \varepsilon e^{i \vartheta_0} \in \lambda(t) + \varepsilon_0 \, S_{(\vartheta_0^-, \vartheta_0^+)} \subset \rho(A(t))
\end{align*}
und, da die $A(t)$ normal sind,
\begin{align*}
\norm{(\lambda(t) + \varepsilon e^{i \vartheta_0} - A(t))^{-1}} = \frac{1}{  \dist \bigl( \lambda(t) + \varepsilon e^{i \vartheta_0} , \, \sigma(A(t)) \bigr)    }\le \frac{M_0}{\varepsilon}
\end{align*}
für alle $t \in I$ (Proposition~\ref{prop: rieszproj für normale A}), wie gewünscht.
\end{proof}

Der folgende Satz verschärft Satz~\ref{thm: allg adsatz ohne sl} im Sonderfall normaler $A(t)$, denn in diesem Fall reduzieren sich die Voraussetzungen von Satz~\ref{thm: allg adsatz ohne sl} auf die Voraussetzungen des folgenden Satzes: 
die $P(t)$ aus Satz~\ref{thm: allg adsatz ohne sl} sind im Sonderfall normaler $A(t)$ automatisch orthogonal, da $P(t)H$ und $(1-P(t))H$ wegen $\ker (A(t)-\lambda(t))^* = \ker (A(t) - \lambda(t))$ (Satz~\ref{thm: eigenschaften normaler A}) orthogonale Unterräume sind (zunächst nur für \emph{fast} alle $t \in I$, ein Stetigkeitsargument zeigt aber, dass dann $P(t)^* = P(t)$ sogar für alle $t \in I$). 

Der Beweis ist angelehnt an Teufels Argumentation in~\cite{Teufel 01}.

\begin{thm} \label{thm: adsatz ohne sl für normale A(t)}
Sei $A(t)$ für jedes $t \in I$ eine normale lineare Abbildung $D \subset H \to H$ mit $\sigma(A(t)) \subset \{ z \in \complex: \operatorname{Re}(z) \le 0 \}$ und sei $t \mapsto A(t)x$ stetig differenzierbar für alle $x \in D$. 
Sei $\lambda(t)$ für jedes $t \in I$ ein Eigenwert von $A(t)$ und seien $\varepsilon_0$, $\vartheta_0$ sowie $M_0$ positive Zahlen, sodass $\lambda(t) + \varepsilon e^{i \vartheta_0} \in \rho(A(t))$ und 
\begin{align*}
\norm{\bigl( \lambda(t) + \varepsilon e^{i \vartheta_0} - A(t) \bigr)^{-1} } \le \frac{M_0}{\varepsilon}
\end{align*}
für alle $\varepsilon \in (0, \varepsilon_0]$ und alle $t \in I$ (was beispielsweise unter der Sektorbedingung von Proposition~\ref{prop: resolvabsch automatisch erf} erfüllt ist). 
Sei $P(t)$ für jedes $t \in I$ eine orthogonale Projektion in $H$, sodass $P(t)H \subset \ker (A(t)-\lambda(t))$ für alle $t \in I$ und $P(t)H = \ker (A(t)-\lambda(t))$ für fast alle $t \in I$, sei $\rk P(0) < \infty$ oder $\rk (1-P(0)) < \infty$ und sei $t \mapsto P(t)$ stetig differenzierbar.
Dann gilt 
\begin{align*}
(1-P) U_T P(0) \longrightarrow 0 \quad (T \to \infty) \text{ \; gleichmäßig auf } I
\end{align*}
und sogar
\begin{align*}
U_{a,T} - U_T \longrightarrow 0 \quad (T \to \infty) \text{ \; gleichmäßig auf } I,
\end{align*}
wenn nur die Zeitentwicklung $U_{a,T}$ existiert für alle $T \in (0, \infty)$.
\end{thm}

\begin{proof}
Wir überzeugen uns zunächst davon, dass die Voraussetzungen von Satz~\ref{thm: allg adsatz ohne sl} erfüllt sind. 

Weil die $A(t)$ normal sind und $\sigma(A(t)) \subset \{ z \in \complex: \operatorname{Re}(z) \le 0 \}$, ist $A$ $(1,0)$-stabil (wie man mithilfe der Spektralintegraldarstellung der Halbgruppen $e^{A(t)\, . \,}$ einsieht). Außerdem ist $t \mapsto \lambda(t)$ nach Proposition~\ref{prop: lambda automatisch stetig db} stetig differenzierbar. 

Zeigen wir nun, dass $P(t)A(t) \subset A(t)P(t)$ für alle $t \in I$. Wir wissen nach Voraussetzung, dass 
\begin{align*}
P(t)H \subset \ker (A(t)-\lambda(t)) = P_{ \{\lambda(t)\} }^{A(t)}H
\end{align*} 
für alle $t \in I$, woraus leicht folgt, dass $P(t) = P(t)P_{ \{\lambda(t)\} }^{A(t)}$ und damit
\begin{align*}
P(t)A(t) & = P(t) P_{ \{\lambda(t)\} }^{A(t)} A(t) \subset P(t) A(t) P_{ \{\lambda(t)\} }^{A(t)} = P(t) \lambda(t) P_{ \{\lambda(t)\} }^{A(t)} = \lambda(t) P(t) = A(t)P(t)
\end{align*}
für alle $t \in I$. 

Schließlich zeigen wir noch, dass $(1-P(t))H \subset \overline{\im(A(t)-\lambda(t))}$ für fast alle $t \in I$.
Wir haben nach Satz~\ref{thm: eigenschaften normaler A}
\begin{gather*}
\bigl( \im(A(t) - \lambda(t)) \bigr) ^\perp = \ker(A(t) - \lambda(t))^* = \ker(A(t) - \lambda(t)) = P(t)H 
\end{gather*}
für fast alle $t \in I$ und daher 
\begin{align*}
\overline{\im(A(t)-\lambda(t))} & =  \left(   \bigl( \im(A(t) - \lambda(t)) \bigr) ^\perp \right)^\perp = (P(t)H)^\perp 
= (1-P(t))H 
\end{align*}
für fast alle $t \in I$.

Also sind tatsächlich alle Voraussetzungen von Satz~\ref{thm: allg adsatz ohne sl} erfüllt und dieser liefert, dass
\begin{align*}
(1-P) U_T P(0) \longrightarrow 0 \quad (T \to \infty) \text{ \; gleichmäßig auf } I,
\end{align*}
wie behauptet.
\\

Wir nehmen nun (zusätzlich) an, dass die Zeitentwicklung $U_{a,T}$ zu $TA + [P',P]$ für jedes $T \in (0, \infty)$ existiert und zeigen den zweiten Teil der behaupteten Aussage. 

Sei $Q_n(t)$ für jedes $n \in \natu$ und jedes $t \in I$ die im Beweis von Satz~\ref{thm: allg adsatz ohne sl} definierte beschränkte lineare Abbildung in $H$ und sei 
\begin{align*}
B_{n \varepsilon} (t) := \bigl( \lambda(t) + \varepsilon e^{i \vartheta_0} - A(t) \bigr)^{-1} \, Q_n(t)P(t)  \; + \; P(t)Q_n(t) \, \bigl( \lambda(t) + \varepsilon e^{i \vartheta_0} - A(t) \bigr)^{-1},
\end{align*}
\begin{align*}
C_{n \varepsilon}(t) := P(t) Q_n(t) \, \varepsilon e^{i \vartheta_0} & \bigl( \lambda(t) + \varepsilon e^{i \vartheta_0} - A(t) \bigr)^{-1} \\
&- \varepsilon e^{i \vartheta_0}  \bigl( \lambda(t) + \varepsilon e^{i \vartheta_0} - A(t) \bigr)^{-1} \, Q_n(t) P(t)
\end{align*}
sowie
\begin{align*}
C_{\varepsilon}(t) := P(t) P'(t) \, \varepsilon e^{i \vartheta_0} & \bigl( \lambda(t) + \varepsilon e^{i \vartheta_0} - A(t) \bigr)^{-1} \\
& - \varepsilon e^{i \vartheta_0}  \bigl( \lambda(t) + \varepsilon e^{i \vartheta_0} - A(t) \bigr)^{-1} \, P'(t) P(t)
\end{align*}
für alle $n \in \natu$, $\varepsilon \in (0, \varepsilon_0]$ und alle $t \in I$.
Dann gilt $Q_n(t) \longrightarrow P'(t) \;\;(n \to \infty)$ für alle $t \in (0,1)$, die Abbildung $t \mapsto Q_n(t)$ ist (stark) stetig differenzierbar für alle $n \in \natu$ und daher ist auch $t \mapsto B_{n \varepsilon}(t)$ stark stetig differenzierbar für alle $n \in \natu$ und alle $\varepsilon \in (0, \varepsilon_0]$ (Lemma~\ref{lm: abschätzung für abl der resolvente}). 
Weiter gilt, wie man leicht nachrechnet, folgende \emph{approximative} Kommutatorgleichung:
\begin{align*}
B_{n \varepsilon}(t) A(t)x - A(t) B_{n \varepsilon}(t)x + C_{n \varepsilon}(t)x = Q_n(t) P(t)x - P(t) Q_n(t)x
\end{align*}
für alle $x \in D$ und alle $t \in I$.

Wir erhalten damit
\begin{align} \label{eq: adsatz ohne sl für normale A(t) 1}
\notag U_{a,T}&(t)x - U_T(t)x \\
=& \int_0^t U_T(t,s) \Bigl( [P'(s),P(s)] - [Q_n(s),P(s)] \Bigr) U_{a,T}(s)x \,ds \notag \\
&+ \int_0^t U_T(t,s) \Bigl(  B_{n \varepsilon}(s) A(s) - A(s) B_{n \varepsilon}(s) \Bigr) U_{a,T}(s)x \,ds  \notag \\
&+ \int_0^t U_T(t,s) \, C_{n \varepsilon}(s) \, U_{a,T}(s)x \,ds \notag \\
=& \int_0^t U_T(t,s) \Bigl( [P'(s),P(s)] - [Q_n(s),P(s)] \Bigr) U_{a,T}(s)x \,ds \notag \\
&+ \frac{1}{T}\, U_T(t,s) \, B_{n \varepsilon}(t) \, U_{a,T}(s)x \Big|_{s=0}^{s=t} \notag \\
&+ \frac{1}{T} \, \int_0^t U_T(t,s) \Bigl(  B_{n \varepsilon}'(s) +  B_{n \varepsilon}(s)[P'(s),P(s)] \Bigr) U_{a,T}(s)x \,ds \notag \\
&+ \int_0^t U_T(t,s) \Bigl( C_{n \varepsilon}(s) - C_{\varepsilon}(s)  \Bigr) U_{a,T}(s)x \,ds   +  \int_0^t U_T(t,s) \,  C_{\varepsilon}(s) \,  U_{a,T}(s)x \,ds
\end{align}
für alle $T \in (0, \infty)$, $t \in I$, alle $n \in \natu$, $\varepsilon \in (0, \varepsilon_0]$ und alle $x \in D$. 
Jetzt schätzen wir nacheinander die Summanden auf der rechten Seite ab. Zunächst bemerken wir, dass wegen der $(1,0)$-Stabilität von $A$, Proposition~\ref{prop: abschätzung für gestörte zeitentw} und Lemma~\ref{lm: skalierung und (M,w)-stabilität} 
\begin{align*}
\norm{U_{a,T}(t,s)} \le e^{c (t-s)} \le  e^c
\end{align*}
gilt für alle $(s,t) \in \Delta$, wobei $c$ eine reelle Zahl sei mit
\begin{align*}
\norm{P(s)}, \norm{P'(s)} \text{\,und\,} \norm{ [P'(s), P(s)] }  \le c
\end{align*}
für alle $s \in I$. Weiter bemerken wir, dass 
\begin{align*}
\norm{ Q_n(t) } \le c
\end{align*}
für alle $n \in \natu$ und alle $t \in I$ (Definition der $Q_n(t)$), und dass nach Lemma~\ref{lm: abschätzung für abl der resolvente}
\begin{align*}
\norm{B_{n \varepsilon}'(t)} \le 2 \, \frac{M_0'}{\varepsilon^2} \, c^2 + 2\, \frac{M_0}{\varepsilon} \, (c^2 + c \, c_n)
\end{align*}
für alle $n \in \natu$, $\varepsilon \in (0, \varepsilon_0]$ und alle $t \in I$, wobei wir $c_n := \sup_{s \in I} \norm{Q_n'(s)}$ gesetzt haben. 

Aus \eqref{eq: adsatz ohne sl für normale A(t) 1} erhalten wir daher
\begin{align} \label{eq: adsatz ohne sl für normale A(t) 2}
\notag &\norm{U_{a,T}(t)x-U_T(t)x} \\
\notag & \qquad  \le   \, \int_0^1 2c \norm{Q_n(s)-P'(s)} \,e^c \,ds \, \norm{x} \\
\notag & \qquad \quad + \frac{4}{T} \, \frac{M_0}{\varepsilon} \, c^2 \, e^c \, \norm{x} 
+ \frac{1}{T} \, \Bigl( 2 \, \frac{M_0'}{\varepsilon^2} \, c^2 + 2\, \frac{M_0}{\varepsilon} \, (c^2 + c \, c_n) +  2\,\frac{M_0}{\varepsilon} \, c^2 \Bigr) \norm{x} \\
& \qquad \quad + \int_0^1 2 M_0 c \norm{Q_n(s) - P'(s) } e^c \,ds \norm{x}
+ \int_0^1 2 M_0 c \norm{C_{\varepsilon}(s)} e^c \,ds \norm{x}
\end{align}
für alle $T \in (0, \infty)$, $t \in I$, alle $n \in \natu$, $\varepsilon \in (0, \varepsilon_0]$ und alle $x \in D$.

Weil nun 
\begin{align*}
\norm{\bigl( \overline{\lambda(t)} + \varepsilon e^{-i \vartheta_0} - A(t)^* \bigr)^{-1}} = \norm{\bigl(\lambda(t) + \varepsilon e^{i \vartheta_0} - A(t)\bigl)^{-1}} \le \frac{M_0}{\varepsilon} 
\end{align*}
für alle $\varepsilon \in (0, \varepsilon_0]$ und alle $t \in I$ (nach Proposition~X.1.18 in~\cite{Conway: fana}) und 
\begin{align*}
\overline{\im \bigl( A(t)^*-\overline{\lambda(t)} \bigr)} = \bigl( \ker(A(t)-\lambda(t))^{**} \bigr)^\perp  = \bigl( \ker(A(t)-\lambda(t)) \bigr)^\perp = (1-P(t))H 
\end{align*} 
für fast alle $t \in I$ (nach Proposition~X.1.6 in~\cite{Conway: fana}), 
folgt mithilfe von Lemma~\ref{lm: zentrales lm für adsatz ohne sl}, angewandt auf $A(t)$ und $\lambda(t)$ bzw. $A(t)^*$ und $\overline{\lambda(t)}$, und mithilfe von Lemma~\ref{lm: aus starker konv mach normkonv} 
\begin{align*}
& \norm{ \varepsilon \bigl( \lambda(t) + \varepsilon e^{i \vartheta_0} - A(t) \bigr)^{-1} P'(t)P(t)} \\
& \qquad \qquad \qquad \qquad = \norm{ \varepsilon \bigl( \lambda(t) + \varepsilon e^{i \vartheta_0} - A(t) \bigr)^{-1} (1-P(t))P'(t)P(t)} \longrightarrow 0 \quad (\varepsilon \searrow 0) 
\end{align*}
bzw.
\begin{align*}
& \norm{P(t)P'(t) \, \varepsilon \bigl( \lambda(t) + \varepsilon e^{i \vartheta_0} - A(t) \bigr)^{-1} } \\
& \qquad \qquad \qquad \qquad =  \norm{\varepsilon \bigl( \overline{\lambda(t)} + \varepsilon e^{-i \vartheta_0} - A(t)^* \bigr)^{-1} (1-P(t))P'(t)P(t)} \longrightarrow 0 \quad (\varepsilon \searrow 0)
\end{align*}
und damit $\norm{ C_{\varepsilon}(t)} \longrightarrow 0 \;\;(\varepsilon \searrow 0)$ für fast alle $t \in I$.

Sei nun $\varepsilon_T :=  T^{-\frac{1}{3}}$ für alle $T \in [\varepsilon_0^{-3}, \infty )$. Anhand von \eqref{eq: adsatz ohne sl für normale A(t) 2} sehen wir dann, dass wir 
\begin{align*}
\sup_{t \in I} \norm{U_{a,T}(t) - U_T(t) } 
\end{align*}
beliebig klein machen können, indem wir zunächst eine (feste) natürliche Zahl $n_0$ wählen, die so groß ist, dass 
\begin{align*}
\int_0^1 2c \norm{Q_{n_0}(s)-P'(s)} \,e^c \,ds \text{ \;  und \; } \int_0^1 2 M_0 c \norm{Q_{n_0}(s) - P'(s) } e^c \,ds
\end{align*}
genügend klein sind, und dann $T_0 \in [\varepsilon_0^{-3}, \infty )$ so groß wählen, dass 
\begin{align*}
\frac{4}{T} \, \frac{M_0}{\varepsilon_T} \, c^2 \, e^c  
+ \frac{1}{T} \, \Bigl( 2 \, \frac{M_0'}{{\varepsilon_T}^2} \, c^2 + 2\, \frac{M_0}{\varepsilon_T} \, (c^2 + c \, c_{n_0}) +  2\,\frac{M_0}{\varepsilon_T} \, c^2 \Bigr) + \int_0^1 2 M_0 c \norm{C_{\varepsilon_T}(s)} e^c \,ds
\end{align*}
genügend klein ist für alle $T \in [T_0, \infty)$.
\end{proof}

Wir weisen darauf hin, dass wir im eben gegebenen Beweis von 
\begin{align} \label{eq: adiabat zeitentw asympotisch gleich der zeitentw, ohne sl}
\sup_{t \in I} \norm{U_{a,T}(t) - U_T(t) } \longrightarrow 0 \quad (T \to \infty)
\end{align}
nur an einer einzigen Stelle \emph{mehr} gebraucht haben, als schon in Satz~\ref{thm: allg adsatz ohne sl} vorausgesetzt ist: nämlich um zu zeigen, dass
\begin{align*}
P(t)P'(t) \varepsilon \, \bigl( \lambda(t) + \varepsilon e^{i \vartheta_0} - A(t) \bigr)^{-1} \longrightarrow 0 \quad (\varepsilon \searrow 0)
\end{align*} 
bzgl. der Normoperatortopologie für fast alle $t \in I$ (für die Konvergenz bzgl. der starken Operatortopologie genügen auch die Voraussetzungen von Satz~\ref{thm: allg adsatz ohne sl}). Im Beweis des obigen Satzes haben wir ausgenutzt, dass die Projektionen sogar orthogonal sind (was aus der Normalität der $A(t)$ folgte). Aber auch nicht mehr als das: um auf~\eqref{eq: adiabat zeitentw asympotisch gleich der zeitentw, ohne sl} schließen zu können genügt es, wenn die Projektionen $P(t)$ aus Satz~\ref{thm: allg adsatz ohne sl} orthogonal sind -- dass die $A(t)$ zusätzlich normal sind (wie im obigen Satz), ist nicht nötig.

Auch $\rk (1-P(0)) < \infty$ ist so eine Zusatzvoraussetzung (zu den Voraussetzungen von Satz~\ref{thm: allg adsatz ohne sl}), über die hinaus nichts gebraucht wird, um sogar $\sup_{t \in I} \norm{U_{a,T}(t) - U_T(t) } \longrightarrow 0 \;\; (T \to \infty)$ zu bekommen.
\\

Satz~\ref{thm: allg adsatz ohne sl} greift nur in Situationen, in denen 
\begin{align*}
\ker (A(t)-\lambda(t)) + \overline{\im (A(t)-\lambda(t))} = X
\end{align*}
für fast alle $t \in I$. Wir haben versucht, ihn so zu erweitern, dass er auch in gewissen Situationen anwendbar ist, in denen nur 
\begin{align*}
\ker (A(t)-\lambda(t))^{m_0} + \overline{\im (A(t)-\lambda(t))^{m_0}} = X
\end{align*}
für fast alle $t \in I$ und eine natürliche Zahl $m_0$ gilt. Solche Situationen sind schließlich nicht untypisch:  denken wir etwa an isolierte (aber nicht gleichmäßig isolierte) Eigenwerte $\lambda(t)$, deren algebraische Vielfachheit nicht mit der geometrischen Vielfachheit übereinstimmt (Satz~\ref{thm: isol spektralwerte}!). Wenn die $\lambda(t)$ nur endlich oft in $\sigma(A(t)) \setminus \{ \lambda(t) \}$ hineinfallen, können wir möglicherweise unseren Adiabatensatz mit nichtgleichmäßiger Spektrallückenbedingung (Satz~\ref{thm: unhandl adsatz mit nichtglm sl}) anwenden. Wenn sie aber unendlich oft in $\sigma(A(t)) \setminus \{ \lambda(t) \}$ hineinfallen (wie etwa in Beispiel~\ref{ex: motivierendes bsp für erweiterten adsatz ohne sl}), hilft dieser Satz~\ref{thm: unhandl adsatz mit nichtglm sl} nicht weiter. 
\\

Ein neuer Satz wäre daher wünschenswert. Wir sind in dieser Sache aber leider (noch) nicht sehr weit gekommen: Satz~\ref{thm: erweiterter adsatz ohne sl} sagt zwar unter zufriedenstellend allgemeinen Voraussetzungen, wann genau die Aussage von Satz~\ref{thm: allg adsatz ohne sl} erfüllt ist, aber die angegebene Bedingung scheint (im Fall $m_0 \ne 1$) nicht sehr leicht nachprüfbar. So dient dieser Satz~\ref{thm: erweiterter adsatz ohne sl} vor allem dazu vor Augen zu führen, wo Schwierigkeiten auftreten, wenn wir die Vorgehensweise von Satz~\ref{thm: allg adsatz ohne sl} auf die neue Situation zu übertragen versuchen. Vielleicht kann dieser Satz (und sein Beweis) ja aber auch den Weg weisen zu einem neuen Satz, der auch Beispiel~\ref{ex: motivierendes bsp für erweiterten adsatz ohne sl} abdeckt -- obwohl wir eher nicht daran glauben: vermutlich muss man einen ganz neuen Weg einschlagen (wenn es denn überhaupt gehen sollte).

\begin{lm}  \label{lm: lm zu erweitertem adsatz ohne sl}
Sei $A(t): D \subset X \to X$ für jedes $t \in I$ Erzeuger einer stark stetigen Halbgruppe auf $X$, sei $A$ $(M,0)$-stabil und $t \mapsto A(t)x$ stetig für alle $x \in D$.
Sei $\lambda(t)$ für jedes $t \in I$ ein Eigenwert von $A(t)$, sodass $\lambda(t) + \varepsilon e^{i \vartheta_0} \in \rho (A(t))$ für alle $\varepsilon \in (0, \varepsilon_0]$ und alle $t \in I$, und sei $t \mapsto \lambda(t)$ stetig.
Sei $P(t)$ für jedes $t \in I$ eine beschränkte Projektion in $X$ mit $P(t) A(t) \subset A(t) P(t)$ und $t \mapsto P(t)x$ stetig für alle $x \in X$.

(i) Sei weiter
\begin{align*}
P(t)X \subset \ker (A(t)-\lambda(t))^{m_0}
\end{align*}
für alle $t \in I$. Dann existiert eine Zahl $M_1$, sodass
\begin{align*}
\norm{ \bigl( \lambda(t) + \varepsilon e^{i \vartheta_0} - A(t) \bigr)^{-1} P(t) } \le \frac{M_1}{\varepsilon^{m_0}}
\end{align*}
für alle $\varepsilon \in (0, \varepsilon_0]$ und alle $t \in I$. 

(ii) Sei weiter 
\begin{align*}
(1-P(t))X \subset \overline{ \im (A(t)-\lambda(t))^{m_0} }
\end{align*}
für fast alle $t \in I$ und es existiere eine Zahl $M_2$, sodass
\begin{align*}
\norm{ \bigl( \lambda(t) + \varepsilon e^{i \vartheta_0} - A(t) \bigr)^{-1} \, (1-P(t)) }  \le \frac{M_2}{\varepsilon}
\end{align*}
für alle $\varepsilon \in (0, \varepsilon_0]$ und alle $t \in I$. Dann gilt
\begin{align*}
\norm{ \varepsilon \, \bigl( \lambda(t) + \varepsilon e^{i \vartheta_0} - A(t) \bigr)^{-1} \, (1-P(t)) } \longrightarrow 0 \quad (\varepsilon \searrow 0)
\end{align*}
bzgl. der starken Operatortopologie für fast alle $t \in I$.
\end{lm}

\begin{proof}
(i) Wir zeigen zunächst, dass $t \mapsto (A(t)-\lambda(t))^k P(t)$ stark stetig ist für alle $k \in \{0, \dots, m_0-1 \}$. 
Sei also $k \in \{0, \dots, m_0-1 \}$ und sei $z$ eine komplexe Zahl, sodass $\Re(\lambda(t) + z) \ge 1$ für alle $t \in I$. Dann gilt aufgrund der $(M,0)$-Stabilität von $A$, dass $\lambda(t) + z \in \rho(A(t))$ und 
\begin{align*}
\norm{ \bigl( A(t)- \lambda(t) - z \bigr)^{-1}  } \le M
\end{align*} 
für alle $t \in I$, und daher ist $t \mapsto \bigl( A(t)- \lambda(t) - z \bigr)^{-1}$ stark stetig nach Lemma~\ref{lm: reg of inv}.

Wegen $P(t)X \subset \ker (A(t)-\lambda(t))^{m_0}$ gilt $(A(t)-\lambda(t))^{k+i} P(t) = 0$ für alle $i \in \{m_0 - k, \, m_0-k+1, \dots \}$ und wir erhalten mit der binomischen Formel
\begin{align*}
&(A(t)-\lambda(t))^k P(t)x \\
& \qquad \quad = (A(t)-\lambda(t))^k \, \bigl( A(t)-\lambda(t)-z \bigr)^{m_0-1} \, \bigl( A(t)-\lambda(t)-z \bigr)^{-(m_0-1)} \,P(t)x  \\
& \qquad \quad = \sum_{i=0}^{m_0-k-1} {m_0-1 \choose i} (A(t)-\lambda(t))^{k+i} \, (-z)^{m_0-1-i} \, \bigl( A(t)-\lambda(t)-z \bigr)^{-(m_0-1)} \, P(t)x
\end{align*}
für alle $x \in X$ und alle $t \in I$.
Da nun $t \mapsto (A(t)-\lambda(t))^{k+i} \,  \bigl( A(t)-\lambda(t)-z \bigr)^{-(m_0-1)}$ für $i \in \{ 0, \dots, m_0-k-1 \}$ stark stetig ist, folgt die starke Stetigkeit von $t \mapsto (A(t)-\lambda(t))^k P(t)$.

Jetzt können wir die behauptete Abschätzung sehr leicht einsehen. Wegen $(A(t)-\lambda(t))^{m_0} P(t) = 0$ haben wir nämlich
\begin{align*}
\bigl( \lambda(t) + \varepsilon e^{i \vartheta_0} - A(t) \bigr)^{-1} P(t) &= \frac{1}{ \varepsilon e^{i \vartheta_0} } \, \biggl( 1 - \frac{ A(t) - \lambda(t) }{  \varepsilon e^{i \vartheta_0} }   \biggr)^{-1} \, P(t) \\
&= \frac{1}{ \varepsilon e^{i \vartheta_0} } \, \sum_{k = 0}^{m_0-1} \biggl( \frac{ A(t) - \lambda(t) }{  \varepsilon e^{i \vartheta_0}  }   \biggr)^k \, P(t)
\end{align*}
und daher existiert wegen $\sup_{t \in I} \norm{ (A(t)-\lambda(t))^k P(t) } < \infty$ wie behauptet eine Zahl $M_1$, sodass
\begin{align*}
\norm{ \bigl( \lambda(t) + \varepsilon e^{i \vartheta_0} - A(t) \bigr)^{-1} P(t) } \le \frac{M_1}{\varepsilon^{m_0}}
\end{align*}
für alle $\varepsilon \in (0, \varepsilon_0]$ und alle $t \in I$.
\\

(ii) Sei $N := \big \{ t \in I: (1-P(t))X \not \subset \overline{ \im (A(t)-\lambda(t))^{m_0} }  \big \}$. Wir zeigen, dass 
\begin{align*}
\norm{ \varepsilon \, \bigl( \lambda(t) + \varepsilon e^{i \vartheta_0} - A(t) \bigr)^{-1} \, (1-P(t)) } \longrightarrow 0 \quad (\varepsilon \searrow 0)
\end{align*}
bzgl. der starken Operatortopologie für alle $t \in I \setminus N$.

Sei $t \in I \setminus N$ und zunächst $y_0 \in \im(A(t)-\lambda(t))^{m_0}$. Dann $y_0 = (\lambda(t)-A(t))^{m_0} x_0$ für ein $x_0 \in D\bigl( (A(t)-\lambda(t))^{m_0} \bigr)$ und daher
\begin{align*}
&\varepsilon \, \bigl( \lambda(t) + \varepsilon e^{i \vartheta_0} - A(t) \bigr)^{-1} \, (1-P(t))y_0 \\
& \qquad \quad = \varepsilon \, \bigl( \lambda(t) + \varepsilon e^{i \vartheta_0} - A(t) \bigr)^{-1} \, (\lambda(t)-A(t))^{m_0} \, (1-P(t))x_0 \\
& \qquad \quad = \varepsilon \, \bigl( \lambda(t) + \varepsilon e^{i \vartheta_0} - A(t) \bigr)^{-1} \, \bigl( - \varepsilon e^{i \vartheta_0} \bigr)^{m_0} (1-P(t))x_0 \\
& \qquad \quad \quad + \varepsilon \, \sum_{k=1}^{m_0} {m_0 \choose k} \bigl( \lambda(t) + \varepsilon e^{i \vartheta_0} - A(t) \bigr)^{k-1} \, \bigl( - \varepsilon e^{i \vartheta_0} \bigr)^{m_0 -k} \, (1-P(t))x_0  \\
& \qquad \quad \longrightarrow 0 \quad (\varepsilon \searrow 0).
\end{align*}
Sei nun $y \in (1-P(t))X$. Dann $y \in \overline{  \im(A(t)-\lambda(t))^{m_0}  }$ und, da 
\begin{align*}
\norm{ \varepsilon \, \bigl( \lambda(t) + \varepsilon e^{i \vartheta_0} - A(t) \bigr)^{-1} \, (1-P(t)) }  \le M_2
\end{align*}
für alle $\varepsilon \in (0, \varepsilon_0]$, folgt aus obigem, dass auch 
\begin{align*}
\varepsilon \, \bigl( \lambda(t) + \varepsilon e^{i \vartheta_0} - A(t) \bigr)^{-1} \, (1-P(t))y \longrightarrow 0 \quad (\varepsilon \searrow 0),
\end{align*}
wie gewünscht.
\end{proof}

Aus diesem Lemma bekommen wir nun leicht den angesprochenen nicht sehr befriedigenden Satz.

\begin{thm}  \label{thm: erweiterter adsatz ohne sl}
Sei $A(t)$ für jedes $t \in I$ eine lineare Abbildung $D \subset X \to X$, die eine stark stetige Halbgruppe auf $X$ erzeugt, sei $A$ $(M,0)$-stabil und sei $t \mapsto A(t)x$ stetig differenzierbar für alle $x \in D$. 
Sei $\lambda(t)$ für jedes $t \in I$ ein Eigenwert von $A(t)$, sodass $\lambda(t) + \varepsilon e^{i \vartheta_0} \in \rho(A(t))$ für alle $\varepsilon \in (0, \varepsilon_0]$ und alle $t \in I$, und sei $t \mapsto \lambda(t)$ stetig differenzierbar.
Sei $P(t)$ für jedes $t \in I$ eine beschränkte Projektion in $X$ mit $P(t)A(t) \subset A(t)P(t)$, sodass 
\begin{align*}
P(t)X \subset \ker (A(t)-\lambda(t))^{m_0}
\end{align*} 
für alle $t \in I$ und 
\begin{align*}
(1-P(t))X \subset \overline{\im(A(t)-\lambda(t))^{m_0}}
\end{align*}
für fast alle $t \in I$, sei $\rk P(0) < \infty$ oder $\rk (1-P(0)) < \infty$ und sei $t \mapsto P(t)x$ zweimal stetig differenzierbar für alle $x \in X$. Schließlich existiere eine Zahl $M_0$, sodass
\begin{align*}
\norm{ \bigl( \lambda(t) + \varepsilon e^{i \vartheta_0} - A(t) \bigr)^{-1} \, (1-P(t)) }  \le \frac{M_0}{\varepsilon}
\end{align*}
für alle $\varepsilon \in (0, \varepsilon_0]$ und alle $t \in I$. 
Dann gilt die Aussage von Satz~\ref{thm: allg adsatz ohne sl} genau dann, wenn
\begin{align*}
\sup_{t \in I} \norm{  (1-P(t)) \int_0^t U_T(t,s) \, T \, \bigl( \lambda(s)-A(s) \bigr) \, V_T(s)P(0) \, ds  } \longrightarrow 0 \quad (T \to \infty),
\end{align*}
wobei $V_T$ die Zeitentwicklung zu $T \lambda + [P',P]$ bezeichnet.
\end{thm}

\begin{proof}
Sei $V_T$ die Zeitentwicklung zu $T \lambda + [P',P]$. Dann gilt $\norm{V_T(t,s)} \le e^{c(t-s)}$ für alle $(s,t) \in \Delta$ und alle $T \in (0, \infty)$ ($c$ eine generische Konstante) und es gilt $V_T(t)P(0) = P(t)V_T(t)$ für alle $t \in I$ (s. den Beweis von Satz~\ref{thm: allg adsatz ohne sl}). 

Die Abbildung $[0,t] \ni s \mapsto U_T(t,s) V_T(s) P(0)x$ ist differenzierbar für alle $x \in X$ mit Ableitung
\begin{align*}
s \mapsto \: &\dds{ U_T(t,s)V_T(s)P(0)x } \\
& = U_T(t,s) \bigl( T (\lambda(s)-A(s)) \bigr) V_T(s)P(0)x + U_T(t,s) \, P'(s) \, V_T(s) P(0)x,
\end{align*}
die wegen der starken Stetigkeit von $s \mapsto (A(s)-\lambda(s))P(s)$ (Beweis von Lemma~\ref{lm: lm zu erweitertem adsatz ohne sl}) stetig ist.
Also 
\begin{align}  \label{eq: erweiterter adsatz ohne sl 1}
&V_T(t)P(0)x - U_T(t)P(0)x \notag \\
& \qquad \quad = \int_0^t U_T(t,s) \bigl( T (\lambda(s)-A(s)) \bigr) V_T(s)P(0)x \,ds \notag \\
& \qquad \quad \quad + \int_0^t U_T(t,s) \, P'(s) \, V_T(s) P(0)x \,ds
\end{align}
für alle $x \in X$ und alle $t \in I$ sowie $T \in (0,\infty)$. 

Wir zeigen nun, dass 
\begin{align*}
\sup_{t \in I} \norm{ \int_0^t U_T(t,s) \, P'(s) \, V_T(s) P(0) \,ds } \longrightarrow 0 \quad (T \to \infty),
\end{align*}
woraus dann aufgrund von~\eqref{eq: erweiterter adsatz ohne sl 1} die behauptete Äquivalenz folgt.

Zunächst gilt
\begin{align}   \label{eq: erweiterter adsatz ohne sl 2}
&\int_0^t U_T(t,s) \, P'(s) \, V_T(s) P(0)x \,ds \notag \\
& \qquad \quad = \int_0^t U_T(t,s) \, (\lambda(s)-A(s)) \, \bigl( \lambda(s) + \varepsilon e^{i \vartheta_0} - A(s) \bigr)^{-1} \, P'(s) \, V_T(s) P(0)x \,ds \notag \\
& \qquad \quad \quad + \int_0^t U_T(t,s) \, \varepsilon e ^{i \vartheta_0}  \, \bigl( \lambda(s) - \varepsilon e^{i \vartheta_0} - A(s) \bigr)^{-1} \, P'(s) \, V_T(s) P(0)x \,ds
\end{align} 
für alle $\varepsilon \in (0, \varepsilon_0]$.
Aus Lemma~\ref{lm: lm zu erweitertem adsatz ohne sl} folgt 
\begin{align*}
\norm{ \bigl( \lambda(s) + \varepsilon e^{i \vartheta_0} - A(s) \bigr)^{-1} } \le  \frac{c}{\varepsilon^{m_0}}
\end{align*}
für alle $s \in I$ und alle $\varepsilon \in (0, \varepsilon_0]$ und daraus nach Lemma~\ref{lm: reg of inv}, dass $s \mapsto \bigl( \lambda(s) + \varepsilon e ^{i \vartheta_0} - A(s) \bigr)^{-1}$ stark stetig differenzierbar ist. Wir können also wieder (wie im Beweis von Satz~\ref{thm: allg adsatz ohne sl}) partiell integrieren und erhalten
\begin{align*}
& \int_0^t U_T(t,s) \, (\lambda(s)-A(s)) \, \bigl( \lambda(s) + \varepsilon e^{i \vartheta_0} - A(s) \bigr)^{-1} \, P'(s) \, V_T(s) P(0)x \,ds \\
& \quad = \frac{1}{T} \, U_T(t,s) \bigl( \lambda(s) + \varepsilon e^{i \vartheta_0} - A(s) \bigr)^{-1} \, P'(s) \, V_T(s) P(0)x \big|_{s=0}^{s=t} \\
& \quad \quad - \frac{1}{T} \, \int_0^t  U_T(t,s) \Bigl( \dds{ \bigl( \lambda(s) + \varepsilon e^{i \vartheta_0} - A(s) \bigr)^{-1} } \Bigr) (1-P(s)) P'(s) P(s) \, V_T(s) P(0)x \,ds \\
& \quad \quad - \frac{1}{T} \, \int_0^t U_T(t,s) \bigl( \lambda(s) + \varepsilon e^{i \vartheta_0} - A(s) \bigr)^{-1} \Bigl( P''(s) + P'(s)^2 \Bigr) V_T(s) P(0)x \,ds.
\end{align*}

Jetzt können wir die Summanden auf der rechten Seite von~\eqref{eq: erweiterter adsatz ohne sl 2} abschätzen. Erstens gilt nach der eben durchgeführten partiellen Integration
\begin{align} \label{eq: erweiterter adsatz ohne sl 3} 
& \norm{  \int_0^t U_T(t,s) \, (\lambda(s)-A(s)) \, \bigl( \lambda(s) + \varepsilon e^{i \vartheta_0} - A(s) \bigr)^{-1} \, P'(s) \, V_T(s) P(0)x \,ds  } \notag \\
& \qquad \le c \, \Bigl( \frac{1}{T \, \varepsilon} + \frac{1}{T \, \varepsilon^{m_0+1}} + \frac{1}{T \, \varepsilon^{m_0}} \Bigr) \norm{x} \le \frac{c}{T \, \varepsilon^{m_0+1}}  \, \norm{x},
\end{align}
denn mithilfe von Lemma~\ref{lm: reg of inv} bekommen wir
\begin{align*}
\norm{   \Bigl( \dds{ \bigl( \lambda(s) + \varepsilon e^{i \vartheta_0} - A(s) \bigr)^{-1} } \Bigr) (1-P(s))   } \le \frac{c}{ \varepsilon^{m_0+1}   }
\end{align*}
für alle $s \in I$.
Zweitens gilt
\begin{align}  \label{eq: erweiterter adsatz ohne sl 4} 
\norm{   \int_0^t U_T(t,s) \, \varepsilon e ^{i \vartheta_0}  \, \bigl( \lambda(s) - \varepsilon e^{i \vartheta_0} - A(s) \bigr)^{-1} \, P'(s) \, V_T(s) P(0)x \,ds   } \notag \\
\le c \, \int_0^t \norm{  \varepsilon e ^{i \vartheta_0}  \, \bigl( \lambda(s) - \varepsilon e^{i \vartheta_0} - A(s) \bigr)^{-1} \, (1-P(s))P'(s)P(s) } \,ds \, \norm{x}
\end{align}

Weiter haben wir nach dem lebesgueschen Satz
\begin{align*}
\int_0^t \norm{  \varepsilon  \, \bigl( \lambda(s) - \varepsilon e^{i \vartheta_0} - A(s) \bigr)^{-1} \, (1-P(s))P'(s)P(s) } \,ds \longrightarrow 0 \quad (\varepsilon \searrow 0),
\end{align*}
da ja wegen Lemma~\ref{lm: lm zu erweitertem adsatz ohne sl} und $\rk P(0) < \infty$ oder $\rk (1-P(0)) < \infty$ 
\begin{align*}
\varepsilon  \, \bigl( \lambda(s) - \varepsilon e^{i \vartheta_0} - A(s) \bigr)^{-1} \, (1-P(s))P'(s)P(s) \longrightarrow 0 \quad (\varepsilon \searrow 0)
\end{align*}
bzgl. der Normoperatortopologie für fast alle $s \in I$.

Sei nun 
\begin{align*}
\varepsilon_T :=  T^{-\frac{1}{m_0 +2}} \text{ \; für alle } T \in \bigl[ \varepsilon_0^{-(m_0 +2)}, \infty \bigr).
\end{align*}
Dann folgt mithilfe von \eqref{eq: erweiterter adsatz ohne sl 3} und \eqref{eq: erweiterter adsatz ohne sl 4} aus \eqref{eq: erweiterter adsatz ohne sl 2}, dass
\begin{align*}
&\sup_{t \in I} \norm{ \int_0^t U_T(t,s) \, P'(s) \, V_T(s) P(0) \,ds } \\
& \qquad \le \frac{c}{T \, {\varepsilon_T}^{m_0+1}} + \int_0^1 \norm{  \varepsilon_T  \, \bigl( \lambda(s) - \varepsilon_T e^{i \vartheta_0} - A(s) \bigr)^{-1} \, (1-P(s))P'(s)P(s) } \,ds \\
& \qquad \longrightarrow 0 \quad (\varepsilon \searrow 0),
\end{align*} 
wie gewünscht.
\end{proof}

\subsection{Quantitative Adiabatensätze ohne Spektrallückenbedingung}

In unserem ersten (qualitativen) Adiabatensatz ohne Spektrallückenbedingung (Satz~\ref{thm: allg adsatz ohne sl}) war es ganz entscheidend, dass
\begin{align*}
\norm{ \varepsilon \, \bigl( \lambda(t) + \varepsilon e^{i \vartheta_0} - A(t) \bigr)^{-1} \, P'(t)P(t) } \longrightarrow 0 \quad (\varepsilon \searrow 0)
\end{align*}
für fast alle $t \in I$. Wenn wir dies nun verschärfen und annehmen, dass
\begin{align*}
\sup_{t \in I} \norm{ \varepsilon \, \bigl( \lambda(t) + \varepsilon e^{i \vartheta_0} - A(t) \bigr)^{-1} \, P'(t)P(t) } \le \eta(\varepsilon)
\end{align*}
für eine (von $t$ unabhängige!) Abbildung $\eta$ mit $\eta(\varepsilon) \longrightarrow 0 \;\;(\varepsilon \searrow 0)$, so erhalten wir eine Konvergenzrate für den uns interessierenden Ausdruck $\sup_{t \in I} \norm{(1-P(t)) U_T(t) P(0) }$, das heißt einen quantitativen Adiabatensatz. 

Der folgende Satz ist direkt aus~\cite{Griesemer 07} übertragen und kommt ohne die Voraussetzung aus, dass $\rk P(0) < \infty$ oder $\rk (1-P(0)) < \infty$, und auch auf die Voraussetzung, dass $(1-P(t))X \subset \overline{ \im (A(t)- \lambda(t)) }$ für fast alle $t \in I$, können wir verzichten. 

\begin{thm} \label{thm: quant adsatz ohne sl}
Seien $A(t)$ und $\lambda(t)$ wie in Satz~\ref{thm: allg adsatz ohne sl}. Sei $P(t)$ für jedes $t \in I$ eine beschränkte Projektion in $X$ mit $P(t)A(t) \subset A(t)P(t)$, sodass $P(t)X \subset \ker (A(t)-\lambda(t))$ für alle $t \in I$, und sei $t \mapsto P(t)x$ zweimal stetig differenzierbar für alle $x \in X$. Sei schließlich $\eta$ eine Abbildung $(0,\varepsilon_0] \to (0,\infty)$, sodass $\eta(\varepsilon) \ge \varepsilon$ und 
\begin{align*}
\sup_{t \in I} \norm{ \bigl( \lambda(t) + \varepsilon e^{i \vartheta_0} - A(t) \bigr)^{-1} \, P'(t)P(t) } \le \frac{\eta(\varepsilon)}{\varepsilon}
\end{align*}
für alle $\varepsilon \in (0, \varepsilon_0]$. Dann existiert eine Zahl $c$, sodass
\begin{align*}
\sup_{t \in I} \norm{ (1-P(t)) U_T(t) P(0)  } \le c \, \eta \bigl( T^{-\frac{1}{2}} \bigr)
\end{align*}
für alle $T \in [\varepsilon_0^{-2}, \infty)$.
\end{thm}

\begin{proof}
Sei $V_T$ die Zeitentwicklung zu $T \lambda + [P',P]$. Dann ergibt sich auf die übliche Weise, dass
\begin{align*}
& V_T(t)P(0)x - U_T(t)P(0)x \\
& \qquad = \frac{1}{T} \, U_T(t,s) \bigl( \lambda(s) + \varepsilon e^{i \vartheta_0} - A(s) \bigr)^{-1} \, P'(s) \, V_T(s) P(0)x \big|_{s=0}^{s=t} \\
& \quad \quad \quad - \frac{1}{T} \, \int_0^t  U_T(t,s) \Bigl( \dds{ \bigl( \lambda(s) + \varepsilon e^{i \vartheta_0} - A(s) \bigr)^{-1} } \Bigr)  P'(s) P(s) \, V_T(s) P(0)x \,ds \\
& \quad \quad \quad - \frac{1}{T} \, \int_0^t U_T(t,s) \bigl( \lambda(s) + \varepsilon e^{i \vartheta_0} - A(s) \bigr)^{-1} \Bigl( P''(s) + P'(s)^2 \Bigr) V_T(s) P(0)x \,ds \\
& \qquad \quad +  \int_0^t U_T(t,s) \, \varepsilon e ^{i \vartheta_0}  \, \bigl( \lambda(s) - \varepsilon e^{i \vartheta_0} - A(s) \bigr)^{-1} \, P'(s) \, V_T(s) P(0)x \,ds 
\end{align*}
für alle $x \in X$ und alle $t \in I$ sowie $T \in (0, \infty)$.

Sei $\varepsilon_T := T^{-\frac{1}{2}}$ für alle $T \in [\varepsilon_0^{-2}, \infty)$. Angesichts von  
\begin{align*}
\norm{  \Bigl( \dds{ \bigl( \lambda(s) + \varepsilon e^{i \vartheta_0} - A(s) \bigr)^{-1} } \Bigr)  P'(s) P(s)  } \le \frac{M_0}{\varepsilon} \, c' \, \frac{\eta(\varepsilon)}{\varepsilon} 
\end{align*}
für alle $\varepsilon \in (0, \varepsilon_0]$ und alle $s \in I$ (vgl. den Beweis von Lemma~\ref{lm: abschätzung für abl der resolvente}) und der getroffenen Voraussetzungen ist die Behauptung nunmehr klar.
\end{proof}

Wenn wir neben 
\begin{align*}
\sup_{t \in I} \norm{ \bigl( \lambda(t) + \varepsilon e^{i \vartheta_0} - A(t) \bigr)^{-1} \, P'(t)P(t) } \le \frac{\eta(\varepsilon)}{\varepsilon}
\end{align*}
auch noch
\begin{align*}
\sup_{t \in I} \norm{ P(t)P'(t) \, \bigl( \lambda(t) + \varepsilon e^{i \vartheta_0} - A(t) \bigr)^{-1} } \le \frac{\eta(\varepsilon)}{\varepsilon}
\end{align*}
für alle $\varepsilon \in (0, \varepsilon_0]$ voraussetzen und die übrigen Voraussetzungen des obigen Satzes unverändert übernehmen, bekommen wir sogar
\begin{align*}
\sup_{t \in I} \norm{ U_{a,T}(t) - U_T(t) } \le c \, \eta \bigl( T^{-\frac{1}{2}} \bigr)
\end{align*}
für alle $T \in [\varepsilon_0^{-2}, \infty)$. Dies sieht man wie im Beweis von Satz~\ref{thm: adsatz ohne sl für normale A(t)} (anhand von~\eqref{eq: adsatz ohne sl für normale A(t) 1}, wo wir $Q_n$ durch $P'$ ersetzen können, da wir $P$ hier sogar als \emph{zweimal} stark stetig differenzierbar annehmen). 
\\

Wir geben noch einen auf den Sonderfall normaler $A(t)$ zugeschnittenen quantitativen Adiabatensatz an, der einen entsprechenden Satz von Avron und Elgart (Korollar~1 in~\cite{AvronElgart 99}) und auch Bemerkung~1 in~\cite{Teufel 01} ein wenig verallgemeinert, insbesondere die dort angegebenen Konvergenzraten geringfügig verbessert.

\begin{thm} \label{thm: quant adsatz für hoelderstet spektrmass}
Sei $A(t)$ wie in Satz~\ref{thm: adsatz ohne sl für normale A(t)}.
Sei $\lambda(t)$ für jedes $t \in I$ ein Eigenwert von $A(t)$ mit 
\begin{align*}
S(t):= \lambda(t) + \varepsilon_0 \, S_{(\vartheta_0^-, \vartheta_0^+)} \subset \rho(A(t)),
\end{align*}
wobei $\varepsilon_0 \in (0,\infty)$ und $\vartheta_0^-$, $\vartheta_0^+ \in [0, 2 \pi)$ von $t$ unabhängige Zahlen seien mit $\vartheta_0^-$ < $\vartheta_0^+$. 
Sei $P(t)$ für jedes $t \in I$ eine orthogonale Projektion mit $P(t)A(t) \subset A(t)P(t)$ und $P(t)H \subset \ker (A(t) - \lambda(t))$ und sei $t \mapsto P(t)x$ zweimal stetig differenzierbar für alle $x \in H$. Seien schließlich die Spektralmaße $P^{A(t)}$ der $A(t)$ $\alpha$-hölderstetig gleichmäßig in $t \in I$ für ein $\alpha \in (0,1]$, das heißt: es gibt ein $c_0 \in (0, \infty)$, sodass 
\begin{align*}
P_{x,x}^{A(t)}(E) \le c_0 \, \lambda(E)^{\alpha} \, \norm{x}^2
\end{align*}
für alle $E \in \mathcal{B}_{\complex}$ mit $\lambda(E) \le 1$, alle $x \in H$ und alle $t \in I$.
Dann existiert eine Zahl $c$, sodass
\begin{align*}
\sup_{t \in I} \norm{ U_{a,T}(t) - U_T(t) } \le c \; T^{ -\frac{\alpha}{2(1+\alpha)} }
\end{align*}
für alle $T \in (0, \infty)$.
\end{thm}

\begin{proof}
Wir zeigen, dass positive Zahlen $\varepsilon_0'$ und $M'$ existieren, sodass
\begin{align*}
\norm{  \bigl( \lambda(t) + \varepsilon e^{i \vartheta_0} - A(t) \bigr)^{-1}  } \le M' \, \frac{ \varepsilon^{\frac{\alpha}{1+\alpha}} }{\varepsilon}
\end{align*}
für alle $\varepsilon \in (0, \varepsilon_0']$ und alle $t \in I$, wobei $\vartheta_0 := \frac{1}{2} \bigl( \vartheta_0^- + \vartheta_0^+ \bigr)$.

Sei $r := \frac{1}{1+\alpha}$ und sei
\begin{align*}
\sigma_{1 \varepsilon}(t) := \sigma(A(t)) \cap U_{\varepsilon^r}(\lambda(t)) \text{ \; sowie \; } \sigma_{2 \varepsilon}(t) := \sigma(A(t)) \cap \complex \setminus U_{\varepsilon^r}(\lambda(t))
\end{align*}
für alle $\varepsilon \in (0, \infty)$ und alle $t \in I$. 
Dann gilt nach Satz~\ref{thm: rechenregeln spektralintegral}
\begin{align}  \label{eq: quant adsatz für hoelderstet spektrmass 1}
&\norm{ \bigl( \lambda(t) + \varepsilon e^{i \vartheta_0} - A(t) \bigr)^{-1}  x}^2 \notag \\
& \qquad = \int_{\sigma_{1 \varepsilon}(t)} \frac{1}{ | \lambda(t) + \varepsilon e^{i \vartheta_0} - z |^2 } \, dP_{x,x}^{A(t)} (z)  
+ \int_{\sigma_{2 \varepsilon}(t)} \frac{1}{ | \lambda(t) + \varepsilon e^{i \vartheta_0} - z |^2 } \, dP_{x,x}^{A(t)} (z) \notag \\
& \qquad \le \sup_{z \in \sigma_{1 \varepsilon}(t) } \frac{1}{ | \lambda(t) + \varepsilon e^{i \vartheta_0} - z |^2 } \; P_{x,x}^{A(t)} (\sigma_{1 \varepsilon}(t))  \notag \\
&\qquad \quad + \sup_{z \in \sigma_{2 \varepsilon}(t) } \frac{1}{ | \lambda(t) + \varepsilon e^{i \vartheta_0} - z |^2 } \; P_{x,x}^{A(t)} (\sigma_{2 \varepsilon}(t)) 
\end{align} 
für alle $x \in X$.
Wir schätzen nun die beiden Summanden auf der rechten Seite ab. 
Sei $t \in I$.

Wenn $z \in \sigma_{1 \varepsilon}(t)$, dann $z \in U_{\varepsilon^r}(\lambda(t)) \setminus S(t)$ und wir sehen wie im Beweis von Proposition~\ref{prop: resolvabsch automatisch erf}, dass 
\begin{align*}
| \lambda(t) + \varepsilon e^{i \vartheta_0} - z | \ge \varepsilon \, \sin \beta_0
\end{align*}
für genügend kleine $\varepsilon \in (0, \varepsilon_0]$, wobei $\beta_0 := \frac{1}{2}\bigl( \vartheta_0^+ - \vartheta_0^- \bigr)$. 
Also gilt wegen der $\alpha$-Hölderstetigkeit der Spektralmaße
\begin{align}                 \label{eq: quant adsatz für hoelderstet spektrmass 2}
&\sup_{z \in \sigma_{1 \varepsilon}(t) } \frac{1}{ | \lambda(t) + \varepsilon e^{i \vartheta_0} - z |^2 } \; P_{x,x}^{A(t)} (\sigma_{1 \varepsilon}(t)) \notag \\
& \qquad \le \frac{1}{(\sin \beta_0)^2 \, \varepsilon^2} \, c_0 \, \lambda\bigl( U_{\varepsilon^r}(\lambda(t)) \bigr)^{\alpha} \, \norm{x}^2 
= c_0 \, \frac{ \pi^{\alpha} }{(\sin \beta_0)^2} \, \biggl( \frac{ \varepsilon^{\alpha r} }{ \varepsilon } \biggr)^2 \, \norm{x}^2 \notag \\
& \qquad = {c_1}^2 \, \biggl( \frac{ \varepsilon^{ \frac{\alpha}{1+\alpha} } }{ \varepsilon } \biggr)^2 \, \norm{x}^2
\end{align} 
für genügend kleine $\varepsilon \in (0, \varepsilon_0]$.

Wenn $z \in \sigma_{2 \varepsilon}(t)$, dann $z \notin U_{\varepsilon^r}(\lambda(t))$ und daher
\begin{align*}
| \lambda(t) + \varepsilon e^{i \vartheta_0} - z | \ge | \lambda(t) - z | - \varepsilon \ge \bigl( 1-\varepsilon^{1-r} \bigr) \, \varepsilon^r \ge \frac{1}{c_2} \, \varepsilon^r
\end{align*}
für genügend kleine $\varepsilon \in (0, \varepsilon_0]$. Also gilt
\begin{align}                     \label{eq: quant adsatz für hoelderstet spektrmass 3}
&\sup_{z \in \sigma_{2 \varepsilon}(t) } \frac{1}{ | \lambda(t) + \varepsilon e^{i \vartheta_0} - z |^2 } \; P_{x,x}^{A(t)} (\sigma_{2 \varepsilon}(t)) \notag \\ 
& \qquad \le {c_2}^2 \, \biggl( \frac{\varepsilon^{1-r}}{\varepsilon} \biggr)^2 \, \norm{x}^2 
= {c_2}^2 \, \biggl( \frac{ \varepsilon^{ \frac{\alpha}{1+\alpha} } }{ \varepsilon } \biggr)^2 \, \norm{x}^2
\end{align}
für genügend kleine $\varepsilon \in (0, \varepsilon_0]$.

Aus \eqref{eq: quant adsatz für hoelderstet spektrmass 1}, \eqref{eq: quant adsatz für hoelderstet spektrmass 2} und \eqref{eq: quant adsatz für hoelderstet spektrmass 3} folgt nun, dass ein $\varepsilon_0' \in (0, \varepsilon_0]$ und eine Zahl $M'$ existiert, sodass
\begin{align*}
\norm{  \bigl( \lambda(t) + \varepsilon e^{i \vartheta_0} - A(t) \bigr)^{-1}  } \le M' \, \frac{ \varepsilon^{\frac{\alpha}{1+\alpha}} }{\varepsilon}
\end{align*}
für alle $\varepsilon \in (0, \varepsilon_0']$  und alle $t \in I$.

Sei nun $\eta(\varepsilon) := \varepsilon^{\frac{\alpha}{1+\alpha}}$ für alle $\varepsilon \in (0, \varepsilon_0']$. Dann folgt mit den Ausführungen nach Satz~\ref{thm: quant adsatz ohne sl}, dass eine Zahl $c$ existiert, sodass
\begin{align*}
\sup_{t \in I} \norm{U_{a,T}(t) - U_T(t) } \le c \, \eta( T^{-\frac{1}{2}} ) = c \, T^{ -\frac{\alpha}{2(1+\alpha)} }
\end{align*}
für alle $T \in (0, \infty)$.
\end{proof}

\subsection{Wie weit reichen die Sätze? Beispiele}

Wir beginnen mit einem Beispiel, in dem $\lambda(t)$ für kein einziges $t \in I$ isoliert ist in $\sigma(A(t))$.

\begin{ex}  \label{ex: adsatz ohne sl, sigma nicht einpunktig}
Sei $X := \ell^2(I_{\infty})$. Sei 
\begin{align*}
A_0 := \begin{pmatrix}
0		& 0 						& 0 				&												&							&  \\
0		&	-1             & 1             & 0               &               &     \\
		&0             & -1            & 1             & 0               &  \\
		&      & 0             & -1            & 1             & \ddots \\
		&			&               & 0             & -1            & \ddots \\
     &      & 				       &              & \ddots        & \ddots
\end{pmatrix}, \;
R(t):= \begin{pmatrix}
\cos t  & \sin t & 0        & \cdots\\
-\sin t & \cos t & 0        &      \\
0       & 0      & 1        & \ddots\\
\vdots  &         & \ddots   & \ddots
\end{pmatrix},
\end{align*}
$\lambda(t) := 0$ für alle $t \in I$ und $P_0$ die orthogonale Projektion auf $\spn\{e_1\}$. Sei 
\begin{align*}
A(t) := R(t)^* A_0 R(t) \text{ \; und \; } P(t):= R(t)^* P_0 R(t)
\end{align*}
für alle $t \in I$. 

Dann ist $A_0 = A_1 \oplus A_2$ mit $A_1:= 0$ auf $\spn\{e_1\}$ und $A_2 := -1 + B$, $B$ der Shift nach rechts. Dann erzeugt $A_0$ nach Proposition~\ref{prop: eigenschaften der lambda(d)} und damit auch $A(t)$ für jedes $t \in I$ eine Kontraktionshalbgruppe auf $X$, das heißt, $\sigma(A(t)) \subset \{ z \in \complex: \Re z \le 0 \}$ und 
\begin{align*}
\norm{ \bigl(\lambda(t) + \varepsilon - A(t) \bigr)^{-1} } \le \frac{1}{\varepsilon}
\end{align*}
für alle $\varepsilon \in (0, \infty)$ und alle $t \in I$.
Weiter vertauscht $P_0$ mit $A_0$, $P_0 X \subset \ker(A_0 - \lambda(t))$ und $(1-P_0)X \subset \overline{\im (A_0 - \lambda(t))}$, und daher gelten die entsprechenden Aussagen für $A(t)$ und $P(t)$ statt $A_0$ und $P_0$. 
Schließlich ist offensichtlich $\rk P(t) = 1$ für alle $t \in I$. 

Also sind alle Voraussetzungen von Satz~\ref{thm: allg adsatz ohne sl} erfüllt. Die Aussage dieses Satzes folgt aber nicht aus den Adiabatensätzen mit Spektrallückenbedingung der vorigen Abschnitts, denn $\sigma(A(t))$ ist hier nach Beispiel~\ref{ex: spektrum der shifts} gleich $\overline{U}_1(-1)$, weshalb $\lambda(t)=0$ für kein $t \in I$ isoliert ist in $\sigma(A(t))$. Und auch die trivialen Adiabatensätze aus Abschnitt~\ref{sect: triviale adsätze} können wir nicht heranziehen.  $\blacktriangleleft$
\end{ex}

Im folgenden Beispiel ist $\lambda(t)$ zwar für jedes $t \in I$ isoliert in $\sigma(A(t))$, fällt aber an unendlich vielen Stellen in $\sigma(A(t)) \setminus \{ \lambda(t) \}$ hinein. Satz~\ref{thm: unhandl adsatz mit nichtglm sl} ist hier also sicher nicht anwendbar, weil es dazu nur endlich viele Überschneidungsstellen geben dürfte, Satz~\ref{thm: adsatz ohne sl für normale A(t)} aber schon. 

\begin{ex} \label{ex: adsatz ohne sl, unendlich viele überschneidungen}
Sei $X := \ell^2(I_2)$. Sei 
\begin{align*}
A_0(t):= \begin{pmatrix} \lambda(t)   & 0 \\
													0 					& 0
				 \end{pmatrix}, \;
R(t) := \begin{pmatrix}
\cos t  & \sin t    \\
-\sin t & \cos t      \\
\end{pmatrix},
\end{align*}
\begin{align*}
\lambda(t) := \begin{cases} 0, & t = 0 \\
														t^2 \bigl( \sin \frac{1}{t} - 1 \bigr), & t \in (0,1]
							\end{cases}
\end{align*}
und $P_0$ die orthogonale Projektion auf $\spn\{e_1\}$. Sei 
\begin{align*}
A(t) := R(t)^* A_0(t) R(t) \text{ \; und \; } P(t):= R(t)^* P_0 R(t)
\end{align*}
für alle $t \in I$.

Dann ist $A(t)$ für jedes $t \in I$ normal, $\lambda$ ist eine stetig differenzierbare Abbildung $I \to (-\infty, 0]$, weshalb 
\begin{align*}
\sigma(A(t)) = \{ \lambda(t), 0 \} \subset \{ z \in \complex: \Re z \le 0 \}
\end{align*}
und $t \mapsto A(t)$ stetig differenzierbar ist. 
Weiter gilt $\lambda(t) + \varepsilon e^{i \vartheta_0} \in \rho(A(t))$ und
\begin{align*}
\norm{ \bigl( \lambda(t) + \varepsilon e^{i \vartheta_0} - A(t) \bigr)^{-1} } &= \norm{ \bigl( \lambda(t) + \varepsilon e^{i \vartheta_0} - A_0(t) \bigr)^{-1} } 
= \norm{ \begin{pmatrix} \frac{1}{i \varepsilon }  & 0 \\ 0 & \frac{1}{\lambda(t) + i \varepsilon}  \end{pmatrix}   } \le \frac{1}{\varepsilon}			
\end{align*}
für alle $\varepsilon \in (0, \infty)$ und alle $t \in I$, wobei $\vartheta_0 := \frac{\pi}{2}$.
Schließlich ist $P(t)$ für jedes $t \in I$ eine orthogonale Projektion, $P(t)$ projiziert für alle $t \in I$ in $\ker(A(t)-\lambda(t))$ hinein und $P(t)$ projiziert für fast alle $t \in I$ auf ganz $\ker(A(t)-\lambda(t))$, weil $P_0$ entsprechend für alle $t \in I$ in $\ker (A_0(t) - \lambda(t))$ hineinprojiziert und für fast alle $t \in I$, nämlich für alle (abzählbar unendlich vielen!) $t \in I$ mit $\lambda(t) \ne 0$, auf ganz $\ker(A_0(t) - \lambda(t))$ projiziert. 

Wir sehen nun mithilfe von Proposition~\ref{prop: gemeinsamkeiten der bsp}, dass auch die übrigen Voraussetzungen von Satz~\ref{thm: adsatz ohne sl für normale A(t)} erfüllt sind. Und die Aussage dieses Satzes folgt nicht schon aus den trivialen Adiabatensätzen in Abschnitt~\ref{sect: triviale adsätze}. $\blacktriangleleft$
\end{ex}

Auch in den Adiabatensätzen ohne Spektrallückenbedingung (Satz~\ref{thm: allg adsatz ohne sl} und Satz~\ref{thm: adsatz ohne sl für normale A(t)}) ist die Voraussetzung der $(M,0)$-Stabilität wesentlich: wenn sie verletzt ist, braucht die Aussage dieser Sätze nicht zu gelten. Wir können, um dies einzusehen, Beispiel~\ref{ex: (M,0)-stabilität wesentlich in den adsätzen mit sl} übernehmen.

\begin{ex} \label{ex: (M,0)-stabilität wesentlich in adsatz ohne sl}
Sei $X := \ell^2(I_2)$. Sei $\lambda$ eine stetig differenzierbare Abbildung $I \to [0, \infty)$, die fast überall von $0$ verschieden ist, und seien $A(t)$, $P(t)$ wie in Beispiel~\ref{ex: (M,0)-stabilität wesentlich in den adsätzen mit sl} definiert.  

Wie man leicht sieht, sind dann alle Voraussetzungen von Satz~\ref{thm: adsatz ohne sl für normale A(t)} (und damit auch die von Satz~\ref{thm: allg adsatz ohne sl})  erfüllt mit der einzigen Ausnahme, dass $A$ hier nicht $(M,0)$-stabil ist, weil sonst $\sigma(A(t))$ in $\{ z \in \complex: \Re z \le 0 \}$ enthalten wäre für alle $t \in I$. $A$ ist hier nur $(1, \omega)$-stabil für $\omega := \sup_{t \in I} \lambda(t)$.

Dasselbe Argument wie in Beispiel~\ref{ex: (M,0)-stabilität wesentlich in den adsätzen mit sl} zeigt, dass die Aussage des Adiabatensatzes denn auch nicht erfüllt ist. $\blacktriangleleft$
\end{ex}

Das folgende Beispiel ist eine Übertragung von Beispiel~\ref{ex: A(t)= multop, mit sl} auf die Situation ohne Spektrallückenbedingung. 

\begin{ex} \label{ex: A(t)= multop, ohne sl}
Sei $X := L^2(\real, \complex)$. Sei
\begin{align*}
f_t := f_0(\,.\,+ t) \text{ \; und \; } f_0 \in C_c^1(\real, i \, \real) \text{ \; mit } f_0 \ne 0,
\end{align*}
$A(t) := M_{f_t}$ auf $X$, $\lambda(t) := 0$ für alle $t \in I$ und $P(t)$ eine orthogonale Projektion in $X$, die in $\ker(A(t) - \lambda(t))$ hinein projiziert für alle $t \in I$ und auf $\ker(A(t)-\lambda(t))$ projiziert für fast alle $t \in I$. Dann sind fast alle Voraussetzungen von Satz~\ref{thm: adsatz ohne sl für normale A(t)} erfüllt, nur ist hier sowohl $\rk P(0) = \infty$ als auch $\rk (1-P(0)) = \infty$ und $t \mapsto P(t)$ ist gemäß den Ausführungen nach Lemma~\ref{lm: multiplikationsop mit char fkt nur dann stark db nach t wenn schon konst} nicht stetig differenzierbar, da $t \mapsto P(t)$ nicht konstant ist: es gilt ja
\begin{align*}
\ker(A(t)-\lambda(t)) = \{ g \in X: (f_t - \lambda(t))g = 0 \} = \{g \in X: g = \chi_{ \{f_t = \lambda(t) \} } \, g \}
\end{align*} 
für alle $t \in I$ und damit 
\begin{align*}
P(t) g = \chi_{ \{f_t = \lambda(t) \} } \, g = \chi_{ \{f_t = 0\} } \, g
\end{align*}
für fast alle $t \in I$ und alle $g \in X$.

Und tatsächlich geht auch die Aussage des Satzes~\ref{thm: adsatz ohne sl für normale A(t)} schief, was man wie in Beispiel~\ref{ex: A(t)= multop, mit sl} sieht. $\blacktriangleleft$
\end{ex}

Was wir im Anschluss an Beispiel~\ref{ex: A(t)= multop, mit sl} bemerkt haben, gilt auch hier (in der entsprechenden Situation ohne Spektrallücke): auf der einen Seite können die Voraussetzungen von Satz~\ref{thm: adsatz ohne sl für normale A(t)} für $A(t) = M_{f_t}$ nur dann erfüllt sein, wenn $t \mapsto P(t)$ konstant ist. Sei nämlich $X := L^2(X_0,\complex)$ für einen Maßraum $(X_0, \mathcal{A}, \mu)$, sei $A(t) = M_{f_t}$ für messbare Abbildungen $f_t$ mit $\Re f_t \le 0$ fast überall, sei $\lambda(t)$ ein Eigenwert von $A(t)$ und sei $P(t)$ eine orthogonale Projektion in $X$, sodass $P(t)X = \ker(A(t)-\lambda(t))$ für fast alle $t \in I$.
Wie im Beispiel oben erhalten wir, dass dann 
\begin{align*}
P(t) g = \chi_{ \{f_t = \lambda(t) \} } \, g 
\end{align*}
für fast alle $t \in I$ und für alle $g \in X$, und damit aufgrund der Ausführungen nach Lemma~\ref{lm: multiplikationsop mit char fkt nur dann stark db nach t wenn schon konst}, dass die Abbildung $t \mapsto P(t)$ nur dann stark stetig differenzierbar sein kann, wenn sie schon konstant ist.

Auf der anderen Seite scheint aber auch die Aussage von Satz~\ref{thm: adsatz ohne sl für normale A(t)} -- wenn überhaupt -- nur sehr selten nichttrivialerweise erfüllt zu sein. Zumindest für beschränkte schiefselbstadjungierte $A(t)$, die stark stetig von $t$ abhängen, ist die Aussage dieses Satzes sicher nicht erfüllt. Dies folgt mit demselben Argument wie in Beispiel~\ref{ex: A(t)= multop, mit sl}.
\\

Wie in Beispiel~\ref{ex: reg von P wesentlich im adsatz mit nichtglm sl} zeigen wir nun, dass die Aussage von Satz~\ref{thm: adsatz ohne sl für normale A(t)} nicht zu gelten braucht, wenn die in diesem Satz getroffene Voraussetzung, dass $t \mapsto P(t)$ stetig differenzierbar ist, als einzige verletzt ist. 

\begin{ex} \label{ex: reg von P wesentlich in adsatz ohne sl}
Seien $A(t)$ und $\lambda(t)$ wie in Beispiel~\ref{ex: reg von P wesentlich im adsatz mit nichtglm sl}.
Dann ist $A(t)$ für jedes $t \in I$ schiefselbstadjungiert, $t \mapsto A(t)$ ist analytisch (insbesondere einmal stark stetig differenzierbar), $\lambda(t)$ ist für jedes $t \in I$ ein Eigenwert von $A(t)$, sodass $\lambda(t) + \varepsilon \in \rho(A(t))$ und
\begin{align*}
\norm{ \bigl(\lambda(t) + \varepsilon - A(t) \bigr)^{-1} } \le \frac{1}{\varepsilon}
\end{align*}
für alle $\varepsilon \in (0, \infty)$ und alle $t \in I$ ($A(t)$ erzeugt ja eine Kontraktionshalbgruppe).
Aber für die orthogonalen Projektionen $P_0(t)$ von $A(t)$ auf $\ker (A(t)-\lambda(t))$ gilt
\begin{align*}
P_0(t) = \begin{cases} P_1, & t \in [0, \frac{1}{2}) \\
										 1, & t = \frac{1}{2} \\
										 P_2, & t \in (\frac{1}{2}, 1]
			 \end{cases},
\end{align*}
wobei $P_1$, $P_2$ die orthogonale Projektion auf $\spn\{e_1\}$ bzw. $\spn\{e_2\}$ bezeichnet. Und daraus folgt, dass keine orthogonalen Projektionen $P(t)$ existieren, die stetig von $t$ abhängen und bis auf eine Nullmenge mit $P_0(t)$ übereinstimmen.

Also sind zwar die Voraussetzungen von Satz~\ref{thm: adsatz ohne sl für normale A(t)} an $A$ und $\sigma$ erfüllt aber nich die an $P$. 
Wie in Beispiel~\ref{ex: reg von P wesentlich im adsatz mit nichtglm sl} folgt nun, dass auch die Aussage des Adiabatensatzes hier nicht erfüllt ist. $\blacktriangleleft$
\end{ex}

Wir schließen mit einem Beispiel, das einen über Satz~\ref{thm: allg adsatz ohne sl} hinausgehenden Adiabatensatz wünschenswert erscheinen lässt. Dieses Beispiel ist eine überaus natürliche und einfache Verbindung 
unserer bisherigen Beispiele zu Spektralwertüberschneidungen (Beispiel~\ref{ex: adsatz mit nichtglm sl, endl viele überschneidungen} und Beispiel~\ref{ex: adsatz ohne sl, unendlich viele überschneidungen}) und leider wissen wir nicht, ob in diesem Beispiel die Aussage des Adiabatensatzes gilt oder nicht. Angesichts der beiden eben genannten Beispiele, in denen die Aussage des Adiabatensatzes jeweils erfüllt ist, vermuten wir aber, dass sie auch hier erfüllt ist.

\begin{ex} \label{ex: motivierendes bsp für erweiterten adsatz ohne sl}
Sei $X := \ell^2(I_3)$. Sei $\lambda$ eine stetig differenzierbare Abbildung $I \to (-\infty, \lambda_2]$, sodass $\lambda(t) = \lambda_2$ für unendlich viele $t \in I$ und $\lambda(t) \ne \lambda_2$ für fast alle $t \in I$, sei
\begin{align*}
A_0(t) := \begin{pmatrix}
\lambda(t)  & 1           & 0 \\
0           & \lambda(t)  & 0  \\
0           & 0           & \lambda_2
\end{pmatrix}, \;
R(t):= \begin{pmatrix}
1  & 0		 	& 0 \\
0  & \cos t & \sin t \\
0  & -\sin t  & \cos t
\end{pmatrix}
\end{align*}
und sei $P_0$ die orhtogonale Projektion auf $\spn\{e_1, e_2\}$. Sei 
\begin{align*}
A(t) := R(t)^* A_0(t) R(t) \text{ \; und \; } P(t):= R(t)^* P_0 R(t)
\end{align*}
für alle $t \in I$.

Dann ist Satz~\ref{thm: allg adsatz ohne sl} nicht anwendbar, weil für alle $t \in I$ gilt, dass $R(t)^* e_2 \notin \ker(A(t)-\lambda(t))$ und $R(t)^* e_2 \notin \im (A(t)-\lambda(t))$ und damit $\ker(A(t)-\lambda(t)) +  \im (A(t)-\lambda(t)) \ne X$. 
Wir zeigen nun, dass die Voraussetzungen von Satz~\ref{thm: erweiterter adsatz ohne sl} hingegen erfüllt sind.
Zunächst ist $A$ $(1,0)$-stabil, weil $\lambda(t) \le \lambda_2$ für alle $t \in I$, $t \mapsto A(t)$ ist stetig differenzierbar, $\lambda(t) + \varepsilon e^{i \vartheta_0} \in \rho(A(t))$ für alle $\varepsilon \in (0, \infty)$ und alle $t \in I$ (wobei $\vartheta_0 := \frac{\pi}{2}$) und $t \mapsto \lambda(t)$ ist stetig differenzierbar. 
Weiter vertauscht $P_0$ mit $A_0(t)$ für alle $t \in I$ und 
\begin{align*}
P_0 X = \spn\{e_1, e_2\} \subset \ker (A_0(t)-\lambda(t))^2
\end{align*}
für alle $t \in I$ und
\begin{align*}
(1-P_0)X = \spn\{ e_3 \} \subset \im (A_0(t)-\lambda(t))^2
\end{align*}
für fast alle $t \in I$ (nämlich für alle $t \in I$ mit $\lambda(t) \ne \lambda_2$), woraus wir erhalten, dass $P(t)$ mit $A(t)$ vertauscht für alle $t \in I$, dass $P(t)$ für alle $t \in I$ in $\ker (A_0(t)-\lambda(t))^2$ hineinprojiziert und $1-P(t)$ für fast alle $t \in I$ in $\im (A_0(t)-\lambda(t))^2$ hineinprojiziert.
Schließlich ist $\rk P(0) = 2$, $t \mapsto P(t)$ ist zweimal stetig differenzierbar und wegen $A_0(t)\big|_{(1-P_0)X} = \lambda_2$ haben wir die Abschätzung
\begin{align*}
&\norm{  \bigl( \lambda(t) + \varepsilon e^{i \vartheta_0} - A(t) \bigr)^{-1} \, (1-P(t)) }
= \norm{  \bigl( \lambda(t) + \varepsilon e^{i \vartheta_0} - A_0(t) \bigr)^{-1} \, (1-P_0) } \\
&\qquad \qquad \le \Big| \frac{1}{ \lambda(t) + i \varepsilon - \lambda_2 } \Big| \, \norm{1-P_0} \le \frac{1}{\varepsilon} 
\end{align*}
für alle $\varepsilon \in (0, \infty)$ und alle $t \in I$. 

Also sind tatsächlich alle Voraussetzungen von Satz~\ref{thm: erweiterter adsatz ohne sl} erfüllt. Um allerdings mithilfe dieses Satzes auf die Aussage des Adiabatensatzes schließen zu können, bräuchten wir einen handlichen Ausdruck für die Zeitentwicklung zu $T (A-\lambda)$, den es nicht zu geben scheint. Jedenfalls vertauschen die $A(t)$ nicht paarweise, weil $A(t)-\lambda(t)$ nur für diejenigen $t \in (0,1]$ mit $A(0) - \lambda(0)$ vertauscht, für die $\lambda(t) = \lambda_2$. Wir können uns also wenigstens nicht auf Korollar~\ref{cor: Dyson} berufen. 
Und auch auf die trivialen Adiabatensätze, Satz~\ref{thm: triv adsatz 1} und Satz~\ref{thm: triv adsatz 2} in der Version mit $M = 1$, können wir nicht zurückgreifen. $\blacktriangleleft$
\end{ex}


\section{Adiabatensätze höherer Ordnung} \label{sect: höhere adsätze}

In diesem Abschnitt entwickeln wir Satz~\ref{thm: unhandl adsatz mit sl} weiter zu Satz~\ref{thm: höherer adsatz} und bekommen auf diese Weise auch einen Adiabatensatz höherer Ordnung (Korollar~\ref{cor: höherer adsatz für suppP' in (0,1)}), der unter geeigneten Voraussetzungen höhere Konvergenzordnungen als die bisher besprochenen Adiabatensätze liefert. Wir folgen dabei der Arbeit~\cite{Nenciu 93} Nencius, die wir ein wenig verallgemeinern. Das Vorgehen Nencius unterscheidet sich deutlich von unserer bisherigen Vorgehensweise: bisher haben wir die Zeitentwicklung $U_{ \frac{1}{\varepsilon} }$ an die Projektionen $P(t)$ angepasst, indem wir übergegangen sind zu einer Zeitentwicklung $V_{ \frac{1}{\varepsilon} }$ (der adiabatischen Zeitentwicklung zu $A$ und $P$), die adiabatisch i. e. S. ist bzgl. $P$ und $1-P$ und die (unter geeigneten Voraussetzungen) die eigentlich interessierende Zeitentwicklung $U_{ \frac{1}{\varepsilon} }$ für $\varepsilon \searrow 0$ gut approximiert. Jetzt passen wir nach Nenciu die Projektionen $P(t)$ an die Zeitentwicklung $U_{ \frac{1}{\varepsilon} }$ an, und zwar gehen wir über zu Projektionen $P_{ \varepsilon }(t)$, sodass
\begin{align*}
\sup_{t \in I} \norm{ P_{ \varepsilon }(t) - P(t) } = O(\varepsilon) \quad (\varepsilon \searrow 0)
\end{align*}
und 
\begin{align*}
\sup_{(s,t) \in \Delta} \norm{ P_{\varepsilon}(t) U_{ \frac{1}{\varepsilon} }(t,s) - U_{ \frac{1}{\varepsilon} }(t,s) P_{\varepsilon}(s) } = O(\varepsilon^{m-1}) \text{ bzw. } O\bigl( e^{-\frac{c}{\varepsilon}} \bigr) \quad (\varepsilon \searrow 0).
\end{align*}
\\

Wir führen zunächst eine abkürzende Sprechweise ein, in der es um verschiedene Stufen der Regularität geht. Diese verschiedenen Stufen kennzeichnen wir mit dem Symbol $m$ (das für eine natürliche Zahl oder $\infty$ steht) bzw. $\omega$ (das an Analytizität erinnern soll). 
\\

Sei $A(t)$ für jedes $t \in I$ eine abgeschlossene lineare Abbildung $D \subset X \to X$, $\sigma(t)$ eine kompakte in $\sigma(A(t))$ isolierte Untermenge von $\sigma(A(t))$ und $P(t)$ die Rieszprojektion von $A(t)$ auf $\sigma(t)$. 
Wir nennen dann $A$, $\sigma$, $P$ (zusammengenommen) \emph{$m$-regulär} für ein $m \in \natu \cup \{ \infty \}$ genau dann, wenn gilt:
\begin{itemize} 
\item [(i)] $A(t)$ erzeugt für jedes $t \in I$ eine stark stetige Halbgruppe auf $X$ und $A$ ist $(M,0)$-stabil, $t \mapsto A(t)x$ ist stetig differenzierbar für alle $x \in D$ und $t \mapsto (A(t)-1)^{-1}x$ ist $(m-1)$-mal stetig differenzierbar für alle $x \in X$
\item [(ii)] zu jedem $t_0 \in I$ existiert ein Zykel $\gamma_{t_0}$ und eine in $I$ offene Umgebung $U_{t_0}$ von $t_0$, sodass $\im \gamma_{t_0} \subset \rho(A(t))$ und $n( \gamma_{t_0}, \sigma(t)) = 1$ und $n( \gamma_{t_0}, \sigma(A(t)) \setminus \sigma(t)) = 0$ für alle $t \in U_{t_0}$
\item [(iii)] $t \mapsto P(t)x$ ist $m$-mal stetig differenzierbar für alle $x \in X$.
\end{itemize}

Wir nennen $A$, $\sigma$, $P$ (zusammengenommen) \emph{$\omega$-regulär} genau dann, wenn $A$, $\sigma$, $P$ $\infty$-regulär sind und zusätzlich eine Zahl $c$ existiert und zu jedem $t_0 \in I$ eine in $I$ offene Umgebung $V_{t_0}$ von $t_0$ mit $V_{t_0} \subset U_{t_0}$, sodass
\begin{align*}
\sup_{(t,z) \in V_{t_0} \times \, \im \gamma_{t_0} } \norm{ \ddtk{ (A(t)-z)^{-1} } }   \le c^k \, k!
\end{align*}
für alle $k \in \natu$.
\\

Zu dieser Vereinbarung ist einiges anzumerken. Zunächst haben wir, wenn $A$, $\sigma$, $P$ $m$-regulär sind für ein $m \in \natu$, dass zu jedem $t_0 \in I$ eine in $I$ offene Umgebung $V_{t_0}$ von $t_0$ existiert mit $V_{t_0} \subset U_{t_0}$ und
\begin{align*}
\sup_{(t,z) \in V_{t_0} \times \, \im \gamma_{t_0} } \norm{  (A(t)-z)^{-1}  }   < \infty.
\end{align*}
Das folgt mithilfe von Lemma~\ref{lm: A stetig im verallg sinn} und Satz~\ref{thm: (A(t)-z)^{-1} stetig in (t,z)} (genauso wie im Beweis von Satz~\ref{thm: unhandl adsatz mit sl}).

Aber es gilt noch mehr (s. den Beweis des folgenden Lemmas): die Abbildung $V_{t_0} \ni t \mapsto (A(t)-z)^{-1}x$ ist $(m-1)$-mal stetig differenzierbar für alle $z \in \im \gamma_{t_0}$ und alle $x \in X$, die Abbildung $\im \gamma_{t_0} \ni z \mapsto  \ddtk{ (A(t)-z)^{-1} }$ ist stetig für alle $t \in V_{t_0}$ und alle $k \in \{ 1, \dots, m-1 \}$ bzw. alle $k \in \natu$ und 
\begin{align*}
\sup_{(t,z) \in V_{t_0} \times \, \im \gamma_{t_0} } \norm{ \ddtk{ (A(t)-z)^{-1} } }  < \infty
\end{align*}
für alle $k \in \{ 1, \dots, m-1 \}$ bzw. alle $k \in \natu$. Insbesondere existieren die in der Definition von $\omega$-Regularität vorkommenden Ableitungen überhaupt.  

Schließlich ist $\sigma(t)$ -- dies entnimmt man dem Beweis von Korollar~\ref{cor: sl glm unter den unhandl vor} -- automatisch gleichmäßig isoliert in $\sigma(A(t))$, wenn $m$-Regularität vorliegt.
\\

Die Zykelbedingung ist (nach dem Beweis von Satz~\ref{thm: handl adsatz mit sl}) beispielsweise dann erfüllt, wenn $t \mapsto \sigma(t)$ stetig ist. Die Bedingung, dass die Abbildung $t \mapsto (A(t)-1)^{-1}x$ $(m-1)$-mal stetig differenzierbar ist für alle $x \in X$, ist etwa dann erfüllt, wenn $t \mapsto A(t)x$  $(m-1)$-mal stetig differenzierbar ist für alle $x \in D$ (Lemma~\ref{lm: reg of inv}).

\begin{lm} \label{lm: reg der En}
Sei $A(t)$ für jedes $t \in I$ eine abgeschlossene lineare Abbildung $D \subset X \to X$, $\sigma(t)$ eine kompakte in $\sigma(A(t))$ isolierte Untermenge von $\sigma(A(t))$ und $P(t)$ die Rieszprojektion von $A(t)$ auf $\sigma(t)$, sodass $A$, $\sigma$, $P$ $m$-regulär sind für ein $m \in \natu$. Sei weiter 
\begin{align*}
E_0(t) := P(t)
\end{align*}
und 
\begin{align*}
E_k(t) &:= \frac{1}{2 \pi i} \, \int_{\gamma_t} (A(t)-z)^{-1} \, \Bigl( P(t) E_{k-1}'(t) \overline{P}(t) - \overline{P}(t) E_{k-1}'(t) P(t) \Bigr) (A(t)-z)^{-1} \,dz \\
& \quad \qquad \qquad \qquad + S_k(t) - 2 P(t) S_k(t) P(t)
\end{align*}
für alle $k \in \{1, \dots, m \}$ und alle $t \in I$, wobei $\overline{P}(t) := 1-P(t)$ und
\begin{align*}
S_k(t) := \sum_{l=1}^{k-1} E_{k-l}(t) E_l(t).
\end{align*}
Dann ist $t \mapsto E_k(t)x$ $(m-k)$-mal stetig differenzierbar für alle $x \in X$ und alle $k \in \{ 0, 1, \dots, m \}$ und es gilt $E_k(t)X \subset D$ für alle $t \in I$ und alle $k \in \{ 0, 1, \dots, m \}$.
\end{lm}

\begin{proof}
Wir beginnen mit ein paar kleinen (oben schon angedeuteten) Vorbereitungen, die es uns erlauben werden mithilfe von Lemma~\ref{lm: vertauschung von abl und wegintegral} die behaupteten Aussagen (mitsamt der Wohldefiniertheit der $E_k(t)$) sehr leicht einzusehen.

Zunächst existiert, wie oben bemerkt, zu jedem $t_0 \in I$ eine in $I$ offene Umgebung $V_{t_0}$ von $t_0$ mit $V_{t_0} \subset U_{t_0}$ und
\begin{align*}
\sup_{(t,z) \in V_{t_0} \times \, \im \gamma_{t_0} } \norm{  (A(t)-z)^{-1}  }   < \infty.
\end{align*}
Insbesondere gilt 
\begin{align*}
&\sup_{(t,z) \in V_{t_0} \times \, \im \gamma_{t_0} } \norm{  \Bigl( 1- (z-1)(A(t)-1)^{-1} \Bigr)^{-1}  } \\
& \qquad \qquad =  \sup_{(t,z) \in V_{t_0} \times \, \im \gamma_{t_0} } \norm{ (A(t)-1) (A(t)-z)^{-1}  }  < \infty
\end{align*}
und daher ist 
\begin{align*}
V_{t_0} \ni t \mapsto (A(t)-z)^{-1} = (A(t)-1)^{-1} \, \Bigl( 1- (z-1)(A(t)-1)^{-1} \Bigr)^{-1}x
\end{align*}
nach Lemma~\ref{lm: reg of inv} $(m-1)$-mal stetig differenzierbar für alle $x \in X$ und alle $z \in \im \gamma_{t_0}$.

Weiter ist 
\begin{align*}
\rho(A(t)) \ni z \mapsto \Bigl( 1- (z-1)(A(t)-1)^{-1} \Bigr)^{-1} = (A(t)-1) (A(t)-z)^{-1}
\end{align*}
holomorph für alle $t \in I$ und daher ist auch $\rho(A(t)) \ni z \mapsto \ddtk{  \bigl( 1- (z-1)(A(t)-1)^{-1} \bigr)^{-1}  }$ holomorph für alle $k \in \{1, \dots, m-1 \}$, denn diese Ableitung besteht aus Summanden, die sich nach dem Beweis von Lemma~\ref{lm: reg of inv} (abgesehen von skalaren Vorfaktoren) zusammensetzen aus $\ddtl{ (A(t)-1)^{-1} }$ und $\bigl( 1- (z-1)(A(t)-1)^{-1} \bigr)^{-1}$ selbst.
Wir sehen nun, dass auch 
\begin{align*}
\rho(A(t)) \ni z \mapsto \ddtk{  (A(t)-z)^{-1}  }
\end{align*}
holomorph und insbesondere stetig ist für alle $t \in I$.

Schließlich sehen wir anhand der eben beschriebenen Zusammensetzung von $\ddtk{  \bigl( 1- (z-1)(A(t)-1)^{-1} \bigr)^{-1}  }$, dass
\begin{align*}
\sup_{(t,z) \in V_{t_0} \times \, \im \gamma_{t_0} } \norm{  \ddtk{ (A(t)-z)^{-1} }  }   < \infty
\end{align*}
für alle $k \in \{ 1, \dots, m-1 \}$.
\\

Aufgrund dieser Vorbereitungen ist es nun leicht mit (endlicher) Induktion über $k \in \{0, 1, \dots, m \}$ zu zeigen, dass die rekursive Definition der $E_k(t)$ überhaupt sinnvoll ist und dass $t \mapsto E_k(t)x$ $(m-k)$-mal stetig differenzierbar ist für alle $x \in X$ und alle $k \in \{ 0, 1, \dots, m \}$, wie behauptet. Wir müssen uns nur an Lemma~\ref{lm: vertauschung von abl und wegintegral} erinnern.
\\

Zuletzt: da  
\begin{align*}
\rho(A(t)) \ni z \mapsto A(t)\, \Bigl( (A(t)-z)^{-1} \, \Bigl( P(t) E_{k-1}'(t) \overline{P}(t) - \overline{P}(t) E_{k-1}'(t) P(t) \Bigr) (A(t)-z)^{-1}  \Bigr)
\end{align*}
stetig (sogar holomorph) ist und damit insbesondere das zugehörige Wegintegral existiert, folgt wegen der Abgeschlossenheit von $A(t)$, dass 
\begin{align*}
\biggl(    \frac{1}{2 \pi i} \, \int_{\gamma_t} (A(t)-z)^{-1} \, \Bigl( P(t) E_{k-1}'(t) \overline{P}(t) - \overline{P}(t) E_{k-1}'(t) P(t) \Bigr) (A(t)-z)^{-1} \,dz \biggr) \, X \subset D
\end{align*}
für alle $t \in I$. Induktiv folgt nun, dass $E_k(t)X \subset D$ für alle $t \in I$ und alle $k \in \{0,1, \dots, m\}$, und wir sind fertig.
\end{proof}

Das folgende Lemma (Lemma~1 in~\cite{Nenciu 93}) ist der entscheidende Schritt hin zu Satz~\ref{thm: höherer adsatz}. Wie man darauf kommt ausgerechnet Operatoren $E_k(t)$ mit den Eigenschaften in diesem Lemma zu suchen, verrät Abschnitt~1 in~\cite{Nenciu 93}.

\begin{lm} \label{lm: konstr der En}
Seien $A(t)$, $\sigma(t)$ und $P(t)$ wie in Lemma~\ref{lm: reg der En} und auch die $E_k(t)$ seien wie in diesem Lemma definiert. Dann gilt
\begin{align*}
E_k(t) = \sum_{l=0}^k E_{k-l}(t)E_l(t)
\end{align*} 
für alle $k \in \{0,1, \dots, m \}$ und
\begin{align*}
E_k'(t) \supset [A(t),E_{k+1}(t)]
\end{align*}
für alle $k \in \{0,1, \dots, m-1 \}$. Weiter sind die $E_k(t)$ dadurch und durch $E_0(t)=P(t)$ eindeutig bestimmt.
\end{lm}

\begin{proof}
Wir können die Argumentation Nencius übernehmen.
\end{proof}

Seien $A(t)$, $\sigma(t)$ und $P(t)$ wie im obigen Lemma und seien $A$, $\sigma$, $P$ $m$-regulär für ein $m \in \natu$ bzw. $\omega$-regulär. Wir setzen dann 
\begin{align*}
T_{\varepsilon}(t) := \sum_{k=0}^{m_{\varepsilon}-1} E_k(t) \varepsilon^k
\end{align*}
für alle $\varepsilon \in (0, \infty)$ und alle $t \in I$, wobei
\begin{align*}
m_{\varepsilon} := \begin{cases} m, &\, \text{wenn } A, \sigma, P \text{ $m$-regulär} \\
																\big\lfloor \frac{1}{g \varepsilon} \big\rfloor, &\, \text{wenn } A, \sigma, P \text{ $\omega$-regulär}  \end{cases}
\end{align*}
und $g$ eine positive Zahl sei, sodass
\begin{align*}
\norm{E_k(t)} \le g^k \, k!   \text{ \; und \; } \norm{E_k'(t)} \le g^{k+1} \, (k+1)!
\end{align*}
für alle $k \in \natu$ und alle $t \in I$. So eine Zahl existiert im Fall von $\omega$-Regularität wirklich, und zwar nach Lemma~4 in~\cite{Nenciu 93} (dessen Aussage auch in unserer leicht abgewandelten Situation gilt, wie eine sorgfältige Analyse des Beweises dieses Lemmas und des Lemmas~3 zeigt).
\\

Das folgende Lemma gibt Lemma~5 aus~\cite{Nenciu 93} wieder.

\begin{lm} \label{lm: T geht gegen P}
Seien $A(t)$, $\sigma(t)$ und $P(t)$ wie in Lemma~\ref{lm: reg der En} und seien $A$, $\sigma$, $P$ $m$-regulär für ein $m \in \natu$ oder $\omega$-regulär. Dann gilt
\begin{align*}
\sup_{t \in I} \norm{ T_{\varepsilon}(t) - P(t) } = O(\varepsilon) \quad (\varepsilon \searrow 0).
\end{align*}
\end{lm}

\begin{proof}
Im Fall von $m$-Regularität folgt das sofort aus 
\begin{align*}
\sup_{t \in I} \norm{E_k(t)} < \infty
\end{align*}
für alle $k \in \{0,1, \dots, m\}$, was sich aus der starken Stetigkeit von $t \mapsto E_k(t)$ (Lemma~\ref{lm: reg der En}) ergibt.

Im Fall von $\omega$-Regularität folgt das mithilfe der Abschätzung
\begin{align*}
\norm{E_k(t)} \le g^k \, k!
\end{align*}
für alle $k \in \natu$ und des stirlingschen Satzes (s. etwa Beispiel~6.13 in~\cite{AmannEscher}). Dieser sagt, dass
\begin{align*}
\frac{k!}{ (2 \pi k)^{\frac{1}{2}} \, k^k \, e^{-k} } \longrightarrow 1 \quad (k \to \infty)
\end{align*}
und es folgt, dass eine Zahl $c_0 \in (1, \infty)$ existiert, sodass
\begin{align*}
k! \le c_0 \, k^{\frac{1}{2}} \, e^{k \, \log k} \, e^{-k}
\end{align*}
für alle $k \in \natu \cup \{0\}$. Sei nun $\beta \in (0,1)$. Dann haben wir aufgrund der Abschätzung von $\norm{E_k(t)}$, der eben angeführten Abschätzung für $k!$ und der Wahl von $m_{\varepsilon}$, dass
\begin{align*}
\norm{ T_{\varepsilon}(t) - P(t) } &\le g \varepsilon \, \sum_{k=1}^{m_{\varepsilon}-1} (g \varepsilon)^{k-1} \, (k-1)! \, k \\
& \le c_0' \, g \varepsilon \, \sum_{k=1}^{m_{\varepsilon}-1} \bigl( k^{\frac{3}{2}} e^{-(1-\beta) k} \bigr) \, e^{-\beta k}
\le c_0'' \, g \varepsilon \, \sum_{k=1}^{\infty} e^{-\beta k} 
\end{align*}
für alle $\varepsilon \in (0, \infty)$ und alle $t \in I$, woraus die Behauptung folgt.
\end{proof}

Wir sehen anhand des obigen Lemmas und anhand des Beweises von Proposition~\ref{prop: sigma(A(t)) oberhstet}, dass ein $\varepsilon_0 \in (0, \infty)$ existiert, sodass $\sigma(T_{\varepsilon}(t)) \subset U_{\frac{1}{3}}\bigl( \sigma(P(t)) \bigr) \subset U_{\frac{1}{3}}( \{0,1\})$ für alle $\varepsilon \in (0, \varepsilon_0]$ und alle $t \in I$.
Also können wir setzen
\begin{align*}
P_{\varepsilon}(t) := \frac{1}{2 \pi i} \, \int_{ \partial U_{\frac{1}{2}}(1) } (z-T_{\varepsilon}(t))^{-1} \,dz
\end{align*}
für alle $\varepsilon \in (0, \varepsilon_0]$ und alle $t \in I$.
\\

Jetzt können wir durch geringfügige Abwandlung der Argumente Nencius eine verallgemeinerte Version von Theorem~2 in~\cite{Nenciu 93} beweisen.
Wir knüpfen damit an die Schlussbemerkung in Nencius Artikel an, in der die Verallgemeinerungsfähigkeit von Theorem~2 schon angedeutet wird: für nichtschiefselbstadjungierte $A(t)$ sei es notwendig geeignete Schranken an die (bei uns) mit $V_{\frac{1}{\varepsilon}}$ bezeichnete Zeitentwicklung zu finden, und in dem Artikel~\cite{NenciuRasche 92} wird auch eine spezielle Situation behandelt, in der das möglich ist. Allerdings werden dort (wie auch in~\cite{Nenciu 93} selbst) keine allgemeinen Voraussetzungen formuliert, unter denen so eine Abschätzung erzielt werden kann. Satz~\ref{thm: höherer adsatz} gibt mit der $(M,0)$-Stabilität von $A$ so eine Voraussetzung an.

\begin{thm} \label{thm: höherer adsatz}
Seien $A(t)$, $\sigma(t)$ und $P(t)$ wie in Lemma~\ref{lm: reg der En} und seien $A$, $\sigma$, $P$ $m$-regulär für ein $m \in \natu$ bzw. $\omega$-regulär. Seien weiter die $P_{\varepsilon}(t)$ wie oben definiert. Dann gilt: \\
(i)
\begin{align*}
\sup_{t \in I} \norm{ P_{\varepsilon}(t) - P(t)} = O(\varepsilon) \quad (\varepsilon \searrow 0)
\end{align*} 
(ii) Es gibt ein $\varepsilon_0' \in (0,\varepsilon_0]$, sodass $t \mapsto P_{\varepsilon}(t)x$ stetig differenzierbar ist für alle $x \in X$ und $[A(t), P_{\varepsilon}(t)]$ fortsetzbar ist zu einer beschränkten linearen Abbildung auf $X$ für alle $\varepsilon \in (0, \varepsilon_0']$, sodass für alle $c \in \bigl( 0,\frac{1}{g} \bigr)$ gilt: 
\begin{align*}
\sup_{t \in I} \norm{ P_{\varepsilon}'(t) - \frac{1}{\varepsilon} \, [A(t), P_{\varepsilon}(t)] } = O(\varepsilon^{m-1}) \text{ bzw. } O\bigl( e^{-\frac{c}{\varepsilon}} \bigr) \quad (\varepsilon \searrow 0).
\end{align*}
(iii) 
\begin{align*}
&\sup_{t \in I} \norm{  (1-P_{\varepsilon}(t)) U_{\frac{1}{\varepsilon}}(t) P_{\varepsilon}(0)  }, \\
& \qquad \qquad \sup_{t \in I} \norm{  P_{\varepsilon}(t) U_{\frac{1}{\varepsilon}}(t) (1-P_{\varepsilon}(0))  }  = O(\varepsilon^{m-1}) \text{ bzw. } O\bigl( e^{-\frac{c}{\varepsilon}} \bigr) \quad (\varepsilon \searrow 0)
\end{align*}
und wenn die Zeitentwicklung $V_{\frac{1}{\varepsilon}}$ zu $\frac{1}{\varepsilon} A + (1-2P_{\varepsilon}) \bigl(  P_{\varepsilon}' - \frac{1}{\varepsilon} \, [A, P_{\varepsilon}]  \bigr)$ für alle $\varepsilon \in (0,\varepsilon_0']$ existiert, dann ist $V_{\frac{1}{\varepsilon}}$ adiabatisch i. e. S. bzgl. $P_{\varepsilon}$ und $1-P_{\varepsilon}$ und es gilt
\begin{align*}
\sup_{t \in I} \norm{   V_{\frac{1}{\varepsilon}}(t) - U_{\frac{1}{\varepsilon}}(t)  } = O(\varepsilon^{m-1}) \text{ bzw. } O\bigl( e^{-\frac{c}{\varepsilon}} \bigr) \quad (\varepsilon \searrow 0).
\end{align*}
\end{thm}

\begin{proof}
Wir zeigen zunächst, dass ein $\varepsilon_0' \in (0,\varepsilon_0]$ existiert und eine Zahl $M_0$, sodass
\begin{align*}
\norm{ (z-T_{\varepsilon}(t))^{-1} } \le M_0
\end{align*}
für alle $z \in \partial U_{\frac{1}{2}}(1)$, alle $\varepsilon \in (0,\varepsilon_0']$ und alle $t \in I$.

Wir haben erstens, dass
\begin{align*}
(z-T_{\varepsilon}(t))^{-1} \, \Bigl( 1- \bigl( T_{\varepsilon}(t) - P(t) \bigr) (z-P(t))^{-1} \Bigr) = (z-P(t))^{-1}
\end{align*}
für alle $z \in \partial U_{\frac{1}{2}}(1)$, alle $\varepsilon \in (0,\varepsilon_0]$ und alle $t \in I$. Weil nun  $\sup_{t \in I} \norm{ T_{\varepsilon}(t) - P(t) } \longrightarrow 0 \;\;(\varepsilon \searrow 0)$ nach Lemma~\ref{lm: T geht gegen P}, existiert zweitens ein $\varepsilon_0' \in (0, \varepsilon_0]$, sodass 
\begin{align*}
\norm{  \bigl( T_{\varepsilon}(t) - P(t) \bigr) (z-P(t))^{-1}  } \le \frac{1}{2}
\end{align*} 
für alle $z \in \partial U_{\frac{1}{2}}(1)$, alle $\varepsilon \in (0,\varepsilon_0']$ und alle $t \in I$. 
Sei nun 
\begin{align*}
M_0 :=  \; 2 \, \sup_{(t,z) \in I \times \partial U_{\frac{1}{2}}(1) } \norm{  (z-P(t))^{-1}  },
\end{align*}
was eine reelle Zahl ist. Dann folgt wie gewünscht
\begin{align*}
\norm{ (z- T_{\varepsilon}(t))^{-1} } \le M_0
\end{align*}
für alle $z \in \partial U_{\frac{1}{2}}(1)$, alle $\varepsilon \in (0,\varepsilon_0']$ und alle $t \in I$ (neumannsche Reihe!).
\\

(i) Aus $(z-P(t))^{-1} = \frac{1}{z-1} \, P(t) + \frac{1}{z} \, (1-P(t))$ folgt, dass
\begin{align*}
P(t) = \frac{1}{2 \pi i} \, \int_{\partial U_{\frac{1}{2}}(1)} (z-P(t))^{-1} \, dz
\end{align*}
und damit
\begin{align*}
P_{\varepsilon}(t) - P(t) &= \frac{1}{2 \pi i} \, \int_{\partial U_{\frac{1}{2}}(1)} (z-T_{\varepsilon}(t))^{-1} - (z-P(t))^{-1} \, dz \\
&= \frac{1}{2 \pi i} \, \int_{\partial U_{\frac{1}{2}}(1)} (z-T_{\varepsilon}(t))^{-1} \, \bigl( T_{\varepsilon}(t)-P(t) \bigr) \, (z-P(t))^{-1} \, dz
\end{align*}
für alle $\varepsilon \in (0,\varepsilon_0']$ und alle $t \in I$. Aus der eingangs bewiesenen Abschätzung und Lemma~\ref{lm: T geht gegen P} ergibt sich nun 
\begin{align*}
\sup_{t \in I} \norm{ P_{\varepsilon}(t) - P(t)} = O(\varepsilon) \quad (\varepsilon \searrow 0),
\end{align*}
wie behauptet.
\\

(ii) Zunächst erhalten wir mithilfe von Lemma~\ref{lm: reg der En}, der eingangs bewiesenen Abschätzung, Lemma~\ref{lm: reg of inv} und Lemma~\ref{lm: vertauschung von abl und wegintegral}, dass $t \mapsto P_{\varepsilon}(t)x$ stetig differenzierbar ist für alle $x \in X$ und alle $\varepsilon \in (0, \varepsilon_0']$ und dass 
\begin{align*}
P_{\varepsilon}'(t)x = \frac{1}{2 \pi i} \, \int_{\partial U_{\frac{1}{2}}(1)} (z-T_{\varepsilon}(t))^{-1} \, T_{\varepsilon}'(t) \, (z-T_{\varepsilon}(t))^{-1}x \,dz 
\end{align*}
für alle $t \in I$. Weil nun 
\begin{align*}
T_{\varepsilon}'(t) &= \sum_{k=0}^{m_{\varepsilon}-1} E_k'(t) \, \varepsilon^k \supset \sum_{k=0}^{m_{\varepsilon}-2} [A(t),E_{k+1}(t)] \, \varepsilon^k + E_{m_{\varepsilon}-1}'(t) \, \varepsilon^{m_{\varepsilon}-1} \\
&= \frac{1}{\varepsilon} \, [A(t), T_{\varepsilon}(t)] + E_{m_{\varepsilon}-1}'(t) \, \varepsilon^{m_{\varepsilon}-1}
\end{align*}
nach Lemma~\ref{lm: konstr der En},
ergibt sich
\begin{align}             \label{eq: höherer adsatz 1}
P_{\varepsilon}'(t)x =& \: \int_{\partial U_{\frac{1}{2}}(1)} (z-T_{\varepsilon}(t))^{-1} \, \frac{1}{\varepsilon} \, [A(t), T_{\varepsilon}(t)] \, (z-T_{\varepsilon}(t))^{-1}x \,dz \notag \\
&+   \int_{\partial U_{\frac{1}{2}}(1)} (z-T_{\varepsilon}(t))^{-1} \, E_{m_{\varepsilon}-1}'(t) \, \varepsilon^{m_{\varepsilon}-1} \, (z-T_{\varepsilon}(t))^{-1}x \,dz
\end{align}
für alle $x \in D$, alle $\varepsilon \in (0, \varepsilon_0']$ und alle $t \in I$. Wir haben dabei benutzt, dass $T_{\varepsilon}(t)X \subset D$ nach Lemma~\ref{lm: reg der En} und dass wegen
\begin{align*}
(z-T_{\varepsilon}(t))^{-1} = \frac{1}{z} + \frac{1}{z} \, T_{\varepsilon}(t) \, (z-T_{\varepsilon}(t))^{-1}
\end{align*}
auch $(z-T_{\varepsilon}(t))^{-1} \, D \subset D$ gilt.

Weiter ist $A(t)T_{\varepsilon}(t)$ wegen $T_{\varepsilon}(t)X \subset D$ und der Abgeschlossenheit von $A(t)$ eine beschränkte lineare Abbildung. Die Abbildung 
\begin{align*}
\partial U_{\frac{1}{2}}(1) \ni z \mapsto A(t) \, (z-T_{\varepsilon}(t))^{-1} x  = \frac{1}{z} \, A(t)x + \frac{1}{z} \, A(t)T_{\varepsilon}(t) \, (z-T_{\varepsilon}(t))^{-1} x 
\end{align*}
ist also stetig für alle $x \in D$, alle $\varepsilon \in (0, \varepsilon_0']$ und alle $t \in I$, insbesondere existiert das zugehörige Wegintegral, und aufgrund der Abgeschlossenheit von $A(t)$ haben wir
\begin{align*}
\int_{\partial U_{\frac{1}{2}}(1)}   (z-T_{\varepsilon}(t))^{-1} x \, dz \in D
\end{align*}
sowie
\begin{align}         \label{eq: höherer adsatz 2}
A(t) \, \int_{\partial U_{\frac{1}{2}}(1)} (z-T_{\varepsilon}(t))^{-1} x \, dz = \int_{\partial U_{\frac{1}{2}}(1)} A(t) \,  (z-T_{\varepsilon}(t))^{-1} x \, dz
\end{align}
für alle $x \in D$.

Aus \eqref{eq: höherer adsatz 1} und \eqref{eq: höherer adsatz 2} folgt nun, dass
\begin{align*}
&P_{\varepsilon}'(t)x - \frac{1}{\varepsilon} \, [A(t), P_{\varepsilon}(t)]x \\
&\qquad \qquad \quad = \frac{1}{2 \pi i} \, \int_{\partial U_{\frac{1}{2}}(1)} (z-T_{\varepsilon}(t))^{-1} \, E_{m_{\varepsilon}-1}'(t) \, \varepsilon^{m_{\varepsilon}-1} \, (z-T_{\varepsilon}(t))^{-1}x \,dz
\end{align*}
für alle $x \in D$, alle $\varepsilon \in (0, \varepsilon_0']$ und alle $t \in I$, insbesondere ist $[A(t), P_{\varepsilon}(t)]$ fortsetzbar zu einer auf ganz $X$ definierten beschränkten linearen Abbildung und 
\begin{align*}
\norm{  P_{\varepsilon}'(t) - \frac{1}{\varepsilon} \, [A(t), P_{\varepsilon}(t)]  } \le \frac{1}{2} \, M_0^2 \, \norm{ E_{m_{\varepsilon}-1}'(t) } \, \varepsilon^{m_{\varepsilon}-1}
\end{align*}
für alle $\varepsilon \in (0, \varepsilon_0']$ und alle $t \in I$.

Wenn $A$, $\sigma$, $P$ $m$-regulär sind, dann folgt  
\begin{align*}
\sup_{t \in I} \norm{ P_{\varepsilon}'(t) - \frac{1}{\varepsilon} \, [A(t), P_{\varepsilon}(t)] } = O(\varepsilon^{m-1}) \quad (\varepsilon \searrow 0)
\end{align*}
sofort, da in diesem Fall $m_{\varepsilon} = m$ unabhängig von $\varepsilon$.

Wenn  $A$, $\sigma$, $P$ sogar $\omega$-regulär sind, dann folgt die entsprechende Aussage ähnlich wie in Lemma~\ref{lm: T geht gegen P}. Sei nämlich $\beta \in (0,1)$ und $c := \frac{\beta}{g}$. Wir haben 
\begin{align*}
\norm{ E_k'(t) } \le g^{k+1} \, (k+1)!
\end{align*}
für alle $k \in \natu$ und alle $t \in I$, und 
\begin{align*}
g^{m_{\varepsilon}} \, m_{\varepsilon}! \, \varepsilon^{m_{\varepsilon}-1} \le c_0 \, g \, \Bigl( \frac{1}{g \varepsilon} \Bigr)^{\frac{3}{2}} \, e^{-m_{\varepsilon}} \le c_0'''  \, g \,   e^{-\frac{\beta}{g \varepsilon}} = c_0''' \, g \, e^{-\frac{c}{\varepsilon}}
\end{align*}
für alle $\varepsilon \in (0,\infty)$, woraus sich 
\begin{align*}
\sup_{t \in I} \norm{ P_{\varepsilon}'(t) - \frac{1}{\varepsilon} \, [A(t), P_{\varepsilon}(t)] } = O\bigl(  e^{-\frac{c}{\varepsilon}}  \bigr) \quad (\varepsilon \searrow 0)
\end{align*}
ergibt, wie behauptet.
\\

(iii) Die beiden Abbildungen $[0,t] \ni s \mapsto U_{\frac{1}{\varepsilon}}(t,s) P_{\varepsilon}(s) U_{\frac{1}{\varepsilon}}(s) P_{\varepsilon}(0)x$ und $[0,t] \ni s \mapsto P_{\varepsilon}(t) U_{\frac{1}{\varepsilon}}(t,s) P_{\varepsilon}(s) U_{\frac{1}{\varepsilon}}(s)x$ sind wegen $P_{\varepsilon}(s) D \subset D$ differenzierbar für alle $x \in D$ (Lemma~\ref{lm: strong db of products}), woraus sich nach (der geläufigen Version von) Lemma~\ref{lm: mws für einseitig db} ergibt, dass
\begin{align*}
\norm{ (1-P_{\varepsilon}(t)) U_{\frac{1}{\varepsilon}}(t) P_{\varepsilon}(0)x } &= \norm{   U_{\frac{1}{\varepsilon}}(t,s) P_{\varepsilon}(s) U_{\frac{1}{\varepsilon}}(s) P_{\varepsilon}(0)x \big|_{s=0}^{s=t}   } \\
& \le \sup_{s \in [0,t]} \norm{      U_{\frac{1}{\varepsilon}}(t,s) \Bigl(  P_{\varepsilon}'(s) - \frac{1}{\varepsilon} \, [A(s), P_{\varepsilon}(s)]  \Bigr) U_{\frac{1}{\varepsilon}}(s) P_{\varepsilon}(0)x      }
\end{align*}
und
\begin{align*}
\norm{ P_{\varepsilon}(t) U_{\frac{1}{\varepsilon}}(t) (1-P_{\varepsilon}(0))x } &= \norm{   P_{\varepsilon}(t)  U_{\frac{1}{\varepsilon}}(t,s) P_{\varepsilon}(s) U_{\frac{1}{\varepsilon}}(s) x \big|_{s=0}^{s=t}   } \\
& \le \sup_{s \in [0,t]} \norm{    P_{\varepsilon}(t)  U_{\frac{1}{\varepsilon}}(t,s) \Bigl(  P_{\varepsilon}'(s) - \frac{1}{\varepsilon} \, [A(s), P_{\varepsilon}(s)]  \Bigr) U_{\frac{1}{\varepsilon}}(s) x      }
\end{align*}
für alle $\varepsilon \in (0, \varepsilon_0']$ und alle $t \in I$. Weil nun $A$ $(M,0)$-stabil ist und $\sup_{\varepsilon \in (0, \varepsilon_0'], \, t \in I} \norm{ P_{\varepsilon}(t) } < \infty$, erhalten wir
\begin{align*}
&\sup_{t \in I} \norm{  (1-P_{\varepsilon}(t)) U_{\frac{1}{\varepsilon}}(t) P_{\varepsilon}(0)  }, \\
& \qquad \qquad \sup_{t \in I} \norm{  P_{\varepsilon}(t) U_{\frac{1}{\varepsilon}}(t) (1-P_{\varepsilon}(0))  }  = O(\varepsilon^{m-1}) \text{ bzw. } O\bigl( e^{-\frac{c}{\varepsilon}} \bigr) \quad (\varepsilon \searrow 0),
\end{align*}
wie gewünscht.

Schließlich existiere die Zeitentwicklung $V_{\frac{1}{\varepsilon}}$ zu $\frac{1}{\varepsilon} A + (1-2P_{\varepsilon}) \bigl(  P_{\varepsilon}' - \frac{1}{\varepsilon} \, [A, P_{\varepsilon}]  \bigr)$ für alle $\varepsilon \in (0,\varepsilon_0']$.
Dann ist die Abbildung $[s,t] \ni \tau \mapsto V_{\frac{1}{\varepsilon}}(t,\tau) P_{\varepsilon}(\tau) V_{\frac{1}{\varepsilon}}(\tau,s)x$ differenzierbar für alle $x \in D$ nach Proposition~\ref{thm: char zeitentwicklung}, denn
\begin{align*}
&\tau \mapsto \bigl(  P_{\varepsilon}'(\tau) - \frac{1}{\varepsilon} \, [A(\tau), P_{\varepsilon}(\tau)]  \bigr) y \\
& \qquad \qquad = \frac{1}{2 \pi i} \, \int_{\partial U_{\frac{1}{2}}(1)} (z-T_{\varepsilon}(\tau))^{-1} \, E_{m_{\varepsilon}-1}'(\tau) \, \varepsilon^{m_{\varepsilon}-1} \, (z-T_{\varepsilon}(\tau))^{-1} y \,dz 
\end{align*}
ist stetig für alle $y \in X$ nach Lemma~\ref{lm: reg der En}, 
und ihre Ableitung verschwindet, denn
\begin{align*}
(1-2P_{\varepsilon})  P_{\varepsilon}' = [ P_{\varepsilon}', P_{\varepsilon}]
\end{align*}
und
\begin{align*}
\frac{1}{\varepsilon} Ay - (1-2P_{\varepsilon}) \, \frac{1}{\varepsilon} \, [A, P_{\varepsilon}]y & = \frac{A}{\varepsilon}\, y - \bigl( P_{\varepsilon} \, \frac{A}{\varepsilon} \, \overline{P_{\varepsilon}} \, y +  \overline{P_{\varepsilon}} \, \frac{A}{\varepsilon} \, P_{\varepsilon} \, y \bigr) \\
& = P_{\varepsilon} \, \frac{A}{\varepsilon} \, P_{\varepsilon} \, y + \overline{P_{\varepsilon}} \, \frac{A}{\varepsilon} \, \overline{P_{\varepsilon}} \, y
\end{align*}
für alle $y \in D$ und alle $\varepsilon \in (0,\varepsilon_0']$ (wobei $\overline{P_{\varepsilon}} := 1-P_{\varepsilon}$). 
Dies zeigt, dass die Zeitentwicklung $V_{\frac{1}{\varepsilon}}$ adiabatisch i. e. S. ist bzgl. $P_{\varepsilon}$ und $1-P_{\varepsilon}$.

Auch die Abbildung $[0,t] \ni  s \mapsto U_{\frac{1}{\varepsilon}}(t,s) V_{\frac{1}{\varepsilon}}(s)x$ ist differenzierbar für alle $x \in D$ und sogar stetig differenzierbar, weil (wie eben bemerkt) 
\begin{align*}
&s \mapsto \bigl(  P_{\varepsilon}'(s) - \frac{1}{\varepsilon} \, [A(s), P_{\varepsilon}(s)]  \bigr) y 
\end{align*}
stetig ist für alle $y \in X$. 
Also haben wir
\begin{align*}
V_{\frac{1}{\varepsilon}}(t)x &- U_{\frac{1}{\varepsilon}}(t)x = U_{\frac{1}{\varepsilon}}(t,s) V_{\frac{1}{\varepsilon}}(s)x \big|_{s=0}^{s=t} \\
&= \int_0^t U_{\frac{1}{\varepsilon}}(t,s) \, (1-2P_{\varepsilon}(s)) \bigl(  P_{\varepsilon}'(s) - \frac{1}{\varepsilon} \, [A(s), P_{\varepsilon}(s)]  \bigr) \, V_{\frac{1}{\varepsilon}}(s)x \, ds,
\end{align*}
woraus wegen (ii) und der $(M,0)$-Stabilitiät von $A$ folgt (Lemma~\ref{lm: skalierung und (M,w)-stabilität} und Proposition~\ref{prop: abschätzung für gestörte zeitentw}), dass
\begin{align*}
\sup_{t \in I} \norm{   V_{\frac{1}{\varepsilon}}(t) - U_{\frac{1}{\varepsilon}}(t)   } = O(\varepsilon^{m-1}) \text{ bzw. } O\bigl( e^{-\frac{c}{\varepsilon}} \bigr) \quad (\varepsilon \searrow 0),
\end{align*}
wie behauptet.
\end{proof}

Aus diesem Satz folgt insbesondere die zweite Aussage unseres früheren Adiabatensatzes mit Spektrallückenbedingung, Satz~\ref{thm: unhandl adsatz mit sl}. Wir müssen dazu nur beachten, dass die dort getroffenen Voraussetzungen nichts anderes bedeuten, als dass $A$, $\sigma$, $P$ $2$-regulär sind.
\\

Außerdem ergibt sich (s. Abschnitt~1 in~\cite{Nenciu 93}) mühelos der folgende Adiabatensatz höherer Ordnung.

\begin{cor} \label{cor: höherer adsatz für suppP' in (0,1)}
Seien $A(t)$, $\sigma(t)$, $P(t)$ wie in Satz~\ref{thm: höherer adsatz} und sei zusätzlich $\supp P' \ne I$. Dann gilt
\begin{align*}
&\sup_{t \in I \setminus \supp P'} \norm{  (1-P(t)) U_{\frac{1}{\varepsilon}}(t) P(0)  }, \\
& \qquad \qquad \sup_{t \in I \setminus \supp P'} \norm{  P(t) U_{\frac{1}{\varepsilon}}(t) (1-P(0))  }  = O(\varepsilon^{m-1}) \text{ bzw. } O\bigl( e^{-\frac{c}{\varepsilon}} \bigr) \quad (\varepsilon \searrow 0).
\end{align*}
\end{cor}

\begin{proof}
Anhand der Definition der $E_k(t)$ (Lemma~\ref{lm: reg der En}) sieht man induktiv, dass $E_k(t) = 0$ für alle $t \in I \setminus \supp P'$ und alle $k \in \{1, \dots, m \}$ (falls $m$-Regularität vorliegt) bzw. alle $k \in \natu$ (falls $\omega$-Regularität vorliegt), woraus $T_{\varepsilon}(t) = P(t)$ und damit auch
\begin{align*}
P_{\varepsilon}(t) = \frac{1}{2 \pi i} \, \int_{\partial U_{\frac{1}{2}}(1)} (z-T_{\varepsilon}(t))^{-1} \, dz = P(t)
\end{align*}
folgt für alle $t \in I \setminus \supp P'$. Satz~\ref{thm: höherer adsatz} liefert nun die behauptete Aussage.
\end{proof}

Avron, Seiler und Yaffe zeigen einen ähnlichen Satz (Theorem~2.8) in~\cite{AvronSeilerYaffe 87}, allerdings nur für den Sonderfall schiefselbstadjungierter $A(t)$. Die Vorgehensweise dieser Arbeit (im wesentlichen mehrfache partielle Integration) scheint ganz wesentlich auf der Schiefselbstadjungiertheit der $A(t)$ zu beruhen und ist daher wohl nicht auf allgemeinere Situationen (wie im obigen Korollar) übertragbar.
\\

Wir weisen darauf hin, dass wir die Aussagen (i) und (ii) von Satz~\ref{thm: höherer adsatz} (ebenso alle Aussagen der vorangehenden Lemmas~\ref{lm: reg der En}, \ref{lm: konstr der En} und~\ref{lm: T geht gegen P}) auch unter deutlich schwächeren Voraussetzungen bekommen hätten: um diese Aussagen zu bekommen, hätte es genügt vorauszusetzen, dass 
\begin{itemize}
\item [(i)] $A(t)$ für jedes $t \in I$ eine abgeschlossene lineare Abbildung $D(A(t)) \subset X \to X$ ist mit $\rho(A(t)) \ne \emptyset$,
\item [(ii)] $\sigma(t)$ für jedes $t \in I$ eine in $\sigma(A(t))$ isolierte kompakte Untermenge von $\sigma(A(t))$ ist und zu jedem $t_0 \in I$ ein Zykel $\gamma_{t_0}$ und eine Umgebung $U_{t_0}$ existiert, sodass $\im \gamma_{t_0} \subset \rho(A(t))$ und $n(\gamma_{t_0},\sigma(t)) =1$ und $n(\gamma_{t_0},\sigma(A(t)) \setminus \sigma(t)) = 0$ für alle $t \in U_{t_0}$, 
\item [(iii)] die Abbildung $U_{t_0} \ni t \mapsto (A(t)-z)^{-1}x$ $(m-1)$-mal stetig differenzierbar bzw. beliebig oft differenzierbar ist für alle $z \in \im \gamma_{t_0}$ und alle $x \in X$, die Abbildung $\im \gamma_{t_0} \ni z \mapsto  \ddtk{ (A(t)-z)^{-1} }$ stetig ist für alle $t \in U_{t_0}$ und alle $k \in \{ 1, \dots, m-1 \}$ bzw. alle $k \in \natu$ und
\begin{align*}
\sup_{(t,z) \in U_{t_0} \times \, \im \gamma_{t_0} } \norm{ \ddtk{ (A(t)-z)^{-1} } }  < \infty
\end{align*}
für alle $k \in \{ 1, \dots, m-1 \}$ bzw. 
\begin{align*}
\sup_{(t,z) \in U_{t_0} \times \, \im \gamma_{t_0} } \norm{ \ddtk{ (A(t)-z)^{-1} } }  \le c^k \, k!
\end{align*}
für alle $k \in \natu$,
\item [(iv)] $P(t)$ für jedes $t \in I$ die Rieszprojektion von $A(t)$ auf $\sigma(t)$ ist und $t \mapsto P(t)x$ $m$-mal stetig differenzierbar ist für alle $x \in X$. 
\end{itemize}
Wir sind aber vor allem an Aussage (iii) von Satz~\ref{thm: höherer adsatz} interessiert und für diese brauchen wir die schärfere Voraussetzung der $m$- bzw. $\omega$-Regularität, insbesondere die $(M,0)$-Stabilität von $A$ (nach Beispiel~\ref{ex: (M,0)-stabilität wesentlich in den adsätzen mit sl}). 
\\

Wenn wir in den eben aufgeführten schwächeren Voraussetzungen (i) bis (iv) überall $m-1$ durch $m$ ersetzen und zudem $\sigma(t)$ als gleichmäßig isoliert in $\sigma(A(t))$ voraussetzen, so erhalten wir die Voraussetzungen aus Nencius Arbeit~\cite{Nenciu 93} (dort $G$ und $S^m$ bzw. $S_0$ genannt) für das beschränkte Grundintervall $J := I$. Jedenfalls verstehen wir die Voraussetzungen in~\cite{Nenciu 93} so -- und dass diese nicht genauso gemeint sein können, wie sie dort formuliert sind, zeigt das folgende Beispiel. In diesem sind zwar die wörtlich genommenen Voraussetzungen Nencius erfüllt, aber $t \mapsto P(t)x$ ist nicht (stetig) differenzierbar für alle $x \in X$. Die in Lemma~1 in~\cite{Nenciu 93} gegebene Definition der $E_k(t)$ ist in diesem Beispiel also nicht sinnvoll.

\begin{ex}
Sei $A(t)$ für jedes $t \in I$ eine beschränkte lineare Abbildung in $X$ mit konstantem Spektrum
\begin{align*}
\sigma(A(t)) = \overline{U}_1(0) \setminus U_{\frac{1}{2}}(0) \cup \{0, -3 i\}
\end{align*}
und sei $t \mapsto A(t)$ $m$-mal stetig differenzierbar bzw. (reell) analytisch. Sei weiter 
\begin{align*}
\sigma(t) := \begin{cases} \overline{U}_1(0) \setminus U_{\frac{1}{2}}(0) \cup \{0\}, & t \in I\setminus \{0\} \\
													 \overline{U}_1(0) \setminus U_{\frac{1}{2}}(0), & t = 0
	           \end{cases}	
\end{align*}
und $P(t)$ für jedes $t \in I$ die Rieszprojektion von $A(t)$ auf $\sigma(t)$.

Dann ist 
\begin{align*}
\partial U_{\frac{d(t)}{2}}(\sigma(t)) = \begin{cases} \partial U_2(0), & t \in I\setminus \{0\} \\
													 														 \partial U_{\frac{5}{4}}(0) \cup \partial U_{\frac{1}{4}}(0), & t = 0,
	                                       \end{cases},	
\end{align*}
wobei $d(t) := \dist(\sigma(t), \sigma(A(t)) \setminus \sigma(t))$, und die positiv einfach geschlossenen Zykel $\gamma_t$, die $\partial U_{\frac{d(t)}{2}}(\sigma(t))$ beschreiben (und die in~\cite{Nenciu 93} wohl gemeint sind), sind gegeben durch
\begin{align*}
\gamma_t = \begin{cases} \partial U_2(0), & t \in I\setminus \{0\} \\
												 \partial  U_{\frac{5}{4}}(0) - \partial U_{\frac{1}{4}}(0), & t = 0
	         \end{cases}.	
\end{align*}

Wir sehen nun, dass die Voraussetzungen $G$ und $S^m$ bzw. $S_0$ aus~\cite{Nenciu 93} wörtlich erfüllt sind, es gilt ja sogar $\im \gamma_{t_0} \subset \rho(A(t))$ für alle $t \in I$, die Abbildung $I \ni t \mapsto (A(t)-z)^{-1}$ ist $m$-mal stetig differenzierbar bzw. beliebig oft differenzierbar für alle $z \in \im \gamma_{t_0}$, die Abbildung $\im \gamma_{t_0} \ni z \mapsto \ddtk{ (A(t)-z)^{-1} }$ ist stetig für alle $k \in \{0,1, \dots, m\}$ bzw. alle $k \in \natu$ und 
\begin{align*}
\sup_{(t,z) \in I \times \, \im \gamma_{t_0}} \norm{ \ddtk{ (A(t)-z)^{-1} } } < \infty
\end{align*}
für alle $k \in \{0,1, \dots, m\}$ bzw.
\begin{align*}
\sup_{(t,z) \in V_{t_0} \times \, \im \gamma_{t_0}} \norm{ \ddtk{ (A(t)-z)^{-1} } } \le c^k \, k!
\end{align*}
für alle $k \in \natu$, wobei sich diese letzte Abschätzung aus den cauchyschen Ungleichungen (die aus Satz~\ref{thm: Cauchy global} folgen) ergibt:  wir müssen nur beachten, dass $A$ fortgesetzt werden kann zu einer holomorphen Abbildung $B$ auf einer in $\complex$ offenen Umgebung von $I$, dann eine so kleine positive Zahl $r_0$ wählen (Theorem~IV.3.1 in~\cite{Kato: Perturbation 80}), dass $\im \gamma_{t_0} \subset \rho(B(w))$ für alle $w \in U_{r_0}(t_0)$, und $V_{t_0} := U_{\frac{r_0}{2}}(t_0) \cap I$ setzen.

Aber $t \mapsto P(t)x$ ist nicht für alle $x \in X$ differenzierbar in $0$, denn
\begin{align*}
2 \pi i \, \frac{P(h)-P(0)}{h} \, x &= \frac{1}{h} \, \int_{\gamma_h} (z-A(h))^{-1} x \,dz - \frac{1}{h} \, \int_{\gamma_0} (z-A(0))^{-1} x \,dz \\
&= \frac{1}{h} \, \int_{\partial U_2(0)} (z-A(h))^{-1} x - (z-A(0))^{-1} x \,dz \\
& \quad + \frac{1}{h} \, \int_{\partial U_{\frac{1}{4}}(0)} (z-A(0))^{-1} x \,dz,
\end{align*}
was nicht konvergiert, schließlich ist $\int_{\partial U_{\frac{1}{4}}(0)} (z-A(0))^{-1} \,dz$ (bis auf Vorfaktor) die Rieszprojektion von $A(0)$ auf $\{0\} \ne \emptyset$ und verschwindet daher nicht.  $\blacktriangleleft$
\end{ex}

Sofern wir die Voraussetzungen in~\cite{Nenciu 93} auf die eben beschriebene Art verstehen (dürfen), ist Satz~\ref{thm: höherer adsatz} allgemeiner als Nencius Theorem~2. Seien nämlich die so verstandenen Voraussetzungen von Theorem~2 erfüllt und die $A(t)$ zusätzlich schiefselbstadjungiert. Dann sind $A$, $\sigma$, $P$ $m$-regulär bzw. $\omega$-regulär, nach Satz~\ref{thm: höherer adsatz} gilt also, dass
\begin{align*}
\sup_{t \in I} \norm{   V_{\frac{1}{\varepsilon}}(t) - U_{\frac{1}{\varepsilon}}(t)  } = O(\varepsilon^{m-1}) \text{ bzw. } O\bigl( e^{-\frac{c}{\varepsilon}} \bigr) \quad (\varepsilon \searrow 0)
\end{align*} 
(wenn nur die Zeitentwicklung $V_{\frac{1}{\varepsilon}}$ zu $\frac{1}{\varepsilon} A + P_{\varepsilon}' - \frac{1}{\varepsilon} \, [A, P_{\varepsilon}]$ existiert) oder, was nach der folgenden Proposition dasselbe bedeutet,
\begin{align*}
\sup_{t \in I} \norm{   V_{\frac{1}{\varepsilon}}(t)^{*} \, U_{\frac{1}{\varepsilon}}(t) - 1  } = O(\varepsilon^{m-1}) \text{ bzw. } O\bigl( e^{-\frac{c}{\varepsilon}} \bigr) \quad (\varepsilon \searrow 0).
\end{align*}
Und dies ist die (letzte und einzige uns hier interessierende) Aussage von Theorem~2.

\begin{prop}  \label{prop: P(eps) orth und V(eps) unitär}
Seien $A(t)$, $\sigma(t)$, $P(t)$ sowie $P_{\varepsilon}(t)$ wie in Satz~\ref{thm: höherer adsatz} und die $A(t)$ seien zusätzlich schiefselbstadjungiert. Dann sind die Projektionen $P_{\varepsilon}(t)$ orthogonal und wenn die Zeitentwicklung $V_{\frac{1}{\varepsilon}}$ zu $\frac{1}{\varepsilon} A + (1-2P_{\varepsilon}) \bigl(  P_{\varepsilon}' - \frac{1}{\varepsilon} \, [A, P_{\varepsilon}]  \bigr)$ existiert, dann ist sie unitär.
\end{prop}

\begin{proof}
Sei $m \in \natu$ und seien $A$ , $\sigma$, $P$ $m$-regulär. Wir zeigen zunächst mit Induktion über $k \in \{0,1, \dots, m\}$, dass $E_k(t)$ symmetrisch ist für alle $k \in \{0, 1, \dots, m \}$ und alle $t \in I$.

Sei $k = 0$.
Dann ist $E_k(t) = P(t)$ als Rieszprojektion der schiefselbstadjungierten linearen Abbildung $A(t)$ orthogonal (Proposition~\ref{prop: rieszproj für normale A}), insbesondere symmetrisch für alle $t \in I$.

Sei $k \in \{1, \dots, m \}$ und $E_l(t)$ sei symmetrisch für alle $l \in \{0, \dots, k-1\}$ und alle $t \in I$.  
Wegen $\sigma(t) \subset \sigma(A(t)) \subset i \, \real$ ist der Zykel $\gamma_t$ für jedes $t \in I$ in $\rho(A(t))$ homolog zu einem Zykel $\gamma_{0 t}$, der aus endlich vielen Kreiswegen mit Mittelpunkten auf der imaginären Achse besteht. Nach Induktionsvoraussetzung ist $P(t) E_{k-1}'(t) \overline{P}(t) - \overline{P}(t) E_{k-1}'(t) P(t)$ schiefsymmetrisch, womit sich durch Ausschreiben des Wegintegrals (beachte die Schiefselbstadjungiertheit von $A(t)$) ergibt, dass   
\begin{align*}
\frac{1}{2 \pi i} \, \int_{\gamma_{0 t}} (A(t)-z)^{-1} \, \Bigl( P(t) E_{k-1}'(t) \overline{P}(t) - \overline{P}(t) E_{k-1}'(t) P(t) \Bigr) (A(t)-z)^{-1} \,dz 
\end{align*}
symmetrisch ist für alle $t \in I$. Weiterhin zeigt die Induktionsvoraussetzung, dass auch $S_k(t)$ symmetrisch ist für alle $t \in I$. Also ist $E_k(t)$ symmetrisch für alle $t \in I$, wie gewünscht.

Jetzt sehen wir, dass 
\begin{align*}
T_{\varepsilon}(t) = \sum_{k=0}^{m_{\varepsilon}-1} E_k(t) \varepsilon^k
\end{align*}
symmetrisch ist, und daher ist $P_{\varepsilon}(t)$ als Rieszprojektion von $T_{\varepsilon}(t)$ auf $\sigma(T_{\varepsilon}(t)) \cap U_{\frac{1}{2}}(1)$ orthogonal für alle $\varepsilon \in (0, \varepsilon_0']$ und alle $t \in I$.

Schließlich sind die $(1-2P_{\varepsilon}(t)) \bigl(  P_{\varepsilon}'(t) - \frac{1}{\varepsilon} \, [A(t), P_{\varepsilon}(t)]  \bigr)$ schiefsymmetrisch fortsetzbar, da die linearen Abbildungen
\begin{align*}
(1-2P_{\varepsilon}(t))  P_{\varepsilon}'(t) = [ P_{\varepsilon}'(t), P_{\varepsilon}(t)]
\end{align*}
und
\begin{align*}
(1-2P_{\varepsilon}(t)) \, \frac{1}{\varepsilon} \, [A(t), P_{\varepsilon}(t)] & =  P_{\varepsilon}(t) \, \frac{A(t)}{\varepsilon} \, \overline{P_{\varepsilon}(t)}  +  \overline{P_{\varepsilon}(t)} \, \frac{A(t)}{\varepsilon} \, P_{\varepsilon}(t)  
\end{align*}
wegen der Schiefselbstadjungiertheit der $A(t)$ und der eben gezeigten Orthogonalität der Projektionen $P_{\varepsilon}(t)$ schiefsymmetrisch auf $D$ sind.
 
Aufgrund von Proposition~\ref{prop: abschätzung für gestörte zeitentw} können wir $V_{\frac{1}{\varepsilon}}$ als Störungsreihe darstellen (die Abbildung $t \mapsto P_{\varepsilon}'(t) - \frac{1}{\varepsilon} \, [A(t), P_{\varepsilon}(t)]$ ist nach dem Beweis von Satz~\ref{thm: höherer adsatz} wirklich stark stetig) und diese Störungsreihendarstellung von $V_{\frac{1}{\varepsilon}}$ liefert wegen der eingangs gezeigten Schiefsymmetrie der Störung $(1-2P_{\varepsilon}(t)) \bigl(  P_{\varepsilon}'(t) - \frac{1}{\varepsilon} \, [A(t), P_{\varepsilon}(t)]  \bigr)$ die gewünschte Unitarität von $V_{\frac{1}{\varepsilon}}$.
\end{proof}

Aufgrund von Satz~\ref{thm: Kato} \emph{existiert} die Zeitentwicklung $V_{\frac{1}{\varepsilon}}$, wenn $t \mapsto P_{\varepsilon}'(t) - \frac{1}{\varepsilon} \, [A(t), P_{\varepsilon}(t)]$ stark stetig differenzierbar ist für alle $\varepsilon \in (0, \varepsilon_0']$, was nach dem Beweis von Satz~\ref{thm: höherer adsatz} beispielsweise dann erfüllt ist, wenn $t \mapsto E_{m_{\varepsilon}-1}'(t)$ noch einmal stark stetig differenzierbar ist.
\\

Wie eben schon erwähnt, ist Satz~\ref{thm: höherer adsatz} (eher) allgemeiner als Theorem~2 in~\cite{Nenciu 93}. Aber er ist -- wie folgendes Beispiel zeigt -- auch \emph{echt} allgemeiner, und zwar weil die $A(t)$ in diesem Beispiel zum einen nicht schiefselbstadjungiert sind und zum andern weniger regulär sind als in Theorem~2 verlangt.

\begin{ex} \label{ex: unser höherer adsatz echt allgemeiner als der von nenciu}
Sei $X := \ell^2(I_3)$. Sei  
\begin{align*}
\lambda(t) := 0 \text{ \; und \; } 
\mu(t) := \begin{cases} \lambda_2 - \bigl( t-\frac{1}{2} \bigr)^2, & t \in [0,\frac{1}{2}) \\
												\lambda_2, & t \in [\frac{1}{2}, 1]
					\end{cases}, 
\end{align*}
\begin{align*}
A_0 := \begin{pmatrix}
\lambda(t)         & 0     & 0         \\
0    & \mu(t)  & 1     \\
0     & 0					& \mu(t)      
\end{pmatrix}, \;
R(t):= \begin{pmatrix}
\cos t  & \sin t & 0    \\
-\sin t & \cos t & 0     \\
0       & 0      & 1   
\end{pmatrix},
\end{align*}
$\sigma(t) := \{ \lambda(t) \} = \{0\}$ für alle $t \in I$ und $P_0$ die orthogonale Projektion auf $\spn \{e_1\}$. Sei 
\begin{align*}
A(t) := R(t)^* A_0 R(t) \text{ \; und \; } P(t):= R(t)^* P_0 R(t)
\end{align*}
für alle $t \in I$.

Wir zeigen, dass dann zwar die Voraussetzungen von Satz~\ref{thm: handl adsatz mit sl} (und damit auch die Voraussetzungen von Satz~\ref{thm: höherer adsatz} mit $m = 2$) erfüllt sind, aber nicht die von Nencius Theorem~2. Daraus folgt dann, wie gewünscht, dass Satz~\ref{thm: handl adsatz mit sl} nicht aus Theorem~2 folgt, und dass Satz~\ref{thm: höherer adsatz} echt allgemeiner ist als Theorem~2.

Zunächst erzeugt jedes $A_0(t)$ und damit auch jedes $A(t)$ eine Kontraktionshalbgruppe auf $X$ und $t \mapsto A(t)$ ist einmal stetig differenzierbar, weil $t \mapsto \mu(t)$ einmal stetig differnzierbar ist.
Weiter ist $\sigma(t)$ gleichmäßig isoliert in $\{0, \mu(t)\} = \sigma(A(t))$ und $t \mapsto \sigma(t) = \{0\}$  ist stetig.
Schließlich ist $P_0$ die Rieszprojektion von $A_0$ auf $\{0\}$ (Nachrechnen oder Proposition~\ref{prop: rieszproj eind}), das heißt, $P(t)$ ist die Rieszprojektion von $A(t)$ auf $\sigma(t)$ für alle $t \in I$ und $t \mapsto P(t) = R(t)^* P_0 R(t)$ ist zweimal stetig differenzierbar.  

Also sind tatsächlich die Voraussetzungen von Satz~\ref{thm: handl adsatz mit sl} erfüllt. Aber die Voraussetzungen von Theorem~2 in~\cite{Nenciu 93} sind nicht erfüllt: erstens sind die $A(t)$ nicht schiefselbstadjungiert und auch nicht normal und zweitens ist $t \mapsto A(t)x$ nicht zweimal stetig differenzierbar für alle $x \in X$, denn $t \mapsto \mu(t)$ und damit $t \mapsto A_0(t) e_2$ ist zwar einmal aber nicht zweimal (stetig) differenzierbar.

Zuletzt sei gesagt, dass die Aussage des Adiabatensatzes hier nicht schon aufgrund der trivialen Adiabatensätze erfüllt ist, denn zum einen ist $t \mapsto P(t)$ nicht konstant und zum andern ist $A$ auch für keine negative Zahl $\omega \in (-\infty, 0)$ $(M, \omega)$-stabil, weil dazu $\sigma(A(t))$ in $\{z \in \complex: \Re z < 0 \}$ enthalten sein müsste.  $\blacktriangleleft$
\end{ex}

Wir haben versucht, Satz~\ref{thm: höherer adsatz} auf Situationen ohne Spektrallücke auszudehnen. Jedoch erfolglos -- selbst im Sonderfall normaler $A(t)$ und einpunktiger $\sigma(t) = \{ \lambda(t) \}$ mit einem nach einer Seite hin in $\sigma(A(t))$ isolierten $\lambda(t)$ (das heißt, $\lambda(t) + \delta e^{i \vartheta_0} \in \rho(A(t))$ für alle $\delta \in (0, \delta_0]$ und alle $t \in I$, wie in den Sätzen von Abschnitt~\ref{sect: adsätze ohne sl}).


\section{Anwendungsbeispiel} \label{sect: anwendbsp}

In diesem letzten Abschnitt besprechen wir eine kleine Anwendung der Adiabatentheorie, und zwar eine Anwendung in der Neutronentransporttheorie. Wir stützen uns dabei auf die Arbeiten~\cite{LehnerWing 55}, \cite{LehnerWing 56} Lehners und Wings und werden auch sonst nur die wichtigsten Schritte ausführlich begründen: es geht uns hier nur um Satz~\ref{thm: anwendung des adsatzes}.

Wir betrachten die folgenden Anfangsrandwertprobleme ($T \in (0, \infty)$) auf $I \times [-a,a] \times [-1,1]$:
\begin{align*}
\frac{ \partial \varphi }{\partial t}(t,x,\mu) = T \Bigl( - \mu \, \frac{ \partial \varphi }{\partial x}(t,x,\mu) + \frac{c(t)}{2} \, \int_{-1}^1 \varphi(t,x,\mu') \, d\mu' - s(t) \varphi(t,x,\mu) \Bigr),
\end{align*}
\begin{gather*}
\varphi(0,x, \mu) = \varphi_0(x,\mu), \\
\varphi(t,a,\mu) = 0 \text{ für alle } \mu \in [-1,0) \text{ und } \varphi(t,-a,\mu) = 0 \text{ für alle } \mu \in (0,1].
\end{gather*}
Diese beschreiben, wie sich die Verteilung $\varphi(t, \,.\,,\,.\,)$ der Neutronen in einer unendlich ausgedehnten Platte der Dicke $2 a$ (umgeben von Vakuum) zeitlich entwickelt. 
$\varphi(t,\,.\,,\,.\,)$ bezeichnet genauer die Anzahldichte (im Sinne des Satzes von Radon, Nikodym) der Neutronen zur Zeit $t$, das heißt
$\int_{E} \varphi(t,x,\mu) \, d(x,\mu)$ ist die Anzahl der Neutronen, die sich zur Zeit $t$ in der Untermenge $E \in \mathcal{B}_{[-a,a] \times [-1,1]}$ des Ortsrichtungsraumes $[-a,a]\times [-1,1]$ aufhalten, insbesondere ist
\begin{align*}
\int_{E_1 \times E_2} \varphi(t,x,\mu) \, d(x,\mu) 
\end{align*}
die Anzahl der Neutronen, die zur Zeit $t$ in $E_1 \in \mathcal{B}_{[-a,a]}$ sind und sich in eine Richtung $\mu$ aus $E_2 \in \mathcal{B}_{[-1,1]}$ bewegen (was bedeuten soll, dass $\mu$ gleich dem Kosinus des Winkels zwischen Bewegungsrichtung des Neutrons und der positiven $x$-Achse ist). Wir nehmen dabei an, dass die Neutronen nur mit dem Medium wechselwirken (und zwar durch Stöße), aber nicht untereinander, dass die Streuung infolge der Stöße isotrop ist und dass alle Neutronen (betragsmäßig) dieselbe Geschwindigkeit haben. Die Zahl $s(t)$ bezeichnet den totalen Wirkungsquerschnitt zur Zeit $t$ (kennzeichnet die Wahrscheinlichkeit einer Wechselwirkung, das heißt eines Zusammenstoßes) und $\frac{c(t)}{s(t)}$ die durchschnittliche Anzahl Neutronen, die aus dem Zusammenstoß eines Neutrons mit einem Kern hervorgehen: $\frac{c(t)}{s(t)} < 1$ bedeutet also Streuung und Absorption, $\frac{c(t)}{s(t)} = 1$ reine Streuung und $\frac{c(t)}{s(t)} > 1$ eine Vervielfältigung.

Die Randbedingungen schließlich bedeuten, dass zu keiner Zeit $t$ Neutronen von rechts oder links aus dem Vakuum in die Platte eindringen. 
\\

Den obigen konkreten Anfangsrandwertproblemen entsprechen die folgenden abstrakten Anfangswertprobleme ($T \in (0, \infty)$) auf $I$:
\begin{align*}
\varphi' = T A(t) \varphi, \;\; \varphi(0) = \varphi_0,
\end{align*}
wobei $A(t) := A_0(c(t)) - s(t)$ und $A_0(c) := A_0 + cB$ für alle $c \in (0, \infty)$. $B$ ist dabei die lineare Abbildung in $X := L^2([-a,a] \times [-1,1], \complex)$, gegeben durch
\begin{align*}
\bigl( B \varphi \bigr)(x,\mu) := \frac{1}{2} \, \int_{[-1,1]} \varphi(x,\mu') \,d\mu' 
\end{align*}
für alle $\varphi \in X$, und $A_0$ ist die wie folgt gegebene lineare Abbildung: $D(A_0)$ ist die Menge (der Äquivalenzklassen) genau der $2$-integrierbaren (das heißt quadratintegrierbaren) Abbildungen $\varphi: [-a,a] \times [-1,1] \to \complex$, für die gilt:
\begin{itemize}
\item [(i)] $(-a,a) \ni x \mapsto \varphi(x, \mu)$ ist schwach differenzierbar für fast alle $\mu \in [-1,1]$
\item [(ii)] $[-a,a] \times [-1,1] \ni (x,\mu) \mapsto \mu \, \frac{ \partial \varphi }{\partial x}(x,\mu)$ ist $2$-integrierbar
\item [(iii)] $\varphi(a, \mu) = 0$ für fast alle $\mu \in [-1,0)$ und $\varphi(-a, \mu) = 0$ für fast alle $\mu \in (0,1]$,
\end{itemize}
und
\begin{align*}
\bigl( A_0 \varphi) (x,\mu) := - \mu \, \frac{ \partial \varphi }{\partial x}(x,\mu)
\end{align*}
für alle $\varphi \in D(A_0)$, wobei $\frac{ \partial \varphi }{\partial x}(x,\mu)$ hier und im folgenden für $\bigl( \varphi(\, . \, , \mu) \bigr)'(x)$ steht und $\bigl( \varphi(\, . \, , \mu) \bigr)'$ die für fast alle $\mu \in [-1,1]$ existierende schwache Ableitung von $(-a,a) \ni x \mapsto \varphi(x, \mu)$ bezeichnet. Weiter meinen wir mit einer \emph{$p$-integrierbaren} Abbildung $f$ auf einem Maßraum $(X_0, \mathcal{A}, \mu)$ natürlich eine messbare Abbildung $X_0 \to \complex$ mit $\int |f|^p \, d\mu < \infty$.

Zu dieser Definition von $A_0$ (s. Abschnitt~1 in~\cite{LehnerWing 56} oder in~\cite{Lehner 56}) ist anzumerken: wenn $\varphi, \tilde{\varphi}: [-a,a] \times [-1,1] \to \complex$ zwei $2$-integrierbare Abbildungen sind, die fast überall übereinstimmen, und $(-a,a) \ni x \mapsto \varphi(x, \mu)$ schwach differenzierbar ist für fast alle $\mu \in [-1,1]$, dann ist auch die Abbildung $(-a,a) \ni x \mapsto \tilde{\varphi}(x, \mu)$ schwach differenzierbar für fast alle $\mu \in [-1,1]$ und es gilt
\begin{align*}
\bigl( \tilde{\varphi}(\, . \, , \mu) \bigr)'(x) = \bigl( \varphi(\, . \, , \mu) \bigr)'(x)
\end{align*}
für fast alle $(x,\mu) \in [-a,a] \times [-1,1]$. Insbesondere hängt der $A_0$ definierende Ausdruck nicht vom gewählten Vertreter ab.

Weiter ist zu beachten: wenn $\varphi$ eine $2$-integrierbare Abbildung mit (i) und (ii) ist, dann ist $(-a,a) \ni x \mapsto \bigl( \varphi(\, . \, , \mu) \bigr)'(x)$ $2$-integrierbar für fast alle $\mu \in [-1,1]$ und daher ist $[-a,a] \ni x \mapsto \varphi(x,\mu)$ absolut stetig. Warum? Weil ganz allgemein $W^{1,p}(J, \complex)$ für beschränkte offene Intervalle $J$ und $p \in [1, \infty)$ beschrieben werden kann als die Menge (der Äquivalenzklassen) der absolut stetigen Abbildungen $f: \overline{J} \to \complex$, deren fast überall gegebene punktweise Ableitung $p$-integrierbar ist. 
Insbesondere ist für eine Abbildung $\varphi$ mit (i) und (ii) die Randbedingung (iii) überhaupt sinnvoll.

Die fast überall gegebene punktweise Ableitung einer absolut stetigen Abbildung $f$ wie eben stimmt mit der schwachen Ableitung von $f$ überein.
All dies folgt aus Satz~VII.4.14 in~\cite{Elstrodt} (Charakterisierung absolut stetiger Abbildungen) und dem bekannten Satz, wonach eine schwach differenzierbare $p$-integrierbare Abbildung $f: J \to \complex$ mit $f' = 0$ (im wesentlichen) konstant ist. 
\\

Wegen $C_c^{\infty}((-a,a)\times(-1,1)) \subset D(A_0)$ ist $D(A_0)$ dicht in $X$. Außerdem ist $A_0$ dissipativ, was man mithilfe partieller Integration gemäß Satz~VII.4.16 in~\cite{Elstrodt} einsieht (Abschnitt~1 in~\cite{Lehner 56}), und $B$ ist eine orthogonale Projektion, wie man leicht bestätigt.

\begin{lm} \label{lm: inhom lin dgl für schwache abl}
Sei $J = (a,b)$ ein beschränktes offenes Intervall und $p \in [1, \infty)$. Seien $f, g \in L^p(J)$, $f$ schwach differenzierbar und $f' = \lambda f + g$ für eine komplexe Zahl $\lambda$. Dann gilt
\begin{align*}
f(x) = f(a) e^{\lambda(x-a)} + \int_{(a,x]} e^{\lambda(x-t)} \, g(t) \, dt
\end{align*}
für fast alle $x \in J$.
\end{lm}

\begin{proof}
Sei 
\begin{align*}
f_1(x):= e^{-\lambda(x-a)} f(x) \text{ \; und \; } f_2(x):= \int_{(a,x]} e^{-\lambda(t-a)} \, g(t) \, dt
\end{align*}
für alle $x \in J$.
Wie man leicht sieht, sind $f_1$ und $f_2$ dann schwach differenzierbar und 
\begin{align*}
f_1'(x) = -\lambda e^{-\lambda(x-a)} \, f(x) + e^{-\lambda(x-a)} \, f'(x) = e^{-\lambda(x-a)} \, g(x) = f_2'(x)
\end{align*}
für fast alle $x \in J$. 
Also existiert eine Zahl $c \in \complex$, sodass $f_1(x) = c + f_2(x)$ für fast alle $x \in J$, oder mit anderen Worten:
\begin{align*}
f(x) = e^{\lambda(x-a)} c + \int_{(a,x]} e^{\lambda(x-t)} \, g(t) \, dt
\end{align*}
für fast alle $x \in J$. Wir können $f$ daher kanonisch auf $\partial J$ fortsetzen, denn die rechte Seite der obigen Gleichung ist ja (absolut) stetig in $x \in \overline{J}$. 
\end{proof}

Aus diesem Lemma folgt: wenn $\varphi \in D(A_0)$, $\psi \in X$ und $(\lambda - A_0)\varphi = \psi$ für ein $\lambda \in \complex$, dann ist $\varphi$ gegeben durch
\begin{align*}
\varphi(x, \mu) = \begin{cases} -\frac{1}{\mu} \int_{[x,a]} e^{-\frac{\lambda}{\mu} \, (x-t)} \, \psi(t,\mu) \, dt & \text{ für fast alle } (x,\mu) \in [-a,a] \times [-1,0) \\
																 \frac{1}{\mu} \int_{[-a,x]} e^{-\frac{\lambda}{\mu} \, (x-t)} \, \psi(t,\mu) \, dt & \text{ für fast alle } (x,\mu) \in [-a,a] \times (0,1]
						      \end{cases}.
\end{align*}
Wenn umgekehrt $\varphi$ durch die rechte Seite der obigen Gleichung definiert ist (mit einem $\psi \in X$) und $\Re \lambda > 0$, dann ist $\varphi$ $2$-integrierbar, $\varphi \in D(A_0)$ und $(\lambda - A_0)\varphi = \psi$. Also ist $\lambda - A_0$ surjektiv für $\Re \lambda > 0$ und mit dem Satz von Lumer, Phillips (Satz~\ref{thm: Lumer, Phillips}) folgt, dass $A_0$ eine Kontraktionshalbgruppe auf $X$ erzeugt. Wegen Satz~\ref{thm: störungssatz halbgruppenerz} erzeugt daher $A_0(c) = A_0 + c B$ eine stark stetige Halbgruppe und 
\begin{align*}
\norm{ e^{A_0(c)s} }  \le e^{c s} 
\end{align*}
für alle $s \in [0, \infty)$. Insbesondere ist die quasikontraktive Wachstumsschranke dieser Halbgruppe kleiner oder gleich $c$.

\begin{prop} \label{prop: quasikontr wschranke von A0(c)}
Die quasikontraktive Wachstumsschranke von $A_0(c)$ ist gleich $c$ für alle $c \in (0, \infty)$.
\end{prop}

\begin{proof}
Sei $c \in (0, \infty)$.
Wir haben eben schon gesehen, dass $\omega_{A_0(c)}' \le c$. Wir zeigen nun, dass $A_0(c) - \omega$ für Zahlen $\omega < c$ nicht dissipativ ist, woraus dann folgt, dass $\omega_{A_0(c)}' \ge c$.

Sei also $\omega < c$. Sei weiter $\varphi \in C_c^1((-a,a), \complex)$ mit $\varphi \ne 0$ und sei $\psi(x, \mu) := \varphi(x)$ für alle $(x, \mu) \in [-a,a] \times [-1, 1]$. Dann gilt $\psi \in D(A_0)$, $B \psi = \psi$ und
\begin{align*}
\scprd{ \psi, A_0 \psi } &= \int_{ [-a,a] \times [-1, 1] } \overline{ \psi(x, \mu) } \, \bigl( - \mu \frac{ \partial \psi }{ \partial x} (x,\mu) \bigr) \, d(x,\mu) \\
&= - \int_{[-a,a]} \overline{ \varphi(x)} \, \varphi'(x) \Bigl( \int_{[-1,1]} \mu \, d\mu \Bigr) \, dx = 0.
\end{align*}
Also erhalten wir
\begin{align*}
\scprd{ \psi, (A_0(c) - \omega) \psi } = \scprd{ \psi, A_0 \psi } + \scprd{ \psi, c B \psi - \omega \psi } = (c-\omega) \norm{\psi}^2 > 0,
\end{align*}
das heißt, $A_0(c) - \omega$ ist nicht dissipativ, was zu zeigen war.
\end{proof}

$A_0$ ist nicht normal, weil $D(A_0^*) \ne D(A_0)$. $D(A_0^*)$ ist nämlich die Menge (der Äquivalenzklassen) genau der $2$-integrierbaren Abbildungen $\varphi: [-a,a] \times [-1,1] \to \complex$, für die $(-a,a) \ni x \mapsto \varphi(x, \mu)$ schwach differenzierbar ist für fast alle $\mu \in [-1,1]$, für die $[-a,a] \times [-1,1] \ni (x,\mu) \mapsto \mu \, \frac{ \partial \varphi }{\partial x}(x,\mu)$ $2$-integrierbar ist und $\varphi(a, \mu) = 0$ für fast alle $\mu \in [-1,0)$ und $\varphi(-a, \mu) = 0$ für fast alle $\mu \in (0,1]$, und
\begin{align*}
\bigl( A_0^* \varphi) (x,\mu) = \mu \, \frac{ \partial \varphi }{\partial x}(x,\mu).
\end{align*}
Auch $A_0(c)$ ist damit nicht normal, 
die Adiabatensätze für nichtnormale (erst recht für nichtschiefselbstadjungierte) $A(t)$ sind also durchaus -- auch von den Anwendungen her -- wünschenswert (s. auch die Anwendungen in Abou Salems Arbeit~\cite{Abou 07}). 
\\

Was können wir über das Spektrum von $A_0$ sagen? Zunächst gilt $\sigma(A_0) \subset \{ z \in \complex: \Re z \le 0\}$, weil $A_0$ ja eine Kontraktionshalbgruppe erzeugt (Satz~\ref{thm: eigenschaften von erzeugern}). Weiter ist $1 \notin \im (\lambda-A_0)$ für $\Re \lambda < 0$, woraus sich ergibt, dass $\sigma(A_0) = \{ z \in \complex: \Re z \le 0\}$.
Schließlich gilt nach der Schlussfolgerung aus Lemma~\ref{lm: inhom lin dgl für schwache abl} $\sigma_p(A_0) = \emptyset$ und $\sigma_p(A_0^*) = \emptyset$, das heißt, $\sigma_r(A_0) = \emptyset$ und daher $\sigma(A_0) = \sigma_c(A_0)$.
 
Die Spektralstruktur von $\sigma(A_0(c))$ für $c \in (0, \infty)$ zu klären ist deutlich anspruchsvoller. Der nächste auf Lehner und Wing zurückgehende Satz tut dies.

\begin{thm} \label{thm: spektrum von A0(c)}
Sei $c \in (0, \infty)$. Dann gilt:
\begin{itemize}
\item [(i)] $\sigma_p(A_0(c))$ ist eine nichtleere endliche Untermenge von $(0, \infty)$:
\begin{align*}
\sigma_p(A_0(c)) = \{ \beta_1(c), \dots, \beta_{m_c}(c)\} \subset (0, \infty),
\end{align*} 
wobei $\beta_1(c), \dots, \beta_{m_c}(c)$ nach geometrischer Vielfachheit gezählt und absteigend geordnet sind. $\sigma_c(A_0(c)) = \{ z \in \complex: \Re z \le 0\}$ und $\sigma_r(A_0(c)) = \emptyset$.
\item [(ii)] $\beta_1(c), \dots, \beta_{m_c}(c)$ sind jeweils isolierte Spektralwerte von $A_0(c)$ der Ordnung $1$, 
insbesondere stimmen geometrische und algebraische Vielfachheit jedes Eigenwertes überein. 
\end{itemize}  
\end{thm}

\begin{proof}
(i) ist die Aussage des in Abschnitt~1 von~\cite{LehnerWing 55} festgehaltenen Satzes. Die Aussage (ii) folgt aus Lemma~2 in~\cite{LehnerWing 56} und Satz~\ref{thm: isol spektralwerte}.
\end{proof}

Zur Abhängigkeit der Eigenwerte von $A_0(c)$ vom Störungsparameter $c \in (0, \infty)$ ist zu sagen: die Anzahl $m_c$ der nach geometrischer Vielfachheit gezählten Eigenwerte wächst monoton und unbeschränkt mit $c$, die Definitionsbereiche $\dom \beta_n$ sind unbeschränkte offene Intervalle mit 
\begin{align*}
\dom \beta_{n+1} \subset \dom \beta_n \subset (0, \infty) = \dom \beta_1 
\end{align*}
für alle $n \in \natu$, und $c \mapsto \beta_n(c) \in (0, \infty)$ ist für jedes $n \in \natu$ eine streng monoton wachsende stetige Abbildung mit $\beta_n(c) \longrightarrow \infty \;\;(c \to \infty)$ und $\beta_n(c) \longrightarrow 0 \;\;(c \searrow \inf \dom \beta_n)$. Außerdem gilt $\dom \beta_n \subsetneq (0, \infty)$ für alle $n \ne 1$. 
Diese Aussagen gehen alle aus der Argumentation Lehners und Wings in~\cite{LehnerWing 55} hervor.
\\

Jetzt können wir zeigen, dass die Aussage des Adiabatensatzes gilt, wenn wir wie oben $A(t) := A_0(c(t)) - s(t)$ setzen für genügend reguläre Abbildungen $c, s: I \to (0, \infty)$ mit $\frac{c(t)}{s(t)} \le 1$ für alle $t \in I$ und weiter $\lambda(t)$ als den größten Eigenwert von $A(t)$ wählen. Übrigens ist nur dieser physikalisch bedeutsam (Abschnitt~5.1 in~\cite{Mokhtar-K 97}).

\begin{thm} \label{thm: anwendung des adsatzes}
Seien $c$ und $s$ zweimal stetig differenzierbare Abbildungen $I \to (0, \infty)$ mit $c(t) \le s(t)$ für alle $t \in I$. Sei $A(t) := A_0(c(t)) - s(t)$, $\lambda(t) := \beta_1(c(t)) - s(t)$ und $P(t)$ die Rieszprojektion von $A(t)$ auf $\{\lambda(t)\}$ für alle $t \in I$. Dann gilt
\begin{align*}
\sup_{t \in I} \norm{ U_{a,T}(t) - U_T(t) } = O\Bigl( \frac{1}{T} \Bigr) \quad (T \to \infty).
\end{align*}
\end{thm}

\begin{proof}
Wir zeigen zunächst, dass $\lambda(t)$ gleichmäßig isoliert ist in $\sigma(A(t))$. Der Schlüssel dazu liegt in Proposition~C-III.3.5 aus~\cite{Nagel 86} (vgl. Theorem~VI.1.12 in~\cite{EngelNagel}), die die Aussage von Satz~\ref{thm: spektrum von A0(c)} entscheidend ergänzt: sie besagt nämlich, dass unter der Voraussetzung, $A_0(c)$ erzeuge eine irreduzible positive Halbgruppe, die algebraische (und damit auch geometrische) Vielfachheit von $\beta_1(c)$ gleich $1$ ist für alle $c \in (0, \infty)$.

Zeigen wir nun, dass diese Voraussetzung der Positivität und Irreduzibilität erfüllt ist. Wir zeigen dazu, dass für alle $c \in (0, \infty)$ ein $\lambda \in (0, \infty)$ existiert, sodass
\begin{align*}
(\lambda - A_0(c))^{-1} \varphi  \ge 0 \text{ bzw.} > 0 \text{ fast überall}
\end{align*}
für alle $\varphi \in X$ mit $\varphi \ge 0 \text{ bzw.} > 0$ fast überall. 
Zunächst gilt
\begin{align} \label{eq: anwendung des adsatzes 1}
\bigl((\lambda-A_0)^{-1} \varphi \bigr)(x, \mu) = \begin{cases} -\frac{1}{\mu} \int_{[x,a]} e^{-\frac{\lambda}{\mu} \, (x-t)} \, \varphi(t,\mu) \, dt, & (x,\mu) \in [-a,a] \times [-1,0) \\
																													      \frac{1}{\mu} \int_{[-a,x]} e^{-\frac{\lambda}{\mu} \, (x-t)} \, \varphi(t,\mu) \, dt, & (x,\mu) \in [-a,a] \times (0,1]
																						      \end{cases}
\end{align}
für alle $\varphi \in X$ und alle $\lambda \in (0, \infty)$. Sei nun $c \in (0, \infty)$ und sei $\lambda := 2c$. Dann gilt weiter, dass
\begin{align*}
\norm{ c B (\lambda - A_0)^{-1} } \le c \, \frac{1}{\lambda} = \frac{1}{2}
\end{align*}
($A_0$ erzeugt ja eine Kontraktionshalbgruppe und $B$ ist eine orthogonale Projektion) und damit
\begin{align} \label{eq: anwendung des adsatzes 2}
(\lambda - A_0(c))^{-1} &= (\lambda - A_0)^{-1} \, \bigl( 1-cB(\lambda - A_0)^{-1} \bigr)^{-1} \notag \\
&= (\lambda - A_0)^{-1} + (\lambda - A_0)^{-1} \, \sum_{n=1}^{\infty} \bigl( cB(\lambda - A_0)^{-1} \bigr)^n.
\end{align}
Aus \eqref{eq: anwendung des adsatzes 1} und \eqref{eq: anwendung des adsatzes 2} und der Positivität von $B$ (im Verbandssinn) folgt nun, dass
\begin{align*}
(\lambda - A_0(c))^{-1} \varphi \ge 0 \text{ bzw.} > 0 \text{ fast überall}
\end{align*}
für alle $\varphi \in X$ mit $\varphi \ge 0 \text{ bzw.} > 0$ fast überall, und damit nach Theorem~VI.1.8 in~\cite{EngelNagel} und der  Definition~C-III.3.1 von Irreduzibiltät (zusammen mit den Ausführungen in Abschnitt~C-I.2 in~\cite{Nagel 86}), dass $A_0(c)$ tatsächlich für jedes $c \in (0, \infty)$ eine irreduzible positive Halbgruppe auf $X$ erzeugt.  

Jetzt können wir Proposition~C-III.3.5 anwenden, die besagt, dass $\beta_1(c)$ für alle $c \in (0, \infty)$ die algebraische Vielfachheit $1$ hat. Wegen der Stetigkeit von $c \mapsto A_0(c)$ im verallgemeinerten Sinn (Lemma~\ref{lm: A stetig im verallg sinn}) und der Stetigkeit von $c \mapsto \{\beta_1(c)\}$ (Ausführugen nach Satz~\ref{thm: spektrum von A0(c)}) folgt daraus mithilfe von Proposition~\ref{prop: zshg isoliert und glm isoliert}, dass $\beta_1(c)$ gleichmäßig in $c \in J$ isoliert ist in $\sigma(A_0(c))$ für jedes kompakte Intervall $J \subset (0,\infty)$, und daraus schließlich, dass $\lambda(t) = \beta_1(c(t)) -s(t)$ gleichmäßig (in $t \in I$) isoliert ist in $\sigma\bigl( A_0(c(t)) - s(t) \bigr) = \sigma(A(t))$, wie gewünscht. 
\\

Was haben wir nun von der gleichmäßigen Isoliertheit von $\lambda(t)$? Aus ihr ergibt sich, dass $t \mapsto P(t)\varphi$ zweimal stetig differenzierbar ist für alle $\varphi \in X$. 
Sei nämlich $t_0 \in I$ und $r_0$ eine positive Zahl, sodass 
\begin{align*}
\overline{U}_{r_0}(\lambda(t)) \setminus \{\lambda(t)\} \subset \rho(A(t))
\end{align*}
für alle $t \in I$. Dann exisitiert wegen der Stetigkeit von $t \mapsto \lambda(t)$ eine in $I$ offene Umgebung $U_{t_0}$ von $t_0$, sodass
\begin{align*}
\{\lambda(t)\} \subset U_{\frac{r_0}{3}}(\lambda(t_0)) \text{ \; und \; } \sigma(A(t)) \setminus \{\lambda(t)\} \subset \complex \setminus \overline{U}_{r_0}(\lambda(t)) \subset \complex \setminus \overline{U}_{\frac{2 r_0}{3}}(\lambda(t_0))
\end{align*}
für alle $t \in U_{t_0}$, das heißt,
\begin{align*}
P(t) \varphi = \frac{1}{2 \pi i} \, \int_{\gamma_{t_0}} (z-A(t))^{-1} \varphi \, dz
\end{align*}
für alle $t \in U_{t_0}$ und alle $\varphi \in X$, wobei $\gamma_{t_0} := \partial  U_{\frac{2 r_0}{3}}(\lambda(t_0))$. Aus der zweimaligen stetigen Differenzierbarkeit der Abbildungen $c$ und $s$ folgt deshalb mit Lemma~\ref{lm: vertauschung von abl und wegintegral}, dass $V_{t_0} \ni t \mapsto P(t) \varphi$ zweimal stetig differenzierbar ist (für jede in $I$ offene Umgebung $V_{t_0}$ von $t_0$ mit $\overline{V}_{t_0} \subset U_{t_0}$) und daher, weil $t_0$ beliebig war in $I$, dass $t \mapsto P(t) \varphi$ zweimal stetig differenzierbar ist für alle $\varphi \in X$. Wir sehen nun, dass alle Voraussetzungen von Satz~\ref{thm: unhandl adsatz mit sl} erfüllt sind, und dieser liefert die Behauptung. 
\end{proof}

Wir weisen darauf hin, dass die gleichmäßige Isoliertheit von $\lambda(t)$ unter der einschränkenderen Voraussetzung, dass $c(t) \notin \operatorname{dom}\beta_2$ für alle $t \in I$, offensichtlich ist, denn $\sigma_p(A_0(c(t))$ ist ja dann für alle $t \in I$ gleich $\{ \beta_1(c(t)) \}$, was wegen der Stetigkeit von $t \mapsto \beta_1(c(t)) \in (0, \infty)$ gleichmäßig isoliert ist in $\{ \Re z \le 0 \} \cup \{ \beta_1(c(t)) \} = \sigma\bigl( A_0(c(t)) \bigr)$.  

\end{document}